\documentclass[fontsize=11pt, headings=normal, headinclude=true, paper=a4, pagesize, cleardoublepage=current, BCOR=10mm, fleqn, numbers=noenddot]{scrartcl}
\pdfoutput=1
\usepackage{lmodern}

\setkomafont{sectioning}{\normalcolor\bfseries}
\recalctypearea 
\setcapindent{0pt}

\usepackage[english]{babel}  
\usepackage[numbers,square]{natbib}
\usepackage{mathrsfs}
\usepackage{amsmath}
\usepackage{amssymb}
\usepackage{dsfont}
\usepackage{upgreek}
\usepackage{amsthm}
\usepackage{enumerate}
\usepackage[utf8x]{inputenc}    
\usepackage[T1]{fontenc}         
\usepackage{graphicx}            
\usepackage[dvipsnames]{xcolor}		
\usepackage[english=british]{csquotes} 
\usepackage{braket}
\usepackage[colorlinks=false,hidelinks]{hyperref}

\newtheorem{thm}{Theorem}[section]
\newtheorem{prop}[thm]{Proposition}
\newtheorem{lem}[thm]{Lemma}
\newtheorem{cor}[thm]{Corollary}

\theoremstyle{definition}

\newtheorem{cond}[thm]{Condition}

\newtheorem{rem}[thm]{Remark}

\newtheorem*{thm*}{Theorem}
\newtheorem*{cor*}{Corollary}

\renewcommand{\phi}{\varphi}
\renewcommand{\Re}{\mathrm{Re}}

\newcommand{\eps}{\varepsilon}
\newcommand{\nn}{\nonumber}
\newcommand{\supp}{\mathrm{supp}}
\newcommand{\ao}{\mathfrak{a}}

\newcommand{\ud}{\mathrm{d}}
\newcommand{\ue}{\mathrm{e}}
\newcommand{\ui}{\mathrm{i}}
\newcommand{\R}{\mathbb{R}}
\newcommand{\C}{\mathbb{C}}
\newcommand{\N}{\mathbb{N}}
\newcommand{\Z}{\mathbb{Z}}

\newcommand{\cH}{\mathcal{H}}

\newcommand{\cN}{\mathcal{N}}
\newcommand{\cL}{\mathcal{L}}

\newcommand{\cB}{\mathcal{B}}
\newcommand{\cE}{\mathcal{E}}

\newcommand{\Ho}{\mathbb{H}_0}
\newcommand{\W}{\mathbb{W}}

\newcommand{\B}{\mathcal{B}}

\newcommand{\1}{\mathds{1}}
\newcommand{\dd}{\mathrm{d}}

\newcommand{\hc}{\mathrm{h.c.}}
\newcommand{\aV}{ \widetilde{\mathfrak{a}}_{V_N}}
\newcommand{\aW}{\widetilde{\mathfrak{a}}_{W_N}}

\newcommand{\T}{\mathbb{T}}

\newcommand{\Vphi}{\widetilde{\phi}_{\mathrm{B}}}
\newcommand{\Vphih}{\widehat{\widetilde{\phi}}_{\mathrm{B}}}
\newcommand{\Wphi}{\widetilde{\phi}_{\mathrm{I}}}
\newcommand{\Wphix}{\widetilde{\phi}_{\mathrm{I},x}}
\newcommand{\Wphih}{\widehat{\widetilde{\phi}}_{\mathrm{I}}}
\newcommand{\Wx}{W_{N,x}}

\newcommand{\mw}{w}

\usepackage{todonotes}

\numberwithin{equation}{section}

\addtokomafont{title}{}
\title{The excitation spectrum of a Bose gas with an impurity in the Gross--Pitaevskii regime}

\author{Jonas Lampart\thanks{CNRS \& Laboratoire Interdisciplinaire Carnot de Bourgogne (UMR 6303), Universit\'e de Bourgogne, 9 Av. A. Savary, 21078 Dijon Cedex, France. \texttt{jonas.lampart@u-bourgogne.fr}}, Arnaud Triay\thanks{Department of Mathematics, LMU Munich, Theresienstrasse 39, 80333 Munich, Germany. \texttt{triay@math.lmu.de}}}

\begin{document}

\maketitle

\begin{abstract}
 We study a dilute system of $N$ interacting bosons coupled to an impurity particle via a pair potential in the Gross--Pitaevskii regime. We derive an expansion of the ground state energy up to order one in the boson number, and show that the difference of excited eigenvalues to the ground state is given by the eigenvalues of the renormalized Bogoliubov--Fröhlich Hamiltonian in the limit $N\to \infty$.
\end{abstract}

\tableofcontents

\section{Introduction}

At very low temperatures a gas of bosons forms a Bose--Einstein condensate (BEC), a quantum state in which a macroscopic fraction of the bosons occupy the same single-particle state. The particles outside the condensate state, often called excitations or phonons, play a crucial role in the thermodynamic properties of the BEC. In 1947, Bogoliubov developed a theory to explain the emergence of superfluidity in dilute Bose gases using Landau's criterion of superfluidity~\cite{Landau-41}. In Bogoliubov's theory \cite{Bogoliubov-47}, the excitations are described by a quantum field governed by a Hamiltonian that is quadratic in the creation and annihilation operators of the excitations. This allowed him to compute explicitly the ground state energy and the low-lying excitation spectrum of the system. The rigorous justification of Bogoliubov's theory has now been studied in depth from a mathematical point of view, both from the spectral point of view \cite{Sei-11,GreSei-13,LewNamSerSol-15,YauYin-09,DerNap-14,BocBreCenSch-18,basti-21,NamTri-23,BreSchSch-22,FouSol-20,FouSol-22,BosPetSei-21,bossmann-22,brooks-23}, and for the dynamics \cite{RodSch-09,GrilMacMar-10,GriMacMar-11,LewNamSch-15,BenOliSch-15,CarOldSch-24,BocCenSch-17}. See also \cite{NamNap-19,CarCenSch-23,fournais-22} for two-dimensional systems and \cite{boccato-23,HabHaiNamSeiTri-23,BreLeeNam-24,HabHaiSchTri-24,fournais-24,CapDeu-23} for the positive temperature case.

 The interaction of the Bose gas with an impurity particle will create additional excitations out of the condensate, which in turn act on the impurity.
In the regime of large boson number, the Bogoliubov-Fr\"ohlich Hamiltonian has been proposed as an effective model (cf.~\cite{grusdt2016, MySe-20, LaPi-22}) for the excitation spectrum of the system. In this Hamiltonian, the interaction between the impurity and the excitation field is given by a coupling term that is linear in the field operator. This model has already been mathematically justified in the mean-field limit~\cite{MySe-20, LaPi-22}. However, Bose--Einstein condensates, as they are prepared in experiments, are better modeled by the dilute Gross--Pitaevskii limit, and the validity of the  Bogoliubov-Fr\"ohlich model has been questioned in the physics literature~\cite{christensen2015} (see Remark~\ref{rem:christensen}).
The goal of this article is to study the ground state energy and spectrum of the Bose gas with an impurity in this limit and prove validity of the approximation by the Bogoliubov-Fr\"ohlich Hamiltonian.

In the Gross--Pitaevskii regime, $N$ bosons are spread out over a volume of diameter $L=N$ in three dimensions. The density is therefore of order $N^{-2}$.
It is however more convenient to rescale the system in a volume of order one. This leads us to consider the Hamiltonian 
\begin{equation}\label{eq:H_N intro}
H_N=-\Delta_x +\sum_{i=1}^N (-\Delta_{y_i}) + \sum_{1\leq i < j \leq N} N^2 V(N(y_i-y_j)) + \sum_{i=1}^N N W(\sqrt N(x-y_i)).
\end{equation}
For simplicity, we will consider a system with periodic boundary conditions.  The position of the impurity will be denoted by $x\in \T^3$ and those of the $N$ bosons by $(y_i)_{1\leq i\leq N} \in \T^{3N}$, where $\T^3=\R^3/\Z^3$ is the three-dimensional unit torus.
We assume that the boson-boson potential $V$ and boson-impurity potential $W$ belong to $L^2(\R^3)$, have compact support and are non-negative. If $N$ is sufficiently large, the scaled potentials $V_N(y)=N^2V(Ny)$ and $W_N(x)=NW(\sqrt{N} x)$ have support inside the cube $[0,1]^3$, therefore we may identify the support with a subset of the torus and replace the potentials by their periodic extension.
The Hamiltonian $H_N$ thus acts on the Hilbert space $\mathscr{H}_N=L^2(\mathbb{T}^3)\otimes L^2(\mathbb{T}^3)^{\otimes_{\mathrm{sym}} N}$ with domain $H^2(\mathbb{T}^{3N+1})\cap \mathscr{H}_N$.
Heuristically, one should think that as $N\to \infty$, the boson-boson interaction behaves like  $N^{-1}\delta(y_i-y_j)$, and the boson-impurity interaction like $ N^{-1/2}\delta(x-y_i)$.
This singular behavior leads to considerable difficulties in the analysis, even at leading order of the energy. For the case without an impurity, such an analysis has been carried out in the works~\cite{LieYng-98, LieSei-02,LieSeiYng-00,LieSei-06,NamRouSei-16,BocBreCenSch-20,brooks-23}.

The scaling for $W$ is chosen so that the interaction between the impurity and the excitations above the condensate is of the same order as the kinetic energy of the Bogoliubov excitations, similarly to~\cite{MySe-20, LaPi-22}, see Remark \ref{rem:scaling_BF} below.

The starting point for our analysis will be to assume Bose-Einstein condensation of low-energy states, in the form of Condition~\ref{cond:BEC} below.
To state this, let $\ao_v$ denote the scattering length of a potential $0\leq v\in L^1(\R^3)$, defined by
\begin{equation}\label{eq:scatt_length-intro}
8 \pi \mathfrak{a}_v:=\inf_{\phi\in \dot{H}^1(\R^3)} \int_{\mathbb{R}^{3}} \Big( 2 |\nabla \phi(x)|^2 + v(x) |1+\phi(x)|^2\Big)\ud x.
\end{equation}
Note the scaling property $\mathfrak{a}_{\lambda^2 v (\lambda\cdot)} = \lambda^{-1} \mathfrak{a}_{v}$.

\begin{cond}[Bose--Einstein condensation]\label{cond:BEC}
There exist $c,d>0$ so that
 \begin{equation*}
H_N - 4\pi \ao_V N \geq  c \sum_{j=1}^N Q_{j} -d\sqrt{N},
 \end{equation*}
  where we  $Q_j $ denotes the projection  onto the orthogonal of the constant function on $\T^3$, acting on the $j$-th factor of $L^2(\mathbb{T}^3)^{\otimes_{\mathrm{sym}} N}$.
\end{cond}
As we are already assuming $W\geq 0$, the boson-impurity interaction can be dropped from $H_N$ for a lower bound, which makes this a condition on $V$  only.
It is known to be satisfied if $V$ is radial~\cite{BocBreCenSch-20, HaiSchTri-22}, where it holds with a rate of order one instead of $\sqrt{N}$.  For our use, the improved rate is of no importance since the impurity will change the energy at order $\sqrt{N}$ in any case.
Condition~\ref{cond:BEC} is, however, expected to hold more generally. A discussion of the non-radial case  can be found in~\cite{NamRicTri-23}.

\begin{thm}\label{thm:asymptotic}
Let $0\leq V,W\in L^2(\R^3)$ be even and compactly supported.
Assume moreover that Bose--Einstein condensation holds in the sense of Condition~\ref{cond:BEC}.
Then as $N\to \infty$
\begin{equation*}
 \inf \sigma (H_N) = 4\pi \ao_{V}N + 8\pi \ao_W\sqrt{N} -  32\pi  (2\pi/3 - \sqrt{3}) \ao_W^4\log N +\mathcal{O}(1).
\end{equation*}
\end{thm}

The first two terms in the expansion  account for the direct boson-boson and the boson-impurity interaction,  respectively. They are natural guesses, as two-body interactions are expected to be responsible for the leading order in the energy, due to the diluteness of the gas. The $\log N$-term, on the other hand, is explained by an effective three-body interaction between the impurity and two bosonic excitations (see Remark~\ref{rem:three-body}).
The emergence of such a term from the Bogoliubov-Fr\"ohlich Hamiltonian with cutoff $\Lambda =N$ was proved in~\cite{La-20}, after being observed in~\cite{grusdt2015, grusdt2016}.

\begin{rem}\label{rem:scaling_BF}
The scaling for $W$ is chosen so as to observe meaningful interactions of the impurity with the excitations. To explain this in more detail, recall that the Bogoliubov Hamiltonian represents the Hessian of the  Gross--Pitaevskii functional~\cite{LewNamSerSol-15}.
In other words, without impurity, the second variation of the Gross--Pitaevskii energy describes the asymptotic excitation spectrum of $H_N$. The Gross--Pitaevskii energy with impurity is obtained by formally replacing the potentials in $H_N$ by $8\pi \ao_V N^{-1} \delta(y_i-y_j)$,  $8\pi \ao_{W} N^{-1/2}\delta(x-y_i)$ and calculating the quadratic form of $H_N$ in an uncorrelated state $\psi\otimes u^{\otimes N}$. Doing so, we arrive at
\begin{align}
\mathcal E_{\textrm{GP}}(\psi,u)
	&=  \int |\nabla_x \psi| + N \int |\nabla_y u|^2 + 4 \pi \mathfrak{a}_V (N-1) \int |u(y)|^4  \notag\\
	&\quad + 8\pi \mathfrak{a}_W \sqrt{N} \int |u(x)|^2|\psi(x)|^2.
\end{align}
With no impurity, $\psi=0$, the constant function $u= \1$  minimizes this functional.
This function represents a condensate of $N$ bosons, so finitely many excitations out of the condensate (as expected in the Gross--Pitaevskii regime \cite{BocBreCenSch-20}) correspond to a variation of $\int |u|^2$ of order $N^{-1}$, i.e., $u = \1 + N^{-1/2}\delta u$.
Inserting this into the functional, we obtain
\begin{align}
&\mathcal E_{\textrm{GP}}(\psi,\1+ N^{-1/2} \delta u)  - \mathcal E_{\textrm{GP}}(0,\1) - 8\pi \ao_W \sqrt{N}\notag \\
&	\approx  \int |\nabla_x \psi |^2 + \int \left( |\nabla_y \delta u |^2 + 8 \pi \mathfrak{a}_V |\delta u (y)|^2\right) \notag\\
	&\quad +  4 \pi \mathfrak{a}_V \int (\delta u(y)^2 +   \overline{\delta u(y)}^2)   + 8\pi \mathfrak{a}_W \int (\delta u(y) + \overline{\delta u(y)})|\psi(x)|^2,
\end{align}
up to higher order terms. Here, the interaction between the impurity and the excitations represented by $\delta u$ is of the same order as the interaction of the excitations among themselves due to our choice of interaction strength, which justifies the scaling for $W$.
\end{rem}

Our second main result is that the low-lying spectrum of $H_N$ is well approximated by the one of the renormalized Bogoliubov-Fr\"ohlich Hamiltonian. This operator acts on
\begin{equation}
 \mathscr{H}_+ = L^2(\T^3) \otimes \Gamma(\mathfrak{H}_+), 
\end{equation}
where $\Gamma(\mathfrak H_+)$ is the bosonic Fock space over $\mathfrak{H}_+ = \{\1\}^\perp$, the orthogonal to the (constant) condensate wave-function $\1_{\T^3}$. Formally, the Bogoliubov-Fr\"ohlich Hamiltonian is given by
\begin{equation}\label{eq:H_BF formal}
  -\Delta_x + \ud \Gamma(\epsilon) + a(\mw_x) + a^*(\mw_x),
\end{equation}
where $\epsilon$ is the Fourier multiplier with Bogoliubov's dispersion relation
\begin{equation}\label{eq:omega-def}
 \epsilon(p)=\sqrt{p^4+ 16\pi \ao_{V} p^2},
\end{equation}
 $a^*, a$ are the bosonic creation and annihilation operators (see~\eqref{eq:a*-a}) and $\mw_x(y)=w(x-y)$, where $w$ is the distribution with Fourier coefficients
\begin{equation}\label{eq:w-def}
 \widehat{  \mw}(p) = 8\pi \ao_W |p| \epsilon(p)^{-1/2}, \; p\in 2\pi\Z^3\setminus\{0\}.
\end{equation}
However, the expression~\eqref{eq:H_BF formal} is ill-defined due to the lack of decay of $\widehat w(p)\sim 8\pi \ao_W$ for $p\to \infty$. For a rigorous definition, it should be replaced by the renormalized Bogoliubov-Fr\"ohlich Hamiltonian on the torus, constructed as in~\cite{La-20}.

\begin{prop}\label{prop:renorm}
Define $\mw^\Lambda$ by its Fourier coefficients $\widehat \mw^\Lambda(k)= \widehat \mw(k) \1_{|k|\leq \Lambda}$ for $\Lambda \in \R_+$, and set
\begin{equation}\label{eq:H_BF cutoff}
 H_\mathrm{BF}^\Lambda= -\Delta_x + \ud \Gamma(\epsilon) +  a(\mw^\Lambda_x) + a^*(\mw^\Lambda_x),
\end{equation}
which defines a self-adjoint operator on $D(-\Delta_x + \ud \Gamma(\epsilon))$.
There exists a family $(E_\Lambda)_{\Lambda\geq 0}\subset \R$ and a self-adjoint operator $(H_\mathrm{BF}$, $D(H_\mathrm{BF}))$ with compact resolvent so that
\begin{equation*}
 H_\mathrm{BF}^\Lambda -E_\Lambda \stackrel{\Lambda\to \infty}{\rightarrow} H_\mathrm{BF}
\end{equation*}
in norm resolvent sense.
\end{prop}

We prove this proposition in Section~\ref{sect:renorm}, where we also give more details on $H_\mathrm{BF}$. With the renormalized Bogoliubov-Fr\"ohlich Hamiltonian, we can resolve the excitation spectrum of $H_N$ for $N\to \infty$.

\begin{thm}\label{thm:spect}
Let $0\leq V,W\in L^2(\R^3)$ be even and compactly supported, and assume that Bose--Einstein condensation in the sense of Condition~\ref{cond:BEC} holds.
Let $(H_\mathrm{BF}$, $D(H_\mathrm{BF}))$ be the renormalized Bogoliubov-Fr\"ohlich Hamiltonian of Proposition~\ref{prop:renorm}, and for $k\in \N_0$ let $e_k(H)$ denote the $k$-th min-max value of the operator $H$.
Then as $N\to \infty$
 \begin{equation*}
  e_k(H_N)-e_0(H_N)= e_k(H_\mathrm{BF})-e_0(H_\mathrm{BF}) + \mathcal{O}(N^{-1/80}).
 \end{equation*}
 \end{thm}

Note that in Theorem \ref{thm:asymptotic} we determine the ground state energy $e_0(H_N) =  \inf \sigma (H_N)$ with an error $\mathcal O(1)$ but in Theorem \ref{thm:spect} we can express the elementary excitations  $ e_k(H_N)-e_0(H_N)$ up to a precision $o(1)$.
From the proof of Theorem~\ref{thm:spect} one can extract an expression for $e_0(H_N)$ with this precision, but it is quite involved and not completely explicit, since the value of $e_0(H_\mathrm{BF})$ is not known.

The proof of Theorem~\ref{thm:spect} also shows that the eigenfunctions of $H_N$ can be expressed in terms of those of $H_\mathrm{BF}$ and two unitary transformations: the excitation map $U_X:\mathscr{H}_+\to \mathscr{H}_N $ (see Section~\ref{sect:excite}) and a Bogoliubov transformation $U_\mathrm{B}$ (see Section~\ref{sect:Bog}).
This information is also sufficient to obtain an approximation of the unitary dynamics.

\begin{cor}\label{cor:operators}
Assume that the hypothesis of Theorem~\ref{thm:spect} hold. Let $U_X$ be the excitation map defined by~\eqref{eq:excitation-map} and $U_\mathrm{B}$ the Bogoliubov transformation whose generator is given by~\eqref{eq:Bog-generator} with $N=\infty$. Then:
\begin{enumerate}[a)]
 \item For any $k\in \N$, the corresponding min-max values $e_k$
  and  $\Psi_N \in \ker(H_N-e_k(H_N))$,  the sequence $U_\mathrm{B}^*U_X^* \Psi_N$ converges, up to a subsequence, to a limit $\Psi \in\ker(H_\mathrm{BF}-e_k(H_\mathrm{BF}))$.
 \item The corresponding spectral projections converge in norm: Let $B$ be an open interval with  $\sigma(H_\mathrm{BF}-e_0(H_\mathrm{BF})) \cap B=e_k(H_\mathrm{BF})-e_0(H_\mathrm{BF})$, then
  \begin{equation*}
 \lim_{N\to \infty} U_X^* \1_{B}\big(H_N=e_k(H_N)\big) U_X=U_\mathrm{B} \1(H_\mathrm{BF}=e_k\big(H_\mathrm{BF})\big)U_\mathrm{B}^*.
 \end{equation*}
\item For any sequence $\Psi_N\in \mathscr{H}_N$ so that $U_X^*\Psi_N$ has a limit $\Psi\in \mathscr{H}_+$,
\begin{equation*}
 \lim_{N\to \infty} U_X^*\ue^{-\ui t(H_N-e_0(H_N))}\Psi_N = U_\mathrm{B} \ue^{-\ui t (H_\mathrm{BF}-e_0(H_\mathrm{BF}))} U_\mathrm{B}^*\Psi
\end{equation*}
in the norm of $\mathscr{H}_+$.
\end{enumerate}

\end{cor}

\begin{rem}\label{rem:christensen}
 There has been some debate in the physics literature concerning the applicability of the Bogoliubov-Fr\"ohlich model.
In~\cite{christensen2015}, the authors performed an expansion of $e_0(H_N)$ in $\ao_W$ and found an expression of order $\ao_W^3 \log \ao_W$. It is argued that this cannot arise from the Bogoliubov-Fr\"ohlich Hamiltonian, whose (formal) perturbation series is even in $\ao_W$.
In our scaling, the corresponding term in the perturbation series is of order $N^{-1/2} \log N$. It is related to the creation of pairs of Bogoliubov phonons by the impurity, which is suppressed by the fact that the boson-boson interaction is weaker than the boson-impurity interaction  (see Section~\ref{sect:Bog}).
Hence, its absence from our results is not in contradiction to~\cite{christensen2015}, where no relation between the two interactions is assumed and the expansion is in $\ao_W$.

Our results show that, in the precise scaling we consider, the spectrum is correctly described by the Bogoliubov-Fr\"ohlich Hamiltonian.
\end{rem}

\begin{rem}

The first application of Bogoliubov's theory was to provide an explanation of superfluidity. Landau argued~\cite{Landau-41}, based on conservation of energy and momentum, that an obstacle moving at velocity below the speed of sound relative to a quantum liquid could not excite phonons and would thus not create friction. This could be made more precise using the Bogoliubov-Fr\"ohlich model, where the impurity serves to probe the superfluid behavior. A step in this direction was recently made~\cite{HinLam2023}, where the Bogoliubov-Fr\"ohlich model in $\R^3$ is analyzed, and it is shown that a moving impurity does not experience friction when moving at velocities smaller than the speed of sound, which equals $4\sqrt{\pi \ao_V}$ in our derivation. Connecting our derivation of the Bogoliubov-Fr\"ohlich Hamiltonian on the torus to the one on $\R^3$ (e.g., by taking the size of the torus to infinity, as is done in the physics literature~\cite{grusdt2016}, see also~\cite{DerNap-14}) would thus provide a detailed explanation of this effect from a many-body quantum model.
This has also been investigated in a semi-classical setting in~\cite{Eglietal.2013, FrohlichGang.2014b,DeFrPiPi-16, Leger.2020}.
\end{rem}

\begin{rem}\label{rem:masses}
 The prefactor of $\log N$ in Theorem~\ref{thm:asymptotic} is due to our choice of the masses of both the bosons and the impurity being equal to one-half.
For a general mass of the particle $m$ and reduced mass $\mu^{-1}=m^{-1}+2$, the prefactor is
\begin{equation}\label{eq:c-log}
 c_{m}=-16 \pi \mu^{-1} \ao_{\mu W}^4 \Big(\tfrac{m}{\mu} \arcsin\big(\tfrac{\mu}{m}\big)- \sqrt{1-\big(\tfrac{\mu}{m}\big)^2}\Big),
\end{equation}
as can be checked by keeping track of the constants in the calculations of Section~\ref{sect:Thm1} (see also~\cite{LaTr-24}).
\end{rem}

\begin{rem}\label{rem:scaling}
The scaling of our Hamiltonian can be expressed in the natural units of a dilute system (see also~\cite{LaTr-24}). Consider a system of bosons with density $\rho$ interacting among themselves through a potential $V$ and with one impurity via $W$, in a box of side-length  $\ell_{\textrm{GP}} = (\rho \ao)^{-1/2}$ (the Gross--Pitaveskii length), where $\ao:= \ao_V$ is the scattering length of $V$ and where we also denote by $\ao_W$ the one of $W$.
 The Hamiltonian for the total system is
 \begin{equation}
  H_{\textrm{GP}} = -\Delta_{x}+\sum_{i=1}^{N}(-\Delta_{y_i})+ \sum_{1\leq i< j \leq N}V(y_i-y_j) +
  \sum_{1\leq j \leq N}W(x-y_j),
 \end{equation}
We consider the dilute regime, in which the dimensionless parameter $\rho \mathfrak a^3$ is small and the polaronic coupling coupling constant (compare~\cite{tempere2009,grusdt2016, LaTr-24}, where definitions differ by factors of $\sqrt{8\pi}$)
\begin{equation}
 \alpha=\ao_W^2 \ell_{\textrm{GP}} \rho=\frac{\ao_W^2}{\ao \ell_{\textrm{GP}}}
\end{equation}
is of order one. The number of bosons is given by $N= \rho \ell_{\textrm{GP}}^3=(\rho \mathfrak{a}^3)^{-1/2} = \ell_{\textrm{GP}}/\ao \gg 1$ and the number of excited bosons is of order one \cite{BocBreCenSch-20}.
The coupling to the impurity satisfies the relation $\ell_{\textrm{GP}}/\ao_W = \sqrt{N/\alpha}$, so its range is still much shorter than $\ell_{\textrm{GP}}$, though larger than $\ao$.

In order to normalize the interactions, we may write $V(y)=\mathfrak{a}^{-2} \widetilde V (\mathfrak{a}^{-1} y)$ and $W(x) =\ao_W^{-2} \widetilde W (\ao_W^{-1} x)$ with potentials $\widetilde V$ and $\widetilde W$ that have scattering length equal to one. Then, rescaling by $\ell_{\textrm{GP}}$ and multiplying with $\ell_\mathrm{GP}^2$, we have $\ell_{\textrm{GP}}^2 V(\ell_{\textrm{GP}}x)=N^2 \widetilde V(Nx)$, and $\ell_{\textrm{GP}}^2 W(\ell_{\textrm{GP}}x)=N/\alpha \widetilde W(\sqrt{N/\alpha} x)$ which leads to a Hamiltonian of the form~\eqref{eq:H_N intro}.
The expansion of Theorem~\ref{thm:asymptotic} can then be written as
\begin{align*}
&\frac{\inf \sigma( H_{\textrm{GP}})}{\ell_{\textrm{GP}}^3} \\
&\quad =  4\pi \mathfrak{a} \rho^2 \Big( 1 + 2 \sqrt{\alpha} (\rho \mathfrak{a}^3)^{1/4} - 8 (2\pi/3-\sqrt{3}) \alpha^2 \sqrt{\rho \mathfrak{a}^3} \log (\rho \mathfrak{a}^3) +\mathcal{O}\big(\sqrt{\rho \mathfrak{a}^3}\big)\Big).
\end{align*}
In the case of different masses, the correct formula is obtained from~\eqref{eq:c-log} (see also~\cite{LaTr-24}).
 
For comparison, note that the $\rho^{5/2} \log \rho$ term in this expansion is larger than the well-known Lee-Huang-Yang correction~\cite{LeeHuangYang-57, FouSol-20, basti-21,FouSol-22, fournais-22,HabHaiNamSeiTri-23}, which is of the order $\rho^{5/2}$. It is however very similar in nature to the term proposed in~\cite{wu1959,sawada-59, HuPi-59} to capture the three-body effects in a Bose gas. The latter is expected to be of order $\rho^{3} \log \rho$ and was very recently derived in the Gross--Pitaevskii limit \cite{CarOlgAubSch-23}.
\end{rem}

\begin{rem}\label{rem:three-body}
 The three-body nature of the $\log N$--contribution to the energy can be understood by the following heuristics. Since the Bose gas is dilute, the effect of the boson-impurity interaction is dominated by simple two-body interactions. These lead to correlations between the impurity and a single boson since, due to the repulsive nature of $W_N$, these will try to avoid getting closer than distances of order $N^{-1/2}$. This is captured precisely by the solution $\phi$ to the minimization problem~\eqref{eq:scatt_length-intro} with $v=W$, and it is responsible for the scattering length $\ao_W$ appearing in the coefficient of $\sqrt N$, instead of just the averaged potential $\int W$ as in the mean-field case.
 If we consider two bosons interacting with the impurity and make the ansatz $\Psi \approx\phi(\sqrt{N}(x-y_1))\phi(\sqrt N( x-y_2)) \Phi(x,y_1,y_2)$ at short distances, then $H_2 \Psi$ will contain a term with
 \begin{equation}
v^{(3)}(x,y_1,y_2) :=  N\nabla\phi\big(\sqrt N(x-y_1)\big )\nabla\phi\big(\sqrt N (x-y_2)\big),
 \end{equation}
coming from the action of $-\Delta_x$. This can be interpreted as an effective three-body interaction potential acting on $\Phi$. Minimizing over $\Phi$ will make the scattering length of this three-body potential appear (cf.~\cite{NamRicTri-22b} and Equation~\eqref{eq:E_NW}, noting that this potential integrates to zero), which multiplied by the number of pairs of bosons $\sim N^2$ will give the term of order $\log N$ in Theorem~\ref{thm:asymptotic}.
That this is of order $\log N$ can be easily checked. Note that for $1 \lesssim |y| $, $\phi(y)$ behaves like $1/ |y |$ and $\nabla\phi(y) \sim 1/  |y |^2$ (cf.~\cite[Theorem 6]{NamRicTri-23}). Then, passing to relative coordinates, the three-body scattering length, defined by a minimization problem analogous to~\eqref{eq:scatt_length-intro}, is approximately given by
\begin{align*}
 & \int v^{(3)} - \langle v^{(3)}, (-(\nabla_{y_1}+\nabla_{y_2})^2-\Delta_{y_1} - \Delta_{y_2} + v^{(3)})^{-1}v^{(3)}\rangle  \notag\\
 &\qquad\qquad\sim -N^2 \hspace{-12pt} \int\limits_{N^{-1/2}\lesssim |y_1|, |y_2|}\hspace{-12pt} \nabla\phi(\sqrt N y_1 )\nabla\phi(\sqrt N y_2) \notag\\
 &\qquad\qquad\qquad \qquad \times (-(\nabla_{y_1}+\nabla_{y_2})^2-\Delta_{y_1} - \Delta_{y_2} )^{-1} \nabla\phi(\sqrt N y_1)\nabla\phi(\sqrt N y_2) \notag \\
 &\qquad\qquad\sim -  \frac{1}{N^2} \int\limits_{N^{-1/2}\lesssim |y_1|, |y_2|} \frac{1}{|y_1|^2} \frac{1}{|y_2|^2} (-\Delta_{y_1} - \Delta_{y_2})^{-1}  \frac{1}{|y_1|} \frac{1}{|y_2|} \sim \frac{\log N}{N^2}.
\end{align*}
The proof of Theorem~\ref{thm:asymptotic} contains a rigorous implementation of these heuristic estimates where we do not use such an ansatz of product form, but rather a Weyl transformation which leads to similar expressions.
\end{rem}

\section{Preliminaries}\label{sect:prelim}

\subsection{Outline of the article}

The following outline of the article may serve as a guide to the proofs of our main results.
In the remainder of Section~\ref{sect:prelim} we introduce some global conventions on notation, and summarize the key properties of the scattering length and the associated scattering solution, as well as its analogues on the torus. These results will be proved in Appendix~\ref{app:scatt}.

In Section~\ref{sect:renorm}, we discuss the renormalization of the Bogoliubov-Fr\"ohlich Hamiltonian on the torus and prove Proposition~\ref{prop:renorm}. The proofs of the necessary technical lemmas are postponed to Section~\ref{sect:conv}, where a more general case with an additional excitation-impurity interaction is treated.

In Section~\ref{sect:trafo}, we write the Hamiltonian $H_N$ in the excitation picture. We then apply several unitary transformations to extract the correct leading asymptotics of the energy, and, at the same time, to make the correct model parameters appear. The general scheme is to deal first with the large momenta, which are responsible for the bigger terms in the energy asymptotics, and then work downwards step-by-step.
Specifically:
\begin{enumerate}[1)]
 \item In Section~\ref{sect:quadratic} we apply a Bogoliubov transformation $U_q$, whose generator is quadratic in boson creation and annihilation operators. It acts on the large boson momenta and makes the scattering length $\ao_V$ appear at the leading order of the energy.
 \item In Section~\ref{sect:Weyl} we apply a Weyl-transformation $U_W$, which additionally depends on the impurity position $x$ (similarly to the Gross transformation known from the literature on singular polaron models, like the Nelson model, cf.~\cite{Nelson-64,GrWu-18}). This transformation also acts on large boson momenta and makes the scattering length $\ao_W$ appear in the energy asymptotics.
 \item In Section~\ref{sect:cubic} we apply a transformation $U_c$ to the boson Hamiltonian that has a generator which is cubic in creation and annihilation operators. This accounts for scattering processes that mix low and high energy bosons and makes the scattering length  $\ao_V$ appear in the boson Hamiltonian restricted to small energies.
 \item In Section~\ref{sect:Bog}, we use Bogoliubov's transformation $U_\mathrm{B}$, which transforms the remaining boson Hamiltonian into the simple expression $\ud \Gamma(\epsilon)$, with $\epsilon$ given by~\eqref{eq:omega-def}.
 \item In Section~\ref{sect:Gross} we apply another Weyl transformation $U_\mathrm{G}$, which is precisely a variant of the Gross transform, which accounts for the two-particle impurity-excitation scattering at low energies.
\end{enumerate}
The final result of Section~\ref{sect:trafo} is Proposition~\ref{prop:U_G}, after which we can write the transformed Hamiltonian on low energy states as (cf.~\eqref{eq:H_NU-error})
\begin{align*}
U^* H_N U \approx 4\pi  \ao_V (N-1)+ 8\pi \ao_W \sqrt{N} +  e^{(U)}_N +  \cH_N^U 
\end{align*}
where $\cH_N^U$ is almost the Bogoliubov--Fröhlich Hamiltonian with an $N$-dependent cutoff $\Lambda_N$ (up to a Gross transform), and $ e^{(U)}_N$ is an energy contribution of order one.

In Section~\ref{sect:conv}, we analyze $\cH_N^U$ and show that after removing the logarithmic divergence, it converges to the renormalized Bogoliubov--Fröhlich Hamiltonian constructed in Proposition~\ref{prop:renorm}. We essentially follow the proof of that proposition, but with one key difference. The excitation Hamiltonian, and also $\cH_N^U$, contains an additional term accounting for the interaction between two excitations and the impurity which 
is not present in the Bogoliubov--Fröhlich Hamiltonian. In our scaling regime, this term is responsible for a global shift of order one in the energy. However, the excitation spectrum is still correctly described by the Bogoliubov--Fröhlich Hamiltonian as stated in Theorem~\ref{thm:spect}, since it only depends on the difference of eigenvalues.

In Section~\ref{sect:proof}, we prove the main Theorems and Corollary~\ref{cor:operators}.

\subsection{Notation}

We generally denote by $\widehat{f}\in \ell^2(2\pi \Z^3)$ the Fourier series of a function or distribution on $\T^3$, with the normalization convention and inversion formula given by
\begin{equation}
 \widehat{f}(k)= \int_{\T^3} \ue^{-\ui kx} f(x) \ud x,\qquad f(x)=\sum_{2\pi \Z^3}\ue^{\ui kx}\widehat{f}(k).
\end{equation}

For a symmetric operator $A$ on a Hilbert space $\mathscr{H}$, we will denote by $D(A)$ its domain, and by $Q(A)$ the domain of the associated quadratic form $x \mapsto \braket{x,Ax}$.

For a function $f,g:\N \to \mathcal{C}$ where $\mathcal{C}$ is a set with a partial order and admitting multiplication by positive real numbers (i.e., $\R_+$ or symmetric quadratic forms on $\mathscr{H}$), we denote
\begin{equation}
 f\lesssim g \Leftrightarrow \exists C\in \R_+: \forall N\in \N : f(N)\leq C g(N).
\end{equation}

For families of symmetric operators  $A_N$, $B_N$, $N\in \N$ associated with quadratic forms on $Q(A)$, $Q(B)$, independent of $N$, we denote by
\begin{equation}
 A_N=\mathcal{O}(B_N)
\end{equation}
the statement that $Q(B)\subset Q(A)$ and that there exists a constant $C$ so that for all $\Psi\in Q(B)$ and $N\in \N$
\begin{equation}
 |\langle \Psi, A_N \Psi\rangle| \leq C \langle \Psi, B_N\Psi\rangle.
\end{equation}
Note that this relation is transitive and invariant under unitary transformations.
 Moreover, we have as a consequence of the Cauchy-Schwarz inequality that
\begin{equation}\label{eq:AB Young}
 A^*_NB_N+B^*_NA_N=\mathcal{O}( A^*_NA_N + B^*_NB_N),
\end{equation}
when the left hand side is extended to the domain of the right.

For any function $h: \T^3 \to \mathbb{C}$, we will use the notation
\begin{align}
h_x(y) = h(x-y).
\end{align}

For a Hilbert $\mathfrak{H}$ denote by
\begin{equation}
 \Gamma(\mathfrak{H}):=\C\otimes \bigoplus_{n=1}^\infty \mathfrak{H}^{\otimes_\mathrm{sym} n}
\end{equation}
the bosonic Fock space over $\mathfrak{H}$. In particular, $\Gamma(\mathfrak{H}_+)$ with $\mathfrak{H}_+=\{\1\}^{\perp}\subset L^2(\T^3)$ (the orthogonal to the constant functions), is the Fock space of excitations.
For $\Psi \in \Gamma(\mathfrak{H})$ we denote by $\Psi^{(n)}\in \mathfrak{H}^{\otimes_\mathrm{sym} n}$ its $n$-particle component,
and by
\begin{equation}
 \cN \Psi = \sum_{n=1}^\infty n \Psi^{(n)}
\end{equation}
the number operator (with the canonical domain). In the case of $\Gamma(\mathfrak{H}_+)$ we denote the number operator by $\cN_+$.

 For $f\in \mathfrak{H}$, we denote by $a^*(f)$, $a(f)$ the bosonic creation and annihilation operators on $\Gamma( \mathfrak{H})$, who satisfy the commutation relations
 \begin{equation}\label{eq:CCR}
  [a(f),a^*(g)]=\langle f, g\rangle,\qquad [a(f),a(g)]=0=[a^*(f),a^*(g)].
 \end{equation}
 In the case $\mathfrak{H}=L^2(\T^3)$ we introduce the operators
 \begin{equation}
 a_p = a(\ue^{\ui p (\cdot) }), \qquad a_p^*=a^*(\ue^{\ui p (\cdot) }),
\end{equation}
and the operator-valued distributions $a_y$, $a^*_y$, $y\in \T^3$ as usual.
For a distribution $v$ on $\T^3$ and a quadratic form $A$ on $L^2(\T^3)$ with Schwartz kernel $A(y,y')$, we also introduce the quadratic forms
\begin{equation}\label{eq:a*-a}
 a(v)=\int v(y) a_y, \quad a^*(v)=\int v(y) a_y^*, \quad \ud \Gamma(A)= \int A(y,y')a_y^*a_{y'}
\end{equation}
on appropriate subspaces of $\Gamma(L^2(\T^3))$ (they are well defined at least on finite combinations of elements of $C^\infty(\T^{3n})$).

Any quadratic form on $\Gamma(L^2(\T^3))$ defines a quadratic form on $\Gamma(\mathfrak{H}_+)$ by restriction, which we denote by the same symbol. Note that this is not true for operators, where we would need to project the image back to $\Gamma(\mathfrak{H}_+)$, which is implicit for quadratic  forms.

The excitation space with an impurity is
\begin{equation}
 \mathscr{H}_+=L^2(\T^3) \otimes \Gamma(\mathfrak{H}_+) = L^2(\T^3,\Gamma(\mathfrak{H}_+)).
\end{equation}
We use the notations $\cN_+$, $\ud \Gamma (\epsilon)$, etc. for operators acting trivially on the first factor (i.e., we omit the tensor product with the identity).
In this context,  the variable $x$ is usually reserved for the impurity, while the variables $y_i$ describe the bosons.
To a function $x\mapsto A(x)$ on $\T^3$ whose values are quadratic forms with $Q(A_x)=D\subset \Gamma(\mathfrak{H}_+)$, we associate the quadratic form on $\mathscr{H}_+$
\begin{equation}
\begin{aligned}
\langle \Psi, A \Psi\rangle &= \int \langle \Psi(x), A(x) \Psi(x)\rangle \ud x,\\
Q(A)&=\Big\{\Psi:\T^3\to D\Big\vert \int |\langle \Psi(x), A(x) \Psi(x)\rangle| < \infty \Big\}
\end{aligned}
\end{equation}
by viewing $\mathscr{H}_+$ as $\Gamma(\mathfrak{H}_+)$-valued functions on $\T^3$.

Throughout the article we will explicitly write variables of integration or summation, their domains and the integration measures only when they are not obvious from the context.

\subsection{The approximate scattering problems}\label{sect:scattering}

By the scaling of the potentials, the interaction becomes increasingly short-ranged as $N\to \infty$, and thus the effective ``interaction time'' between two particles decreases to zero. For this reason, concepts from scattering theory of Schrödinger operators become relevant to the discussion, although we are working with periodic boundary conditions and the spectra of all operators are discrete.

The scattering length $\ao_v$ of a potential $0 \leq v \in L^1(\mathbb{R}^{3})$  is the first order in the low-energy expansion of the corresponding scattering amplitude~\cite{albeverio-1982}. In the expansion of the operator, it is accompanied by the generalized eigenfunction, also called scattering solution, at energy zero. This function $f:\R^3\to \C$ is asymptotic to one (i.e., $\ue^{-\ui k x}$ with $k=0$) and we may write it as $f=1+\phi_{\R^3}$, where $ \phi_{\R^3}\in \dot{H}^1(\R^3)$ solves
\begin{align} \label{eq:phi_full_space}
-\Delta \varphi_{\mathbb{R}^{3}} + \tfrac{1}{2} v \varphi_{\mathbb{R}^{3}} = - \tfrac{1}{2}v.
\end{align}
This is also the Euler-Lagrange equation of  the variational problem~\eqref{eq:scatt_length-intro}, and $\phi_{\R^3}$ is the unique minimizer.
The scattering length can thus be expressed as
\begin{align} \label{eq:rel_scat_R3}
8 \pi \ao_v = \int v (1+\varphi_{\mathbb{R}^{3}}).
\end{align}
As the support of the potentials $V_N$, $W_N$ shrinks to a point, the boundary becomes less relevant and approximations of $\phi_\R^3$ naturally appear in the analysis.
Following \cite{HaiSchTri-22}, we will use solutions to the equation analogous to~\eqref{eq:phi_full_space} on the torus, instead of truncations of $\phi_{\R^3}$, and call these (approximate) scattering solutions. Their properties are discussed in detail in Appendix~\ref{app:scatt}. Here, we collect the information of practical relevance for our analysis. The scattering solution $\phi_{\mathrm{B}}$ associated to the boson-boson interaction potential $V_N$ on the torus is defined to be the unique solution in $\ell^2(2\pi\mathbb{Z}^{3}\setminus\{0\})$ of
\begin{equation}
	\label{eq:scattV}
  p^2 \widehat\phi_{\mathrm{B}}(p) + \frac{1}{2} \sum_{q \in 2\pi\Z^3 \setminus\{0\}} \widehat\phi_{\mathrm{B}}(q) \widehat V_N(p-q) = -\frac{1}{2}\widehat V_N(p).
\end{equation}
For $\alpha>0$, we define its truncated version by the Fourier coefficients
\begin{equation}\label{eq:Vphi}
 \widehat{\widetilde{\phi}}_{\mathrm{B}}(p):=\widehat{\phi}_{\mathrm{B}}(p) \1_{|p| >  N^{\alpha}}.
\end{equation}
In analogy with~\eqref{eq:rel_scat_R3}, the torus ``scattering length'' is then defined by
\begin{equation}\label{eq:aV_N}
 8\pi \aV:=\int_{\T^3} N V_N(y)(1+\phi_{\mathrm{B}}(y))=  \widehat V(0) + \int N V_N \phi_{\mathrm{B}}.
\end{equation}

Similarly, for the interaction potential between the bosons and the particle $W_N$, we define  the scattering solution $\phi_\mathrm{I}$ to be the unique solution in $\ell^2(2\pi\mathbb{Z}^{3}\setminus\{0\})$ of
\begin{equation}
	\label{eq:scattW}
  p^2 \Wphi(p) + \frac{1}{2} \sum_{q \in 2\pi\Z^3 \setminus\{0\}} \Wphi(q) \widehat W_N(p-q) = -\frac{1}{2}\widehat W_N(p).
\end{equation}
Its truncated version is denoted, for $\alpha>0$, by
\begin{equation}\label{eq:Wphi}
 \Wphih(p)=\widehat{\phi}_{\mathrm{I}}(p) \1_{|p| >  N^{\alpha}}.
\end{equation}
and the associated tours scattering length is defined by
\begin{equation}\label{eq:aW_N}
 8\pi \aW=\int_{\T^3} \sqrt{N} W_N(y)(1+\phi_{\mathrm{I}}(y)). 
\end{equation}
We gather in the following lemmas useful properties on $\phi_{\mathrm{B}}, \phi_{\mathrm{I}}$, $\aV$ and $\aW$.
\begin{lem}[Properties of $\phi_{\mathrm{B}}$]\label{lem:phi_V}
We have for $0\leq \alpha \leq 1$
\begin{align}
\|\widetilde{\phi}_{\mathrm{B}}\|_2 &\lesssim N^{-1-\alpha/2}, &  \|\widetilde{\phi}_{\mathrm{B}}\|_\infty &\lesssim 1, \notag\\
  \|\widehat{\widetilde{\phi}}_{\mathrm{B}}\|_\infty &\lesssim N^{-1-2\alpha}, &\|p^2\widehat{\widetilde{\phi}}_{\mathrm{B}}\|_\infty &\lesssim N^{-1},\notag \\ 
\| |\nabla|^{1/2} \Vphi\|_2 &\lesssim N^{-1} \sqrt{\log N},  & \|\widetilde{\phi}_{\mathrm{B}}-\phi_{\mathrm{B}}\|_\infty &\lesssim N^{-1+\alpha} .\notag %
\end{align}
\end{lem}
\begin{proof}
This follows from Lemma \ref{lem:reg_est} in Appendix \ref{app:scatt}, by taking $v = V_N$
\end{proof}

\begin{lem}[Properties of $\phi_{\mathrm{I}}$]\label{lem:phi_W}
We have for $ 0\leq \alpha \leq 1$ and $0<s\leq 1/2$
\begin{align*}
 \|\Wphi\|_2 &\lesssim N^{-1/2-\alpha/2}, & \|\Wphi\|_\infty &\lesssim 1 \\
\||\nabla|^{1/2+s}\Wphi\|_2 &\lesssim N^{-1/2+s/2}, & \||\nabla|^{1/2}\Wphi\|_2 &\lesssim N^{-1/2} \sqrt{\log N},\\
 \|p^2 \Wphih\|_\infty &\lesssim N^{-1/2},& \|\phi_\mathrm{I}-\Wphi\|_\infty &\lesssim N^{-1/2+\alpha}.
\end{align*}
\end{lem}

\begin{proof}
This follows from Lemma  \ref{lem:est_phi_alpha} in Appendix \ref{app:scatt}, by taking $v =  W_N$.
\end{proof}

As we perform all our calculations on the torus, we need to control the difference of the quantities defined there to the real scattering lengths in $\R^3$.

\begin{lem}\label{lem:aM-difference} The differences between the scattering length on the torus and in $\mathbb{R}^{3}$ satisfy
 \begin{align}
  |\aV- \ao_V| &\lesssim N^{-1}, \notag \\
	\label{eq:diff_aw}
  |\aW- \ao_W| &\lesssim N^{-1/2}. \notag
 \end{align}
\end{lem}

\begin{proof}
This follows from Lemma \ref{lem:diff_scatt_gen} in Appendix \ref{app:scatt}, by taking respectively $v = V_N$ and $v  = W_N$.
\end{proof}

\begin{rem}\label{rem:diff_scat}
The error estimate above is in fact sharp. When $V$ is radial it can be computed that \cite{BocBreCenSch-18}
\begin{align*}
\lim_{N\to \infty} 4\pi N \left( \aV -  \ao_{V}\right) = 2 - \lim_{M\to \infty} \sum_{\substack{p\in \mathbb{Z}^{3} \\ \|p\|_{\infty} \leq M}} \frac{\cos(|p|)}{|p|^2}.
\end{align*}

\end{rem}

\section{Renormalization of the Bogoliubov-Fr\"ohlich Hamiltonian}\label{sect:renorm}

In this section we explain the renormalization of the Bogoliubov-Fr\"ohlich Hamiltonian and prove Proposition~\ref{prop:renorm}. For the proof of the technical Lemmas, we will refer to Section~\ref{sect:conv}, where the whole procedure is performed again, but on a slightly more general Hamiltonian. Although the proof is based on the ideas of~\cite{La-19,La-20}, we adapt it here for the use in the proof of Theorem~\ref{thm:spect}.

The problem with the formal expression of the Hamiltonian~\eqref{eq:H_BF formal} is that the interaction form factor doesn't decay for large momentum. Since the form factor is not square-integrable, the corresponding creation operator is not a densely defined operator on $\Gamma(\mathfrak{H}_+)$.

To approach this problem, we first consider the operator with ultra-violet cutoff $\Lambda\geq 0$,
\begin{equation}\label{eq:def_BF}
 H_\mathrm{BF}^\Lambda = -\Delta_x +\ud \Gamma(\epsilon) + a(\mw_x^\Lambda) + a^*(\mw_x^\Lambda),
\end{equation}
where
\begin{equation} \label{eq:w_Lambda}
 \widehat \mw^\Lambda(p)= 8\pi \ao_W |p| \epsilon(p)^{-1/2}  \1_{|p|\leq \Lambda}.
\end{equation}
This defines a self-adjoint operator on $\mathscr{H}_+=L^2(\T^3) \otimes \Gamma(\mathfrak{H}_+)$.
In our derivation of the Hamiltonian from $H_N$ we will encounter a similar but somewhat more complicated Hamiltonian with $\Lambda$ a function of the boson number $N$. It should be kept in mind that, as a consequence of this, the notation we introduce below is specific to this section.

We now perform a Weyl transformation on the Hamiltonian. This is a well-known procedure that is used for example in the renormalization of the Nelson Hamiltonian~\cite{Nelson-64,GrWu-18}. It will make the divergence of the Hamiltonian explicit, though in our case it will not yet suffice to arrive at an expression from which the cutoff may be removed.
For $\kappa\geq 0$, which will be chosen sufficiently large later, we define $x\mapsto \widehat{f}^\Lambda_{\kappa,x}\in L^2(\T^3,\ell^2(2\pi\Z^3\setminus\{0\})$ by
\begin{equation} \label{eq:def_f_sec_renorm}
\widehat f_{\kappa, x}^\Lambda (p)= - (-\Delta_x +\epsilon(p))^{-1}\widehat \mw_x^\Lambda(p) \1_{|p|> \kappa},
\end{equation}
and, denoting by $f^\Lambda_{\kappa,x}$ the inverse Fourier transform,
\begin{equation}
	\label{eq:def_U_kappa_Lambda}
 U_\kappa^\Lambda = \exp\big( a^*(f_{\kappa, x}^\Lambda) - a(f_{\kappa, x}^\Lambda )\big).
\end{equation}
This transformation satisfies the identities~\cite[Sect.2]{GrWu-18}
\begin{subequations}
 \begin{align} \label{eq:Gross_1}
 (U_\kappa^\Lambda)^* a(g) U_\kappa^\Lambda &= a( g) + \langle  g, f_{\kappa, x}^\Lambda\rangle, \\
 (U_\kappa^\Lambda)^* \ui \nabla_x U_\kappa^\Lambda &= \ui \nabla_x + a(\ui\nabla_x f_{\kappa, x}^\Lambda) + a^*(\ui \nabla_x f_{\kappa, x}^\Lambda). \label{eq:Gross_2}
\end{align}
\end{subequations}
The choice of  $f_{\kappa, x}^\Lambda$ leads to cancellations between terms in $ (U_\kappa^\Lambda)^*(-\Delta_x+\ud \Gamma(\epsilon))U_\kappa^\Lambda$ and $a^*(\mw_x)+a(\mw_x)$, e.g., we have
\begin{equation}
 (U_\kappa^\Lambda)^*\ud \Gamma(\epsilon)U_\kappa^\Lambda=\ud \Gamma(\epsilon) + a^*(\epsilon f_{\kappa, x}^\Lambda) + a(\epsilon f_{\kappa, x}^\Lambda) + \langle f_{\kappa, x}^\Lambda ,\epsilon f_{\kappa, x}^\Lambda \rangle.
\end{equation}
The terms coming from $-\Delta_x$ are similar but more involved, since taking the square of~\eqref{eq:Gross_2} also yields new quadratic expressions in $a^*$, $a$ (see~\cite[Sect.2]{GrWu-18}).
Collecting all the terms and using $\ui\nabla_x f_{\kappa, x}^\Lambda(y) = -\ui \nabla f_{\kappa, x}^\Lambda(y)$, we obtain for $\Lambda\geq \kappa$
\begin{align}
 (U_\kappa^\Lambda)^* H_\mathrm{BF}^\Lambda U_\kappa^\Lambda&= -\Delta_x + \ud \Gamma(\epsilon) + a^*(\mw^\kappa_x) + a( \mw^\kappa_x) \notag\\
 &\quad + a^*(\ui \nabla f_{\kappa, x}^\Lambda)^2 + a(\ui \nabla f_{\kappa, x}^\Lambda)^2 + 2 a^*(\ui \nabla f_{\kappa, x}^\Lambda)a(\ui \nabla f_{\kappa, x}^\Lambda)\notag\\
 &\quad - 2 a^*(\ui \nabla f_{\kappa, x}^\Lambda)\ui \nabla_x -  2\ui \nabla_x  a(\ui \nabla f_{\kappa, x}^\Lambda) + \langle \mw_x , f^\Lambda_{\kappa, x}\rangle.
 \label{eq:H_BF Weyl}
\end{align}
The numbers
\begin{equation}\label{eq:E^1}
 E_\Lambda^{(1)}=\langle  \mw^\Lambda_x , f_{\kappa, x}^\Lambda\rangle =-(8\pi \ao_W)^2 \sum_{p\in 2\pi \Z^3} \frac{p^2}{\epsilon(p)(p^2 +\epsilon(p))} \1_{\kappa \leq |p|\leq \Lambda}\sim -c \Lambda
\end{equation}
diverge as $\Lambda\to \infty$, but the rest of the operator is now more regular than before.
Indeed, the terms $a^*(\ui \nabla f_{\kappa, x}^\Lambda)a(\ui \nabla f_{\kappa, x}^\Lambda)$ and $a^*(\ui \nabla f_{\kappa, x}^\Lambda)\ui \nabla_x +\hc$ now define quadratic forms on $Q(-\Delta_x+\ud \Gamma(\epsilon))$, even for $\Lambda=\infty$ (which was not the case before, as $\epsilon^{-1/2} \mw$ is not square-integrable).
 However, the term $a^*(\ui \nabla f_{\kappa, x}^\Lambda)^2+\hc$ is still singular.
To deal with this, we employ an idea developed in~\cite{LaSch-19,La-19,La-20}.
We write
\begin{align}\label{eq:H-G intro}
 &(U_\kappa^\Lambda)^* H_\mathrm{BF}^\Lambda U_\kappa^\Lambda - E_\Lambda^{(1)}+1 \\
 &\quad = (1-G^*_\Lambda)(-\Delta_x+ \ud \Gamma(\epsilon)+1)(1-G_\Lambda) -G^*_\Lambda (-\Delta_x+ \ud \Gamma(\epsilon)+1)G_\Lambda + R_\Lambda ,\notag
\end{align}
with
\begin{align}
 G_\Lambda&=-  (-\Delta_x+ \ud \Gamma(\epsilon)+1)^{-1}a^*( \ui \nabla f_{\kappa, x}^\Lambda )^2,\label{eq:G_Lambda-def}\\
 R_\Lambda&=2a^*( \ui \nabla f_{\kappa, x}^\Lambda ) a(\ui \nabla f_{\kappa, x}^\Lambda ) - \big(2 \ui \nabla_x a(\ui \nabla f_{\kappa, x}^\Lambda ) + a(w^\kappa_x) +\hc\big).
\end{align}
Note that~\eqref{eq:H-G intro} is indeed just a simple identity, as expanding the first term and cancelling the term with $G^*_\Lambda(\cdot)G_\Lambda$ returns the expression~\eqref{eq:H_BF Weyl}.
After rewriting the Hamiltonian in this way, the singularity is contained in the term $G^*_\Lambda(\cdot)G_\Lambda$ and can be extracted, as we will see. We will now analyze each term in (\ref{eq:H-G intro}) as $\Lambda \to \infty$.

\paragraph{The dressing factor $G_{\Lambda}$.}
We begin with the dressing factor $G_\Lambda$ for which we have the following Lemma.
\begin{lem}[Regularity and limit of $G_\Lambda$]\label{lem:G_Lambda bound}
For all $0\leq s<1/2$, there exists a family $(C_\kappa)_{\kappa>0}$ with $\lim_{\kappa\to \infty} C_\kappa=0$, so that for all $\Lambda\in \R_+\cup\{\infty\}$, $\kappa>0$, $\Psi\in D(\cN_+)$
 \begin{equation*}
  \|(-\Delta_x+\ud \Gamma(\epsilon))^sG_\Lambda \Psi\|_{\mathscr{H}_+}\leq C_\kappa \|(\cN_++1)^{s} \Psi\|_{\mathscr{H}_+}.
 \end{equation*}
Moreover, for all $0\leq s<1/2$, there exists $(C_\Lambda)_{\Lambda>0}$ with $\lim_{\Lambda\to \infty}C_\Lambda =0$, so that for all $\kappa, \Lambda>0$, $\Psi\in D(\cN_+)$,
\begin{equation*}
 \|(-\Delta_x+\ud \Gamma(\epsilon))^s(G_\Lambda-G_\infty)\Psi\|_{\mathscr{H}_+}\leq C_\Lambda\|\cN_+^s\Psi\|_{\mathscr{H}_+}.
\end{equation*}
\end{lem}
The proof follows the one of Lemma~\ref{lem:G_0_bound}. The key point is that $\epsilon^{-1/2+s/2} k f_{\kappa}^\infty\lesssim |k|^{s}f_{\kappa}^\infty\in \ell^2$ for $s<1/2$.

\paragraph{The dressed kinetic operator $K_{\Lambda}$.}
We can now discuss the behavior as $\Lambda\to \infty$ of the principal term,
\begin{equation}
 K_\Lambda := (1-G^*_\Lambda)(-\Delta_x+ \ud \Gamma(\epsilon)+1)(1-G_\Lambda),
\end{equation}
with,
\begin{equation}
 D(K_\Lambda)=\big\{\Psi \in \mathscr{H}_+ : (1-G_\Lambda)\Psi \in D(-\Delta_x+ \ud \Gamma(\epsilon))\big\}.
\end{equation}

\begin{lem}[Comparison estimates for $K_\Lambda$]\label{lem:K sa}
 For $\kappa$ sufficiently large and all $\Lambda\in \R_+\cup\{\infty\}$,  $K_\Lambda$ is self-adjoint with compact resolvent.
 Moreover, for $s<1/2$ there is $C>0$ such that for all $\Lambda \in \R_+\cup\{\infty\}$ and all $\Psi \in D(K_{\Lambda})$
  \begin{equation*}
  \|(\cN_++1)^{1-s} (-\Delta_x+\ud \Gamma(\epsilon)+1)^{s} \Psi\|\leq C \|K_{\Lambda} \Psi\|.
 \end{equation*}
\end{lem}
\begin{proof}
 By Lemma~\ref{lem:G_Lambda bound}, $1-G_\Lambda$ is invertible for $\Lambda\in \R_+\cup\{\infty\}$ by Neumann series for $\kappa$ large enough.
 In particular, $D(K_\Lambda)$ is dense.
 As $-\Delta_x+\ud \Gamma(\epsilon)+1$ is invertible, the symmetric operator $K_\Lambda$ is a product of three invertible operators, and thus self-adjoint. Since the resolvent of $-\Delta_x+\ud \Gamma(\epsilon)+1$ is compact, so is that of $K_\Lambda$.

 The bound follows by inserting $\pm G_\Lambda$ and using the properties of $G_\Lambda$.  The details, and more general estimates, are given in Lemma~\ref{lem:K_0-props}.
\end{proof}

Note that $D(K_\Lambda)$ coincides with $D(-\Delta+\ud \Gamma(\epsilon))$ for $\Lambda<\infty$, but from Lemma~\ref{lem:G_Lambda bound} for $\Lambda=\infty$ we only have $D(K_\infty)\subset D((-\Delta+\ud \Gamma(\epsilon))^{s})$ for $s<1/2$ (and in fact the intersection is trivial for $s\geq 1/2$, cf.~\cite{LaSch-19}).

\paragraph{The singular term.}

We now consider the second term in (\ref{eq:H-G intro}), which contains the remaining singularity.
We can observe that, denoting the Fock vacuum by $\varnothing$,
\begin{align}
 E_\Lambda^{(2)}:&= -\Big\langle \1\otimes \varnothing, G^*(-\Delta_x + \ud \Gamma(\epsilon)+1)G \1\otimes  \varnothing \Big\rangle
  \nn \\
  &= - \Big\langle \1\otimes  \varnothing ,a(\ui \nabla f_{\kappa, x}^\Lambda)^2(-\Delta_x + \ud \Gamma(\epsilon)+1)^{-1} a^*(\ui \nabla f_{\kappa, x}^\Lambda)^2 \1\otimes  \varnothing \Big\rangle \nn\\
  &=- {2} \sum_{p,q\in 2\pi\Z^3} \frac{(p\cdot q)^2 |\widehat f_{\kappa, 0}^\Lambda(p)|^2 |\widehat f_{\kappa, 0}^\Lambda(q)|^2}{(p+q)^2 + \epsilon(p)+\epsilon(q)+1} \sim - c \log \Lambda.\label{eq:E^2}
\end{align}
Subtracting this divergent contribution, we define
\begin{equation}\label{eq:T_Lambda-def}
 T_{\Lambda} = -G^*(-\Delta_x + \ud \Gamma(\epsilon)+1)G- E_\Lambda^{(2)}.
\end{equation}
In order to exhibit the cancellation of the singularities, we move all creation operators to the left of the annihilation operators. To do so, we use that $\ui \nabla_x \ue^{-\ui kx}=\ue^{- \ui kx} (\ui \nabla_x + k)$ and $\ud\Gamma(\epsilon) a_k^* = a_k^*(\ud\Gamma(\epsilon)+\epsilon(k))$, which implies
\begin{align}
\Big(-\Delta_x + \ud \Gamma(\epsilon)\Big)^{-1} e^{-\ui k\cdot x} a^*_k = e^{-\ui k\cdot x} a^*_k \, \Big( (\ui \nabla_x + k)^2 + \epsilon(k) + \ud \Gamma(\epsilon)\Big)^{-1}.
\end{align}
This gives
\begin{equation}\label{eq:T_Lambda-Theta}
 T_\Lambda
 =\sum_{j=0}^2\Theta_{\Lambda,j},
\end{equation}
where $\Theta_{\Lambda,j}$ contains the terms with $j$ remaining creation/annihilation operators after reordering.
These are given explicitly by
\begin{align}
 \Theta_{\Lambda,0} &= - {2} \sum_{p,q \in 2\pi\Z^3}
 \frac{(p\cdot q)^2 |\widehat f_{\kappa, 0}^\Lambda(p)|^2 |\widehat f_{\kappa, 0}^\Lambda(q)|^2}{(\ui \nabla_x +p +q)^2 + \epsilon(p)+\epsilon(q) + \ud \Gamma(\epsilon)+1} - E_\Lambda^{(2)}\\
 \Theta_{\Lambda,1}&= {4} \sum_{p,q \in 2\pi\Z^3} \ue^{-\ui p x} a^*_p \theta_{\Lambda,1}(p,q)a_q \ue^{\ui q x}\\
 \Theta_{\Lambda,2}&= \sum_{p,q,k,\ell \in 2\pi\Z^3} \ue^{-\ui p x}\ue^{-\ui q x} a^*_{p}a^*_{q} \theta_{\Lambda,2}(p,q,k,\ell)a_{k} a_{\ell} \ue^{\ui k x}\ue^{\ui\ell x},
\end{align}
with
\begin{align}
&\theta_{\Lambda,1}(p,q)\nn \\
 &= - \sum_{k\in 2\pi \Z^3}  \frac{\widehat f_{\kappa, 0}^\Lambda(p)(p\cdot k)(q\cdot k) \widehat f_{\kappa, 0}^\Lambda(k)^2\widehat f_{\kappa, 0}^\Lambda(q)}{(\ui \nabla_x +p +q+k)^2+\epsilon(k) + \epsilon(p)+\epsilon(q)+\ud \Gamma(\epsilon)+1},
 \end{align}
 and
 \begin{align}
 &\theta_{\Lambda,2}(p,q,k,\ell)\nn \\
 &=-\frac{\widehat f_{\kappa, 0}^\Lambda(p)\widehat f_{\kappa, 0}^\Lambda(q)(p\cdot q)(k\cdot \ell)\widehat f_{\kappa, 0}^\Lambda(k)\widehat f_{\kappa, 0}^\Lambda(\ell)}{(\ui \nabla_x +p+q +k+\ell)^2+\epsilon(p)+\epsilon(q) + \epsilon(k)+\epsilon(\ell)+\ud \Gamma(\epsilon)+1}.
\end{align}
Note that the kernels $\theta_{\Lambda,1}, \theta_{\Lambda,2}$ as functions of the commuting operators $\ui \nabla_x$ and $\dd\Gamma(\epsilon)$, are well defined also for $\Lambda=\infty$, since in $\theta_{\Lambda,1}$ the sum over $k$ converges. Moreover, writing the terms in $ \Theta_{\Lambda,0}$ with a common denominator we obtain
\begin{align}
 \Theta_{\Lambda,0} =  \sum_{k,\ell \in 2\pi\Z^3} \bigg(&
 \frac{(k\cdot \ell)^2\widehat f_{\kappa, 0}^\Lambda(k)^2 \widehat f_{\kappa, 0}^\Lambda(\ell)^2}
 {(\ui \nabla_x +k +\ell)^2 + \epsilon(k)+\epsilon(\ell) + \ud \Gamma(\epsilon)+1} \nn \\
 &\qquad \times \frac{((\ui \nabla_x +k +\ell)^2 +\ud \Gamma(\epsilon) - (k+\ell)^2))}{(k +\ell)^2 + \epsilon(k)+\epsilon(\ell)+1}\bigg),
\end{align}
where the series also converges for $\Lambda=\infty$ to a well defined, although unbounded, function of $\ui \nabla_x$ and $\dd\Gamma(\epsilon)$. With these observations, we can make sense of $T_\infty$.
The key properties of $T_\Lambda$ are summarized in the following lemma.  In particular, $T_\infty$ defines an operator on $D(K_\infty)$ and  $K_\infty + T_\infty$ is self-adjoint on $D(K_\infty)$.

\begin{lem}(Estimates and convergence of $T_{\Lambda}$)\label{lem:T_Lambda relative}
Given $\delta>0$, the following holds for $\kappa$ sufficiently large.
For all $\Lambda,\Lambda'\in \R_+\cup\{\infty\}$, $\Psi\in D(K_{\Lambda'})$
\begin{equation*}
\|T_{\Lambda}\Psi\| \leq  \delta\|K_{\Lambda'}\Psi\|.
\end{equation*}
Moreover, there exists $(C_\Lambda)_{\Lambda>0}$ with $\lim_{\Lambda\to \infty} C_\Lambda=0$ so that for all $\Lambda,\Lambda'\in \R_+\cup\{\infty\}$, $\Psi\in D(K_\Lambda') $
\begin{equation*}
 \|(T_\Lambda-T_\infty)\Psi\|\leq C_\Lambda \|K_{\Lambda'}\Psi\|.
\end{equation*}
\end{lem}
The proof is similar to the one of Proposition~\ref{prop:T_N_relative} and is omitted here. The main idea is to estimate $T_\Lambda$, respectively the difference $T_\Lambda-T_\infty$, by operators of the form $\cN_+^{1-s}(-\Delta_x+\ud \Gamma(\epsilon))^s$, which are themselves controlled by $K_{\Lambda'}$ using Lemma \ref{lem:K sa}.

\paragraph{The remainder term $R_\Lambda$.}
Since $p\mapsto p \widehat f_{\kappa,x}^\infty(p)$ is not square summable, the expression $a^*(-\ui \nabla f_{\kappa,x}^\Lambda) \ui \nabla_x$
occuring in $R_\Lambda$ only makes sense as a quadratic form for $\Lambda=\infty$.
\begin{lem}\label{lem:R_Lambda bound}
 Given $\delta>0$, $|s|<1/8$, the following holds for $\kappa$ sufficiently large. There is $C>0$, so that for all $\Lambda,\Lambda' \in \R_+\cup\{\infty\}$ and $\Psi,\Phi\in D(K_{\Lambda'})$
 \begin{equation*}
  |\langle \Phi, R_\Lambda\Psi\rangle|\leq \delta \|(K_{\Lambda'}^{1/2-s}+C)\Phi\| \|(K_{\Lambda'}^{1/2+s}+C)\Psi\|.
  \end{equation*}
Moreover, there exists  $(C_\Lambda)_{\Lambda>0}$ with $\lim_{\Lambda\to \infty} C_\Lambda=0$ so that for all $\Lambda,\Lambda' \in \R_+\cup\{\infty\}$ and $\Psi\in Q(K_{\Lambda'})$
\begin{equation*}
  |\langle \Psi, (R_\Lambda -R_\infty)\Psi \rangle| \leq C_\Lambda \langle\Psi, K_{\Lambda'}\Psi\rangle .
 \end{equation*}
\end{lem}
  The proof of this lemma is a combination of Lemma~\ref{lem:R_N_bound} and Lemma~\ref{lem:K sa}.

\begin{prop}\label{prop:H_BF trafo conv}
For $\kappa$ sufficiently large, the quadratic form
\begin{equation*}
 K_\infty +T_\infty + R_\infty
\end{equation*}
is associated with a unique self-adjoint operator with compact resolvent and domain contained in $Q(K_\infty)=(1-G_\infty)^{-1}Q(-\Delta_x + \ud \Gamma(\epsilon))$. Moreover,
\begin{equation}
 \lim_{\Lambda \to \infty} \big(K_\Lambda +T_\Lambda + R_\Lambda \pm \ui \big)^{-1} = \big(K_\infty +T_\infty + R_\infty \pm \ui \big)^{-1}
\end{equation}
in the operator norm.
\end{prop}

\begin{proof}
By Lemma~\ref{lem:T_Lambda relative} and the Kato-Rellich Theorem, $K_\Lambda+T_\Lambda$ is self-adjoint on $D(K_\Lambda)$ for $\kappa$ sufficiently large.
Lemma~\ref{lem:R_Lambda bound} shows that $R_\Lambda$ is a perturbation of the quadratic form of $K_\Lambda+T_\Lambda$ with relative bound $\delta<1$. This implies existence of the self-adjoint realization by the KLMN theorem, whose resolvent is compact since the resolvent of $K_\Lambda$ is.
Using the resolvent formula and the properties of $T_\Lambda$, $R_\Lambda$ and $K_\Lambda$ (in particular, the uniformity of the relative bounds), one shows norm-resolvent convergence by the same argument as in the proof of Proposition~\ref{prop:HsingW}.
\end{proof}

We can now conclude by proving the proposition on the renormalization of the Bogoliubov-Fr\"ohlich Hamiltonian.

\subsection{Proof of Proposition~\ref{prop:renorm}}

Setting $E_\Lambda=E_\Lambda^{(1)} + E_\Lambda^{(2)}$, with the latter defined in~\eqref{eq:E^1},~\eqref{eq:E^2}.
Let $\kappa$ be sufficiently large for Proposition~\ref{prop:H_BF trafo conv} to hold.
Then, with the notation of (\ref{eq:H-G intro}) and (\ref{eq:T_Lambda-def}), we have
\begin{equation}
 (U_\kappa^\Lambda)^* H_\mathrm{BF}^\Lambda U_\kappa^\Lambda - E_\Lambda+1 = K_\Lambda + T_\Lambda + R_\Lambda.
\end{equation}
By Proposition \ref{prop:H_BF trafo conv}, this converges to $K_\infty + T_\infty+R_\infty$ in norm resolvent sense as $\Lambda\to \infty$. It remains to deal with the convergence of the transformation $U_\kappa$. Note that $f_{\kappa,x}^\infty \in \ell^2(2\pi \Z^3)$, and thus (cf.~\cite[Lem.C.2]{GrWu-18}) 
for $\Psi \in Q(K_{\Lambda'})$, $\Lambda' \in \mathbb{R}_+ \cup \{\infty\}$,
\begin{equation}\label{eq:U_kappa-conv}
 \| (U_\kappa^\Lambda- U_\kappa^\infty)\Psi\| \leq 2 \|f^\Lambda_\kappa - f^\infty_\kappa \| \|(\cN_++1)^{1/2} \Psi\|\stackrel{\Lambda\to \infty}{\rightarrow} 0.
\end{equation}
Since $\|\cN_+^{1/2}\Psi\| \lesssim \|K_{\infty}^{1/2}\Psi\| \lesssim \|(K_{\infty} + T_{\infty} + T_{\infty})^{1/2}\Psi\|$ from Lemma \ref{lem:K sa} and Proposition \ref{prop:H_BF trafo conv} (inclusion of the domains), this proves that
\begin{equation}
 \lim_{\Lambda\to \infty} \big(H_\mathrm{BF}^\Lambda -E_\Lambda +1\pm\ui\big)^{-1}
 = U_\kappa^\infty\big(K_\infty + T_\infty+R_\infty+\pm\ui \big)^{-1}(U_\kappa^\infty)^*
\end{equation}
in norm. We have thus proved Proposition~\ref{prop:renorm} with the renormalized Bogoliubov--Fröhlich Hamiltonian given by
\begin{equation}\label{eq:H_BF_Ukappa_infty}
 H_\mathrm{BF}=U_\kappa^\infty(K_\infty + T_\infty+R_\infty-1)(U_\kappa^\infty)^*.
\end{equation}
The resolvent of this operator is compact by Proposition~\ref{prop:H_BF trafo conv}.
\qed

 The proof above shows that
 \begin{equation}
  Q(H_\mathrm{BF})=U_\kappa^\infty  Q(K_\infty)= U_\kappa^\infty (1-G_\infty)^{-1}Q(-\Delta_x+\ud \Gamma(\epsilon)),
 \end{equation}
since $R_\infty$ is only a form perturbation of $K_\infty$.
More precise information on the operator domain is obtained using the slightly different formulation of~\cite{La-19,La-20}.

\section{Transformation of the Hamiltonian}\label{sect:trafo}

\subsection{The excitation Hamiltonian}\label{sect:excite}

In a first step, we transfer the Hamiltonian 
 \begin{equation}
 H_N = -\Delta_x +\sum_{i=1}^N (-\Delta_{y_i}) + \sum_{1\leq i < j \leq N} V_N(y_i-y_j) + \sum_{i=1}^N W_N(x-y_i)
\end{equation}
from the $N$-particle space $\mathscr{H}_N$ to the excitation space $\mathscr{H}_+$ using the excitation map introduced in~\cite{LewNamSerSol-15}. This map is given by
\begin{align}\label{eq:excitation-map}
\begin{aligned}
  U_X:\;&\mathscr{H}_+\to \mathscr{H}_N \\
  (U_X \Psi)(x,y_1, \dots, y_N) =& \sum_{j=0}^N  S_N  \Big(\1(y_{j+1},\dots ,y_N) \Psi^{(j)}(x,y_1, \dots, y_j)\Big)
\end{aligned}
\end{align}
where $\1$ denotes the constant function and $S_N$ the projection to symmetric functions in $y_1, \dots , y_N$.
In words, $U_X$ completes the function $\Psi^{(j)}(x,y_1, \dots, y_j)$ to a function on $\T^{3(N+1)}$ by multiplying with the constant function in the variables $y_{j+1}, \dots, y_N$ and then symmetrizes the result.
This defines a partial isometry, i.e., we have
\begin{equation}
 U_X^*U_X=\1_{\cN_+\leq N},\qquad U_X U_X^*=\1_{\mathscr{H}_N},
\end{equation}
where $\cN_+$ denotes the number operator on $\mathscr{H}_+$.

The Bose--Einstein condensation condition that we will use in the proof of our main theorems ensures that there are at most $\sqrt{N}$ excitations in a state of the bosons whose energy has the correct leading order for $N\to \infty$. Within the notations of second quantization, Condition \ref{cond:BEC} is equivalent to
 \begin{equation}\label{eq:equiv_cond}
  H_N - 4\pi \ao_V N \geq  c U_X \cN_+ U_X^*-d\sqrt{N}.
 \end{equation}

We will now evaluate the action of $U_X$ on $H_N$ in detail.

\begin{prop}\label{prop:H-ex}
In the sense of quadratic forms on $Q(-\Delta_x +\ud \Gamma(-\Delta))\subset \mathscr{H}_+$, we have the identity
 \begin{equation*}
 U_X^* H_N U_X = \1_{\cN_+ \leq N} \mathcal H_N  \1_{\cN_+ \leq N},
\end{equation*}
where
\begin{align*}
	\mathcal H_N &:= \tfrac12(N-1)\widehat V(0) + \sqrt{N } \widehat{W}(0) + \mathcal H_{\mathrm{B}} + \mathcal H_{\mathrm{I}}, \label{eq:HN Fock}
\end{align*}
	with
\begin{align*}
		 \mathcal H_{\mathrm{B}} &= \ud\Gamma(-\Delta) +(N-\cN_+)_+\ud \Gamma(\widehat V_N(-\ui\nabla))+ \mathcal L_0 + \mathcal L_2 + \mathcal L_3 + \mathcal L_4, \notag \\
		 \mathcal L_0 &= -\frac{1}{2N} \widehat V(0) \cN_+ (\cN_+-1)\mathds{1}_{\cN_+ \leq N}, \\
\mathcal L_2 &= \frac{1}{2} \sqrt{(1-\cN_+/N)_+}\sqrt{(1-(\cN_++1)/N)_+}  \int N V_N(y_1-y_2) a_{y_1} a_{y_2} + \hc, \\
\mathcal L_3 &= \frac{1}{2} \sqrt{(1-\cN_+/N)_+} \int \sqrt N V_N(y_1-y_2) a^*_{y_1} a_{y_2} a_{y_1} + \hc, \\
\mathcal L_4 &= \frac{1}{2} \int V_N(y_{1}-y_2) a^*_{y_1}a^*_{y_2} a_{y_1} a_{y_2},
\end{align*}
and
\begin{align*}
	 \mathcal H_{\mathrm{I}} &=  -\Delta_x + Q_1 +\ud \Gamma(W_{N,x}), \notag\\
	 Q_1&=\sqrt{(1-\cN_+/N)_+ }a(\sqrt{N} W_{N,x}) + a^*(  \sqrt{N} W_{N,x}) \sqrt{(1-\cN_+/N)_+ }.
\end{align*}
\end{prop}
\begin{proof}
 This follows by applying well-known identities, going back to~\cite{LewNamSerSol-15}, hence we only sketch the calculations.
The kinetic operator is unchanged by $U_X$ since the constant function has derivative zero, thus only the potentials are transformed in a non-trivial way. For the boson-boson interaction, we can write
 \begin{equation}
  \sum_{1\leq i<j\leq N}V_N(y_i-y_j)= \frac12 \sum_{k,p,q \in 2\pi \Z^3} \widehat V_N(k)a^*_{p +k} a^*_{q} a_{q+k}a_p
 \end{equation}
by considering $\mathscr{H}_N$ as a subspace of $L^2(\mathbb T^3)\otimes \Gamma(L^2(\mathbb T^3))$.
With 
\begin{equation}
 U_X^* a_0^*a_0U_X=(N-\cN_+)_+, \qquad  \sum_{p\in 2\pi\Z^3\setminus\{0\}} U_X^* a^*_pa_p U_X =\cN_+,
\end{equation}
and
\begin{equation}
 U_X^* a_p^*a_0 U_X = a_p^* \sqrt{(N-\cN_+)_+},
\end{equation}
for $p\neq 0$, we can rewrite
\begin{multline}
\frac12 \sum_{k,p,q \in 2\pi \Z^3} \widehat V_N(k) U_X^* a^*_{p +k} a^*_{q} a_{q+k}a_p U_X = \frac{1}{2N} \widehat{V}(0) (N-\mathcal N_+)_+(N-\mathcal N_+-1)_+ \\ + (N-\cN_+)_+\ud \Gamma(\widehat V_N(-\ui \nabla) + N^{-1}\widehat{V}(0))+ \cL_2 + \cL_3 + \cL_4 \1_{\cN_+\leq N}.
\end{multline}
With
\begin{equation}
 \frac{\widehat{V}(0)}{2N}\left\{ (N-\mathcal N_+)(N-\mathcal N_+-1) + 2(N-\mathcal N_+) \mathcal N_+ \right\} = \tfrac12(N-1)\widehat V(0) + \mathcal L_0
\end{equation}
this gives the claimed form of $\mathcal{H}_B$.
For the boson-impurity interaction, we find similarly
\begin{align}
 U_X^*& \sum_{i=1}^N W_N(x-y_j) U_X = \sum_{k,p \in 2\pi \Z^3}  \widehat W_N(k) \ue^{\ui k x}  U_X^* a_p^*a_{p+k} U_X \notag \\
 &= \sqrt N \widehat W(0) + \sum_{k\in 2\pi \Z^3\setminus\{0\}}  \widehat W_N(k) \ue^{\ui k x}  \Big(\sqrt{(N-\cN_+)_+} a_k + a^*_{-k}\sqrt{(N-\cN_+)_+} \Big) \notag \\
 &\quad + \sum_{k,p \in 2\pi \Z^3\setminus\{0\}\atop p +k \neq 0}  \widehat W_N(k) \ue^{\ui k x}   a_p^*a_{p+k},
 \end{align}
where we may identify $\ud \Gamma(W_{N,x})$ in the last line and $Q_1$ in the line before. This proves the identity.
\end{proof}

\subsection{Interactions of high-energy bosons: a first quadratic transformation}\label{sect:quadratic}

We start by implementing a transformation whose generator is quadratic in the creation and annihilation operators.
This transformation makes explicit correlations among the bosons that are due to their interaction at high energies. More precisely it makes the correct leading order of the Bose gas $4\pi \ao_V (N-1)$ appear and renormalizes the quadratic term $\mathcal L_2$ into a less singular version $\widetilde{ \mathcal L}_2$.  Variants of this transformation have been used in many works on the dilute Bose gas~\cite{Sei-11,ErdSchYau-08,YauYin-09,BenOliSch-15,BreSch-19,BocBreCenSch-18,Hai-21,NamTri-23}. Throughout all the paper, we will distinguish between high and low energies using the cutoff $N^\alpha$, where the parameter $0 < \alpha <1/2$ will be chosen sufficiently small for our estimates to hold.

Let $\Vphi$, defined in~\eqref{eq:Vphi}, be the truncated scattering solution at momentum $N^\alpha$, $0< \alpha \leq 1/2$.
We set, with $Q(\cB_q)=Q(\cN_+)$,
\begin{equation}\label{eq:B_q}
 \cB_q= \frac12 \int_{\T^3} N \Vphi(y_1-y_2) a^*_{y_1}a^*_{y_2} \ud y_1\ud y_2 - \hc
\end{equation}
Using the commutation relations~\eqref{eq:CCR} together with the identity $\cN_+=\int a_y^*a_y\ud y$ and the Cauchy-Schwarz inequality, we have for $\Psi \in D(\cN_+^{1/2})$
\begin{align}
 &\langle \Psi , [\cN_+, \cB_q] \Psi \rangle
 = \mathrm{Re} \int \Big\langle \Psi,  N\Vphi(y_1-y_2) a^*_{y_1}a^*_{y_2}  \Psi\Big\rangle \ud y_1\ud y_2 \nn \\
 &=\mathrm{Re}\int \Big\langle (\cN_++1)^{-1/2} a_{y_2}a_{y_1}\Psi,  N\Vphi(y_1-y_2) (\cN_++1)^{1/2}  \Psi\Big\rangle \ud y_1\ud y_2 \nn \\
 &\leq \|N \Vphi\|_2  \bigg(\int \Big\|(\cN_++1)^{-1/2} a_{y_2}a_{y_1}\Psi\|^2\bigg)^{1/2}  \|\cN_+^{1/2}\Psi\| \nn\\
 &\leq \|N \Vphi\|_2  \|\cN_+^{1/2}\Psi\|^2.\label{eq:B_q-N commute}
\end{align}
As $ \|N\widetilde{\phi}_{\mathrm{B}}\|_2 \lesssim N^{-\alpha/2}$, see Lemma \ref{lem:phi_V}, this implies
\begin{equation}
 [\cN_+, \cB_q]= \mathcal{O}(N^{-\alpha/2} (\cN_++1)).
\end{equation}
Hence, by the commutator theorem~\cite[Thm.X.36]{ReeSim2}, there exists an anti-self-adjoint operator, denoted also by $\cB_q$, whose form is the above and for which $D(\cN_+)$ is a core. We can thus define the unitary
\begin{equation}
 U_q=e^{\cB_q},
\end{equation}
and $U_q(t)=e^{t\cB_q}$.
The effect of this transformation on $\cH_N$ is summarized in the following proposition.
\begin{prop}\label{prop:U_q}
Let $0\leq \alpha\leq 1/6$, and $\ao_{V}$ be the scattering length defined by~\eqref{eq:scatt_length-intro}.
We have, in the sense of quadratic forms on $\mathscr{H}_+$,
\begin{align*}
 U_q^* \cH_N U_q
 	&= 4 \pi \ao_V (N-1) +   e_N^{(U_q)} + \sqrt N\widehat W(0) \notag \\
 	&\quad +  \dd\Gamma(-\Delta + N \widehat{V}_N(-\ui\nabla) + \widehat{V}(0) - 8\pi \ao_{V}) + \widetilde{ \mathcal L}_2 + {\mathcal L}^{U_q}_3 + \mathcal L_4 \notag \\
	&\quad  -\Delta_x + Q_1^{U_q} + \ud \Gamma(W_{N,x}) +  \cE^{(U_q)},
 \end{align*}
with
\begin{align*}
 \widetilde{ \mathcal L}_2 &=    \int  \widetilde V_N(x-y) a^*_{x}a^*_{y} + \hc ,\qquad \widehat{\widetilde V}_N(p) = 4\pi  \ao_{V} \1_{0<|p| \leq N^{\alpha}}, \\
 Q_1^{U_q} &=\sqrt{\left(1-\frac{U_q^*\cN_+U_q}{N}\right)_+} a(\sqrt N W_{N,x}) +\hc,\\
{\mathcal L}^{U_q}_3 &= \frac{1}{2}\int \sqrt N V_N(x-y) a^*_x a_y a_x + \hc, 
\end{align*}
the scalar correction
\begin{equation*}
 e_N^{(U_q)} = 4 \pi N (\aV - \ao_V)  + \sum_{p\in 2\pi\Z^3 \atop |p|\leq N^{\alpha}} \frac{(4\pi \ao_{V})^2}{p^2}, 
\end{equation*}
as well as the error bound
\begin{equation*}
 \pm \cE^{(U_q)} \lesssim N^{-\alpha/2} (\mathcal L_4 + \mathcal N_++1 + N^{-1/2} \ud \Gamma(W_{N,x}) ) + N^{\alpha/2-1} \mathcal N_+^2  .
\end{equation*}
\end{prop}

\begin{rem}
The term $e_N^{(U_q)}$ is a lower order energy contribution, where the first term is a boundary effect of order one, see Remark \ref{rem:diff_scat}, and the second term is of order $N^{\alpha}$. This divergence is due to the fact that we solely renormalized the quadratic part with momenta higher than $N^{\alpha}$. It will later combine with a similar term coming from the diagonalization of the low momenta part of the Hamiltonian, yielding together a  contribution of order one, which is the well-known Lee-Huang-Yang correction.
\end{rem}

\begin{proof}[Proof of Proposition \ref{prop:U_q}]
The proof relies on several lemmas about the transformation of individual terms that we will prove later. We will be rather brief, as the treatment of the bosonic part of the Hamiltonian is standard, except for some estimates we need to adapt to our analysis.

The correlation energy that renormalizes $N\widehat V(0)$ into the torus scattering length $8\pi \aV$ defined in~\eqref{eq:aV_N} is hidden in $\ud \Gamma(-\Delta)$, $\cL_2$, and $\cL_4$.  We can see this by combining Lemmas \ref{lem:kin_T1}, \ref{lem:L2_T1}, \ref{lem:L4_T1} below. These give
\begin{subequations}
 \begin{align}
 &U_q^*(\ud \Gamma(-\Delta) + \cL_2 + \cL_4)U_q \notag \\
 &=\frac{1}{2} \int N^2 V_N \widetilde{\phi}_{\mathrm{B}}^2 + \int N^2 V_N \widetilde{\phi}_{\mathrm{B}}\Big(1-\frac{1}{2N}\Big) + N^2\langle -\Delta \Vphi, \Vphi\rangle \label{eq:U_q scalar} \\
 &\quad +  \int N \Big( \tfrac12 V_N\widetilde{\phi}_{\mathrm{B}}+ \tfrac12  V_N-\Delta\Vphi\Big)(y_1-y_2) a^*_{y_1} a^*_{y_2}  +\hc  \label{eq:U_q quad}\\
 & \quad + \ud \Gamma(-\Delta) - \cN_+ \int N V_N \Vphi + \cL_4 + \cE, \notag
\end{align}
\end{subequations}
with,
\begin{equation}
 \pm \cE\lesssim N^{-\alpha/2} (\mathcal N_++1) + N^{-\alpha/2} \mathcal L_4 +  N^{-1+\alpha/2}(\mathcal N_+ +1)^2.
\end{equation}
Adding to the scalar terms~\eqref{eq:U_q scalar} the contribution $\tfrac 12 (N-1)\widehat V(0)$ from $\cH_N$, we find using the definition of $\aV$ in~\eqref{eq:aV_N}, and the scattering equation~\eqref{eq:scattV},
\begin{align}
	&\frac{N-1}{2}\int N V_N + \frac{1}{2} \int N^2 V_N \widetilde{\phi}_{\mathrm{B}}^2 + \int N^2 V_N \widetilde{\phi}_{\mathrm{B}}\Big(1-\frac{1}{2N}\Big) + N^2\langle -\Delta \Vphi, \Vphi\rangle  \notag\\
 	&\quad =\begin{aligned}[t]
 	4\pi \aV (N-1)+ \Big(\frac{1}{2}-&\frac{1}{2N}\Big)\int N^2V_N (\Vphi - \phi_\mathrm{B}) \notag \\
 	&\qquad  + N^2 \Big \langle \Vphi, -\Delta \Vphi + \tfrac12 V_N \Vphi + \tfrac12 V_N\Big \rangle
 	        \end{aligned}
 	\notag \\
	&\quad = 4\pi \aV (N-1)  + \Big(\frac{1}{2}-\frac{1}{2N}\Big)\int N^2V_N ( \widetilde{\phi}_{\mathrm{B}}-\phi_\mathrm{B})  +\int N^2V_N  \widetilde{\phi}_{\mathrm{B}} ( \widetilde{\phi}_{\mathrm{B}}-\phi_\mathrm{B}).
\end{align}
Following the computation in the proof of \cite[Proposition 4]{HaiSchTri-22} and using that $\aV = \ao_{V} + \mathcal O(N^{-1})$ from Lemma \ref{lem:aM-difference} gives
\begin{equation}
 \frac{1}{2}\int N^2V_N (1+ \widetilde{\phi}_{\mathrm{B}})  ( \widetilde{\phi}_{\mathrm{B}}-\phi_\mathrm{B}) = \sum_{|p|\leq N^{\alpha}} \frac{(4\pi \ao_{V})^2}{p^2} + \mathcal O \mathcal (N^{2\alpha-1}).\label{eq:aV_N rest}
\end{equation}
The additional term with prefactor $1/(2N)$ is of order $N^{\alpha-1}$ by Lemma~\ref{lem:phi_V}.

We now turn to the quadratic terms~\eqref{eq:U_q quad} and show that they are approximated by $\widetilde \cL_2$. Again using the scattering equation~\eqref{eq:scattV}, we have
\begin{subequations}
 \begin{align}
&\int N \Big( \tfrac12 V_N\widetilde{\phi}_{\mathrm{B}}+ \tfrac12  V_N-\Delta\Vphi\Big)(y_1-y_2) a^*_{y_1} a^*_{y_2}\notag\\
 	&\quad  = \frac{1}{2} \sum_{|p|\leq N^\alpha }
 	N( \widehat V_N(p) + \widehat {V_N \phi}_\mathrm{B}(p))a_p^*a_{-p}^* \label{eq:L_2 rem-1} \\
 	%
 	&\qquad +  \frac{1}{2} \int N V_N (\Vphi- \phi_\mathrm{B})(y_1-y_2) a_{y_1}^* a_{y_2}^*  . \label{eq:L_2 rem-2}
\end{align}
\end{subequations}
The first term~\eqref{eq:L_2 rem-1} may be approximated by $\widetilde \cL_2$. Indeed, noting that $ \widehat {V}_N(0) +\widehat {V_N \phi}_\mathrm{B}(0) = 8\pi \aV$, we have 
\begin{align}
\big| N\widehat{ V_N(1+ \phi_\mathrm{B})}(p) - 8\pi \ao_{V} \big| 
	&\leq \int NV_N(1+ \phi_\mathrm{B})(y) |e^{-ip\cdot y} - 1| \nn \\
	&\quad +  |\aV - \ao_{V}| \lesssim |p| N^{-1}, \label{eq:V_N renorm error}
\end{align}
where we used that $|y| V \in L^1(\mathbb{R}^{3})$ and that $\|\phi_\mathrm{B}\|_\infty \lesssim 1$, from Lemma \ref{lem:phi_V}.
Then, estimating the operator using the $\ell^2$-norm of its kernel,
\begin{align}
 (\eqref{eq:L_2 rem-1} + \hc) - \widetilde \cL_2 &=  \frac{1}{2} \sum_{|p|\leq N^\alpha }\Big(N\widehat{ V_N(1+ \phi_\mathrm{B})}(p) - 8\pi \ao_{V} \Big) a_p^*a_{-p}^* +\hc \notag\\
 &=\mathcal{O}( N^{5\alpha/2-1} (\cN_+ + 1)).
\end{align}
To estimate the second term~\eqref{eq:L_2 rem-2}, we use the Cauchy-Schwarz inequality to get
\begin{align}
 &\bigg|\bigg\langle\Psi, \int N V_N (\Vphi- \phi_\mathrm{B})(y_1-y_2) a_{y_1}^* a_{y_2}^* \Psi \bigg\rangle\bigg| \notag \\
 & \leq \bigg(\int V_N(y_1-y_2)  \| a_{y_1} a_{y_2} \Psi\|^2\bigg)^{1/2} N\| (\Vphi- \phi_\mathrm{B}) V_N^{1/2}\|_2 \|\Psi\| \notag \\
 & \leq \delta \langle \Psi, \cL_4 \Psi \rangle + (4\delta)^{-1} \big(N\| (\Vphi- \phi_\mathrm{B}) V_N^{1/2}\|_2\big)^2 \|\Psi\|^2. \label{eq:L_2-L_4 bound}
\end{align}
With Lemma~\ref{lem:phi_V}, the choice $\delta=N^{-\alpha/2}$ and $\alpha\leq 1/6$ this gives
\begin{equation}
 \eqref{eq:L_2 rem-2} + \hc =\mathcal{O}\big(N^{-\alpha/2} \mathcal{L}_4 +  N^{-\alpha/2} \big).
\end{equation}
Of the terms coming from $\ud \Gamma(-\Delta) +\cL_2+\cL_4$, we still have the term $\cN_+ \int N V_N \Vphi$ that does not appear in the statement as such.
It can be written as
\begin{multline}
 -\cN_+\int N V_N \Vphi = \left(\widehat V(0)- 8\pi \ao_{V} \right)\cN_+ \\ + \left(8\pi(\aV - \ao_{V}) + \int N V_N (\phi_\mathrm{B}-\Vphi)\right) \cN_+.
\end{multline}
The first term now appears in the statement, and the second can be absorbed in the error by Lemma~\ref{lem:phi_V} and Lemma \ref{lem:aM-difference}.
For $\ud\Gamma(N \widehat V_N(-\ui \nabla))$, we easily find
\begin{equation}
 U_q^* \ud\Gamma(N \widehat V_N(-\ui \nabla))U_q^* = \ud\Gamma(N \widehat V_N(-\ui \nabla)) + \mathcal{O}(N^{-\alpha/2}\cN_+),
\end{equation}
using that $\|NV_N\ast N\Vphi\|_2\lesssim N^{-\alpha/2}$ by Lemma~\ref{lem:phi_V}. Writing  
\begin{equation}
 (N-\cN_+)_+\ud\Gamma( \widehat V_N(-\ui \nabla)) = \ud\Gamma(N \widehat V_N(-\ui \nabla)) + \mathcal O(\mathcal N_+^2/N),
\end{equation}
the previous estimates and Lemma~\ref{lem:N_T1} below show that
$ U_q^* (N-\cN_+)_+\ud\Gamma( \widehat V_N(-\ui \nabla))  U_q$ is equal to $\ud\Gamma(N \widehat V_N(-\ui \nabla))$ plus a term that can be absorbed in the error $\cE^{(U_q)}$.

The transformation of $\cL_0=-(\cN_++1)\cN_+/2N$ is absorbed in the error $\mathcal E_{\mathcal B_q}$ as well.
The remaining terms are treated by applying Lemmas~\ref{lem:L3_T1} and \ref{lem:Hint_T1} below, after which the proof is complete.
\end{proof}

\begin{lem}\label{lem:N_T1}
 For $\alpha>0$ we have uniformly in $|s|\leq 1$
 \begin{equation*}
  U_q(s)^*\cN_+ U_q(s) =\cN_+ + \mathcal{O}(N^{-\alpha/2} (\cN_++1)),
 \end{equation*}
and for $t\geq 0$
 \begin{equation*}
  U_q(s)^*\cN_+^t U_q(s) = \mathcal{O}(\cN_+^t+1).
 \end{equation*}
\end{lem}
\begin{proof}
 The first claim follows directly from the bound~\eqref{eq:B_q-N commute} and Gr\"onwall's inequality. For the generalization to arbitrary exponents see~\cite[Lem.3.1]{BreSch-19} for the case $t\in \N$, which implies the general case by operator monotonicity of $x\mapsto x^s$ for $0\leq s \leq 1$.
\end{proof}

\begin{lem}\label{lem:L4_T1}
In the sense of quadratic forms on $\mathscr{H}_+$, we have for $\alpha>0$
\begin{multline*}
U_q^* \mathcal L_4 U_q
	= \mathcal L_4 + \frac{1}{2} \int N V_N\widetilde{\phi}_{\mathrm{B}}(y_1-y_2) a^*_{y_1} a^*_{y_2} + \hc
	  + \frac{1}{2} \int N^2 V_N \widetilde{\phi}_{\mathrm{B}}^2 + \mathcal E_{\mathcal L_4}^{(U_q)}, 
\end{multline*}
with
\begin{equation*}
\pm \mathcal E_{\mathcal L_4}^{(U_q)}
	\lesssim N^{-\alpha/2} \mathcal N_+ + N^{-1} (\mathcal N_+ + 1 )^2 + N^{-\alpha/2} \mathcal L_4.
\end{equation*}
Moreover, we have the a priori bound
\begin{align}
	\label{eq:aprio_L4}
U_q^*(s) \mathcal L_4 U_q(s) \lesssim \mathcal L_4 + N^{-1}(\mathcal N_+ +1)^2 + N
\end{align}
uniformly in $s\in[0,1]$.
\end{lem}

\begin{proof}
From Duhamel's formula, we obtain
\begin{align}
U_q^* \mathcal L_4 U_q
	&= \mathcal L_4 +  \int_0^1 U_q^*(s) [\mathcal L_4, \cB_q]  U_q(s) \dd s.
\end{align}
Calculating the commutator and putting terms in normal order, we find (using that $V_N$ is even)
 \begin{align}
[\mathcal L_4, \cB_q]
	&=  \overbrace{\frac{1}{2}\int N V_N\widetilde{\phi}_{\mathrm{B}}(y_1-y_2) a^*_{y_1} a^*_{y_2}}^{X_1} + \hc  \notag \\ 
	&\quad  +  \underbrace{\int N V_N (y_1-y_2) \widetilde{\phi}_{\mathrm{B}}(y_3-y_2) a^*_{y_1} a^*_{y_2} a^*_{y_3} a_{y_1}}_{X_2}  + \hc \label{eq:L4Bq-2}
\end{align}
Using the Cauchy-Schwarz inequality as in~\eqref{eq:L_2-L_4 bound} gives
\begin{align}\label{eq:CS_L4}
 |\langle \Psi,X_1 \Psi \rangle|\leq \langle \Psi, \cL_4\Psi\rangle +
 N^2\| \widetilde \phi_B V_N^{1/2}\|_2^2 \|\Psi\|^2 \lesssim  \langle \Psi, \cL_4\Psi\rangle + N \|\Psi\|^2,
\end{align}
where we used Lemma~\ref{lem:phi_V} in the last step. With the same reasoning and another application of Cauchy-Schwarz, we obtain
\begin{align}
  &|\langle \Psi, X_2 \Psi \rangle|\nn \\
  &\leq \delta\langle \Psi, \cL_4\Psi\rangle + \frac{1}{\delta} \int N^2 V_N (y_1-y_2)  \bigg\|\int\widetilde{\phi}_{\mathrm{B}}(y_3-y_2)  a_{y_3}^* a_{y_1}\Psi \ud y_3\bigg\|^2  \ud y_1 \ud y_2\notag \\
 &\leq \delta\langle \Psi, \cL_4\Psi\rangle + \delta^{-1} N^2 \|\Vphi\|_2^2 \|V_N\|_1 \int \| a_{y_2} a_{y_1}\Psi\|^2 +\int \|a_{y_1}\Psi\|^2 \notag\\
 & \lesssim \delta\langle \Psi, \cL_4\Psi\rangle + \delta^{-1} N^{-1-\alpha} \|(\cN_+ +1)^2\Psi\|.
\label{eq:L_4-B_2 comm2}
\end{align}
Here, we used the estimates of Lemma~\ref{lem:phi_V} and the bound for $g\in L^2(\T^{3})$,
\begin{align}
& \int \Big \| \int g(y_1-y_2) a^*_{y_1}a_{y_2} \Psi \ud y_1\Big\|^2 \ud y_2\nn \\
 	&=  \int g(y_1-y_2) \overline{g(y_3-y_2)}\langle \Psi, a^*_{y_2}a_{y_1}a^*_{y_3}a_{y_2} \Psi\rangle \ud y_1\ud y_3 \ud y_2 \nn \\
	 &= \int |g(y_1-y_2)|^2 \langle \Psi ,a^*_{y_2}a_{y_2} \Psi\rangle \ud y_1 \ud y_2 \nn \\
	 &\quad + \int g(y_1-y_2) \overline{g(y_3-y_2)} \langle \Psi, a^*_{y_2}a^*_{y_3}a_{y_1}a_{y_2}\Psi\rangle \ud y_1\ud y_3 \ud y_2 \nn\\
 	&\leq \|g\|_2^2 \|(\cN_++1)\Psi\|^2,
 	\label{eq:est_CS}
\end{align}
where we used the Cauchy-Schwarz inequality to bound the second term, similarly as in (\ref{eq:B_q-N commute}). Now taking $\delta=1$ in (\ref{eq:L_4-B_2 comm2}) and using Lemma~\ref{lem:N_T1} gives
\begin{equation*}
 U_q^*(s) \mathcal L_4 U_q(s) \lesssim \cL_4 + N + N^{-1}(\cN_++1)^2 + \int_0^s U_q^*(\tau) \mathcal L_4  U_q(\tau) \dd \tau,
\end{equation*}
from which the a priori bound (\ref{eq:aprio_L4}) on $U_q^*(s) \mathcal L_4 U_q(s)$ follows by Gr\"onwall's inequality.
To get the more precise statement for $s=1$, we continue the Duhamel expansion to obtain
\begin{multline}
U_q^* \mathcal L_4 U_q
	= \mathcal L_4 + X_1 +\hc +   \int_0^1\int_0^t U_q^*(s) \big(\left[X_1,\cB_q \right] +\hc\big) U_q(s) \dd s \dd t  \\
	 + \int_0^1 U_q^*(s) (X_2 +\hc) U_q(s) \dd s.\label{eq:L_4-T_1 expand2}
\end{multline}
We may compute
\begin{align}
\left[X_1,\cB_q \right] = \frac12 \int N^2 V_N \Vphi^2 + \int N V_N \Vphi(y_1-y) N \Vphi(y_2-y) a^*_{y_1}a_{y_2}\ud y\ud y_1 \ud y_2.
\end{align}
The first term contributes to the main term in the statement (with an additional factor $1/2$ from the double integral). The second is estimated by $\|N V_N \Vphi \|_1 \|N \Vphi\|_2 \cN_+$, so its contribution can be absorbed in the error by Lemma~\ref{lem:N_T1} and Lemma~\ref{lem:phi_V}.
In order to control the final contribution to~\eqref{eq:L_4-T_1 expand2}, we write
\begin{subequations}
 \begin{align}
 & U_q^*(s) X_2  U_q(s)  \notag \\
 &=   \int N V_N (y_1-y_2) \widetilde{\phi}_{\mathrm{B}}(y_3-y_2) U_q^*(s) a^*_{y_1} a^*_{y_2}U_q(s) U_q^*(s) a^*_{y_3} a_{y_1}  U_q(s)  \notag \\
 &=   \int N V_N (y_1-y_2) \widetilde{\phi}_{\mathrm{B}}(y_3-y_2)  a^*_{y_1} a^*_{y_2} U_q^*(s) a^*_{y_3} a_{y_1}  U_q(s)   \label{eq:L_4-T_1 expand3}\\
 &\quad + \int_0^s
 \begin{aligned}[t]
  \int N V_N (y_1-y_2) \widetilde{\phi}_{\mathrm{B}}(y_3-y_2) U_q^*(t) &\big[ a^*_{y_1} a^*_{y_2}, \cB_q\big]U_q(t) \\
  &\times U_q^*(s) a^*_{y_3} a_{y_1}  U_q(s) \ud t .
 \end{aligned}\label{eq:L_4-T_1 expand4}
\end{align}
\end{subequations}
With the reasoning of~\eqref{eq:L_4-B_2 comm2} and $\delta=N^{-\alpha/2}$, we find
\begin{align}
 |\langle \Psi, \eqref{eq:L_4-T_1 expand3} \Psi\rangle| \lesssim N^{-\alpha/2} \langle\Psi,\cL_4 \Psi \rangle + N^{-1-\alpha/2} \| (\cN_+ +1)^2 U_q(s)\Psi\|^2.
\end{align}
To bound \eqref{eq:L_4-T_1 expand4}, we use the identity
\begin{align}
 [ a^*_{y_1} a^*_{y_2}, \cB_q\big]= - N\Vphi(y_1-y_2) - N a_{y_1}^* a(\widetilde{\varphi}_{B}(\cdot-y_2)) - N a_{y_2}^*a(\widetilde{\varphi}_{B}(\cdot-y_1)),
\end{align}
to obtain
\begin{align}
&\eqref{eq:L_4-T_1 expand4}
	=
-  \int_0^s \int N (V_N \Vphi)(y_1-y_2) U_q^*(s) a^*(N\widetilde{\phi}_{\mathrm{B}} (\cdot-y_2)) a_{y_1}  U_q(s) \ud t \nn \\
	&\qquad -  \int_0^s \int N V_N (y_1-y_2) U_q^*(t) \left( a_{y_1}^* a(N \widetilde{\varphi}_{B}(\cdot-y_2)) +  a_{y_2}^*a(N\widetilde{\varphi}_{B}(\cdot-y_1))\right)\nn \\
	&\qquad\qquad\qquad \times U_q(t)  U_q^*(s) a^*(\widetilde{\phi}_{\mathrm{B}} (\cdot-y_2)) a_{y_1}  U_q(s) \ud t .
\end{align}
Using the Cauchy-Schwarz inequality and the estimate (\ref{eq:est_CS}), we have
\begin{align}
 |\langle \Psi,\eqref{eq:L_4-T_1 expand4} \Psi\rangle| 
 	&\lesssim  \int_0^s N \|V_N \Vphi\|_1 \| N \widetilde{\varphi}_{B} \|_{2} \|(\cN_++1)^{1/2} U_q(s)\Psi\|^2 \ud t \notag\\
	&\quad + \int_0^s   \|V_N\|_{1} \| N \widetilde{\varphi}_{B} \|_{2}^2 \|\left(\mathcal N_++1\right) U_q(t) \Psi\|^2 \ud t\notag \\
	&\quad	+  \int_0^s \|V_N\|_{1} \| N \widetilde{\varphi}_{B} \|_{2}^2 \|(\cN_++1)U_q(s)\Psi\|^2 \ud t\notag \\
 	&\lesssim  N^{-\alpha/2} \mathcal N_+ + N^{-1} \left(\mathcal N+1\right)^2,
\end{align}
where we used that $U_q(s)^*\cN_+^k U_q(s) = \mathcal{O}(\cN_+^k+1)$, see Lemma~\ref{lem:N_T1}. This can be absorbed in $\mathcal E_{\mathcal L_4}^{(U_q)}$ and the proof is complete.
\end{proof}

\begin{lem}\label{lem:L2_T1}
In the sense of quadratic forms on $\mathscr{H}_+$, we have for $\alpha>0$
\begin{align*}
 U_q^* \mathcal{L}_2 U_q & = \frac{1}{2} \int N V_N({y_1}-{y_2}) a_{y_1} a_{y_2} + \hc \\
 &\qquad +  \left(1 - \frac{\mathcal N_+}{N} - \frac{1}{2N} \right) \int N^2 V_N \widetilde{\phi}_{\mathrm{B}} + \mathcal E_{\mathcal L_2}^{(U_q)},
\end{align*}
with, for all $0<\delta\leq N^{2\alpha} $,
\begin{align*}
\pm \mathcal E_{\mathcal L_2}^{(U_q)} \lesssim N^{-\alpha/2} (\mathcal N_++1) + \delta \mathcal L_4 + \delta^{-1} N^{-1}(\mathcal N_+ +1)^2.
\end{align*}
\end{lem}
\begin{proof}
Using that where $\mathcal N_+ \leq N$
\begin{align}
\bigg|  \sqrt{1-\frac{\cN_++1}{N}}\sqrt{1-\frac{\cN_+}{N}}  -  \left(1 - \frac{\mathcal N_+}{N} - \frac{1}{2N} \right)\bigg| \lesssim \frac{\mathcal N_+ + 1}{N^2},
\end{align}
we obtain with the argument of~\eqref{eq:L_2-L_4 bound}
\begin{equation}
\begin{aligned}
\mathcal L_2 &=  \frac{1}{2}  \left(1 - \frac{\mathcal N_+}{N} + \frac{1}{2N} \right) \int N V_N({y_1}-{y_2}) (a_{y_1} a_{y_2} + \hc)
+ \mathcal E, \\
	\pm \mathcal E &\lesssim N^{-2} \mathcal L_4 + N^{-1} (\mathcal N_++1)^2.
\end{aligned}
\end{equation}
Using the a priori estimate of Lemma~\ref{lem:L4_T1} and Lemma~\ref{lem:N_T1}, the error term $U_q^* \mathcal E U_q$ can be absorbed in $\mathcal E_{\mathcal L_2}^{(U_q)}$. Denoting $\chi(\mathcal N_+) := \left(1 - \frac{\mathcal N_+}{N} - \frac{1}{2N} \right)$ and using Duhamel's formula, we obtain
\begin{subequations}
\begin{align}
&\frac{1}{2} U_q^* \left\{ \chi(\mathcal N_+) \int  N V_N({y_1}-{y_2})  a_{y_1} a_{y_2} + \hc \right\} U_q\notag \\
 	& = \frac{1}{2} \chi(U_q^* \mathcal N_+U_q)  \int N V_N({y_1}-{y_2})  a_{y_1} a_{y_2} +\hc \label{eq:U_q-L_2-1}\\
 &\;\; + \frac12   \chi(U_q^*\mathcal N_+U_q)  \int_0^1U_q(s)^* \Big[\int  N V_N({y_1}-{y_2})  a_{y_1} a_{y_2} ,\cB_q \Big] U_q(s) \dd s + \hc \label{eq:U_q-L_2-2}
\end{align}
\end{subequations}
Let us start with the first term. Using the definition of $\chi(\mathcal N_+)$, we have by the argument of~\eqref{eq:L_2-L_4 bound}
\begin{align}
\pm\left(\eqref{eq:U_q-L_2-1} - \frac{1}{2} \int  N V_N({y_1}-{y_2})  a_{y_1} a_{y_2} +\hc\right) \lesssim \delta \mathcal L_4 + \delta^{-1} N^{-1}(\mathcal N_+ +1)^2,
\end{align}
which can be absorbed in $\mathcal E_{\mathcal L_2}^{(U_q)}$. Now for the commutator term, we have by~\eqref{eq:B_q-N commute}
\begin{align}
  \Big[\int  N V_N({y_1}-{y_2})  a_{y_1} a_{y_2} ,\cB_q \Big]
  	&= \int   N^2 V_N \widetilde{\phi}_{\mathrm{B}}+ \underbrace{\dd\Gamma( N \widehat{V}_N(-\ui\nabla) N\Vphih (-\ui\nabla) )}_{\lesssim \, \cN_+ N^{-2\alpha}},
\end{align}
where we used Lemma~\ref{lem:phi_V}.
From Lemma~\ref{lem:N_T1}, we get
\begin{equation}
 \chi(U_q^*\mathcal N_+U_q) 
 =  \chi(\mathcal N_+) + \mathcal O(N^{-1-\alpha/2} (\mathcal N_++1)),
\end{equation}
and together this yields
\begin{align}
\eqref{eq:U_q-L_2-2} =  \chi(\mathcal N_+) \int \,  N^2 V_N \widetilde{\phi}_{\mathrm{B}} + \mathcal O\left( N^{-\alpha/2}(\cN_+ + 1)+N^{-1-2\alpha} \cN_+^2\right).
\end{align}
This finishes the proof of Lemma \ref{lem:L2_T1}.
\end{proof}

\begin{lem}\label{lem:kin_T1}
In the sense of quadratic forms on $\mathscr{H}_+$, we have for $\alpha>0$
\begin{multline*}
U_q^* \dd\Gamma(-\Delta) U_q
	= \dd\Gamma(-\Delta) + N\int (-\Delta \widetilde{\phi}_{\mathrm{B}})(y_1-y_2) a^*_{y_1}a^*_{y_2} + \hc \notag \\
	  + N^2 \langle -\Delta \widetilde{\phi}_{\mathrm{B}}, \widetilde{\phi}_{\mathrm{B}} \rangle + \mathcal{O}(N^{-2\alpha} \cN_+). \label{eq:lem:kin_T1}
\end{multline*}
\end{lem}

\begin{proof}
From Duhamel's formula, we obtain
\begin{align}
U_q^* \dd \Gamma(-\Delta) U_q
	&= \dd \Gamma(-\Delta) + [\dd \Gamma(-\Delta), \cB_q]  \\
	&\quad +  \int_0^1 \int_0^t U_q^*(s) [[\dd \Gamma(-\Delta), \cB_q] ,\cB_q] U_q(s) \dd s \dd t. \label{eq:duh_kin_T1}
\end{align}
Moreover,
\begin{align}
 \Big[\dd \Gamma(-\Delta), \cB_q\Big]&=
 N\int (-\Delta \widetilde{\phi}_{\mathrm{B}})(y_1-y_2) a^*_{y_1}a^*_{y_2} +\hc,
\end{align}
and with Lemma~\ref{lem:phi_V}
\begin{align}
\Big[[\dd \Gamma(-\Delta), \cB_q], \cB_q\Big]&= 2 N^2 \langle -\Delta \widetilde{\phi}_{\mathrm{B}}, \widetilde{\phi}_{\mathrm{B}} \rangle
 + 2 \underbrace{\dd \Gamma\big(N^2 \Vphih(-\ui\nabla) (\widehat{-\Delta \widetilde{\phi}_{\mathrm{B}})}(-\ui\nabla)\big)}_{=\mathcal{O}(\cN_+ N^{-2\alpha})}, \label{eq:kin_B1B1}
\end{align}
where we used that $ \|\Vphih\|_\infty \lesssim N^{-1-2\alpha}$ and $\|p^2\Vphih\|_\infty \lesssim N^{-1}$ from Lemma \ref{lem:phi_V}. Plugging this into (\ref{eq:duh_kin_T1}) and using that $U_q^* (s) \mathcal N_+ U_q(s) =\mathcal{O}(\mathcal N_+ + 1)$ proves the claim (note that we obtain a factor $1/2$ from the double integral in~\eqref{eq:duh_kin_T1}).
\end{proof}

\begin{lem} \label{lem:Hint_T1}
As quadratic forms on $\mathscr{H}_+$, we have for $\alpha>0$
\begin{align*}
U_q^* \mathcal H_{\mathrm{I}} U_q = -\Delta_x + \sqrt{\left(1-\frac{U_q^*\cN_+U_q}{N}\right)_+} a(\sqrt{N} W_{N,x}) + \hc + \ud\Gamma(W_{N,x}) + \mathcal E_{\mathrm{I}}^{(U_q)}
\end{align*}
where
\begin{align*}
\pm \mathcal E_{\mathrm{I}}^{(U_q)} \lesssim N^{-\alpha/2}\left( N^{-1/2}\ud \Gamma(W_{N,x}) +  \left(\mathcal N_++1\right)\right).
\end{align*}
\end{lem}
\begin{proof}
First note that $U_q^* (-\Delta_x) U_q = -\Delta_x$, since $U_q$ does not depend on the variable $x$. Let us now turn to $Q_1$. From Duhamel's formula we obtain
\begin{multline}
U_q^* Q_1 U_q 	= \sqrt{\Big(1-\frac{U_q^*\cN_+U_q}{N}\Big)_+} \bigg(  a(\sqrt{N} W_{N,x}) \\
 + \int_0^1 U_q^*(s)[a(\sqrt{N} W_{N,x}), \cB_q] U_q(s) \dd s  \bigg) +\hc
\end{multline}
We have (using that $V$, $\phi_B$ are even)
\begin{align}
 \big[a(\sqrt{N} W_{N, x}), \cB_q \big]
 &=\int \sqrt N W_N(x-y) \big[ a_y, \cB_q \big]\ud y \\
 &=N^{3/2} a^*(\Vphi\ast W_{N,x}) =\mathcal{O}\Big( \cN_+^{1/2} N^{3/2} \|\widetilde{\phi}_{\mathrm{B}}\|_2 \|W_N\|_1\Big).
 \nn
\end{align}
With $U_q^*(s) \mathcal N_+ U_q(s) \lesssim (\mathcal N_+ + 1)$ and Lemma~\ref{lem:phi_V} we thus obtain
\begin{align}
U_q^* Q_1 U_q =  \sqrt{\left(1-\frac{U_q^*\cN_+U_q}{N}\right)_+} a(\sqrt{N} W_{N,x}) + \hc + \mathcal{O}( \cN_+^{1/2} N^{-\alpha/2}).
\end{align}
Let us now turn to $\ud \Gamma(W_{N,x})$. Using the Duhamel formula again, we obtain for $t\in [-1,1]$
\begin{align}
	\label{eq:duh_Q2_T1}
U_q^*(t) \ud \Gamma(W_{N,x}) U_q(t) = \ud \Gamma(W_{N,x}) + \int_0^t U_q^*(s)[\ud \Gamma(W_{N,x}), \cB_q] U_q(s) \dd s.
\end{align}
Then, using the Cauchy-Schwarz inequality as in~\eqref{eq:L_2-L_4 bound},
\begin{align}
 \big[\ud \Gamma(W_{N,x}), \cB_q \big] &=  N \int W_N(x-y_1) \widetilde{\phi}_{\mathrm{B}}(y_1-y_2) a^*_{y_1} a^*_{y_2}  + \hc \notag \\
 	&=\mathcal{O}\Big( N^{-\alpha/2-1/2} \ud\Gamma(W_{N,x}) + N^{-\alpha/2}(\mathcal N_+ + 1)\Big).
\end{align}
From (\ref{eq:duh_Q2_T1}), the above estimate, Lemma~\ref{lem:N_T1} and the Grönwall lemma, we first deduce the a priori estimate $U_q^*(t) \ud \Gamma(W_{N,x}) U_q(t) \lesssim (\ud \Gamma(W_{N,x}) + \mathcal N_+ + 1)$, for all $t\in [-1,1]$. Taking $t=1$ in \eqref{eq:duh_Q2_T1}, we then obtain
\begin{align}\label{eq:Q_2-T_1-aprio}
U_q^* \ud \Gamma(W_{N,x}) U_q &= \ud \Gamma(W_{N,x}) + \mathcal{O}\Big( N^{-\alpha/2-1/2} \ud\Gamma(W_{N,x}) + N^{-\alpha/2}(\mathcal N_+ + 1)\Big).
\end{align}
This proves the claim.
\end{proof}

\begin{lem}\label{lem:L3_T1}
For $\alpha>0$, we have as quadratic forms on $\mathscr{H}_+$
\begin{equation*}
 U_q^* \mathcal L_3 U_q
	=  \frac{1}{2}\int \sqrt{N} V_N(y_1-y_2) a^*_{y_1} a_{y_2} a_{y_1} + \hc + \mathcal E_{\mathcal L_3}^{(U_q)},
\end{equation*}
with, for all $N^{-\alpha/2}\leq \delta<1$,
\begin{equation*}
\pm \mathcal E_{\mathcal L_3}^{(U_q)}
	\lesssim N^{-\alpha/2} \mathcal N_+ + \delta^{-1} N^{-1} (\mathcal N_+ + 1 )^2 + \delta \mathcal L_4.
\end{equation*}
\end{lem}
\begin{proof}
Let us denote $\chi(\mathcal N_+) = \sqrt{(1-\frac{\cN_+}{N})_+}$. From Duhamel's formula, we obtain
\begin{align}
&U_q^* \mathcal L_3 U_q
	= \chi(U_q^* \mathcal N_+ U_q) \int \sqrt{N} V_N(y_1-y_2) a^*_{y_1} a_{y_2} a_{y_1} + \hc \\
	&\quad  +  \chi(U_q^* \mathcal N_+ U_q)  \int_0^1 U_q^*(s) \int \sqrt{N} V_N(y_1-y_2) [ a^*_{y_1} a_{y_2} a_{y_1} , \cB_q]  U_q(s) \dd s + \hc \notag
\end{align}
To obtain $\mathcal L_3^{(U_q)}$, we need to extract the main contribution from the first term.
With $|\sqrt{1-x} - 1| \leq  x$ and $U_q^* \mathcal N_+ U_q \lesssim (\mathcal N_++1)$, we find
\begin{equation}
 |1-\chi(U_q^* \mathcal N_+ U_q)|\lesssim N^{-1}(\mathcal N_++1).
\end{equation}
This allows us to bound for $\delta>0$, as in~\eqref{eq:L_2-L_4 bound},
\begin{align}
\pm &\left\{ \left(1 - \chi(U_q^* \mathcal N_+ U_q)\right) \int \sqrt{N} V_N(y_1-y_2) a^*_{y_1} a_{y_2} a_{y_1} + \hc \right\} \nn \\
	&\lesssim \delta \mathcal L_4 + \delta^{-1} \left(1 - \chi(U_q^* \mathcal N_+ U_q)\right)^2 \|V\|_{1} \mathcal N_+ \nn \\
	&\lesssim \delta \mathcal L_4 + \delta^{-1} N^{-1} \mathcal N_+^2. \label{eq:bound_L3}
\end{align}
To bound the commutator terms, one proceeds as in the proof of Lemma~\ref{lem:L4_T1}.
That is, one calculates the commutator, noticing that the terms linear in $a^*$, $a$ vanish on $\mathscr{H}_+$, and then inserts an additional factor $U^*_q(s)U_q(s)$ after a pair of creation/annihilation operators to apply Duhamel's formula once more. The details, with slightly different conventions, can be found in~\cite[Lem. 6]{HaiSchTri-22}.
\end{proof}

\subsection{High-energy boson-impurity interactions: a Weyl transformation}\label{sect:Weyl}

As a next step we apply an $x$-dependent Weyl transformation to the Hamiltonian. Its role is very similar to the one of $U_q$ used in the previous section. This transformation corrects the scalar $N \widehat W(0)$ from the excitation Hamiltonian to make the scattering length $\ao_W$ appear. Additionally, it renormalizes the high momentum part of the interaction term $Q_1$ by transforming it into a less singular interaction. 

Let $\Wphi$ be the solution to the approximate impurity-boson scattering equation~\eqref{eq:scattW} cut-off at momentum $N^\alpha$, for the same $\alpha\geq 0$ as in the previous section (cf. Equation~\eqref{eq:Wphi}).
Define the unitary
\begin{equation}
 U_W(t)=\exp\Big(t\int  \sqrt{N} \Wphi(x-y)a_y^* \dd y- \text{h.c}\Big),\qquad U_W=U_W(1),
\end{equation}
whose generator is clearly (essentially) anti-self-adjoint on $D(\cN_+^{1/2})$.
This transformation is very similar to the transformation $U_\kappa^\Lambda$ from Section~\ref{sect:scattering} with $\kappa=N^\alpha$, except that the definition of $\Wphi$ also takes the interaction term $\ud \Gamma(W_{N,x})$ into account.
Note that the generator is linear in $a$, $a^*$, but if we were to include creation and annihilation operators $b^*$, $b$, for the impurities, the kernel would be $\sqrt N\Wphi(x-y)b_x^*b_x (a_y^*-a_y)$, which is cubic in total. The transformation satisfies the identities
\begin{equation}\label{eq:U_Wa_y}
  U_W^* a_y U_W= a_y + \sqrt{N} \Wphi(y-x),
\end{equation}
and
\begin{align}
 U_W^* \cN_+ U_W &= \cN_+ + \big(a(\sqrt{N} \widetilde{\phi}_{\mathrm{I},x}) + \hc\big) + \|\sqrt N \Wphi\|^2_2 \notag\\
 &= \cN_+ + \mathcal{O}\big(N^{-\alpha/2}(\cN_+ +1)\big). \label{eq:N-U_W}
\end{align}
The result of transforming the Hamiltonian with $U_W$ is given below.
\begin{prop}\label{prop:U_W}
 Let $0<\alpha\leq 1/6$, $\ao_W$ be the scattering length of $W$ defined by~\eqref{eq:scatt_length-intro}, and recall the Hamiltonian $U_q^*\mathcal{H}_N U_q$ from Proposition~\ref{prop:U_q}. Then, we have as quadratic forms on $\mathscr{H}_+$
 \begin{align*}
  U_W^* U_q^* \mathcal{H}_N U_q U_W
  & = 4\pi  \ao_V (N-1)  + 8\pi  \ao_W \sqrt N + e_N^{(U_W)} \\
  &\quad  + \mathcal H_{\mathrm{BB}}^{(U_W)} + \mathcal H_{\mathrm{IB}}^{(U_W)} + \mathcal{E}^{(U_W)}, \notag
  \end{align*}
  where
  \begin{align}
  \mathcal H_{\mathrm{BB}}^{(U_W)}
  	&=  \dd \Gamma(-\Delta + 2\widehat{V}(0) - 8\pi \ao_{V})  +\widetilde{\mathcal{L}}_2 + \mathcal{L}_3^{U_q} +
  \mathcal{L}_4, \notag 
  \end{align}
  and (with the notation $f_x(y)=f(x-y)$)
  \begin{align*}
   \mathcal H_{\mathrm{IB}}^{(U_W)}
	&=  -\Delta_x +\Big(a^*(\widetilde W_{N,x} ) +\hc \Big) +\Big(a^*\big(\sqrt N \ui\nabla \Wphi)_x\big)^2 +\hc \Big) + \ud \Gamma(W_{N,x}) \nn \\
  &\quad   -2 a^*\big(\sqrt N  (\ui\nabla \Wphi)_x\big)\ui\nabla_x  + \hc + 2 a^*\big(\sqrt N (\ui\nabla \Wphi)_x\big)a\big(\sqrt N (\ui\nabla \Wphi)_x\big), \notag \\
  \widehat{\widetilde W}_N(p) &= 8\pi \ao_W \1_{0<|p|\leq  N^{\alpha}},\\
  e_N^{(U_W)} &= 4 \pi N (\aV - \ao_V) + 8 \pi \sqrt{N} (\aW - \ao_W) + \hspace{-6pt}\sum_{p\in 2\pi\Z^3 \atop 0<|p|\leq N^{\alpha}} \hspace{-6pt}\frac{(4\pi  \ao_{V})^2+2(4\pi \ao_W)^2}{p^2}, \notag
  \end{align*}
and the error term $\mathcal{E}^{(U_W)}$ satisfies
\begin{align*}
	 \pm \mathcal{E}^{(U_W)} &\lesssim N^{-\alpha/2}(1+\cN_+ +\mathcal{L}_4+\ud \Gamma(W_{N,x}))  + N^{-1+\alpha/2}\cN_+^2 + N^{-1/4} \ud\Gamma(|\ui \nabla|).
\end{align*}
\end{prop}

\begin{proof}
Let us rewrite the result of Proposition~\ref{prop:U_q} as
\begin{align}
 U_q^* \cH_N U_q
 	&= 4 \pi \ao_V (N-1) +e_N^{(U_q)}+   \mathcal{H}_{\mathrm{IB}} +  \mathcal{H}_{\mathrm{BB}}  +  \cE^{(U_q)},
\end{align}
 with
  \begin{align}
     \mathcal{H}_{\mathrm{IB}} &= \sqrt N \widehat W(0) -\Delta_x + \ud \Gamma(-\Delta) +Q_1^{U_q}+\ud \Gamma(W_{N,x}), \notag \\
    \mathcal{H}_{\mathrm{BB}} &= \ud\Gamma(N\widehat{V}_N(-\ui\nabla) + \widehat V(0)-8\pi \ao_{V}) + \widetilde{\mathcal{L}}_2+\mathcal{L}_3^{U_q} + \mathcal{L}_4.
 \end{align}
In Lemma~\ref{lem:Weyl-BB} below, we prove that $\mathcal{H}_{\mathrm{BB}}$ remains essentially unchanged by the transformation $U_W$, except that we will be able to simplify $ N\widehat{V}_N(-\ui\nabla)$ to $\widehat V(0)$.  The relevant changes thus occur in the part of the Hamiltonian describing non-interacting bosons coupled to the impurity, that is in $\mathcal{H}_{\mathrm{IB}}$. Let us therefore consider $U_W^*\mathcal{H}_{\mathrm{IB}}U_W$.

\noindent\textit{Estimating $U_W^*  Q_1^{U_q} U_W$}. The action of $U_W$ on $ Q_1^{U_q}$ is easily evaluated using the identity~\eqref{eq:U_Wa_y}:
\begin{align}
  U_W^* Q_1^{U_q} U_W &= \sqrt{\left(1-\frac{U_W^* U_q^*\cN_+U_q U_W}{N}\right)_+} \Big( a(\sqrt N W_{N,x}) + N \langle W_N, \Wphi \rangle \Big) +\hc\label{eq:Weyl-Q1-lin}
  \end{align}
 By Equation~\eqref{eq:N-U_W} and Lemma~\ref{lem:N_T1}, we have for all $\Psi$
 \begin{multline}
\bigg\| \bigg(1-\sqrt{\Big(1-\frac{U_W^* U_q^*\cN_+U_q U_W}{N}\Big)_+}\bigg) \Psi\bigg\| \\
\lesssim \| N^{-1} U_W^* U_X^*\cN_+U_q U_W \Psi\| \lesssim \| N^{-1} (\cN_+ +1)\Psi\|.
 \end{multline}
By the Cauchy-Schwarz inequality,
\begin{equation}
 2  \Big|\big\langle \Phi, a(\sqrt N W_{N,x}) \Psi\big\rangle\Big|  \leq \bigg(\underbrace{\int W_{N}(x-y)\| a_y \Psi\|^2}_{=\langle \Psi,\ud \Gamma(W_{N,x}) \Psi\rangle}\bigg)^{1/2}  \bigg(\underbrace{\int N W_{N}}_{=\sqrt{N} \widehat W(0)}\bigg)^{1/2}\|\Phi\|. \label{eq:a-Q_2 bound}
\end{equation}
Thus Lemma~\ref{lem:phi_W} with~\eqref{eq:AB Young} and using that $N \langle W_N, \Wphi \rangle = \mathcal O(N^{1/2})$ by Lemma \ref{lem:phi_W} yield
\begin{align}
 U_W^* Q_1^{U_q} U_W &= 2N \Re \langle W_N, \Wphi \rangle +   (a(\sqrt N W_{N,x})+\hc) \notag\\
 &\quad + \mathcal{O}\big(N^{-1/2}\ud \Gamma(W_{N,x}) + N^{-1/2}(\cN_++1) + N^{-1}(\cN_++1)^2\big), \label{eq:Weyl Q_1}
\end{align}
where the last line can be absorbed in $\mathcal{E}^{(U_W)} $.

\noindent\textit{Estimating $U_W^* (\dd\Gamma(-\Delta)+ \ud \Gamma(W_{N,x})) U_W$}.  With the shift property~\eqref{eq:U_Wa_y}, we find
\begin{subequations}
  \begin{align}
   U_W^* \ud \Gamma(W_{N,x}) U_W&=  \ud \Gamma(W_{N,x}) +\big(a\big( \sqrt N (W_N \Wphi)_x\big) +\hc\big) + N \langle \Wphi, W_N \Wphi \rangle, \label{eq:Weyl-Q2} \\
U_W^* \dd\Gamma(- \Delta) U_W&= \ud \Gamma(-\Delta) -\big(a\big(\sqrt{N} (\Delta \Wphi)_x\big) +\hc\big) -  N \langle \Wphi, \Delta \Wphi\rangle \label{eq:Weyl-DeltaN}.
 \end{align}
\end{subequations}
We keep these terms to combine them with those coming from $Q_1$ and $-\Delta_x$ later. We also note that, applying the Cauchy-Schwarz inequality as in~\eqref{eq:a-Q_2 bound}, we have the a priori estimate
\begin{align}
U_W^* \ud \Gamma(W_{N,x}) U_W
	&\lesssim \ud \Gamma(W_{N,x}) + N \langle \Wphi, W_N \Wphi \rangle \lesssim \ud \Gamma(W_{N,x}) + \sqrt N. \label{eq:apriori_Q2_T2}
\end{align}

\noindent\textit{Estimating $U_W^* \Delta_x U_W$}.
As the generator of the transformation $U_W$ depends on the position $x$ of the particle, $-\Delta_x$ also transforms non-trivially.
Using that $\nabla_x \widetilde\phi_{\textrm{I},x}=- (\nabla \Wphi)_x$, the shift property~\eqref{eq:U_Wa_y}, and that $\Wphi$ is real, we find
\begin{align}
 U_W^* (-\ui \nabla_x) U_W &=-\ui \nabla_x + \int_0^1 \ud s\, U_W^*(s) \big[-\ui \nabla_x, a^*(\sqrt{N} \Wphix)-\hc\big] U_W(s)\notag  \\
 &=-\ui \nabla_x - \big(  a^*( \sqrt N  (-\ui \nabla \Wphi)_x)+\hc\big), \label{eq:Weyl-x-shift}
\end{align}
and thus
\begin{align}
 &U_W^* (-\Delta_x) U_W  = \Big(-\ui \nabla_x -\big( a^*\big(\sqrt{N}  (-\ui \nabla \Wphi)_x\big)+\hc\big)\Big)^2 \notag \\
 &=-\Delta -a^*\big(\sqrt N (\ui \nabla \Wphi)_x\big)2\ui \nabla_x  +\hc  \notag\\
  &\quad + 2 a^*\big(\sqrt N (-\ui \nabla \Wphi)_x\big) a\big(\sqrt N (-\ui \nabla \Wphi)_x\big)
  + a^*\big(\sqrt N (-\ui \nabla \Wphi)_x\big)^2 +\hc \notag \\
  &\quad  + a^*\big(\sqrt{N} (-\Delta \Wphi)_x\big) +\hc + N\langle \Wphi, (-\Delta) \Wphi).\label{eq:Weyl-Delta_y}
 \end{align}

\noindent\textit{Combining the scalar and linear terms}.
We now combine the scalars and terms linear in $a, a^*$ from these equations, which add up in a non-trivial way. The remaining errors will be estimated at the end.
Together with $\sqrt{N}\widehat W(0)$, the scalar terms from~\eqref{eq:Weyl Q_1},~\eqref{eq:Weyl-Q2},~\eqref{eq:Weyl-DeltaN} add up to
 \begin{align}
  &\sqrt{N}\widehat W(0)+2N  \Re\langle W_N, \Wphi\rangle + N \langle (-2\Delta + W_N) \Wphi,  \Wphi \rangle.
    %
 \end{align}
Similarly as in~\eqref{eq:aV_N rest}, using the scattering equation~\eqref{eq:scattW}, the definition of $\aW$ in \eqref{eq:aW_N} and that $\aW = \ao_W + \mathcal O(N^{-1/2})$ from Lemma \ref{lem:aM-difference}, we obtain
 \begin{align}
  &\sqrt{N}  8\pi \aW + \langle W_N(1+\Wphi),\Wphi-\phi_\mathrm{I}\rangle \notag\\
  \quad&= \sqrt{N}  8\pi \aW + 2\sum_{|p|<N^{\alpha}} \frac{(4\pi  \ao_W)^2}{p^2} + \mathcal{O}(N^{2\alpha-1/2}).
 \end{align}

Gathering the terms linear in $a$, $a^*$ from~\eqref{eq:Weyl Q_1},~\eqref{eq:Weyl-Q2},~\eqref{eq:Weyl-Delta_y} and~\eqref{eq:Weyl-DeltaN} yields
 \begin{align}
  a\Big( \sqrt{N} \big(\underbrace{ (-2 \Delta+ W_N) \Wphi  + W_N }_{=:f}\big)_x\Big)+\hc
 \end{align}
Using the scattering equation~\eqref{eq:scattW} for $\phi_{\textrm{I}}$, we see that the Fourier transform of the function $f$ above equals
\begin{align}
 \widehat f (p)&= - p^2 \widehat \phi_\mathrm{I}(p)\1_{0<|p|\leq N^{\alpha}} -   \widehat{W_N(\phi_\mathrm{I}- \Wphi)} (p) \notag \\
 &=\widehat{W_N(1+\phi_{\textrm{I}})}(p) \1_{0<|p|\leq N^{\alpha}} -  \widehat{W_N(\phi_\mathrm{I}- \Wphi)} (p).
\end{align}
By the same reasoning as~\eqref{eq:V_N renorm error} and using that $\aW = \ao_W + \mathcal O(N^{-1/2})$ from Lemma \ref{lem:aM-difference}, we have with $\alpha\leq 1/6$
\begin{align}
 %
 \|(\widehat{W_N(1+\phi_{\textrm{I}})}(p) -8\pi \ao_W)\1_{|p|\leq N^{\alpha}}\|_2 &\lesssim N^{5\alpha/2-1} + N^{3\alpha/2 -1/2} \leq N^{-\alpha/2}.\label{eq:W_N renorm error}
\end{align}
Moreover, using the Cauchy-Schwarz inequality as in~\eqref{eq:a-Q_2 bound} and Lemma~\ref{lem:phi_W}, we have
\begin{align}
 \pm \big (a(\sqrt N W_{N,x}(\phi_\mathrm{I}- \Wphi)_x) + \hc) \lesssim \delta \ud \Gamma(W_{N,x}) +\delta^{-1} N^{-1/2+2\alpha}.
\end{align}
Choosing $\delta=N^{-\alpha/2}$ and using that $\alpha<1/4$, this yields
\begin{equation}
 a(\sqrt N f_x)+\text{h.c} = a(\widetilde W_{N,x}) +\text{h.c} + \mathcal{O}(N^{-\alpha/2} (1+\ud \Gamma(W_{N,x}))),
\end{equation}
with the modified interaction $\widetilde W_N$.
This proves that
\begin{multline}
 U_W^* \cH_{\mathrm{IB}} U_W =8\pi  \aW \sqrt{N}+ 2\sum_{|p|\leq N^\alpha} \frac{(4\pi  \ao_W)^2 }{p^2}+ \ud \Gamma(-\Delta) + \cH_{\mathrm{I}}^{(U_W)}  \\
  + \mathcal{O}\big(N^{-\alpha/2} (\ud \Gamma(W_{N,x})+ \cN_+ +1) + N^{-1}(\cN_++1)^2 \big).
\end{multline}

\noindent\textit{Estimating $U_W^* \cE^{(U_q)}  U_W $}.
It remains to track the behavior of the error terms $\mathcal{E}^{(U_q)}$ under $U_W$.
Using the a priori estimates from Lemma \ref{lem:Weyl-BB} together with \eqref{eq:apriori_Q2_T2} (note that $\ud \Gamma(W_{N,x})$ appears with an additional factor $N^{-1/2}$ in the error $\cE^{(U_q)}$), we obtain that $U_W^* \mathcal{E}^{(U_q)} U_W $ can be absorbed in $\mathcal{E}^{(U_W)}$.
This completes the proof.
\end{proof}

\begin{lem}\label{lem:Weyl-BB}
 For $\alpha>0$, we have as quadratic forms on $\mathscr{H}_+$
 \begin{equation*}
  U_W^*\big(\widetilde{\mathcal{L}}_2+\ud\Gamma(N\widehat{V}_N(-\ui\nabla)) +\mathcal{L}_3^{U_q} + \mathcal{L}_4\big)U_W= \widetilde{\mathcal{L}}_2 +\widehat{V}(0)\cN_++\mathcal{L}_3^{U_q} + \mathcal{L}_4 +\mathcal{E}_\mathrm{BB}^{(U_W)},
 \end{equation*}
with
\begin{equation*}
 \pm\mathcal{E}_\mathrm{BB}^{(U_W)} \lesssim N^{-\alpha/2}(1+\cN_++\mathcal{L}_4) +N^{-1/3+\alpha/6}\ud\Gamma(|\ui \nabla|).
 \end{equation*}
Moreover, we have for $t \geq 0$
 \begin{align*}
  U_W^* (\mathcal N_+ + 1)^t U_W &\lesssim (\mathcal N_+ +1)^t, 
\end{align*}
and
\begin{align*}
 U_W^* \mathcal L_4 U_W &\lesssim \mathcal L_4 + N^{-1/3+\alpha/6} \ud\Gamma(|\ui \nabla|) + N^{-\alpha/2}(\cN_++1). 
\end{align*}
\end{lem}

\begin{proof}
\textit{Estimating $U_W^*(\mathcal N_+ + 1)^k U_W$.} For integer $t$, this bound follows directly from the identity~\eqref{eq:N-U_W} together with~\eqref{eq:AB Young}. The general case then follows by interpolation~\cite{ReeSim2}.

\noindent\textit{Estimating $ U_W^* \widetilde{\mathcal{L}}_2 U_W $.} For the quadratic term $\widetilde{\mathcal{L}}_2$ we calculate using the shift property~\eqref{eq:U_Wa_y}
\begin{align}
 U_W^* \widetilde{\mathcal{L}}_2 U_W
 & =   \widetilde{\mathcal{L}}_2 +2  \int \sqrt N \widetilde{V}_N(y_1-y_2) \Wphi(x-y_1) a_{y_2} +\hc\notag \\
 &\quad  + \underbrace{\int N \widetilde{V}_N(y_1-y_2)\Wphi(y_1)\Wphi(y_2)}_{\lesssim \|\widehat{\widetilde V}_N\|_\infty \|\sqrt N\Wphi\|_2^2 \lesssim N^{-\alpha}} .
\end{align}
With $\| \widetilde{V}_N\ast \sqrt N \Wphi\|_2 \lesssim \|\widehat{\widetilde{V}}_N\|_\infty \|\sqrt N \Wphi\|_2 \lesssim N^{-\alpha/2}$ from Lemma~\ref{lem:phi_W}, we thus find
\begin{equation}
 U_W^* \widetilde{\mathcal{L}}_2 U_W=  \widetilde{\mathcal{L}}_2 + \mathcal{O}\big(N^{-\alpha/2}(\cN_++1)\big).
\end{equation}

\noindent\textit{Estimating $ U_W^* \ud\Gamma(N\widehat{V}_N(-\ui\nabla)) U_W $.} Similarly as above, we find
\begin{align}
 U_W^*\dd \Gamma( N \widehat{V}_N(-\ui\nabla))U_W&= \dd \Gamma(N\widehat {V}_N(-\ui\nabla)) +\Big( a\big( \underbrace{N^{3/2} (\widehat {V}_N(-\ui\nabla)\Wphi)_x}_{\|\cdot\|_2 \lesssim N^{-\alpha/2}}\big) +\hc\Big) \notag \\
 &\qquad + \underbrace{N^2\langle \widetilde{\phi}_{\textrm{I}} , \widehat{V}_N(-\ui\nabla) \widetilde{\phi}_{\textrm{I}}\rangle}_{\lesssim N^{-\alpha}},
\end{align}
and 
\begin{multline}
\dd \Gamma(N\widehat{V}_N(-\ui\nabla))\\ = \sum_{0\neq p\in 2\pi\Z^3} \widehat{V} (p/N)a^*_p a_p = \widehat{V}(0) \mathcal N_+ + \underbrace{\sum_p \left(\widehat{V} (p/N) - \widehat{V}(0)\right)  a^*_p a_p}_{\lesssim N^{-1} \dd \Gamma(|\ui \nabla|)},
\end{multline}
where it is used that $|y| V \in L^1(\mathbb{R}^{3})$, as in (\ref{eq:V_N renorm error}).

\noindent\textit{Estimating $ U_W^* \widetilde{\mathcal{L}}_4 U_W $.} For the quartic term we have
\begin{subequations}
 \begin{align}
 & U_W^*\mathcal{L}_4 U_W - \mathcal{L}_4=\frac{1}{2} \int V_N(y_1-y_2) U_W^* a^*_{y_1}a^*_{y_2} a_{y_1} a_{y_2}U_W - \mathcal{L}_4 \nn\\
 &\quad=   \int    \sqrt NV_N(y_1-y_2) \Wphi(x-y_1) a^*_{y_1}a^*_{y_2} a_{y_1} +\hc\label{eq:Weyl-cubic-error} \\
 &\qquad +\frac12 \int   NV_N(y_1-y_2) \Wphi(x-y_1) \Wphi(x-y_2) a^*_{y_1}a^*_{y_2} +\hc \label{eq:Weyl-L_4-quad} \\
 &\qquad +\int    NV_N(y_1-y_2) \Wphi(x-y_1) \Wphi(x-y_2) a^*_{y_1}a_{y_2}
 \label{eq:Weyl-L_4-qmixed}\\
 &\qquad + \ud \Gamma(NV_N \ast |\widetilde{\phi}_{\mathrm{I},x}|^2) +\Big( N^{3/2} a\big( ( V_N \ast |\widetilde{\phi}_{\mathrm{I},x}|^2)  \widetilde{\phi}_{\mathrm{I},x}\big)+ \hc\Big) \label{eq:Weyl-L4-dGamma-lin} \\
 &\qquad + \frac12 {\int  N^2|\widetilde{\phi}_{\mathrm{I},x}(y_1)|^2|\widetilde{\phi}_{\mathrm{I},x}(y_2)|^2 V_N(y_1-y_2)} \label{eq:Weyl-L_4-scal}
\end{align}
\end{subequations}
To start with, the Hölder inequality yields for the scalar in the last line
\begin{align}
0 \leq \eqref{eq:Weyl-L_4-scal} \lesssim \|N V_N\|_1 \|\sqrt{N}\Wphi\|_2^2 \|\Wphi\|_\infty^2 \lesssim  N^{-\alpha}.
\end{align}
The second term in~\eqref{eq:Weyl-L4-dGamma-lin} can easily be estimated using Lemma~\ref{lem:phi_W} by
\begin{multline}
 \pm \Big(N^{3/2} a\big( ( V_N \ast |\widetilde{\phi}_{\mathrm{I},x}|^2)  \widetilde{\phi}_{\mathrm{I},x}\big)+ \hc\Big)\\
 \lesssim \|N V_N\|_{1} \|\Wphi\|_\infty^2 \|\sqrt{N}\Wphi\|_2(\mathcal N_+ +1 )^{1/2} \lesssim N^{-\alpha/2} (\mathcal N_+ +1)^{1/2}.
\end{multline}
The first term in~\eqref{eq:Weyl-L4-dGamma-lin} is non-negative and will also be used as a reference to bound other error terms later. By the Sobolev embedding of $H^{1/2}(\T^3)$ into $L^3$, the norm of a multiplication operator from $H^{1/2}$ to $H^{-1/2} $ is controlled by the $L^3$-norm of the multiplier. By Lemma~\ref{lem:phi_W} and Young's inequality we thus have the bound
\begin{align}
 |\langle \Psi, \ud \Gamma(NV_N \ast |\widetilde{\phi}_{\mathrm{I},x}|^2) \Psi \rangle|
 &\lesssim \|NV_N \ast |\Wphi|^2\|_3 \langle  \Psi, \ud\Gamma(|\ui \nabla|) \Psi\rangle \notag\\
 &\lesssim N^{-(1+\alpha)/3} \langle  \Psi, \ud\Gamma(|\ui \nabla|) \Psi\rangle. \label{eq:Weyl-error-dGamma}
\end{align}

We can now estimate the other error terms, starting with the cubic term \eqref{eq:Weyl-cubic-error}. Using the Cauchy-Schwarz inequality as in~\eqref{eq:L_2-L_4 bound} we obtain
\begin{align}
\pm \eqref{eq:Weyl-cubic-error}
	&\lesssim N^{-\alpha/2} \mathcal L_4 + N^{\alpha/2}  \ud \Gamma(NV_N \ast |\widetilde{\phi}_{\mathrm{I},x}|^2).
\end{align}

For the quadratic term \eqref{eq:Weyl-L_4-quad}, we have similarly by the Cauchy-Schwarz inequality
\begin{align}
\pm \eqref{eq:Weyl-L_4-quad} &\lesssim N^{-\alpha/2} \mathcal{L}_4 + N^{\alpha/2} \eqref{eq:Weyl-L_4-scal} \lesssim N^{-\alpha/2} \left(\mathcal L_4 + 1\right).
\end{align}
A similar reasoning gives
\begin{align}
  \pm\eqref{eq:Weyl-L_4-qmixed} \lesssim  \ud \Gamma(NV_N \ast |\widetilde{\phi}_{\mathrm{I},x}|^2) \lesssim  N^{-(1+\alpha)/3} \ud\Gamma(|\ui \nabla|). \label{eq:Weyl-L_4-qmixed-error}
\end{align}
We have thus obtained that
\begin{equation}
 U_W^*\mathcal{L}_4 U_W = \mathcal{L}_4 +\mathcal{O}\Big( N^{-\alpha/2}\mathcal{L}_4  + N^{-\alpha/2}(\cN_++1) + N^{-1/3+\alpha/6 } \ud\Gamma(|\ui \nabla|)  \Big).
\end{equation}
In particular, this proves the claim on $U_W^*\mathcal{L}_4 U_W$.

\noindent\textit{Estimating $ U_W^* \mathcal{L}_3^{U_q} U_W $.}
For the cubic term $\mathcal{L}_3$ we have
\begin{subequations}
 \begin{align}
& U_W^*\mathcal{L}_3^{U_q} U_W -\mathcal{L}_3^{U_q}  = \frac{1}{2} \int \sqrt N V_N(y_1-y_2) U_W^* a^*_{y_1} a^*_{y_2} a_{y_1}  U_W +\hc  - \mathcal{L}_3^{U_q} \notag \\
 &\quad= \dd \Gamma(NV_N \ast \Wphix) +\frac12\int  N V_N(y_1-y_2) \Wphix(y_1) a^*_{y_2}a_{y_1} +\hc\label{eq:Weyl-L3-dGamma} \\
 %
   %
 & \qquad +  \frac12\int  N V_N(y_1-y_2) \Wphi(x-y_1) a^*_{y_1} a^*_{y_2} +\hc \label{eq:Weyl-L_3-quad}\\
 &\qquad + \frac32 a\big( (N^{3/2} V_N \ast \Wphix) \Wphix\big) +\hc \label{eq:Weyl-L_3-lin} \\
&\qquad + \frac12 \int  N^{2} V_N(y_1-y_2) \Wphix(y_2)|\Wphix(y_1)|^2 .\label{eq:Weyl-L_3-scal}
 \end{align}
\end{subequations}
These terms are estimated by similar means as for $\cL_4$.
Similarly to~\eqref{eq:Weyl-error-dGamma} and~\eqref{eq:Weyl-L_4-qmixed-error}, we obtain
\begin{align}
\pm\eqref{eq:Weyl-L3-dGamma}&\lesssim  (\|N V_N \ast |\Wphi|\|_3 + \|N V_N\|_1\|\Wphi\|_3)  \ud\Gamma(|\ui \nabla|) \lesssim N^{-1/3-\alpha/3} \ud\Gamma(|\ui \nabla|).
\end{align}
For the quadratic terms in \eqref{eq:Weyl-L_3-quad}, we have by the Cauchy-Schwarz inequality
\begin{align}
\pm \eqref{eq:Weyl-L_3-quad}
	&\lesssim N^{-\alpha/2} \mathcal{L}_4 +  N^{\alpha/2} N^2 \int  V_N(y_1-y_2)  |\Wphi(y_2)|^2 \notag \\
	&\lesssim N^{-\alpha/2} \left( \mathcal{L}_4 + 1\right).
\end{align}
The linear terms in \eqref{eq:Weyl-L_3-lin} are estimated by
\begin{align}
\pm \eqref{eq:Weyl-L_3-lin} \leq 2\|(V_N\ast \Wphi)\Wphi\|_2 \cN_+^{1/2} \lesssim N^{-\alpha/2} (\mathcal N_+ + 1).
\end{align}
Finally, the scalar term (\ref{eq:Weyl-L_3-scal}) is bounded by the Hölder inequality and Lemma~\ref{lem:phi_W}
\begin{align}
| \eqref{eq:Weyl-L_3-scal} | \lesssim \|N V_N\|_1  \| \Wphi\|_2  N \|  \Wphi\|_4^2 \lesssim N^{-\alpha}.
\end{align}
Together, these bounds show that
\begin{equation}
U_W^*\mathcal{L}_3^{U_q} U_W =  \mathcal{L}_3^{U_q}+\mathcal{O}\Big( N^{-\alpha/2} \left( \mathcal{L}_4 +  \cN_+ + 1\right) + N^{-1/3 + \alpha/6} \ud\Gamma(|\ui \nabla|) \Big).
\end{equation}
This completes the proof of the Lemma.
\end{proof}

\subsection{Soft boson pairs: a cubic transformation}\label{sect:cubic}

We now apply a transformation in the boson variables whose generator is cubic in creation and annihilation operators.
This transformation eliminates the cubic term $\cL_3^{U_q}$ while turning $\ud \Gamma(2\widehat V (0) - 8 \pi \ao_{V})$ into $\ud \Gamma(8 \pi \ao_{V})$.

Such a renormalization has its roots in \cite{YauYin-09}, where a cubic non-unitary transformation was used to derive the upper bound to the Lee-Huang-Yang formula. Later in \cite{BocBreCenSch-20,BocBreCenSch-18}, a unitary version was implemented to prove optimal rate of condensation and derive the low-lying excitation spectrum in the Gross--Pitaevskii regime.

Its role is to account for the interactions of the so-called ``soft pairs'' of high energy modes, which, in contrast to hard pairs that annihilate to yield two particles in the condensate, only annihilate to give one particle in the condensate and one low energy mode. As before, we will distinguish between low and high energy modes with the cutoff $N^\alpha$.

In \cite{NamTri-23}, a cutoff in the number of excitations $\theta_M(\mathcal N_+) \simeq \mathds{1}_{\mathcal N_+ \leq M}$ was introduced in the kernel of the cubic transformation. The idea is that if we know a priori a condensation rate on the low-lying eigenvectors $\braket{\Psi,\mathcal N_+\Psi} \lesssim M_0 \ll N$, it is enough to renormalize the sectors with $\mathcal N_+ \leq M$ for $M \gg M_0$. 
The rate of condensation provided by Condition~\ref{cond:BEC} means that  we can work on sectors with $\mathcal N_+ \leq M$ for some $M \gg \sqrt N $.

Our implementation of the cubic transformation essentially follows \cite{HaiSchTri-22}, although an important difference is that we have to use $\dd\Gamma(|\ui \nabla|)$ instead of $\ud \Gamma(-\Delta)$ in the error estimates, as the latter would interfere with the renormalization of the impurity-boson interaction explained in Section~\ref{sect:renorm}, because $Q(H_\mathrm{BF})\nsubseteq Q(\ud \Gamma(-\Delta))$.

Let $\Vphi$ be the truncated scattering solution for the bosons~\eqref{eq:Vphi}, as in Section~\ref{sect:quadratic}. Let $\theta:\R_+\to [0,1]$ be a smooth function satisfying $\theta(x) = 1$ for $|x| \leq 1/2$ and $\theta(x) = 0$ for $|x|\geq 1$, and set $\theta_M(x)=\theta(x/M)$.
We define for $M\geq 1$
\begin{equation}
 U_c(t)=\ue^{t\cB_c}, \quad U_c=U_c(1),
\end{equation}
with the generator
\begin{subequations}
 \begin{align}
 \mathcal B_c &= \theta_M(\cN_+) \mathcal B_c^+ - \mathcal B_c^{-} \theta_M(\cN_+), \\
 \label{eq:Bc}
\mathcal B_c^+ &= \sum_{p,q \in 2\pi \Z^3\setminus\{0\}} \sqrt{N} \Vphih(p) \1_{|q|\leq  N^{\alpha}} a^*_{p+q}a^*_{-p} a_{q}, \quad \quad  \mathcal B_c^- = (\mathcal B_c^+)^*.
\end{align}
\end{subequations}
Note that the sum is only over a finite set since $\Vphi$ contains a cutoff. Due to the cutoff $\theta_M$ in $\cN_+$, $\cB_c$ is a bounded operator, since
\begin{align}
 |\langle \Phi, \cB_c^- \Psi\rangle| & \leq \bigg(\sum_{p,q} N \Vphih(p)^2 \|a_q (\cN_++1)^{-1/2} \Phi\|^2\bigg)^{1/2} \notag \\
 &\qquad \times \bigg(\sum_{p,q}  \|a_{p+q}a_{-p}(\cN_++2)^{1/2} \Psi\|^2\bigg)^{1/2} \notag\\
 &\leq N^{-1/2} \|N \Vphi\|_2  \|\Phi\| \|(\cN_++2)^{3/2}\Psi\|. \label{eq:B^- bound}
\end{align}
Denoting by $\chi\in L^2(\T^3)$  the (even) function with Fourier coefficients $\widehat \chi(p)=\1_{|p|\leq N^\alpha}$, we can also write
\begin{align}\label{eq:B^+-position}
 \cB_c^+ = \sqrt{N} \int \chi(y_1-y_3)\Vphi(y_2-y_3) a^*_{y_3}a^*_{y_2} a_{y_1}.
\end{align}

The main result of this section is the following proposition.
\begin{prop}\label{prop:U_c}
 Let $0<\alpha\leq 1/8$ and $M=N^{1/2+\alpha/2}$.
 Let $\cH_N$ be given by Proposition~\ref{prop:H-ex} and recall its transformation by $U_qU_W$ given in Proposition~\ref{prop:U_W}.  We have as quadratic forms on $\mathscr{H}_+$
 \begin{align*}
 U_c^*& U_W^* U_q^* \mathcal{H}_N U_q U_W U_c \\
  & = 4\pi  \ao_V (N-1)  + 8\pi  \ao_W \sqrt{N} +  e_N^{(U_W)} \\
  %
  &\quad -\Delta_x +\Big(a(\widetilde W_{N,x}) +\hc \Big) +\Big(a^*\big(\sqrt{N}  (-\ui \nabla \Wphi)_x\big)^2 +\hc \Big) + \ud \Gamma(W_{N,x}) \\
  &\quad   -a^*\big(\sqrt{N} (\ui \nabla \Wphi)_x\big)2\ui \nabla_x  + \hc   + 2 a^*\big(\sqrt{N} (\ui \nabla \Wphi)_x\big)a\big(\sqrt{N} (\ui \nabla \Wphi)_x\big)\\
    &\quad +  \dd \Gamma(-\Delta + 8\pi \ao_{V})  + \widetilde{\mathcal{L}}_2 +
  \mathcal{L}_4 + \mathcal{E}^{(U_c)},
 \end{align*}
where $\mathcal{E}^{(U_c)}$ satisfies
\begin{multline*}
 \pm \mathcal{E}^{(U_c)}
 	\lesssim  N^{-\alpha/2}  \ud \Gamma(W_{N,x})  +N^{-\alpha/4}\mathcal{L}_4 + N^{-1/2-\alpha/4}\cN_+^2 \\
	+   N^{-\alpha/4}\big((\log N) \dd\Gamma(|\ui \nabla|)  +1 +(\mathcal N_++1)^{1/2} |\nabla_x|\big).
\end{multline*}
\end{prop}

We first give the proof of Proposition \ref{prop:U_c} using the commutator bounds and the a priori estimates of Lemmas \ref{lem:Uc_Delta}, \ref{lem:Uc_L_4}, \ref{lem:cancel_Q3}, \ref{lem:L2_Bc}, \ref{lem:L3_Bc}, \ref{lem:com_a_Bc}, \ref{lem:com_Q2_Bc}, \ref{lem:quad_part_Bc}, \ref{lem:Uc_N} and \ref{lem:Uc_p} below, which we then prove in the remaining part of this section.
\begin{proof}[Proof of Proposition \ref{prop:U_c}]
 From Proposition \ref{prop:U_W}, we have (with the definitions given there)
\begin{multline}
U_c^*U_W^* U_q^* \mathcal{H}_N U_q U_W U_c
   = 4\pi  \mathfrak{a}_{V_N} (N-1) +e_N^{(U_W)} \\
   + U_c^* \mathcal H_{\mathrm{BB}}^{(U_W)} U_c + U_c^* \mathcal H_{\mathrm{IB}}^{(U_W)} U_c + U_c^* \mathcal E^{(U_W)} U_c.
\end{multline}
 Let us estimate the three terms above separately.

\noindent
\emph{Estimate of $ U_c^* \mathcal H_{\mathrm{BB}}^{(U_W)} U_c$}.
Recall that
\begin{align*}
  \mathcal H_{\mathrm{BB}}^{(U_W)} 
  	&=  \dd \Gamma(-\Delta + 2\widehat{V}(0) - 8\pi \ao_{V})  +\widetilde{\mathcal{L}}_2 + \mathcal{L}_3^{U_q} +
  \mathcal{L}_4.
\end{align*}
We first deal with the term $\dd\Gamma(-\Delta) +\mathcal{L}_3^{U_q} + \mathcal{L}_4$. 
With 
\begin{align}
 U_c^* \mathcal{L}_3^{U_q} U_c - \int_0^1 U_c^*(s) \mathcal{L}_3^{U_q} U_c(s) \ud s =  \int_0^1 \int_s^1 U_c^*(t) [\mathcal{L}_3^{U_q}, \mathcal B_c] U_c(t) \ud t \ud s
\end{align}
and  Duhamel's formula, we may write
\begin{subequations}
 \begin{align}
& U_c^*\left( \dd\Gamma(-\Delta) + \mathcal{L}_3^{U_q} + \mathcal{L}_4\right) U_c
\notag\\
&\quad = \dd\Gamma(-\Delta)  + \mathcal{L}_4 + \int_0^1 U_c^*(s) \left\{ [\dd\Gamma(-\Delta)  + \mathcal{L}_4, \mathcal B_c]  + \mathcal{L}_3^{U_q} \right\} U_c(s) \dd s \label{eq:U_c-L_3-cancel}\\
&\qquad + \int_0^1 \int_{s}^{1} U_c^*(t) [\mathcal{L}_3^{U_q}, \mathcal B_c] U_c(t) \ud t \ud s. \label{eq:U_c-L_3-remainder}
\end{align}
\end{subequations}
With the cancellation showed in Lemma~\ref{lem:cancel_Q3} below, the error bounds from Lemmas~\ref{lem:Uc_Delta}, \ref{lem:Uc_L_4} and the conservation estimates of Lemmas~\ref{lem:Uc_N} and~\ref{lem:Uc_p} we find,
\begin{equation}
 \eqref{eq:U_c-L_3-cancel} = \dd\Gamma(-\Delta)  + \mathcal{L}_4 + \cE
\end{equation}
with, using that $M= N^{1/2 + \alpha/2}$, $\alpha\leq 1/8$ and $\ud \Gamma(|\ui \nabla|)\geq \cN_+$ on $\mathscr{H}_+$,
\begin{equation}
 \pm \cE \lesssim N^{-\alpha/4}( \mathcal{L}_4+ \ud\Gamma(|\ui \nabla|)+1) + N^{-1/2-\alpha/4} \cN_+^2   .
\end{equation}
Using Lemma~\ref{lem:L3_Bc}  for the commutator with $\cL_3^{U_q}$ yields with $M=N^{1/2+\alpha/2}$ (note that the double integral gives a factor of $1/2$)
\begin{align}
 &\eqref{eq:U_c-L_3-remainder} =\int_0^1 \int_{s}^{1} U_c^*(t) (4 \ud \Gamma(8\pi \ao_{V} - \widehat V(0)) +\cE_{[\cL_3,\cB_c]}) U_c(t)\ud t\ud s   \\
  &= 2 \ud \Gamma(8\pi \ao_{V} - \widehat V(0)) + \mathcal{O}\big(N^{-\alpha/2} \mathcal{L}_4+ N^{-\alpha/4}( \ud \Gamma(|\ui \nabla|)+1) +  N^{-1/2-\alpha/2} \mathcal N_+^2\big), \notag
\end{align}
by the conservation estimates of Lemmas~\ref{lem:Uc_N} and~\ref{lem:Uc_p}.

Together with the estimate on the transformation of $\widetilde \cL_2$ from Lemma~\ref{lem:L2_Bc},
this shows that 
\begin{equation}
 U_c^* \mathcal H_{\mathrm{BB}}^{(U_W)} U_c = \dd \Gamma(-\Delta + 8\pi \ao_{V})  +\widetilde{\mathcal{L}}_2 +   \mathcal{L}_4 + \mathcal{E}^{(U_c)}_\mathrm{B},
\end{equation}
where $\mathcal{E}^{(U_c)}_\mathrm{B}$ satisfies the error bounds claimed for $\mathcal{E}^{(U_c)}$.

\noindent\emph{Estimate of $ U_c^* \mathcal H_{\mathrm{IB}}^{(U_W)}U_c$}.
Recall that
\begin{multline}
 \mathcal H_{\mathrm{IB}}^{(U_W)}=  -\Delta_x +\Big(a^*(\widetilde W_{N,x} ) +\hc \Big) +\Big(a^*\big(\sqrt N (-\ui\nabla \widetilde{\phi}_{\textrm{I}})_x\big)^2 +\hc \Big) + \ud \Gamma(W_{N,x})  \\
    -2 a^*\big(\sqrt N  (\ui\nabla \widetilde{\phi}_{\textrm{I}})_x\big)\ui\nabla_x  + \hc + 2 a^*\big(\sqrt N (\ui\nabla \widetilde{\phi}_{\textrm{I}})_x\big)a\big(\sqrt N (\ui\nabla \widetilde{\phi}_{\textrm{I}})_x\big),
\end{multline}
and note  that $U_c^*\Delta_x U_c=\Delta_x$ since $\mathcal B_c$ does not depend on $x$.
By the Duhamel formula, we then have
\begin{align}
U_c^* \mathcal H_{\mathrm{IB}}^{(U_W)} U_c =  \mathcal H_{\mathrm{IB}}^{(U_W)} + \int_0^1 U_c^*(t) [\mathcal H_{\mathrm{I}}^{(U_W)}+\Delta_x, \mathcal B_c] U_c(t) \dd t
\end{align}
The commutators of individual terms are estimated in Lemmas~\ref{lem:com_a_Bc},~\ref{lem:com_Q2_Bc},~\ref{lem:quad_part_Bc} below. Combining these bounds with the conservation estimates of Lemmas \ref{lem:Uc_N} and \ref{lem:Uc_p}, we obtain, setting $M=N^{1/2+\alpha/2}$ and using that $\alpha\leq 1/8$
\begin{align}
\pm &\int_0^1U_c^*(t) [\mathcal H_{\mathrm{I}}^{(U_W)}+\Delta_x, \mathcal B_c] U_c(t) \dd t \\
	&\lesssim N^{-\alpha/2}  \ud \Gamma(W_{N,x})  +    N^{-\alpha/4}    (\log N\dd\Gamma(|\ui \nabla|) + (\mathcal N_++1)^{1/2} |\nabla_x| +1),\notag
\end{align}
which can be absorbed in the final error term $\mathcal{E}^{(U_c)}$.

\noindent\emph{Estimate of $ U_c^* \mathcal E^{(U_W)} U_c$}.
Recalling the bound on the error $\mathcal E^{(U_W)}$ from Proposition~\ref{prop:U_W}, the expression $U_c^* \mathcal E^{(U_W)} U_c$ can be absorbed in the error $U^{(U_c)}$ by Lemmas \ref{lem:Uc_N} and \ref{lem:Uc_p}.
\end{proof}

\paragraph{Commutator estimates: Boson Hamiltonian.}
Here, we treat the commutators of $\cB_c$ with the different terms of the boson-Hamiltonian.
To remove the cutoff $\theta_M(\cN_+)$ from the leading terms, we will make use of the inequality
\begin{equation}\label{eq:theta-diff}
 | \theta_M(t-1) - \theta_M(t)|\lesssim M^{-1} \1_{t\leq M+1},
\end{equation}
which follows from Taylor's theorem.

\begin{lem}\label{lem:Uc_Delta}
For all $\alpha>0$, $M \geq 1$ we have
\begin{multline*}
 [\dd\Gamma(-\Delta), \mathcal B_c] 
	= - \hspace{-12pt}\sum_{p ,q \in 2\pi \Z^3\setminus\{0\}\atop  |q| \leq N^{\alpha}} \hspace{-12pt}\sqrt N \big(\widehat{V}_N(p) + \widehat{V}_N \ast \widehat \phi_\mathrm{B} (p)\big) a^*_{p+q} a^*_{-p} a_q  + \hc + \mathcal{E}_{[\dd\Gamma(-\Delta),\B_c]},
\end{multline*}
where the error satisfies uniformly in $M$
\begin{multline*}
 \pm \mathcal{E}_{[\dd\Gamma(-\Delta),\B_c]} 
		\lesssim N^{-\alpha/4} \mathcal{L}_4 + N^{\alpha/4}M^{-1} (\mathcal N_++1)^2 \\+\sqrt M N^{\alpha/2-1/2} \log N \ud \Gamma(|\ui \nabla|) + \sqrt M N^{3\alpha/2-1/2}  \mathcal (N_++1).
\end{multline*}
\end{lem}

\begin{proof}
 A simple computation gives
\begin{multline}
\theta_M(\cN_+)[\dd\Gamma(-\Delta), \mathcal B_c^+] \\
=  \theta_M \sum_{p,q} \sqrt N (2p^2+2pq) \Vphih(p) \1_{|q| \leq N^\alpha} a^*_{p+q} a^*_{-p} a_q.\label{eq:Delta_Bc}
\end{multline}
By the scattering equation~\eqref{eq:scattV} (note that $\widehat \phi_\mathrm{B}$ has no cutoff)
\begin{equation}
 2p^2 \Vphih(p) = - \big(\widehat{V}_N(p) + \widehat{V}_N \ast \widehat \phi_\mathrm{B} (p)\big)\1_{|p|>N^\alpha},
\end{equation}
so we need to remove the cutoff $\1_{|p|>N^\alpha}$ and $\theta_M$ from the term with $p^2$ in order to obtain the leading term from the claim.
Using the Cauchy-Schwarz inequality as in~\eqref{eq:B^- bound}, we find
\begin{align}
&\bigg| \Big\langle \Psi, \theta_M(\cN_+) \sum_{|p|,|q|\leq N^\alpha} \sqrt N \big(\widehat{V}_N(p) + \widehat{V}_N \ast \widehat \phi_\mathrm{B} (p)\big) a^*_{p+q} a^*_{-p} a_q \Psi \Big\rangle \bigg| \notag \\
&\quad \lesssim \sqrt N\|(\widehat{V}_N + \widehat{V}_N \ast \widehat \phi_\mathrm{B} )\1_{|p|\leq N^{\alpha}}\|_2 \|\cN_+^{1/2} \Psi\| \| \theta_M(\cN_+)\cN_+ \Psi\| \notag \\
&\quad \lesssim  \sqrt M N^{-1/2+3\alpha/2}   \|\cN_+^{1/2} \Psi\|^2, \label{eq:bound_l2_linfty}
\end{align}
which can be absorbed in the error and allows us to drop $\1_{|p|>N^\alpha}$.
To remove $\theta_M$, we first write
\begin{align}
 &\sum_{p,q} \sqrt N \big(\widehat{V}_N(p) + \widehat{V}_N \ast \widehat \phi_\mathrm{B} (p)\big) \1_{|q|\leq N^\alpha} a^*_{p+q} a^*_{-p} a_q \notag \\
 &\quad = \int \chi(y_3-y_1) \sqrt N V_N(1+\phi_\mathrm{B})(y_2-y_1) a^*_{y_1}a^*_{y_2}a_{y_3},
\end{align}
where $\chi\in L^2(\T^3)$ denotes the function with Fourier coefficients $\widehat \chi(p)=\1_{|p|\leq N^\alpha}$. We then use that $1-\theta_M(\cN_+)\lesssim M^{-1}\cN_+$ together with the Cauchy-Schwarz inequality (as in~\eqref{eq:L_2-L_4 bound},~\eqref{eq:L_4-B_2 comm2}), to obtain for $\delta>0$
\begin{align}
&\pm \Big( (1-\theta_M(\cN_+)) \sum_{p,q \atop |q|\leq N^\alpha} \sqrt N \big(\widehat{V}_N(p) + \widehat{V}_N \ast \widehat \phi_\mathrm{B} (p)\big) a^*_{p+q} a^*_{-p} a_q + \hc \Big)\notag \\
&\quad \lesssim \delta \int V_N(1+\phi_\mathrm{B})(y_2-y_1)a^*_{y_1}a^*_{y_2}a_{y_1}a_{y_2}  +  \delta^{-1}  M^{-1}  \cN_+ \int \chi\ast \chi(y-y')a^*_{y'}a_{y} \notag \\
&\quad \lesssim \delta \cL_4 + \delta^{-1}  M^{-1} \cN_+^2. \label{eq:theta_M-remove}
\end{align} 
Choosing $\delta=N^{-\alpha/4}$ here gives the required bound on this term.
It remains to bound the term coming from the $pq$-part in~\eqref{eq:Delta_Bc}.
This satisfies, for all $\delta>0$,
\begin{align}
&\pm\Big(2 \theta_M(\cN_+) \sum_{p,q}   p\cdot q \sqrt N \Vphih(p) \1_{|q| \leq N^\alpha} a^*_{p+q} a^*_{-p} a_q + \hc\Big) \notag \\
	&\lesssim 
\delta \theta_M^2 \sum_{p,|q| \leq N^{\alpha}} |p| |q|  a^*_{p+q} a^*_{-p} a_{-p} a_{p+q} + \delta^{-1} N \sum_{p,|q| \leq N^{\alpha}} |p| |q|  | \Vphih(p) |^2 \, a^*_{q} a_{q} \notag\\
	&\lesssim \delta N^{\alpha}  \theta_M^2  \dd\Gamma(|\ui \nabla|)\mathcal N_+ + \delta^{-1} N \||\nabla|^{1/2} \Vphi\|_2^2 \dd\Gamma(|\ui \nabla|) .
\end{align}
Using that $\cN_+ \theta_M^2\leq M$, Lemma~\ref{lem:phi_V}, and choosing $\delta = M^{-1/2}N^{-\alpha/2-1/2}\log N$ then completes the proof.
\end{proof}

\begin{lem}\label{lem:Uc_L_4}
For all $\alpha>0$, $M \geq 1$ we have as quadratic forms on $\mathscr{H}_+$
\begin{equation*}
 [\cL_4, \mathcal B_c] = \sum_{p ,q \in 2\pi \Z^3\setminus\{0\}\atop  |q| \leq N^{\alpha}} \sqrt N \widehat{V}_N \ast \Vphih(p) a^*_{p+q} a^*_{-p} a_q + \hc + \mathcal{E}_{[\mathcal{L}_4,\B_c]}
\end{equation*}
where the error satisfies uniformly in $M$ and $\delta>0$
\begin{equation*}
 \pm \mathcal{E}_{[\mathcal{L}_4,\B_c]}
	\lesssim \delta \mathcal{L}_4 + \delta^{-1}(MN^{-1}+1) M N^{-1 +3\alpha}(\mathcal N_+ +1) + M^{-1}\delta^{-1} \mathcal N_+^2.
\end{equation*}
\end{lem}

\begin{proof}
 Following the proof of \cite[Lemma 12]{HaiSchTri-22} (with $a_0$ replaced by $\sqrt{N}$ and using $\theta_M$ to bound any powers of $\cN_+$ in excess of one), we arrive at
\begin{multline}
\theta_M (\cN_+)[\mathcal L_4,\cB_c^+]\\
=  \theta_M (\cN_+) \sum_{r,p,q} \sqrt N \widehat{V}_N(p-r)\Vphih(r)  \1_{|q| \leq N^\alpha} a^*_{p+q} a^*_{-p} a_q + \mathcal{E}
\end{multline}
with
\begin{align}
\pm \mathcal E + \hc \lesssim \delta \cL_4 + \delta^{-1} (MN^{-1} +1) M N^{-1 +3\alpha}\mathcal (\cN_+ +1).
\end{align}
By the bound~\eqref{eq:theta_M-remove}, the error due to the removal of $\theta_M$ from the main term can be absorbed in $\mathcal{E}_{[\mathcal{L}_4,\B_c]}$.
\end{proof}

\begin{lem}\label{lem:cancel_Q3}
For all $\alpha>0$, $M \geq 1$ and with the definitions from Lemma~\ref{lem:Uc_Delta}, Lemma~\ref{lem:Uc_L_4}, we have as quadratic forms on $\mathscr{H}_+$
\begin{align*}
\big[ \dd\Gamma(-\Delta) + \mathcal L_4 , \mathcal B_c \big] + \mathcal L_3^{U_q} = \mathcal{E}_{[\dd\Gamma(-\Delta),\B_c]} +  \mathcal{E}_{[\cL_4,\B_c]} + \mathcal{O}\big(N^{-\alpha/2} \left( \mathcal L_4 + \dd\Gamma(|\ui \nabla|)\right)\big).
\end{align*}
\end{lem}

\begin{proof}
 Combining Lemma~\ref{lem:Uc_Delta}, Lemma~\ref{lem:Uc_L_4} and the Fourier representation of $\mathcal{L}_3^{U_q}$, we obtain
\begin{multline}
 [\dd\Gamma(-\Delta)  + \mathcal{L}_4, \mathcal B_c] + \mathcal{L}_3^{U_q}= \sum_{p ,q} \sqrt N \widehat{V}_N(p)\1_{|q|>N^\alpha}  a^*_{p+q} a^*_{-p} a_q  + \hc \\ +  \mathcal{E}_{[\dd\Gamma(-\Delta),\B_c]} +  \mathcal{E}_{[\cL_4,\B_c]}.
\end{multline}
We have by the reasoning of~\eqref{eq:theta_M-remove}
\begin{align}
&\sum_{p ,q} \sqrt N \widehat{V}_N(p)\1_{|q|>N^\alpha}  a^*_{p+q} a^*_{-p} a_q  + \hc  \lesssim \delta \cL_4+ \delta^{-1} \sum_{|p|>N^\alpha} a^*_pa_p.
\end{align}
Clearly, the last term is bounded by $N^{-\alpha}\ud \Gamma(|\ui \nabla|)$, so choosing $\delta=N^{-\alpha/2}$ proves the claim.
\end{proof}

\begin{lem}
	\label{lem:L2_Bc}
For all $\alpha>0$ and uniformly in $M \geq 1$ we have as quadratic forms on $\mathscr{H}_+$
\begin{align}
 [\widetilde{\mathcal L}_2,\mathcal B_c]  
	&=\mathcal{O}\big( \sqrt M  N^{\alpha-1/2} (\mathcal N_++1)\big).\notag
\end{align}
\end{lem}

\begin{proof}
 From the definition of $\widetilde \cL_2$ in Proposition~\ref{prop:U_q} and using that $a_p \mathcal N_+ = (\mathcal N_++1) a_p$, we have
\begin{equation}
 [\widetilde{\mathcal L}_2,\theta_M(\cN_+)]=  \big(\theta_M(\cN_+)-\theta_M(\cN_++2)\big)\sum_{|p|\leq N^\alpha} 4\pi \ao_{V} a_pa_{p} +\hc
\end{equation}
Using the bound~\eqref{eq:theta-diff} on the difference for $\theta_M$ together
 with the bound~\eqref{eq:B^- bound} on $\cB_c^-$, gives
\begin{align}
 |\langle \Psi, [\widetilde{\mathcal L}_2,\theta_M(\cN_+)] \cB_c^+ \Psi\rangle| & \lesssim N^{\alpha -\frac12} M^{-1}\| \cN_+ \1_{\cN_+\leq 2M} \Psi\| \|(\cN_++1)^{3/2} \1_{\cN_+\leq 2M} \Psi\| \notag \\
 &\lesssim \sqrt M  N^{\alpha-\frac12} \|(\cN_++1)^{1/2}  \Psi\|^2.
\end{align}
Due to the cut-off $\1_{|p|> N^{\alpha}}$ in the definitions of $\Vphih$ and $\mathcal B_c^+$, we have
\begin{align}
 [\widetilde{\mathcal L}_2,\cB_c^+]&= 8\pi \ao_{V} \sum_{p,q} \sqrt N \Vphih(p)\1_{|q| \leq N^{\alpha}}   a^*_{-p}a^*_{p+q} a^*_q  \notag \\
 &\quad -8\pi \ao_{V} \sum_{p,q} \sqrt N \Vphih(p)\1_{|p+q| \leq N^{\alpha}}\1_{|q| \leq N^{\alpha}}a^*_{-p} a_{-(p+q)} a_q.
\end{align}
With the reasoning of~\eqref{eq:bound_l2_linfty} we get from this
\begin{align}
 |\langle \Psi, \theta_M(\cN_+)[\widetilde{\mathcal L}_2,\cB_c^+]\Psi\rangle| 
 &\lesssim \sqrt{N} \|\Vphih\|_\infty N^{3\alpha/2} \|\cN_+^{1/2}\Psi \| \| 1_{\cN_+\leq 2M} \cN_+\Psi\| \notag \\
 &\lesssim \sqrt M N^{-1/2-\alpha/2} \|\cN_+^{1/2}\Psi \|^2.
\end{align}
Together with the bound on the commutator with $\theta_M$ this proves the claim.
\end{proof}

\begin{lem}\label{lem:L3_Bc}
 For all $\alpha>0$, $M \geq 1$ we have as quadratic forms on $\mathscr{H}_+$
 \begin{align*}
  [\cL_3^{U_q}, \cB_c] = 4 ( 8\pi \ao_{V}-\widehat V(0)) \cN_+ + \cE_{[\cL_3, \cB_c]},
 \end{align*}
where the error satisfies uniformly in $M$
\begin{equation*}
\pm \cE_{[\cL_3, \cB_c]} \lesssim N^{-\alpha/2} \mathcal{L}_4+ N^{-\alpha} \ud \Gamma(|\ui \nabla|) + MN^{-1 + 5\alpha/2} (\mathcal N_++1) + M^{-1} \mathcal N_+^2.
\end{equation*}
\end{lem}

\begin{proof}
 The proof follows the computations carried out in \cite[Lem.15]{HaiSchTri-22}, with the only difference  being the cut-off $\theta_M$ that does not entirely commute with $\mathcal L_3^{U_q}$, which we deal with in the same way as in \cite{NamTri-23}.
 First, write $\mathcal L_3^{U_q}=\cL_3^+ + \cL_3^-$ with
 \begin{equation}
  \cL_3^+ = \frac{1}{2}\int \sqrt{N} V_N(y_1-y_2) a^*_{y_1} a_{y_2}^* a_{y_1}
 \end{equation}
and $\cL_3^-=(\cL_3^+)^*$.
With this, we find
\begin{multline}
 [\cL_3^{U_q}  ,\theta_M(\cN_+)] \cB_c^+=[\cL_3^+ ,\theta_M(\cN_+)] \cB_c^+\\
 +  (\theta_M(\cN_+ +1)-\theta_M(\cN_+))(\cB_c^+\cL_3^- + [\cL_3^-,\cB_c^+]).\label{eq:L_3-theta-comm}
\end{multline}
Using the bounds~\eqref{eq:theta-diff} on $\theta_M$, \eqref{eq:B^- bound} on $\cB_c^\pm$ and the Cauchy-Schwarz inequality (as in~\cite[Lem.26]{NamTri-23}), gives
\begin{multline}
 \big(\cL_3^+ [\theta_M(\cN_+), \cB_c^+] +  (\theta_N(\cN_+ +1)-\theta_M(\cN_+))\cB_c^+\cL_3^- +\hc \big)\\
 \lesssim N^{-\alpha/2} \mathcal{L}_4 + MN^{-1 - \alpha/2}  (\mathcal N_+ +1).
\end{multline}
In view of~\eqref{eq:L_3-theta-comm}, we thus have
\begin{multline}
 [\cL_3^{U_q}  ,\cB_c] = \theta_M(\cN_+ )[\cL_3^+,\cB_c^+] + \theta_M(\cN_+ +1)[\cL_3^-,\cB_c^+]   + \hc  \\
 + \mathcal{O}\big(N^{-\alpha/2} \mathcal{L}_4 + MN^{-1 - \alpha/2}  (\mathcal N_+ +1)\big).
\end{multline}
The remaining relevant terms correspond to those denoted (I) and (II) in~\cite[Eq.(77)]{HaiSchTri-22}.
Following the proof there yields (replacing $a_0^\dagger$, $a_0$ by $\sqrt {N}$ and thus ignoring their commutators)
\begin{equation}
 \theta_M(\cN_+ )[\cL_3^+,\cB_c^+] = \mathcal{O}\big( N^{-\alpha/2} \mathcal{L}_4 + MN^{-1+\alpha}  (\mathcal N_+ +1) \big).
\end{equation}
Now following the proof from~\cite[Eq.(79)]{HaiSchTri-22} on, we find
(note that the term $\text{(II)}_{b}$ in our case is just equal to $\text{(II)}_{b_1}$ with $a_0^\dagger a_0$ is replaced by $N$)
\begin{multline}
 \theta_M(\cN_+ +1) [\mathcal L_3^-,\mathcal B_c^+] = \theta_M(\cN_+ +1)  \sum_{q} N \big(\widehat{V}_N \ast \Vphih(0) + \widehat{V}_N \ast \Vphih(q)\big)  \1_{|q| \leq N^\alpha} a_q^* a_q \\
 +  \mathcal{O}\big( N^{-\alpha/2} \mathcal{L}_4 + M N^{-1+5\alpha/2}  (\mathcal N_+ +1) \big).
\end{multline}
As in (\ref{eq:V_N renorm error}), we have 
\begin{align}
 N |\widehat{V}_N \ast \Vphih(q) - \widehat{V}_N \ast \Vphih(0)| \lesssim |q| N^{-1},
\end{align}
and together $ \|\widetilde{\phi}_{\mathrm{B}}-\phi_{\mathrm{B}}\|_\infty \lesssim N^{-1+\alpha}$ from Lemma \ref{lem:phi_V}, we obtain for $|q| \leq N^{\alpha}$
\begin{align}
 N |\widehat{V}_N \ast \Vphih(q) -  N \widehat{V}_N \ast \widehat{\varphi}_B(0)| \lesssim N^{-1+\alpha}.
\end{align}
Using this, we have
\begin{multline}
 \theta_M(\cN_+ +1)  \sum_{q} N  \widehat{V}_N \ast \Vphih(q)  \1_{|q| \leq N^\alpha} a_q^* a_q \\
 = \theta_M(\cN_+ +1)  \sum_{q} N  \widehat{V}_N \ast \varphi_B(0)  \1_{|q| \leq N^\alpha}a_q^* a_q + \mathcal{O}(N^{-1+\alpha}\cN_+).
\end{multline}
By definition~\eqref{eq:aV_N}, $N \widehat{V}_N \ast \widehat{\varphi}_B(0) = 8\pi \aV - \widehat V(0)$, which is bounded in $N$. To remove the cutoffs $\theta_M$ and $\1_{|q|\leq N^\alpha}$, we use that
\begin{equation}
 \sum_{q} N \widehat{V}_N \ast \widehat{\varphi}_B(0)\1_{|q| > N^\alpha} a_q^* a_q = \mathcal{O}\big(N^{-\alpha} \ud \Gamma(|\ui \nabla|)),
\end{equation}
and 
\begin{equation}
 (1-\theta_M(\cN_+ +1)) \sum_{q} N \widehat{V}_N \ast \Vphih(0) a_q^* a_q =\mathcal{O}(M^{-1}\cN_+^2).
\end{equation}
Since $\aV = \ao_V + \mathcal O(N^{-1})$ by Lemma \ref{lem:aM-difference} this proves the claim.
\end{proof}

\paragraph{Commutator estimates: Interaction Hamiltonian.}
We now prove bounds on the commutators of $\mathcal B_c$ with the terms in the interaction Hamiltonian.
\begin{lem}	\label{lem:com_a_Bc}
For all $\alpha>0$ and uniformly in $M\geq 1$ we have as quadratic forms on $\mathscr{H}_+$
 \begin{align*}
  [a^*( \widetilde W_{N,x}), \mathcal B_c] + \hc =\mathcal{O}\big( N^{\alpha-1/2} (\mathcal N_+ +1)\big).
\end{align*}
\end{lem}
\begin{proof}
 Recall that $\widehat{\widetilde W}_N(p)=8\pi \ao_W\1_{|p|<N^\alpha}$.
 Since
 \begin{align}
  [a^*( \widetilde W_{N,x}), \theta_M(\cN_+)] = \big(\theta_M(\cN_+-1) - \theta_M(\cN_+)\big) a^*( \widetilde W_{N,x})
 \end{align}
and with~\eqref{eq:theta-diff} as well as~\eqref{eq:B^- bound} we find
\begin{align}
&| \langle \Psi, [a^*( \widetilde W_{N,x}), \theta_M]  \mathcal B_c^+  \Psi\rangle| \lesssim \|\widetilde W_{N}\|_2 \|\cN_+^{1/2}\Psi\|    M^{-1} \|\1_{\cN_+\leq 2M}\cB_c^+ \Psi\| \notag\\
&\quad \lesssim  N^{3\alpha/2} M^{-1} \|\sqrt{N}\Vphi \|_2  \|\cN_+^{1/2}\Psi\|  \|\1_{\cN_+\leq 2M}(\cN_++1)^{3/2} \Psi\| \notag\\
&\quad  \lesssim N^{-1/2+\alpha} \|(\cN_++1)^{1/2} \Psi\|^2.
\end{align}
For the commutators with $\cB_c^\pm$, we have
\begin{align}
   \theta_M [a^*( \widetilde W_{N,x}),   B_c^+] + \hc &=   \sum_{p,q} \sqrt N \Vphih(p)  \theta_M a^*_{p+q}a^*_{-p} \widehat{\widetilde W}_{N,x}(q) + \hc \notag \\
	&=\mathcal{O}\big( N^{-1/2} \| N\Vphi\|_2 \|\widetilde W_{N,x}\|_{2} (\mathcal N_+ + 1)\big) \nn\\
	&=\mathcal O(N^{\alpha-1/2} (\mathcal N_+ +1)).
\end{align}
A similar bounds holds for $\theta_M[a^*( \widetilde W_{N,x}),   B_c^-]$. Together, these inequalities imply the claim.
\end{proof}

\begin{lem}
	\label{lem:com_Q2_Bc}
We have as quadratic forms on $\mathscr{H}_+$ for all $\alpha>0$ and uniformly in $M \geq 1$
\begin{align*}
 \pm [\ud \Gamma(W_{N,x}),\mathcal B_c] \lesssim  (M/N)^{1/2} N^{-1/4+\alpha} (\ud \Gamma(W_{N,x}) + \mathcal N_+).
\end{align*}
\end{lem}
\begin{proof}
First, note that $\ud \Gamma(W_{N,x})$ commutes with $\theta_M(\cN_+)$.
Using the representation~\eqref{eq:B^+-position} of $\cB^+_c$, we compute
\begin{align}
&[\ud \Gamma(W_{N,x}),\mathcal B_c^+] \label{eq:[Q_2,B^+]}\\
&	= 	\int \sqrt N \Vphi(y_2-y_3)(\Wx (y_3) + \Wx (y_2)-\Wx(y_1))\chi(y_1-y_3) a^*_{y_3}a^*_{y_2} a_{y_1} \notag.
\end{align}
Using the Cauchy Schwarz inequality similarly to~\eqref{eq:a-Q_2 bound}, we get for any $\delta>0$ the bound
\begin{align}
 &\pm\bigg(\int \sqrt N \Vphi(y_2-y_3)\chi(y_1-y_3)\Wx(y_1) a^*_{y_3}a^*_{y_2} a_{y_1} + \hc\bigg)\notag \\
 &\quad \leq \delta \ud \Gamma(W_{N,x}) + \delta^{-1} N \|W_{N}\|_1 \|\chi\|_2^2 \|\Vphi\|_2^2 \cN_+ (\cN_+ -1).
\end{align}
The other terms in~\eqref{eq:[Q_2,B^+]} satisfy the same bound with $\cN_+^2$ replacing $\cN_+(\cN_+-1)$.
With $N \|W_{N}\|_1 \|\chi\|_2^2 \|\Vphi\|^2\lesssim N^{-3/2+2\alpha}$ and choosing $\delta=\sqrt{M} N^{-3/4+\alpha}$, this yields
\begin{multline}
\pm \big(\theta_M(\cN_+)[\ud \Gamma(W_{N,x}),\mathcal B_c^+] + \hc \big) \\\lesssim (M/N)^{1/2} N^{-1/4+\alpha} (\ud \Gamma(W_{N,x}) + \cN_+)
\end{multline}
and proves the claim.
\end{proof}

\begin{lem}	\label{lem:quad_part_Bc}
Denoting $g = \sqrt N \ui \nabla \Wphi$, we have that for all $0<\alpha\leq 1/6$ and uniformly in $M \geq 1$, as quadratic forms on $\mathscr{H}_+$,
\begin{subequations}
 \begin{align} \label{eq:lem:quad_part_Bc_1}
 \pm \big([a^*(g_x)^2, \mathcal B_c] + \hc \big)  & \lesssim  \sqrt M  N^{-\frac14-\frac\alpha2}\big( \log N   \dd\Gamma(|\ui \nabla|) +1), \\
 \label{eq:lem:quad_part_Bc_2}
 \pm \big([a^*(g_x)a(g_x), \mathcal B_c] + \hc \big)
 &\lesssim  \sqrt M  N^{-\frac14-\frac\alpha2}\big( \log N   \dd\Gamma(|\ui \nabla|) +1), \\
  \label{eq:lem:quad_part_Bc_3}
 \pm \big([a^*(g_x)\ui \nabla_x, \mathcal B_c] + \hc \big)  & \lesssim  \sqrt M  N^{-\frac14-\frac\alpha2}\left( \dd\Gamma(|\ui \nabla|) +  (\cN_++1)^{1/2} |\ui \nabla_x| + 1\right).
\end{align}
\end{subequations}
\end{lem}

\begin{proof}
Note that $\|g\|_2\lesssim N^{1/4}$ by Lemma~\ref{lem:phi_W}. We prove successively the three inequalities.

\emph{Proof of (\ref{eq:lem:quad_part_Bc_1}).}
We start by observing that $[a^*(g_x), \cB^+_c]=0$ due to the cut-off $\1_{|q|\leq N^{\alpha}}$ in the definition of $\mathcal B_c^+$ and $\1_{|q|\leq N^{\alpha}} \widehat g(q)=0$.
For the commutator with $\cB_c^-$ we find
\begin{subequations}
\begin{align}
 [\cB_c^-,a^*( g_x)^2] 
 	&= 2 \sum_{p,q} \sqrt N \Vphih(p) \1_{|q|\leq N^{\alpha}} \widehat{g}_x(p+q)\widehat{g}_x(-p) a^*_{q}  \label{eq:B_c,a^2-1}\\
	&\quad  +2  \sum_{p,q} \sqrt N \Vphih(p) \1_{|q|\leq N^{\alpha}} \widehat{g}_x(p+q) a^*(g_x) a^*_{q}  a_{-p}\label{eq:B_c,a^2-2} \\
	&\quad + 2 \sum_{p,q} \sqrt N \Vphih(p) \1_{|q|\leq N^{\alpha}} \widehat{g}_x(-p)a^*(g_x) a^*_{q}a_{p+q} \label{eq:B_c,a^2-3}.
\end{align}
\end{subequations}
 With $\| \widehat g \ast \widehat g \Vphih\|_2 \leq \|\widehat g\|^2_2 \| \Vphih\|_2  \lesssim  N^{-1/2-\alpha/2}$ by Lemmas~\ref{lem:phi_V},~\ref{lem:phi_W},
 we find
 \begin{equation}
  \eqref{eq:B_c,a^2-1} = \mathcal{O}(N^{-\alpha/2} (\cN_++1)^{1/2}).
 \end{equation}
To bound the other terms, we will use the inequality
\begin{equation}
 \|a(\nabla f) \Psi \| \leq \||\nabla|^{1/2}f \|_2 \|(\ud \Gamma(|\ui \nabla|)+1)^{1/2} \Psi\| \label{eq:a(nablaf)-bound}.
\end{equation}
For the first one, this yields with $\delta>0$
\begin{equation}
\pm( \eqref{eq:B_c,a^2-2} + \hc) \lesssim \delta \| |\nabla|^{1/2} \Wphi\|_2^2 (\ud \Gamma(|\ui \nabla|)+1)+ \delta^{-1} \|g\|_2^2 \|\sqrt N \Vphi\|_2^2 (\cN_++1)^2.
\end{equation}
For the second one, we have similarly
\begin{equation}
\pm( \eqref{eq:B_c,a^2-3} + \hc) \lesssim \delta \| |\nabla|^{1/2} \Wphi\|_2^2 \ud \Gamma(|\ui \nabla|)+ \delta^{-1} N^{3\alpha} \|\widehat g\|_\infty^2 \|\sqrt N \Vphi\|_2^2 (\cN_++1)^2.
\end{equation}
Using Lemma~\ref{lem:phi_W} and $\alpha\leq 1/6$, combining with $\theta_M(\cN_+)$ and then choosing $\delta=\sqrt M N^{-1/4-\alpha/2}$, we find
\begin{multline}
\pm\big( [\cB_c^-,a^*( g_x)^2] \theta_M(\cN_+) + \hc\big) \\
\lesssim  \sqrt M N^{-1/4-\alpha/2} \big( \log N   \dd\Gamma(|\ui \nabla|) + \cN_++1\big).
\end{multline}
It remains to estimate 
\begin{equation}
 [a^*( g_x)^2, \theta_M(\cN_+)]\cB_c^+ + \hc = (\theta_M(\cN_+)- \theta_M(\cN_++2))a^*( g_x)^2\cB_c^+ + \hc
\end{equation}
With the bound~\eqref{eq:B^- bound} on $\cB_c^-$ and~\eqref{eq:a(nablaf)-bound} we obtain 
\begin{align}
& | \langle \Psi,  [\cB_c^-,a^*( g_x)^2] \theta_M(\cN_+) \Psi\rangle| \notag \\
&\quad \lesssim \|g\|_2 \| |\nabla|^{1/2} \sqrt N \Wphi\|_2 \|\cN_+^{1/2} \Psi\|  M^{-1} \|\sqrt N \Vphi\|_2 \|\1_{\cN_+ \leq 2M} (\cN_+ + 2)^{2} \Psi\| \notag \\
&\quad \lesssim  \sqrt M N^{-1/4-\alpha/2} \|(\cN_++1)^{1/2} \Psi\| ,
\end{align}
by Lemmas~\ref{lem:phi_V},~\ref{lem:phi_W}. This proves the claim on the commutator with $a^*(g_x)^2$. 

\emph{Proof of (\ref{eq:lem:quad_part_Bc_2}).} 
The argument for the commutator with $a^*(g_x)a(g_x)$ is similar to the bounds on~\eqref{eq:B_c,a^2-2},~\eqref{eq:B_c,a^2-3}, as all other terms vanish in this case. 

\emph{Proof of (\ref{eq:lem:quad_part_Bc_3}).}
Let us now deal with the term involving $a^*(g_x) \ui \nabla_x$. We have
\begin{subequations}
 \begin{align}
&[a^*(g_x) \ui \nabla_x, \mathcal B_c]= [a^*(g_x)\ui \nabla_x, \theta_M(\mathcal N_+) \mathcal B_c^+ - \mathcal B_c^- \theta_{M}(\mathcal N_+)] \notag \\
	&= (\theta_{M}(\mathcal N_+ -1) - \theta_{M}(\mathcal N_+)) a^*(g_x) \ui \nabla_x \mathcal B_c^+\label{eq:B_c-nabla-a_commute_1} \\
	&\quad + (\theta_{M}(\mathcal N_++1)- \theta_{M}(\mathcal N_+)) a^*(g_x)\ui \nabla_x \mathcal B_c^- \label{eq:B_c-nabla-a_commute_2}\\
	&\quad + \theta_{M}(\mathcal N_++1) \hspace{-12pt}\sum_{p,q \in 2\pi \Z^3\setminus\{0\}}\hspace{-12pt} \sqrt{N} \Vphih(p) \1_{|q|\leq  N^{\alpha}} \left(a^*_{-p}\widehat{g}_x (p+q)  + a^*_{p+q} \widehat{g}_x (p)\right) a_{q} \ui \nabla_x .\label{eq:B_c-nabla-a_commute_3}
\end{align}
\end{subequations}
We used here that $\nabla_x$ commutes with $\cB_c$, and that $[a^*(g_x) \ui \nabla_x, \mathcal B_c^+] =0$ because of the cutoffs in the definitions of $g$ and $\mathcal B_c^+$. To obtain the desired bound, we need distribute $\nabla_x$ to the left and right argument of the quadratic form, since otherwise we would obtain a bound by $-\Delta_x$ and not $|\nabla_x|$. To do this, we use that
\begin{align}
a^*(g_x) | \nabla_x|=  |\nabla_x|^{1/2} a^*(g_x) |\nabla_x|^{1/2}  + [a^*(g_x), | \nabla_x|^{1/2}] | \nabla_x|^{1/2}.
\end{align}
With $\ui \nabla_x \ue^{-\ui kx}=\ue^{- \ui kx} (\ui \nabla_x + k)$ and ${||p-k|^{1/2}-|p|^{1/2} |\leq |k|^{1/2}}$, we get that
\begin{align}
 \| [| \nabla_x|^{1/2}, a(g_x)]\Psi\| &\lesssim \sum_k |k|^{1/2} |\hat{g}(k)| \| a_k \Psi\| \lesssim \|g\|_{2} \| \ud \Gamma (|\ui \nabla|)^{1/2}\Psi\|.
\end{align}
Using this and the estimate on $\cB_c^-$ and $\cB_c^+$ (\ref{eq:B^- bound}), we obtain
\begin{align*}
&|\braket{\Psi,\eqref{eq:B_c-nabla-a_commute_1} \Psi} |\\
	&\lesssim M^{1/2} N^{-1/2} \|N \Vphi\|_2 \|g\|_2 \| (\mathcal N_+ +1)^{1/4} |\nabla_x|^{1/2} \Psi\|^2 \notag\\
	&\qquad + M^{1/4} N^{-1/2} \|N \Vphi\|_2 \|g\|_2  \| \ud \Gamma (|\ui \nabla|)^{1/2} \Psi\|  \| (\mathcal N_+ +1)^{1/4} |\nabla_x|^{1/2} \Psi\| \notag \\
	&\lesssim M^{1/2} N^{-1/4 -\alpha/2} \| (\mathcal N_+ +1)^{1/4} |\nabla_x|^{1/2} \Psi\|^2\notag \\
	&\qquad + M^{1/4} N^{-1/4 -\alpha/2}  \| \ud \Gamma (|\ui \nabla|)^{1/2} \Psi\|  \| (\mathcal N_+ +1)^{1/4} |\nabla_x|^{1/2} \Psi\|,
\end{align*} 
where we used that $|(\theta_{M}(\mathcal N_+-1)- \theta_{M}(\mathcal N_+))| \lesssim M^{-1} \mathds{1}_{\mathcal N_+\leq 2M}$ and ${\|N \Vphi\|_2 \lesssim N^{-\alpha/2}}$. Similarly, one shows that~\eqref{eq:B_c-nabla-a_commute_2} satisfies the same bound. For the third term, we commute again $|\nabla_x|^{1/2}$, and a similar reasoning leads to
\begin{align}
|\braket{\Psi,\eqref{eq:B_c-nabla-a_commute_3} \Psi}| &\lesssim  M^{1/2} N^{-1/2} \|N \Vphi\|_2  \|g\|_2 \| (\mathcal N_+ +1)^{1/4} |\nabla_x|^{1/2} \Psi\|^2 \\
	&\quad + M^{1/4} N^{-1/2} \|N \Vphi\|_2 \|\1_{|\cdot| \leq N^{\alpha}}\|_{2} \||\widehat{\nabla|^{1/2}  g}\|_{\infty} \notag \\
	&\qquad \times \|(\mathcal N_++1)^{1/2}\Psi\| \| (\mathcal N_+ +1)^{1/4} |\nabla_x|^{1/2} \Psi\| \notag \\
	&\lesssim  M^{1/2} N^{-1/4-\alpha/2} \| (\mathcal N_+ +1)^{1/4} |\nabla_x|^{1/2} \Psi\|^2  \notag\\
	&\quad + M^{1/4} N^{-1/2 +\alpha/2}  \|(\mathcal N_++1)^{1/2}\Psi\|  \|(\mathcal N_+ +1)^{1/4} |\nabla_x|^{1/2} \Psi\|, \notag
\end{align}
where we used that $\|\widehat{|\nabla|^{1/2} g}\|_{\infty} \lesssim \sqrt{N} \|p^{3/2} \Wphih\|_\infty \lesssim \sqrt{N} \|p^2 \Wphih\|_{\infty}^{3/4} \|\Wphih\|_{\infty}^{1/4} \lesssim N^{-\alpha/2}$ from Lemma \ref{lem:phi_W}. The above can be  absorbed in (\ref{eq:lem:quad_part_Bc_3}) as $\cN_+\leq \ud \Gamma(|\ui \nabla)$. This concludes the proof.
\end{proof}

\paragraph{Conservation estimates.}
Here, we prove that the operators appearing in the bound on the error term remain of the same order when transformed with $U_c(t)=\ue^{t\cB_c}$.

\begin{lem}
\label{lem:Uc_N}
For all $s\geq 0$, we have uniformly in $M \geq 1$ and $t\in [-1,1]$
\begin{align*}
 U_c(t)^* \mathcal N_+ U_c(t)  &=  \mathcal N_+ +\mathcal{O}\big( N^{-\alpha/2} (\mathcal N_+ +1)\big),\\
U_c(t)^*  (\mathcal N_+ +1)^{s} U_c(t)& =\mathcal{O}\big( (\mathcal N_+ +1)^{s}\big),
\end{align*}
as quadratic forms on $\mathscr{H}_+$.
\end{lem}
\begin{proof}
 This follows from a simple calculation of the commutator $[\cN_+,\cB_c]$ and Gr\"onwall's Lemma, see~\cite[Lem.10]{HaiSchTri-22}, and interpolation for non-integer $s$.
\end{proof}

\begin{lem}	\label{lem:Uc_p}
Let $0<\alpha\leq 1/6$, then we have as quadratic forms on $\mathscr{H}_+$
\begin{subequations}
 \begin{align}
U_c(t)^* \dd\Gamma(|\ui \nabla|) U_c(t) &=\mathcal{O}\big( \dd\Gamma(|\ui \nabla|)\big), \label{eq:Tc_p} \\
U_c(t)^* \mathcal{L}_4 U_c(t) &=\mathcal{O}\big( \mathcal{L}_4 + \mathcal N_+ + M^{-1}\cN_+^2 +1\big), \label{eq:Tc_L4} \\
U_c(t)^* \ud \Gamma(W_{N,x}) U_c(t) &=\mathcal{O}\big( \ud \Gamma(W_{N,x}) + \mathcal N_+ \big), \label{eq:Tc_Q2}
\end{align}
\end{subequations}
 uniformly in $s\in [-1,1]$  and  $ M \leq N^{1-3\alpha}$.
\end{lem}

\begin{proof}
We start with the inequality~\eqref{eq:Tc_p} and compute
\begin{align} \label{eq:dGamma(p)-B_c_commute}
 [\dd\Gamma(|\ui \nabla|),\mathcal B_c^+]
 	=  \sum_{p,q} \sqrt N \Vphih(p) \1_{|q|\leq N^\alpha} (|p+q| + |p| - |q|)  a^*_{p+q}a^*_{-p} a_{q}.
\end{align}
Using that $||p+q| - |q| | \leq |p|$, we obtain by the Cauchy-Schwarz inequality, for $\delta>0$
\begin{align}
\pm &\theta_M(\cN_+) [\dd\Gamma(|\ui \nabla|),\mathcal B_c^+] + \hc 
\notag \\
 & \lesssim \delta \theta_M^2 \sum_{p,q} |p| a^*_{p+q}a^*_{p} a_p a_{p+q} + \delta^{-1}\sum_{p,q}  N|p| \Vphih(p)^2 \1_{|q|\leq N^\alpha}    a_{q}^* a_q\notag \\
  &  \lesssim \delta M \dd\Gamma(|\ui \nabla|) + \delta^{-1} N^{-1}\log N \cN_+,
\end{align}
since $\||\ui\nabla|^{1/2} \Vphi\|_2^2 \lesssim N^{-2}\log N$ by Lemma~\ref{lem:phi_V}.
Taking $\delta = \sqrt{N/(M \log N)} $ this gives
\begin{equation}
\pm  [\dd\Gamma(|\ui \nabla|),\mathcal B_c] \lesssim \sqrt {M \log N /N} (\cN_+ + \dd\Gamma(|\ui \nabla|)).
\end{equation}
Let $f(t)=U_c(t)^* \dd\Gamma(|\ui \nabla|) U_c(t)$. Then, as $\sqrt {M \log N /N} \lesssim 1$ by assumption, we obtain with $\cN_+\leq \dd\Gamma(|\ui \nabla|)$ on $\mathscr{H}_+$.
\begin{equation}
 \tfrac{\ud}{\ud t} f(t) = U_c(t)^* [\dd\Gamma(|\ui \nabla|),\cB_c^+]  U_c(t) \lesssim f(t).
\end{equation}
The bound~\eqref{eq:Tc_p} now follows from Gr\"onwall's inequality.

To prove the inequality~\eqref{eq:Tc_L4}, we use the calculation of the commutator $[\cL_4,\cB_c]$ from Lemma~\ref{lem:Uc_L_4}.
The principal term satisfies, using the Cauchy-Schwarz inequality as in~\eqref{eq:theta_M-remove},
\begin{align}
 \sum_{p ,q \in 2\pi \Z^3\setminus\{0\}\atop  |q| \leq N^{\alpha}} \widehat{V}_N \ast \Vphi(p) a^*_{p+q} a^*_{-p} a_q + \hc =\mathcal{O}( \cL_4 + \cN_+).
\end{align}
The difference of the principal term and the full commutator is bounded by using Lemma~\ref{lem:Uc_L_4} with $\delta=1$ and the constraint $M\leq N^{1-3\alpha}$. The claim then follows from Lemma~\ref{lem:Uc_N} and Gr\"onwall's inequality, as above.

The inequality~\eqref{eq:Tc_Q2} follows from Lemma~\ref{lem:com_Q2_Bc} and  Gr\"onwalls inequality.
\end{proof}

\subsection{Low-energy bosons:  Bogoliubov's transformation}\label{sect:Bog}

In this section we apply the well known Bogoliubov transformation $U_\mathrm{B}$ (defined by its generator in~\eqref{eq:Bog-generator}) that diagonalizes the boson Hamiltonian for momenta below $N^\alpha$.
This goes back to the original work of Bogoliubov~\cite{Bogoliubov-47}, and has been employed in all subsequent works on the topic. In addition to the standard computations, that we will omit, we will also need to transform the boson-impurity interaction terms. 
The article~\cite{christensen2015} found a term in the expansion of the ground state energy $H_N$ in powers of the scattering lengths that does not come from the Bogoliubov--Fröhlich model. This arises from interaction terms that create or annihilate two excitations from the condensate, which are generated by acting on $\ud \Gamma(W_{N,x})$ with $U_\mathrm{B}$. In Lemma~\ref{lem:U_B-Q_2} below, we show that these terms are small relative to $\ud \Gamma(W_{N,x})$. Since the latter contributes to the energy at order one, this shows that the contribution of these interaction terms is small in $N$, in our setting.

The generator of the transformation is
\begin{equation}\label{eq:Bog-generator}
 \cB_\mathrm{B} = \frac{1}{8} \sum_{p\in 2\pi \Z^3 \setminus\{0\} \atop |p|\leq N^{\alpha}} \log\Big(1+\frac{16\pi \ao_{V}}{p^2}\Big)(a_p a_{-p} - \hc) .
\end{equation}
The action of the unitary $U_\mathrm{B}=\ue^{\cB_\mathrm{B}}$ on the creation and annihilation operators can be expressed by the explicit formulas
\begin{equation}\label{eq:Bogtrafo}
 U_\mathrm{B}^* a^*(f) U_\mathrm{B} = a^*(Cf) + a(\overline{Sf}), \qquad U_\mathrm{B}^* a(f) U_\mathrm{B} = a( Cf) + a^*(\overline{Sf}),
\end{equation}
where $C,S$ are the Fourier multipliers with coefficients
\begin{subequations}
 \begin{align}
 C(p)&=\frac12 \bigg(\frac{|p|}{(p^4+ 16\pi \ao_{V} p^2)^{1/4}}+ \frac{(p^4+ 16\pi \ao_{V} p^2)^{1/4}} {|p|}\bigg)\1_{|p|\leq N^\alpha}+\1_{|p|>N^\alpha} \\
 S(p)&=\frac12 \bigg( \frac{|p|}{(p^4+ 16\pi \ao_{V} p^2)^{1/4}}-\frac{(p^4+ 16\pi \ao_{V} p^2)^{1/4}} {|p|}\bigg)\1_{|p|\leq N^\alpha}.
\end{align}
Notice that $C(p)=\cosh(-\tfrac14 \log (1+16\pi \ao_{V}/p^2))$, and $S(p)=\sinh(-\tfrac14 \log (1+16\pi \ao_{V}/p^2))$.
\end{subequations}

The main result of this section is the following proposition, which states that for momenta below $N^\alpha$ the Hamiltonian acts almost like the Bogoliubov-Fr\"ohlich Hamiltonian~\eqref{eq:H_BF cutoff} with cutoff $\Lambda=N^\alpha$, except for the continued presence of the positive terms $\ud \Gamma(W_{N,x})$ and $\cL_4$.

\begin{prop}\label{prop:U_B}
Let $0<\alpha\leq 1/8$ and $U_c^* U_W^* U_q^* \mathcal{H}_N U_q U_W U_c$ be given by Proposition~\ref{prop:U_c}.
Then as quadratic forms on $\mathscr{H}_+$
\begin{align*}
 U_\mathrm{B}^* &U_c^* U_W^* U_q^* \mathcal{H}_N U_q U_W U_c U_\mathrm{B}  \\
  & = 4\pi  \mathfrak{a}_{V} (N-1)  + 8\pi  \ao_W \sqrt{N} + e_N^{(U_\mathrm{B})}  \\
    &\quad +\ud \Gamma(\epsilon) -\Delta_x + a(w_{N,x}) +a^*(w_{N,x})   +\mathcal{H}_{\mathrm{Int},>} +\ud \Gamma(W_{N,x}) + \mathcal{L}_4 + \mathcal{E}^{(U_\mathrm{B})},
 \end{align*}
 with
 \begin{equation*}
\widehat w_{N,x} (p)= \widehat{w}_x(p) \1_{0<|p|< N^\alpha}, \qquad \widehat{w}_x(p) = \frac{8\pi \ao_W |p| \ue^{-\ui px}}{(p^4+ 16\pi \ao_{V} p^2)^{1/4}},
 \end{equation*}
the high momentum interaction
 \begin{align*}
 \mathcal{H}_{\mathrm{Int},>} &=\Big(a^*\big(\sqrt{N}  (-\ui \nabla \Wphi)_x\big)^2 +\hc \Big)+ 2 a^*\big(\sqrt{N} (\ui \nabla \Wphi)_x\big)a\big(\sqrt{N} (\ui \nabla \Wphi)_x\big) \\
  &\quad   -2 a^*\big(\sqrt{N} (\ui \nabla \Wphi)_x\big)\ui \nabla_x  + \hc,
 \end{align*}
 the scalar,
 \begin{align*}
  e_N^{(U_\mathrm{B})}&=  4 \pi N (\aV - \ao_V) + 8 \pi \sqrt{N} (\aW - \ao_W) \\
  	&\quad
   +\sum_{0\neq p\in 2\pi\Z^3 \atop |p|\leq N^{\alpha}} \bigg(\frac12 \sqrt{p^4+16\pi \ao_{V} p^2}-\frac12p^2 -4\pi \ao_{V} +\frac{(4\pi  \ao_{V})^2}{p^2}+ \frac{2(4\pi \ao_W)^2}{p^2}\bigg),
 \end{align*}
and where the error satisfies
\begin{align*}
 \pm \mathcal{E}^{(U_\mathrm{B})} &\lesssim  N^{-\alpha/2}  \ud \Gamma(W_{N,x})  +N^{-\alpha/4}\mathcal{L}_4 + N^{-1/2-\alpha/4}\cN_+^2 \\
	& \quad +   N^{-\alpha/4}\big((\log N) \dd\Gamma(|\ui \nabla|) +(\log N)^2  +(\mathcal N_++1)^{1/2} |\nabla_x|\big).
\end{align*}
\end{prop}
As usual, we will first prove Proposition \ref{prop:U_B} assuming known the bounds given by Lemmas \ref{lem:U_B-Q_2} and \ref{prop:U_G}, which we will discuss in the remainder of this section.
\begin{proof}[Proof of Proposition \ref{prop:U_B}]
 First, note that $-\Delta_x$ and $\mathcal{H}_{\mathrm{Int},>}$ (which already appears in Proposition~\ref{prop:U_c}) are left unchanged by $U_\mathrm{B}$ since the generator does not depend on $x$ and acts only on momenta $|p|\leq N^\alpha$.
 By standard calculations (see e.g.~\cite{HaiSchTri-22}) using the relations~\eqref{eq:Bogtrafo}, one finds that
\begin{align}
U_\mathrm{B}^* &\big(\dd \Gamma(-\Delta + 8\pi \ao_{V})  + \widetilde{\mathcal{L}}_2\big)U_\mathrm{B} \notag \\
&= \sum_{p} \Big(\sqrt{p^4+16\pi \ao_{V} p^2} \1_{|p| \leq N^\alpha} +(p^2 +8\pi \ao_{V})\1_{|p|> N^\alpha} \Big) a^*_p a_p \notag \\
&\quad  + \sum_{ |p|\leq N^{\alpha}} \Big(\frac12 \sqrt{p^4+16\pi \ao_{V} p^2}-\frac12p^2 -4\pi \ao_{V} \Big).
\end{align}
The second term contributes to $e_N^{(U_\mathrm{B})}$ and the first equals $\ud \Gamma(\epsilon) + \mathcal O(N^{-2\alpha}\mathcal N_+)$ using that
\begin{align}
 |\epsilon(p) - (p^2 +8\pi \ao_{V})| \1_{|p|> N^\alpha} \lesssim N^{-2\alpha}.
\end{align}
For the interaction terms on momenta below $N^\alpha$, one finds with~\eqref{eq:Bogtrafo}
\begin{align}
 U_\mathrm{B}^*\big(a(\widetilde W_{N,x}) + a^*(\widetilde W_{N,x})\big)U_\mathrm{B}
 &= \sum_{p}(C(p)+S(p)) \ue^{-\ui px} \widehat{\widetilde W}_{N}(p) (a_p^*+a_{-p}) \notag \\
 &=\sum_{|p|\leq N^\alpha} \frac{\ue^{-\ui px}|p| 8\pi \ao_W}{(p^4+16\pi \ao_{V} p^2)^{1/4}}(a_p^*+a_{-p}),
\end{align}
which gives $w_{N,x}$.

The terms $\ud \Gamma(W_{N,x})$ and $\cL_4$ are preserved by $U_\mathrm{B}$ up to errors that can be absorbed in $\cE^{(U_\mathrm{B})}$ by Lemmas~\ref{lem:U_B-Q_2},~\ref{lem:U_B-conserve} below. Moreover, Lemma \ref{lem:U_B-conserve} also shows that $U_\mathrm{B}^* \cE^{(U_\mathrm{c})} U_\mathrm{B}$ can also be absorbed in the error, where in the case of the error term $(\cN_+ +1)^{1/2}|\nabla_x|$ we additionally make use of the fact that $\nabla_x$ commutes with both $U_\mathrm{B}$, and $\cN_+$, which gives
\begin{align}
U_\mathrm{B}^* (\cN_+ +1)^{1/2}|\nabla_x|U_\mathrm{B} &= |\nabla_x|^{1/2} U_\mathrm{B}^*(\cN_+ +1)^{1/2} U_\mathrm{B}|\nabla_x|^{1/2} \notag \\
&= \mathcal{O}\big( (\cN_+ +1)^{1/2}|\nabla_x|\big),
\end{align}
by Lemma~\ref{lem:U_B-conserve}.
This proves the claim.
\end{proof}

\begin{lem}\label{lem:U_B-Q_2}
We have for $\alpha>0$ as quadratic forms on $\mathscr{H}_+$
 \begin{equation*}
  U_\mathrm{B}^* \ud \Gamma(W_{N,x})  U_\mathrm{B} =\ud \Gamma(W_{N,x}) + \mathcal{O}\big(N^{-1/4}\ud \Gamma(W_{N,x}) + N^{-1/4}(\cN_+ +1)\big).
 \end{equation*}
 \end{lem}
\begin{proof}
 Writing $\ud \Gamma(W_{N,x})$ as
 \begin{equation}
  \ud \Gamma(W_{N,x}) = \sum_{p,q}  \ue^{-\ui x(p-q)}\widehat W_N(p-q) a^*_pa_q
 \end{equation}
we find from~\eqref{eq:Bogtrafo}
\begin{align}
 U_\mathrm{B}^* \ud \Gamma(W_{N,x})  U_\mathrm{B} &= \sum_{p,q} (C(p)C(q)+S(p)S(q)) \ue^{-\ui x(p-q)}\widehat W_N(p-q) a^*_pa_q  \notag \\
 &\quad + \sum_{p,q}C(p)S(q) \ue^{-\ui x(p-q)}\widehat W_N(p-q) a^*_pa_q^* + \hc \notag \\
 &\quad + \sum_p S(p)^2 \widehat W_N(0). \label{eq:Bogtrafo:2}
\end{align}
We have
\begin{equation}\label{eq:S-HS}
 S(p)^2 = \frac{\big(1-\sqrt{1+16\pi \ao_{V}/p^2}\big)^2}{4\sqrt{1+16\pi \ao_{V}/p^2}}\1_{|p|\leq N^\alpha}\leq \frac{(4\pi \ao_V)^2}{p^4}\in \ell^1(2\pi \Z^3).
\end{equation}
From this the last term in (\ref{eq:Bogtrafo:2}) is a $\mathcal O(N^{-1/2})$. To bound the $a^*a^*$ term, we use that $C(p) = 1 + H(p)$ with $H(p) = \sqrt{S(p)^2+1} -1$.
By the Cauchy-Schwarz inequality, we obtain
\begin{align}
 &\pm\sum_{p,q}S(q) \ue^{-\ui x(p-q)}\widehat W_N(p-q) a^*_pa_q^* + \hc = \pm \int W_{N,x}(y) \check S(y-y') a_y^* a_{y'}^* +\hc \notag\\
 &\quad \lesssim \delta \ud \Gamma(W_{N,x}) + \delta^{-1} \int W_{N,x} (y) \check S(y-y_1) \check S(y-y_2) a_{y_1} a^*_{y_2} \notag \\
 &\quad \lesssim N^{-1/4} \ud \Gamma(W_{N,x}) + N^{1/4} \|W_N\|_1 \|S\|_2^2 (\cN_+ +1 ),
\end{align}
where we took $\delta=N^{-1/4}$ and which can be absorbed it in the error term. The remaining term is estimated as
\begin{align}
 &\pm\sum_{p,q}H(p)S(q) \ue^{-\ui x(p-q)}\widehat W_N(p-q) a^*_pa_q^* + \hc \nn \\
 &\quad \lesssim \|\widehat W_N\|_{\infty} \|S\|_{2} \|H\|_2 (\mathcal N_+ +1)  \lesssim N^{-1/2} (\mathcal N_++1).
\end{align}
From~\eqref{eq:S-HS} we also immediately obtain 
\begin{equation}
 \sum_{p,q} S(p)S(q) \ue^{-\ui x(p-q)}\widehat W_N(p-q) a^*_pa_q = \mathcal{O}(N^{-1/2} \cN_+),
\end{equation}
and $\sum S(p)^2 \widehat  W(0) \lesssim N^{-1/2}$.
Writing again $C(p) = 1 + H(p)$ and arguing as above gives
\begin{align}
 &\sum_{p,q} C(p)C(q) \ue^{-\ui x(p-q)}\widehat W_N(p-q) a^*_pa_q \notag\\
 &\quad = \ud \Gamma(W_{N,x}) + \mathcal{O}\big(N^{-1/4}\ud \Gamma(W_{N,x}) + N^{-1/4}(\cN_+ +1)\big).
\end{align}
This proves the claim.
\end{proof}

 \begin{lem}\label{lem:U_B-conserve}
 For all $t\geq 0$ and $0\leq \alpha < 1/2$ we have as quadratic forms on $\mathscr{H}_+$
  \begin{align*}
   U_\mathrm{B}^* \cN_+^{t}  U_\mathrm{B}&= \mathcal{O}\big(\cN_+^{t} +1\big), \\
   U_\mathrm{B} \cN_+^{t}  U_\mathrm{B}^*&= \mathcal{O}\big(\cN_+^{t} +1\big),\\
    U_\mathrm{B}^* \cL_4  U_\mathrm{B}&= \cL_4 + \mathcal{O}\big(N^{-1+2\alpha}+N^{-1}\cN_+^2 \big), \\
    U_\mathrm{B}^* \ud \Gamma(|\ui \nabla|)  U_\mathrm{B}&=   \mathcal{O}\big(\ud \Gamma(|\ui \nabla|) + \log N  \big).
  \end{align*}
 \end{lem}
\begin{proof}
The first three estimates follow from the same arguments as in~\cite[Lem.18]{HaiSchTri-22} (first, for integer $t$, and then for all $t$ by interpolation).
The final inequality follows from the fact that $\| |p|^{1/2}S \|_2 \lesssim \log N$  and $|p|S(p)^2\lesssim 1$, $C(p)\lesssim 1$, by an argument similar to Lemma~\ref{lem:U_B-Q_2}.
\end{proof}

\subsection{Excitation-impurity scattering: a second Weyl transformation}\label{sect:Gross}

In this section we apply another Weyl transformation, similar to the one from Section~\ref{sect:Weyl}. However, it acts now on small momenta and the form factor of the transformation is obtained from an equation involving the Bogoliubov dispersion relation $\epsilon$, instead of simply $-\Delta$, and the interaction $w_N$ instead of the potential $W_N$.
The transformation is closely related to the first step of the renormalization in Section~\ref{sect:renorm}.

Let $\kappa>0$, to be chosen later, and let $\alpha>0$ be the parameter of the previous sections. We define the function $f\in L^2(\T^3)$ by
\begin{equation}
\widehat f_x(p) = -\frac{\widehat w_{N}(p)\ue^{\ui px}}{p^2 + \epsilon(p)} \1_{\kappa < |p| }
= \frac{-8\pi \ao_W |p| \ue^{\ui px} \1_{\kappa < |p| \leq  N^\alpha}}{(p^2 + \sqrt{p^4+ 16\pi \ao_{V} p^2})(p^4+16\pi \ao_{V} p^2)^{1/4}}.
\end{equation}
Note that this corresponds to $f_\kappa^\Lambda$ with $\Lambda = N^{\alpha}$ in (\ref{eq:def_f_sec_renorm}). One easily checks that $\|f_x\|_2\lesssim 1$, uniformly in $\kappa$ and $N$.
We then define the following unitary transformation, which is similar to the Gross transformation used in the renormalization of some polaron models~\cite{Nelson-64,GrWu-18},
\begin{equation}
 U_\mathrm{G}(t) = \exp\big( t(a^*(f_x) -a(f_x))\big), \qquad U_\mathrm{G}=U_\mathrm{G}(1).
\end{equation}
After this transformation, the Hamiltonian from Proposition~\ref{prop:U_B} resembles closely the operator~\eqref{eq:H_BF Weyl} appearing in the renormalization procedure.

\begin{prop}\label{prop:U_G}
Let $\kappa>0$, $0<\alpha \leq 1/10$ and $M = N^{1/2+\alpha/2}$, denote $U=U_\mathrm{G} U_q U_W U_c U_\mathrm{B}$ and recall the transformation of $\cH_N$ from Proposition~\ref{prop:U_B}. Then, for $N^\alpha>\kappa$
 \begin{align}
  U^* \cH_N U  &=   4\pi  \mathfrak{a}_{V} (N-1)  + 8\pi  \ao_W \sqrt{N}   + e^{(U)}_N \nn \\
   &\quad  -\Delta_x +\ud\Gamma(\epsilon)  + a(w^\kappa_x) +a^*(w^\kappa_x) + a^*( \ui \nabla v_{N,x})^2 +a(  \ui \nabla v_{N,x})^2    \nn \\
  &\quad + 2 a^*( \ui \nabla v_{N,x}) a( \ui \nabla v_{N,x}) + 2 a^*( \ui \nabla_x v_{N,x}) \ui \nabla_x  + 2 \ui  \nabla_x a( \ui \nabla_x v_{N,x}) \nn \\
  &\qquad + \ud \Gamma(W_{N,x}) +\cL_4 +\cE^{(U)}, \label{eq:prop:U_G}
 \end{align}
where
\begin{align*}
 \widehat w^\kappa_x (p) &= 8 \pi a(W)|p| \epsilon(p)^{-1/2} \ue^{\ui p x} \1_{0<|p|\leq \kappa}, \\
 \widehat v_{N,x} &=\left\{
\begin{aligned}
 \frac{8 \pi \ao_W}{p^2 + \epsilon(p)}  \frac{ |p|}{ \sqrt{\epsilon(p)}} \ue^{\ui p x} &\qquad \kappa<|p|\leq N^\alpha, \\
  \sqrt N \Wphih(p)  \ue^{\ui p x} &\qquad N^\alpha<|p|,
\end{aligned}\right.
\end{align*}
the scalar term has the value
\begin{align*}
 e^{(U)}_N &=  4 \pi N (\aV - \ao_V) + 8 \pi \sqrt{N} (\aW - \ao_W)  \\
 &\quad + \frac12\sum_{0\neq p\in 2\pi\Z^3 \atop |p|\leq N^{\alpha}} \bigg( \sqrt{p^4+16\pi \ao_{V} p^2}- p^2 - 8\pi \ao_{V} +2 \frac{(4\pi  \ao_{V})^2}{p^2}\bigg) \\
  &\quad + (8\pi \ao_W)^2 \sum_{0\neq p\in 2\pi\Z^3 \atop |p|<N^{\alpha}} \bigg(\frac{1}{2p^2} - \frac{ p^2 \1_{|p|>\kappa}}{(p^2+\epsilon(p)) \epsilon(p)}\bigg),
  \end{align*}
and the error term satisfies
\begin{multline}
 \pm \cE^{(U)} \lesssim N^{-\alpha/2}  \ud \Gamma(W_{N,x})  +N^{-\alpha/4}\mathcal{L}_4 + N^{-1/2-\alpha/4}\cN_+^2  \\
	 +   N^{-\alpha/4}\big((\log N) \dd\Gamma(|\ui \nabla|) +(\log N)^2 +(\mathcal N_++1)^{1/2} |\nabla_x|\big). \label{eq:est_E_U}
\end{multline}
\end{prop}

\begin{proof}
Recall from Proposition \ref{prop:U_B} that
\begin{align}
 U^* \cH_N U &= 4\pi  \mathfrak{a}_{V} (N-1)  + 8\pi  \ao_W \sqrt{N} + e_N^{(U_\mathrm{B})}+\ud \Gamma(\epsilon) -\Delta_x \nn\\
  &\quad  + a(w_{N,x}) +a^*(w_{N,x})   +\mathcal{H}_{\mathrm{Int},>} +\ud \Gamma(W_{N,x}) + \mathcal{L}_4 + \mathcal{E}^{(U_\mathrm{B})}.
\end{align}
We have the identities (compare~\eqref{eq:Gross_1},~\eqref{eq:Gross_2} and~\eqref{eq:Weyl-x-shift})
 \begin{subequations}
\begin{align}
 U_\mathrm{G}^* a_y U_\mathrm{G} &= a_y + f_x(y), \label{eq:U_G-shift} \\
 U_\mathrm{G}^* (\ui \nabla_x) U_\mathrm{G} &= \ui \nabla_x + a^*(\ui \nabla_x f_x) + a(\ui \nabla_x f_x),
\end{align}
 \end{subequations}
which gives with $\nabla_x f_x=-(\nabla f)_x$
 \begin{align}
  &U^*_\mathrm{G} \big(-\Delta_x + \ud \Gamma(\epsilon) + a(w_{N,x})+a^*(w_{N,x})\big) U_\mathrm{G} \notag \\
  &\quad = -\Delta_x + \ud \Gamma(\epsilon) + a(w_{x}^\kappa)+a^*(w_{x}^\kappa)
   + a(\ui \nabla f_x)^2 + a^*(\ui \nabla f_x)^2 \notag\\
  &\qquad + 2 a(\ui \nabla f_x)a^*(\ui \nabla f_x) - 2a^*(\ui \nabla f_x) \ui \nabla_x - 2 \ui \nabla_x a(\ui \nabla f_x) \notag \\
  &\qquad - \sum_{\kappa<|p|\leq N^\alpha}\frac{\widehat w(p)^2}{(p^2+\epsilon(p))}.
  \label{eq:U_G-main}
 \end{align}
We now pass to $\mathcal{H}_{\mathrm{Int},>}$, which we recall is defined in Proposition \ref{prop:U_B}. Note that the transformation $U^*_\mathrm{G}$ leaves invariant the terms involving momenta greater than $N^\alpha$. Therefore, all the terms in $\mathcal{H}_{\mathrm{Int},>}$ remain unchanged except the ones containing the gradient $\nabla_x$. This gives, with $\nabla_x f_x=-(\nabla f)_x$ 
\begin{align}
  U_\mathrm{G}^*\mathcal{H}_{\mathrm{Int},>} U_\mathrm{G} &=
  \mathcal{H}_{\mathrm{Int},>}
  + 2 a^*\big(\sqrt{N} \ui \nabla \Wphix\big)(a(\ui \nabla f_x) + a^*(\ui \nabla f_x))  + \hc \label{eq:U_G-H>}
\end{align}
Introducing $v_{N,x} =  f_x + \sqrt{N}\Wphix$ and noting that the additional terms in~\eqref{eq:U_G-H>} complete the square with (\ref{eq:U_G-main}), we find
\begin{align}
 \eqref{eq:U_G-main}+\eqref{eq:U_G-H>} &=-\Delta_x + \ud \Gamma(\epsilon) + a(w_{x}^\kappa)+a^*(w_{x}^\kappa)- \sum_{\kappa<|p|\leq N^\alpha}\frac{\widehat w(p)^2}{(p^2+\epsilon(p))} \notag \\
  &\quad + a^*( \ui \nabla v_{N,x})^2 +a( \ui \nabla v_{N,x})^2 + 2 a^*(  \ui \nabla v_{N,x}) a(  \ui \nabla v_{N,x}) \notag \\
  &\quad - 2 a^*( \ui \nabla v_{N,x}) \ui \nabla - 2 \ui  \nabla a( \ui \nabla v_{N,x})
.
\end{align}
It remains to consider $\ud \Gamma(W_{N,x})$, $\cL_4$ and $\cE^{(U_\mathrm{B})}$.
With the shift property~\eqref{eq:U_G-shift}, we obtain
\begin{equation}
 U_\mathrm{G}^*\ud \Gamma(W_{N,x})  U_\mathrm{G}  = \ud \Gamma(W_{N,x}) + a^*((W_N f)_x) + a((W_N f)_x) + \langle f, W_N f\rangle.
\end{equation}
Next, one easily checks that
\begin{equation}
 \langle f, W_N f\rangle \lesssim \|W_N\|_{1} \|f\|_\infty^2 \lesssim N^{-1/2+ 2\alpha},
\end{equation}
and, using the Cauchy-Schwarz inequality, one obtains
\begin{align}
  a^*((W_N f)_x) + a((W_N f)_x) 
  	&\lesssim N^{-\alpha/2}  \ud \Gamma(W_{N,x}) +  N^{\alpha/2} \langle f, W_N f\rangle \nn \\
  	&\lesssim  N^{-\alpha/2}  \ud \Gamma(W_{N,x}) + N^{-\alpha/2}, \label{eq:UG_Q2}
\end{align}
for $\alpha \leq 1/6$. This gives
\begin{equation}\label{eq:Q_2-U_G}
 U_\mathrm{G}^*\ud \Gamma(W_{N,x})  U_\mathrm{G}  = \ud \Gamma(W_{N,x}) + \mathcal{O}(N^{-\alpha/2} \left(\ud \Gamma(W_{N,x}) + 1\right)).
\end{equation}
We now estimate the transformation of the quartic term $\cL_4$, which gives
\begin{subequations}
 \begin{align}
 & U_\mathrm{G}^*\mathcal{L}_4 U_\mathrm{G} -\cL_4=   \int  V_N(y_1-y_2) f_x(y_1) a^*_{y_1}a^*_{y_2} a_{y_2} +\hc\label{eq:U_G-cubic-error} \\
 &\quad +\frac12 \int   V_N(y_1-y_2) f_x(y_1) f_x(y_2) a^*_{y_1}a^*_{y_2} +\hc \label{eq:U_G-L_4-quad} \\
 &\quad +\int    V_N(y_1-y_2) f_x(y_1) f_x(y_2) a^*_{y_1}a_{y_2}
 \label{eq:U_G-L_4-qmixed}\\
 &\quad + \ud \Gamma(V_N \ast f_x^2) +\Big( a\big( ( V_N \ast f_x^2)  f_x\big)+ \hc\Big) \label{eq:U_G-L4-dGamma-lin} \\
 &\quad + \frac12 {\int  f_x(y_1)^2f_x(y_2)^2 V_N(y_1-y_2)}. \label{eq:U_G-L_4-scal}
\end{align}
\end{subequations}
Since $\|f_x\|_\infty \lesssim \|\widehat f_x\|_1 \lesssim N^\alpha$ and $\|f_x\|_2 \lesssim 1$ we have
\begin{align}
 \eqref{eq:U_G-L_4-scal} \lesssim \|V_N\|_1 \|f_x\|_\infty^2 \|f_x\|_2^2 \lesssim  N^{-1+2\alpha}.
\end{align}
Similarly, we obtain
\begin{align}
 \eqref{eq:U_G-L4-dGamma-lin} = \mathcal{O}(N^{-1+2\alpha} (\cN_+ +1)).
\end{align}
Using the  Cauchy-Schwarz inequality as in~\eqref{eq:L_2-L_4 bound} and~\eqref{eq:bound_L3} yields
\begin{subequations}
\begin{align}
 %
 \eqref{eq:U_G-L_4-quad} &= \mathcal{O}\big(N^{-\alpha/2}\cL_4 + N^{-1+5\alpha/2}\big) \\
 \eqref{eq:U_G-cubic-error}&=\mathcal{O}\big(N^{-\alpha/2} \cL_4 + N^{\alpha/2} \ud \Gamma(V_N \ast f_x^2)\big)=\mathcal{O}\big(N^{-\alpha/2} \cL_4 + N^{-1+5\alpha/2}\cN_+\big).
\end{align}
\end{subequations}
These estimates, together with  $\alpha\leq 1/10$, show that
\begin{equation}\label{eq:L_4-U_G}
 U_\mathrm{G}^*\mathcal{L}_4 U_\mathrm{G} = \cL_4 + \mathcal{O}\big( N^{-\alpha/2}( \mathcal L_4 + \cN_+ +1)\big).
\end{equation}

Finally, we estimate the action of $U_G$ on $\cE^{(U_\mathrm{B})}$. Similarly as in (\ref{eq:UG_Q2}), $\||p|^{1/2} \widehat f \|_2^2 \lesssim \log N$ yields
\begin{equation}\label{eq:U_G-p}
 U_\mathrm{G}^* \ud \Gamma(|\ui \nabla|)U_\mathrm{G} = \mathcal{O}(\ud \Gamma(|\ui \nabla|) + \log N)
\end{equation}
With $\|f\|_2\lesssim 1$, one easily obtains
\begin{equation}\label{eq:U_G-N}
 U_\mathrm{G}^*(t)\cN_+^{k} U_\mathrm{G}(t) = \mathcal{O}(\cN_+^{k}+1)
\end{equation}
for $k\in \N$.
Moreover,
\begin{align}
 &U_\mathrm{G}^*(\cN_+ + 1)^{1/2}|\nabla_x| U_\mathrm{G} \\
 &\quad = (\cN_+ + 1)^{1/2}|\nabla_x| +\int_0^1 U_\mathrm{G}(t)^*\big( [(\cN_+ + 1)^{1/2}|\nabla_x|,a^*(f_x)] +\hc\big) U_\mathrm{G}(t). \notag
\end{align}
Let us estimate the commutator
	\begin{align}
& [(\cN_+ + 1)^{1/2}|\nabla_x|,a^*(f_x)] \nn \\ 
 	&=  (\cN_+ + 1)^{1/2} [|\nabla_x|,a^*(f_x)]  +  [(\cN_+ + 1)^{1/2},a^*(f_x)] |\nabla_x|. \label{eq:U_G-comm}
 	\end{align}
To estimate the first term, we use that $\ui \nabla_x \ue^{-\ui kx}=\ue^{- \ui kx} (\ui \nabla_x + k)$ and the triangle inequality $ | |\ui \nabla_x+ p| - | \ui \nabla_x| | \leq |p|$ for all $p$, to obtain that
\begin{align}
 &|\langle \Psi, [|\ui  \nabla_x|,a^*(f_x)] \Psi\rangle|= \Big|\Big\langle\Psi, \sum_{p}  (|\ui  \nabla_x|-|\ui \nabla_x- p|) \widehat f_x(p) a^*_p \Psi\Big|\Big\rangle \notag\\
 &\quad\leq \| \ud \Gamma(|\ui \nabla|)^{1/2} \Psi\| \bigg(\sum_{p}\Big\|\frac{(|\ui \nabla_x|-|\ui \nabla_x- \ui p|)^2 |\widehat f_x(p)|^2}{|p|} \Psi \Big\|^2\bigg)^{1/2} \notag \\
 &\quad \lesssim \sqrt{\log N} \| \ud \Gamma(|\ui \nabla|)^{1/2} \Psi\| \|\Psi\|,
\end{align}
where we used that $\||p|^{1/2} \widehat{f}\|_2^2 \lesssim \log N$. Hence, the first term in (\ref{eq:U_G-comm}) is bounded by
\begin{equation}
 (\cN_+ + 1)^{1/2} [|\nabla_x|,a^*(f_x)] + \hc = \mathcal{O}\big((\ud \Gamma(|\ui \nabla|) +\log N (\cN_+ +1)) \big).
\end{equation}
To bound the second term in (\ref{eq:U_G-comm}), we write
\begin{align}
 a^*(f_x)|\nabla_x| &= |\nabla_x|^{1/2} a^*(f_x)|\nabla_x|^{1/2} -[|\nabla_x|^{1/2}, a^*(f_x)]  |\nabla_x|^{1/2}
\end{align}
and use again that $\ui \nabla_x \ue^{-\ui kx}=\ue^{- \ui kx} (\ui \nabla_x + k)$ and that $ | |\ui \nabla_x+ p|^{1/2} - | \ui \nabla_x|^{1/2} | \leq |p|^{1/2}$ from which we obtain
\begin{align}
  [(\cN_+ + 1)^{1/2},a^*(f_x)]|\nabla_x| 
  	&= ((\cN_+ + 1)^{1/2} - \cN_+^{1/2})a^*(f_x)|\nabla_x| \notag \\
	 & = \mathcal{O}\big(|\nabla_x|  + \log N  \big), \label{eq:comm_a_nabla}
\end{align}
where we used that $\||p|^{1/2} \widehat{f}\|_2^2 \lesssim \log N$ and that $(\cN_+ + 1)^{1/2} - \cN_+^{1/2} \leq (\cN_+ + 1)^{-1/2}$. 

It thus follows from (\ref{eq:U_G-p}),~\eqref{eq:U_G-N} and Gr\"onwall's inequality that
\begin{equation}
 U_\mathrm{G}^* (\cN_++1)^{1/2}|\nabla_x|U_\mathrm{G}= \mathcal{O}\Big((\cN_++1)^{1/2}|\nabla_x| + \ud \Gamma(|\ui \nabla|) + \log N (\cN_++1)  \Big).
\end{equation}
This completes the proof of the proposition as $\cN_+\leq \dd\Gamma(|i\nabla|)$ on $\mathscr{H}_+$.
\end{proof}

\section{Convergence of the transformed Hamiltonian}\label{sect:conv}

In the section, we study the main part of the Hamiltonian appearing at the end of the previous section in Proposition \ref{prop:U_G}.
The remaining singular contribution, which will give rise to the $\log N$-term in the energy, is the interaction term $a^*(\ui \nabla_x v_{N,x})^2 +\hc$ from \eqref{eq:prop:U_G} (compare Section~\ref{sect:renorm}). For technical reasons, we will cut this term off at particle numbers $\cN_+> M=N^{1-\alpha}$ for the same $\alpha>0$ as before.

We can rephrase the result of Proposition \ref{prop:U_G} as
\begin{multline}\label{eq:H_NU-error}
U^*\cH_N U = \mathcal H_{\mathrm{Sing}}  + R_N+ 4\pi  \mathfrak{a}_{V} (N-1)  + 8\pi  \ao_W \sqrt{N} + e_N^{(U)} \\ +\cL_4  + \cE^{(U)} + \mathcal{O}(N^{-3/2+2\alpha} (\cN_++1)^3).
\end{multline}
Here,
\begin{align}
 \mathcal H_{\mathrm{Sing}} &= -\Delta_x +\ud\Gamma(\epsilon) +\ud\Gamma(W_{N,x}) \nn\\
 &\qquad + \1_{\cN_+\leq N^{1-\alpha}} a^*(\ui \nabla v_{N,x})^2 +a(\ui \nabla v_{N,x})^2 \1_{\cN_+\leq N^{1-\alpha}}\label{eq:Hsing_W}
\end{align}
generalizes the singular Hamiltonian from Section~\ref{sect:renorm}, and
\begin{align}\label{eq:R_N:sec5}
  R_N &=2 a^*(\ui \nabla v_{N,x}) a(\ui \nabla v_{N,x})-2 a^*(\ui \nabla v_{N,x}) \ui \nabla_x  - 2 \ui  \nabla_x a(\ui \nabla v_{N,x}) \notag\\
	 	&\quad +  a^*(w^\kappa_x) +a(w^\kappa_x),
\end{align}
collects  the remaining terms from the transformed operator of Proposition \ref{prop:U_G}, except $\cL_4$, the constants, and the error terms, and
we recall that for $\alpha>0$
\begin{align}\label{eq:w,v-def}
\begin{aligned}
  \widehat w^\kappa_x (p) &= 8 \pi \ao_W|p| \epsilon(p)^{-1/2} \ue^{\ui p x} \1_{0<|p|\leq \kappa} \\
 \widehat v_{N,x}(p) &=\left\{
\begin{aligned}
 \frac{8 \pi \ao_W}{p^2 + \epsilon(p)} \frac{ |p| }{\sqrt{\epsilon(p)}} \ue^{\ui p x} &\qquad \kappa<|p|\leq N^\alpha \\
  \sqrt N \Wphih(p)  \ue^{\ui p x} &\qquad N^\alpha<|p|.
\end{aligned}\right.
\end{aligned}
\end{align}
Note that, since $ \sqrt N \Wphih(p) \lesssim |p|^{-2}$ by Lemma~\ref{lem:phi_W},
\begin{equation}\label{eq:v_N-bound}
 \widehat  v_N(p) \lesssim |p|^{-2},
\end{equation}
and that $ \widehat v_{\infty}$ is also a well-defined function. Also by Lemma~\ref{lem:phi_W} with $s=1/2$, we have $\|\nabla v_{N,x}\|_2 \lesssim N^{1/4}$, so, by the argument of~\eqref{eq:B_q-N commute},
\begin{align}\label{eq:a^2>M-bound}
 \1_{ \1_{\cN_+> M}} a^*(\ui \nabla_x v_{N,x})^2 + \hc =\mathcal{O}\big(M^{-2} \sqrt{N} (\cN_++1)^3\big).
\end{align}
Cutting off the interaction in~\eqref{eq:Hsing_W} thus indeed gives rise to errors of the order stated in~\eqref{eq:H_NU-error}.

For later use, we introduce the notation for the relevant part of the transformed Hamiltonian
\begin{equation}\label{eq:H_N-U}
 \cH_N^U:=\mathcal H_{\mathrm{Sing}}  + R_N.
\end{equation}
Without the term $\dd\Gamma(W_{N,x})$ in $\mathcal H_{\mathrm{Sing}}$, one should think of $\cH_N^U$ as the Bogoliubov-Fr\"ohlich Hamiltonian $H_\mathrm{BF}^\Lambda$ defined in (\ref{eq:def_BF}) for $\Lambda =N^\alpha$ dressed with the appropriate Weyl transformation (compare with the expression given in (\ref{eq:H_BF Weyl})).
This will be discussed in detail in Section~\ref{sect:W=0}.

We will then show in Section~\ref{sec:H_W} that the extra term $\dd\Gamma(W_{N,x})$ only shifts the limiting operator by a constant of order one. Indeed, the divergence in $\cH_{\mathrm{sing}}$ is expressed by the scalar (compare~\eqref{eq:E^2})
\begin{align}
&E_{N,W}:= -2 \Big\langle \nabla  v_N \otimes \nabla v_N , \label{eq:E_NW}\\
&  \Big((\nabla_{y_1}+ \nabla_{y_2})^2 + \epsilon(\ui\nabla_{y_1})+\epsilon(\ui\nabla_{y_2})+ W_N(y_1)+W_N(y_2)+1\Big)^{-1} \nabla v_N \otimes \nabla v_N\Big \rangle,  \nn
\end{align}
which differs from the same quantity with $W_N$ omitted from the resolvent at order one, as shown in the proof of Theorem~\ref{thm:asymptotic}.
The results of this Section are summarized by:
\begin{prop} \label{prop:CV_transf_Ham}
Recall the unitary operator $U_{\kappa}^\infty$ defined in (\ref{eq:def_U_kappa_Lambda}). For $\kappa$ large enough and $0<\alpha\leq 1/2$ we have
\begin{align*}
\lim_{N\to \infty}\big(\cH_N^U - E_{N,W}\big)= (U_\kappa^\infty)^* H_\mathrm{BF}U_\kappa^\infty
\end{align*}
in norm resolvent sense.
\end{prop}
\begin{proof}[Proof of Proposition \ref{prop:CV_transf_Ham}]
In Proposition \ref{prop:H_BF trafo conv}, we defined a self-adjoint operator associated to the quadratic form $K_{\infty} +T_{\infty} + R_\infty $, whose domain is contained in $Q(K_{\infty})=(1-G_{\infty})^{-1}Q(-\Delta_x + \ud \Gamma(\epsilon))$.
In (\ref{eq:H_BF_Ukappa_infty}) we proved that
\begin{align}
 H_\mathrm{BF}=U_\kappa^\infty(K_{\infty} + T_{\infty}+R_\infty-1)(U_\kappa^\infty)^*.
\end{align}
On the other side of the equation, we will show in Proposition \ref{prop:HsingW} that $\cH_N^U-E_{N,W}$ converges to $K_{\infty} + T_{\infty} + R_{\infty}-1$ in norm resolvent sense.
Together, these prove the claim.
\end{proof}

\subsection{Analysis of the simplified singular Hamiltonian} \label{sect:W=0}
Let us denote
\begin{align}
\mathcal H_{\mathrm{Sing},0} = -\Delta_x +\ud\Gamma(\epsilon) + a^*( \ui \nabla v_{N,x})^2 +a(  \ui \nabla v_{N,x})^2.
\end{align}
Note that here we do not cut off the interaction at $\cN_+\leq N^{1-\alpha}$, as this is unnecessary when the interaction term $\ud \Gamma(W_{N,x})$ is absent, and such a cutoff is not used in Section~\ref{sect:renorm}.
Following the procedure from Section~\ref{sect:renorm}, we introduce
\begin{equation}
G_{N}:=-  (-\Delta_x+ \ud \Gamma(\epsilon) +1 )^{-1}a^*(  \ui \nabla v_{N,x} )^2,
\end{equation}
and rewrite
\begin{align}\label{eq:Hsing_0}
 \mathcal H_{\mathrm{Sing},0}  = K_N + T_N + E_N-1
\end{align}
with
\begin{equation}
\begin{aligned}
K_{N } &:= (1-G^*_{N })(-\Delta_x+ \ud \Gamma(\epsilon)  +1)(1- G_{N }), \\
T_{N } &:=  -a(\ui \nabla v_{N,x} )^2(-\Delta_x+ \ud \Gamma(\epsilon)+1)^{-1} a^*(\ui \nabla v_{N,x} )^2 - E_{N } \\
\end{aligned}
\end{equation}
and
\begin{equation}\label{eq:E_N}
 E_{N}:= -2 \Big\langle \nabla  v_N \otimes \nabla v_N ,
 \Big((\nabla_{y_1}+ \nabla_{y_2})^2 + \epsilon(\ui\nabla_{y_1})+\epsilon(\ui\nabla_{y_2})+1\Big)^{-1} \nabla v_N \otimes \nabla v_N\Big \rangle.
\end{equation}
We also use the definitions of $G_{\infty}$, $K_{\infty}$, $T_{\infty}$ from Section~\ref{sect:renorm}.

We now show regularity estimates for $G_N$, and its convergence rate to $G_{\infty}$. We then use these to infer comparison estimates between the dressed $K_{N}$ and the bare Hamiltonian $-\Delta_x+\ud \Gamma(\epsilon)$, as well as its convergence rate towards $K_{\infty}$.

\begin{lem}[Properties of $G_{N}$]\label{lem:G_0_bound}
 For all $0\leq s<1/2$ there exists $(C_\kappa)_{\kappa>0}$, $C$ with $\lim_{\kappa\to \infty} C_\kappa=0$ so that for all $\kappa>0$, $N\in \N\cup\{\infty\}$
\begin{subequations}
  \begin{equation}
  \|G_N^*(-\Delta_x+\ud \Gamma(\epsilon))^s (\cN_++1)^{-s} \| \leq C_\kappa,\label{eq:lem:G_0_bound_1}
 \end{equation}
and for all $\alpha>0$
\begin{equation}
 \|(G_N^*-G_{\infty}^*) (-\Delta_x+\ud \Gamma(\epsilon))^s(\cN_++1)^{-s}\|\leq C N^{(2s-1)\alpha}.
 \label{eq:lem:G_0_bound_diffN}
\end{equation}
Moreover, for $1/2\leq s\leq 1$, $\eps>0$ there is $C$ so that for all $\kappa>0$, $N\in \N$
\begin{equation}
 \|G_N^*(-\Delta_x+\ud \Gamma(\epsilon))^{s}(\cN_++1)^{-s}\|  \leq C N^{(s-\frac12)_++\eps}.\label{eq:lem:G_0_bound_2}
\end{equation}
\end{subequations}
\end{lem}

\begin{proof}
By the Cauchy-Schwarz inequality, we have
  \begin{align}
&a^*(\ui \nabla v_{N,x})^2a(\ui \nabla v_{N,x})^2 \\
&= \sum_{p,q,k,r} (k\cdot r) (p\cdot q) \widehat v_{N,x}(k) \widehat v_{N,x}(r) \widehat v_{N,x}(p) \widehat v_{N,x}(q) a^*_k a^*_r a_p a_q \nn \\
&\leq \sum_{p,q,k,r} (p\cdot q)^2 \frac{|\widehat v_N(p)|^2}{\epsilon(p)^{1-s}} \frac{|\widehat v_N(q)|^2}{\epsilon(q)^{1-s}}  \epsilon(k)^{1-s} \epsilon(r)^{1-s}a^*_k a^*_r a_r a_k
\leq \bigg\| |k| \frac{\widehat v_N}{\epsilon^{\frac{1-s}{2}}}\bigg\|_2^4 \dd\Gamma(\epsilon^{1-s})^2 \nn.
 \end{align}
  for all $s\geq 0$, $N\in \N\cup\{\infty\}$ for which the right hand side is finite.
 Using H\"older's inequality on each particle sector $\sum_{j=1}^k \epsilon_j^{1-s} \leq k^{s} \left(\sum_{j=1}^k\epsilon_j \right)^{1-s}$, we obtain
 \begin{equation}
 a^*(\ui \nabla v_{N,x})^2a(\ui \nabla v_{N,x})^2
  \leq \bigg\| |k| \frac{\widehat v_N}{\epsilon^{\frac{1-s}{2}}}\bigg\|_2^4 \left(- \Delta_x +\dd\Gamma(\epsilon)\right)^{2(1-s)} \cN_+^{2s},
 \label{eq:lem:G_N_bound_5}
 \end{equation}
  since $0\leq -\Delta_x$ and $\dd\Gamma(\epsilon)$ commute.
Using~\eqref{eq:v_N-bound}, we have for $0 \leq s <1/2$ that
\begin{equation}
C_\kappa:=\bigg\| |k| \frac{\widehat v_N}{\epsilon^{\frac{1-s}{2}}}\bigg\|_2^2 \lesssim \sum_{ |k|>\kappa} \frac{1}{|k|^2\epsilon(k)^{1-s}} \lesssim \kappa^{2s-1}.
\end{equation}
Thus, denoting $\mathbb{H}_0:=-\Delta_x+\ud \Gamma(\epsilon)+1$,
\begin{equation}
  (\cN_++1)^{-s}\mathbb{H}_0^{s} G_{N}  G_{N}^*\mathbb{H}_0^{s} (\cN_++1)^{-s} \leq C_\kappa^2,
\end{equation}
which proves~\eqref{eq:lem:G_0_bound_1}.

To prove the statement for the difference~\eqref{eq:lem:G_0_bound_diffN}, note that $v_N(k)-v_\infty(k)$ is supported on the set where $|k|>N^\alpha$, so this follows from the same argument where now $\kappa=N^\alpha$.

By the same reasoning, we have for $s\geq 1/2$ and $N<\infty$
\begin{equation}\label{eq:G_N0-bound-asympt}
(\cN_++1)^{-s} \mathbb{H}_0^{s}  G_{N}  G_{N}^*\mathbb{H}_0^{s} (\cN_++1)^{-s}\lesssim C_N^2,
\end{equation}
with, using Lemma~\ref{lem:phi_W},
\begin{align} \label{eq:moment_vN diverge}
C_N=\Big\| |k| \frac{\widehat v_N}{\epsilon^{\frac{1-s}{2}}}\Big\|_2^2 \lesssim \sum_{ |k| \leq N^{\alpha}} \frac{1}{|k|^{4-2s}} + N \| |\nabla|^{s} \Wphih\|_2^2 \lesssim N^{(s-\frac12)_+ +\eps }.
\end{align}
\end{proof}

\begin{lem}[Properties of $K_{N}$]\label{lem:K_0-props}
 For $\kappa$ sufficiently large, the following hold. The operators $K_{N}$ with domain
  \begin{equation*}
  D(K_{N}):=(1-G_{N})^{-1}D(-\Delta_x+\ud \Gamma(\epsilon)),
 \end{equation*}
 are self-adjoint for any $N \in \mathbb{N}\cup\{\infty\}$.
Moreover, for $0\leq s\leq 1$  there is $C>0$, such that for all  $N \in \mathbb{N}\cup\{\infty\}$ we have  \begin{subequations}
\begin{align}\label{eq:K_0 control}
 \|(-\Delta_x+\ud \Gamma(\epsilon))^{s}(1-G_{N})K_{N}^{-s}\|\leq C, \\
 \|\cN_+^s K_{N}^{-s}\|\leq C. \label{eq:bound_cN-K_0}
\end{align}
Furthermore, $K_{N}$ converges to $K_{\infty}$ in norm resolvent sense, and  for $0\leq s,\alpha<1/2$, there exists $C>0$ so that for all $N\in \N$
 \begin{align}\label{eq:K_conv}
\Big\| K_{N,0}^s \big(K_{N,0}^{-1} - K_{\infty}^{-1}\big)K_{\infty}^s\Big\| &\leq C N^{(2s-1) \alpha}.
\end{align}
\end{subequations}
\end{lem}
\begin{proof}
 The operator $\mathbb{H}_0=-\Delta+\ud \Gamma(\epsilon)+1$ is invertible and, for $\kappa$ sufficiently large, so is $1-G_{N}$ (by Lemma~\ref{lem:G_0_bound}), for all $N\in \N\cup\{\infty\}$. From this, one easily checks that $K_{N}$ is invertible on its domain. Since it is manifestly symmetric, $K_N$ must be self-adjoint.

 We now prove the estimate (\ref{eq:K_0 control}). We will argue by interpolation, so first let $s=1$.
By Lemma~\ref{lem:G_0_bound}, for $N\in \N\cup\{\infty\}$, inserting $(1-G_{N}^*)$ and its inverse yields
\begin{equation}
 \|(-\Delta_x +\ud \Gamma(\epsilon) +1)(1-G_{N})K_{N}^{-1}\| = \|(1-G_{N}^*)^{-1}\|\leq (1-C_\kappa)^{-1}.
\end{equation}
For $s=0$, we simply have, also by Lemma~\ref{lem:G_0_bound},
\begin{equation}
 \|(-\Delta_x +\ud \Gamma(\epsilon) + 1)^{0}(1-G_{N})K_{N}^{0}\| \leq 1+C_\kappa.
\end{equation}
The estimate (\ref{eq:K_0 control}) now follows by complex interpolation (cf.~\cite[Theorem IX.20]{ReeSim2}).
To prove the bound~\eqref{eq:bound_cN-K_0} on $\cN_+$, we use that $ \cN_+ G_{N } =  G_{N }(\cN_++2)$. Adding and subtracting $G_N$, we obtain
\begin{equation}
\|\cN_+^s K_{N }^{-s}\| \leq \|\cN_+^s(1-G_{N }) K_{N }^{-s}\| + C_\kappa \|(\cN_++2)^s K_{N }^{-s}\|,
\end{equation}
where we bounded $G_{N }$ using Lemma~\ref{lem:G_0_bound}. As $x\mapsto x^s$ is sub-additive for $s\leq 1$  this implies the claim~\eqref{eq:bound_cN-K_0} if $\kappa$ is sufficiently large for $C_\kappa<1$ to hold.

We now show the convergence by proving the bound \eqref{eq:K_conv}. With the resolvent formula,
 we have
\begin{align}
&K_{N }^s \big(K_{N }^{-1} - K_{\infty}^{-1}\big)K_{\infty}^s
= K_{N }^{-1+s}\Big(K_{\infty}-K_{N })K_{\infty}^{-1+s}  \\
 &= K_{N }^{s-1}\Big((G_{N }^*-G_{\infty }^*)\mathbb{H}_0(1-G_{\infty })  + (1-G_{N }^*)\mathbb{H}_0(G_{N }-G_{\infty })\Big)K_{\infty}^{s-1}.\notag
\end{align}
Thus~\eqref{eq:lem:G_0_bound_diffN} and~\eqref{eq:K_0 control} imply that for $0\leq s<1/2$
\begin{align}
& \Big\| K_{N }^s \big(K_{N }^{-1} - K_{\infty}^{-1}\big)K_{\infty}^s\Big\| \nn\\
 	&\lesssim \|K_{N }^{-1+s} (G_{N }^*-G_{\infty }^*)\mathbb{H}_0^s\| \|\mathbb{H}_0^{1-s} (1-G_{\infty }) K_{\infty}^{-1+s}\| \nn \\
	&\quad + \|K_{N }^{-1+s} (1-G_{N }^*)\mathbb{H}_0^{1-s}\| \|\mathbb{H}_0^{s}(G_{N }-G_{\infty })K_{\infty}^{-1+s} \| \nn \\
 	&\lesssim N^{(2s-1) \alpha}( \| (\cN_++2)^s K_{N }^{-1+s}\| + \| (\cN_++2)^s K_\infty^{-1+s}\|),
\end{align}
where we used that $G^*_{N } \cN_+ = (\cN_++2) G^*_{N }$. In  view of~\eqref{eq:bound_cN-K_0} this proves \eqref{eq:K_conv}.
\end{proof}

We now discuss the operator $T_N$. Rewriting $\cH_{\mathrm{sing},0}$ as in~\eqref{eq:Hsing_0} extracts the divergent scalar contribution $E_{N }$, and $T_{\infty }$ is well defined as the limit of $T_{N }$. To make this explicit, we follow Section~\ref{sect:renorm}. Using the commutation rules
\begin{align} \label{eq:comm_rules}
\ui \nabla_x \ue^{-\ui kx}=\ue^{- \ui kx} (\ui \nabla_x + k), \qquad \ud\Gamma(\epsilon) a_k^* = a_k^*(\ud\Gamma(\epsilon)+\epsilon(k))
\end{align}
we write  $T_{N }$ as a sum  of terms in ``normal order'',
\begin{align} \label{eq:decompo_TN0}
T_{N }=\Theta_{N,0}+\Theta_{N,1}+\Theta_{N,2},
\end{align}
with
\begin{subequations}
 \begin{align}
 \Theta_{N,0}&=-2\sum_{k,\ell \in 2\pi\Z^3 \atop |k|,|\ell|>\kappa}  \frac{(k\cdot\ell)^2 \widehat  v_N(k)^2  \widehat  v_N(\ell)^2 }{(\ui \nabla_x-k-\ell)^2+\epsilon(k)+\epsilon(\ell)+\ud\Gamma(\epsilon)+1} -E_{N }  \\
 \Theta_{N,1}&= 4\sum_{k,\ell \in 2\pi\Z^3 \atop |k|,|\ell|>\kappa} e^{-\ui kx}a_k^* \theta_{N,1}(k,\ell)  a_\ell e^{\ui \ell x} \\
 \Theta_{N,2}&= \sum_{k_1,k_2 \in 2\pi\Z^3 \atop |k_1|,|k_2|>\kappa} \sum_{\ell_1,\ell_2 \in 2\pi\Z^3 \atop |\ell_1|,|\ell_2|>\kappa} e^{-\ui k_1x} e^{-\ui k_2x}a_{k_1}^* a_{k_2}^* \theta_{N,2}(k_1,k_2,\ell_1,\ell_2)  a_{\ell_1}a_{\ell_2} e^{-\ui \ell_1 x} e^{\ui \ell_2 x},
\end{align}
\end{subequations}
where the kernels are defined as
\begin{align}
\theta_{N,1}(k,\ell)
 &= - \sum_{\xi \in 2\pi\Z^3 \atop |\xi|>\kappa}\frac{(k\cdot \xi) (\ell\cdot \xi) \widehat  v_N(k) \widehat  v_N(\ell) \widehat  v_N(\xi)^2}{(\ui \nabla_x+k-\ell-\xi)^2 + \epsilon(k)+\epsilon(\ell)+\epsilon(\xi)+\ud\Gamma(\epsilon)+1}
\end{align}
and
\begin{align}
&\theta_{N,2}(k_1,k_2,\ell_1,\ell_2)\\
 &= - \frac{(k_1\cdot k_2) \widehat  v_N(k_1) \widehat  v_N(k_2)(\ell_1\cdot \ell_2) \widehat  v_N(\ell_1) \widehat  v_N(\ell_2)}{(\ui \nabla_x+k_1+k_2-\ell_1-\ell_2)^2 + \epsilon(k_1)+\epsilon(k_2)+\epsilon(\ell_1)+\epsilon(\ell_2)+\ud\Gamma(\epsilon)+1}. \nn
\end{align}

Observe that the series $\theta_{N,j}$ converge as functions of $\ui\nabla_x$ and $\dd\Gamma(\epsilon)$. As we will see, the operators $\Theta_{N,j}$ are bounded relatively to $K_{N }$. Because of this, $T_{\infty }$ is well defined on $D(K_{N })$.
This entails:

\begin{prop}\label{prop:T_N_relative}
For $\kappa$ sufficiently large, $0 < \alpha \leq 1/2$, $N,N'\in \N\cup\{\infty\}$ and all $\Psi\in D(K_{N' })$
\begin{equation*}
\|T_{N }\Psi\| \leq  C_\kappa\|K_{N' }\Psi\|
\end{equation*}
where $\lim_{\kappa\to \infty} C_\kappa=0$.
Moreover, for all $0< \varepsilon < 1/4$ there is $C>0$ so that for $N,N' \in \mathbb{N} \cup \{\infty\}$, $0 < \alpha \leq 1/2$
\begin{align*}
\| (T_{N } - T_{\infty }) K_{N' }^{-1} \| \leq C N^{-\alpha\varepsilon}.
\end{align*}
In particular, for $\kappa>0$ large enough, the operator $K_{\infty }+T_{\infty }$ is self-adjoint on
\begin{align*}
  D(K_{\infty}):=(1-G_{\infty })^{-1}D(-\Delta_x+\ud \Gamma(\epsilon)).
\end{align*}
and $\cH_{\mathrm{sing,0} }-E_{N }$ converges to $K_{\infty }+T_{\infty }$ in norm resolvent sense.
\end{prop}

\paragraph{Proof of Proposition \ref{prop:T_N_relative}.}

We split the proof of the bounds on $T_{N }=\Theta_{N,0}+\Theta_{N,1}+\Theta_{N,2}$ into two parts, one on uniform bounds for $\Theta_{N,j}$, and one on the differences $\Theta_{N,j}-\Theta_{\infty,j}$.

\begin{lem}\label{lem:T_N bound}
Let $0<s<1/2$. There is $C_\kappa$ with $\lim_{\kappa\to \infty} C_\kappa=0$ so that for all $\kappa>0$, $N\in \N\cup \{\infty\}$ and $\Psi\in D(\cN_+^{1-s}(-\Delta_x+\ud\Gamma(\epsilon)^{s})$
\begin{align*}
 \| \Theta_{N,1} \Psi\| &\leq C_\kappa \|\cN_+^{1/2-s}\ud \Gamma(\epsilon)^s\Psi\| \\
 \| \Theta_{N,2} \Psi\| &\leq C_\kappa \|\cN_+^{1-s}\ud \Gamma(\epsilon)^{s}\Psi\|.
\end{align*}
Moreover, for $\varepsilon>0$ there is $C$ so that for all $N\in \N$ and $\kappa>0$
\begin{align*}
 \|\Theta_{N,0}\Psi\| &\leq C \kappa^{-\varepsilon} \|(-\Delta_x +\ud \Gamma(\epsilon))^\eps \Psi\|.
\end{align*}
\end{lem}

\begin{proof}
 We start by proving the estimate on $\Theta_{N,1}$. First note that using~\eqref{eq:v_N-bound} we have the point-wise bound
 \begin{align}
 |\theta_{N,1}(k,\ell)|
 &\lesssim \frac{1}{|k||\ell| (\epsilon(k)+\epsilon(\ell)+\ud\Gamma(\epsilon))^{1/2}}. \label{eq:theta_1-bound}
\end{align}
Using this and the Cauchy-Schwarz inequality, we have that for any $\Phi \in \mathscr{H}_+,\Psi \in \mathcal D(\dd\Gamma(\epsilon)^{s} \cN_+^{1-s})$ and $\delta>0$, denoting $t = s + 1/2 \in (1/2,1)$,
\begin{align}
&\langle\Phi, \Theta_{N,1} \Psi \rangle = \hspace{-6pt}\sum_{ |k|,|\ell|>\kappa}\hspace{-4pt}\Big\langle \epsilon(\ell)^{-(1+t)/2} \epsilon(k)^{t/2}a_k \Phi, \epsilon(\ell)^{(t+1)/2}\epsilon(k)^{-t/2}\theta_{N,1}(k,\ell) a_\ell \Psi  \Big\rangle \nn\\
	&\lesssim \delta\sum_{ |k|,|\ell|>\kappa} \epsilon(\ell)^{-(1+t)} \epsilon(k)^{t} \Big\| \frac{1}{(\epsilon(k) +\epsilon(\ell) + \dd\Gamma(\epsilon))^{t/2}}  a_k \cN_+^{-(1-t)/2} \Phi \Big\|^2 \nn\\
	&\quad +\frac{1}{\delta}\sum_{ |k|,|\ell|>\kappa}  \epsilon(k)^{-(1+t)} \epsilon(\ell)^{t}  \Big\|\frac{1}{(\epsilon(k) +\epsilon(\ell) + \dd\Gamma(\epsilon))^{1/2-t/2}}   a_\ell \cN_+^{(1-t)/2}  \Psi\Big \|^2 \nn\\
	&\lesssim C_\kappa \Big(\delta\|\dd\Gamma(\epsilon^t)^{1/2} \dd\Gamma(\epsilon)^{-t/2} \cN_+^{-(1-t)/2} \Phi\|^2 \nn\\
	&\qquad\qquad  + \delta^{-1}\|\dd\Gamma(\epsilon^t)^{1/2} \dd\Gamma(\epsilon)^{-(1-t)/2} \cN_+^{(1-t)/2} \Psi\|^2\Big) \nn\\
	&\lesssim C_\kappa \left(\delta\|\Phi\|^2 + \delta^{-1}\|\dd\Gamma(\epsilon)^{t-1/2} \cN_+^{1-t}\Psi\|^2\right) \nn\\
	&= C_\kappa \left(\delta\|\Phi\|^2 + \delta^{-1} \|\dd\Gamma(\epsilon)^{s} \cN_+^{1/2-s}\Psi\|^2\right),
\end{align}
where we used that
$
\dd\Gamma(\epsilon^t) \leq \dd\Gamma(\epsilon)^t \cN_+^{1-t}
$
and denoted
\begin{align}
C_\kappa = \sum_{|k|>\kappa}  \epsilon(k)^{-(1+t)} \underset{\kappa \to \infty}{\longrightarrow} 0.
\end{align}
Optimizing over $\delta$ yields the claim.

We now turn to the proof of the bound on $\Theta_{N,2}$.
With the same reasoning as~\eqref{eq:theta_1-bound} we find the bound
\begin{align}
 &|\theta_{N,2}(k_1,k_2,\ell_1,\ell_2)|
 \lesssim \frac{1}{|k_1||k_2||\ell_1||\ell_2|(\epsilon(k_1)+\epsilon(k_2)+\epsilon(\ell_1)+\epsilon(\ell_2)+\ud\Gamma(\epsilon))}.\label{eq:theta_2-bound}
\end{align}
 Similarly as above, let $\Phi \in \mathscr{H}_+,\Psi \in \mathcal D(\cN_+^{1-2s}\ud \Gamma(\epsilon)^{2s})$, $\delta>0$ using the Cauchy-Schwarz inequality and denoting again $t = (1+s)/2 \in (1/2,3/4)$, we obtain
\begin{align}
&\langle\Phi, \Theta_{N,2} \Psi \rangle \nn\\
&\lesssim \delta  \sum_{|k_1|,|k_2|>\kappa \atop |\ell_1|,|\ell_2|>\kappa}
	 \epsilon(\ell_1)^{-(1+t)} \epsilon(\ell_2)^{-(1+t)} \epsilon(k_1)^{t} \epsilon(k_2)^t \nn\\
	 &\qquad\qquad\times \|(\epsilon(k_1)+\epsilon(k_2)+\epsilon(\ell_1)+\epsilon(\ell_2)+\ud\Gamma(\epsilon))^{-t} a_{k_1} a_{k_2} \cN_+^{-2(1-t)} \Phi\|^2 \nn\\
	 &\quad +\frac{1}{\delta} \sum_{|k_1|,|k_2|>\kappa \atop |\ell_1|,|\ell_2|>\kappa}
	 \epsilon(k_1)^{-(1+t)} \epsilon(k_2)^{-(1+t)} \epsilon(\ell_1)^{t} \epsilon(\ell_2)^t  \nn\\
	 &\qquad\qquad\times \|(\epsilon(k_1)+\epsilon(k_2)+\epsilon(\ell_1)+\epsilon(\ell_2)+\ud\Gamma(\epsilon))^{-1+t} a_{k_1} a_{k_2} \cN_+^{2(1-t)} \Psi\|^2 \nn\\
	 &\lesssim C_\kappa^2 (\delta \|\dd\Gamma(\epsilon^t) \dd\Gamma(\epsilon)^{-t}\cN_+^{-(1-t)} \Phi \|^2
	 + \delta^{-1} \| \dd\Gamma(\epsilon^t) \dd\Gamma(\epsilon)^{-(1-t)}\cN_+^{1-t} \Psi \|^2) \nn\\
	 &\lesssim  C_\kappa^2(\delta \|\Phi \|^2 + \delta^{-1} \|\dd\Gamma(\epsilon)^{s} \cN_+^{1-s} \Psi\|^2).
\end{align}
Optimizing over $\delta$ proves the desired bound.

Finally, we prove the bound on $ \Theta_{N,0}$. Since $\Theta_{N,0}$ is a function of the commuting operators $\ui \nabla_x$ and $\dd\Gamma(\epsilon)$, it is sufficient to prove a point-wise bound on this function, by the functional calculus.
First, rewrite $\Theta_{N,0}$ as
\begin{align}
 \Theta_{N,0} =  \sum_{k,\ell \in 2\pi\Z^3}\bigg(&
 \frac{2(k\cdot \ell)^2\widehat  v_N (k)^2 \widehat  v_N(\ell)^2}{(\ui \nabla_x +k +\ell)^2 + \epsilon(k)+\epsilon(\ell) + \ud \Gamma(\epsilon)+1}\nn \\
 &\qquad \times\frac{(\ui \nabla_x +k +\ell)^2 +\ud \Gamma(\epsilon) - (k+\ell)^2)}{(k +\ell)^2 + \epsilon(k)+\epsilon(\ell)+1}\bigg).
\end{align}
Using that $| \widehat  v_N(k)| \lesssim k^{-2}$, we have for all $\varepsilon >0$,
\begin{align}
| \Theta_{N,0}|
	&\lesssim \sum_{  |k|,|\ell|>\kappa} \frac{1}{|k|^2|\ell|^2(\epsilon(k)+\epsilon(\ell))} \frac{(\ui \nabla_x)^2 + |\ui \nabla_x| |k| +\ud \Gamma(\epsilon) }{((\ui \nabla_x +k +\ell)^2 + \epsilon(k)+\epsilon(\ell) + \ud \Gamma(\epsilon))} \nn\\
	&\lesssim (\ui \nabla_x)^{2} \sum_{ |k|,|\ell|>\kappa} \frac{1}{|k|^3|\ell|^3} \frac{ 1}{(\ui \nabla_x +k +\ell)^2 + |k|^2+ |\ell|^2} \nn\\
	&\quad +|\ui \nabla_x| \sum_{ |k|,|\ell|>\kappa} \frac{1}{|k|^3|\ell|^3} \frac{|k|}{(\ui \nabla_x +k +\ell)^2 +|k|^2+ |\ell|^2} \nn\\
	&\quad + \dd\Gamma(\epsilon) \sum_{ |k|,|\ell|>\kappa} \frac{1}{|k|^3|\ell|^3} \frac{1}{\dd\Gamma(\epsilon) + |k|^2+ |\ell|^2}.
	%
\end{align}
Dropping $\ud\Gamma(\epsilon)$, $(\ui \nabla_x + k +\ell)^2$ from a power $(1-\eps)$ of the denominator and $|k|^2+|\ell|^2$ from the remaining power $\eps$, and using that
that
\begin{align}
\frac{ 1}{(\ui \nabla_x +k +\ell)^2 + |k|^2+ |\ell|^2} \lesssim \frac{1}{(\ui\nabla_x)^{2}}.
\end{align}
we obtain
\begin{align}
 | \Theta_{N,0}| \lesssim \kappa^{-\varepsilon} ((\ui \nabla_x)^2 + \dd\Gamma(\epsilon))^{\varepsilon}.
\end{align}
This proves the claim.
\end{proof}

\begin{lem} \label{lem:theta_cv}
For $0< \alpha,s <1/2$, $0 <\varepsilon < s$ there is $C$ so that for all $N\in \N$, $\kappa>0$ and $\Psi\in D(\cN_+^{1-s}(-\Delta_x+\ud\Gamma(\epsilon))^{s})$
\begin{align*}
\| (\Theta_{N,1}-\Theta_{\infty,1}) \Psi \| &\leq C N^{-\alpha \varepsilon} \|\cN_+^{\frac{1}{2}+\frac{\eps}{2}-s}\ud \Gamma(\epsilon)^{s}\Psi\|, \\
\| (\Theta_{N,2}-\Theta_{\infty,2})\Psi \| &\leq C N^{-\alpha  \varepsilon} \|\cN_+^{1-s}\ud\Gamma(\epsilon)^{s}\Psi\|.
\end{align*}
Moreover, for all $0<\eps\leq 2$ there is $C>0$ so that for $N\in \N$, $\kappa>0$
\begin{align*}
\| (\Theta_{N,0}-\Theta_{\infty,0})\Psi\| &\leq C N^{-\alpha\varepsilon} \|(-\Delta_x +\ud \Gamma(\epsilon))^{\eps} \Psi\|.
\end{align*}
\end{lem}

\begin{proof}
 To estimate the convergence rate as $N\to\infty$, we use~\eqref{eq:w,v-def}, \eqref{eq:v_N-bound} in the form
 \begin{equation}
  |\widehat  v_N(p) - \widehat  v_\infty(p)|\lesssim \1_{|p|>N^\alpha} |p|^{-2} \leq N^{-\alpha \eps} |p|^{-2+\eps}.
 \end{equation}
 From this it follows that
\begin{align}
 |\theta_{N,1}(k,\ell) - \theta_{\infty,1}(k,\ell)| \lesssim \frac{N^{-\alpha\varepsilon}}{|k|^{1-\eps}|\ell|^{1-\eps} (\epsilon(k)+\epsilon(\ell)+\ud\Gamma(\epsilon))^{1/2-\eps/2}}
\end{align}
and
\begin{align}
 &|\theta_{N,2}-\theta_{\infty,2}|(k_1,k_2,\ell_1,\ell_2) \\
 &\lesssim \frac{N^{-\alpha\varepsilon}}{|k_1|^{1-\eps}|k_2|^{1-\eps}|\ell_1|^{1-\eps}|\ell_2|^{1-\eps}(\epsilon(k_1)+\epsilon(k_2)+\epsilon(\ell_1)+\epsilon(\ell_2)+\dd\Gamma(\epsilon))}\notag.
\end{align}

The rest of the proof follows as in the proof of of Lemma \ref{lem:T_N bound}. We only give the details for $\Theta_{N,1}$. Choosing $t = s + \frac{1}{2} + \frac{\eps}{2}$, from the condition on $s$ and $\varepsilon$, we have $1-\varepsilon+t > 3/2$, and therefore
\begin{align}
& N^{\alpha\eps}|\langle \Phi, (\Theta_{N,1}-\Theta_{\infty,1}) \Psi \rangle| \nn\\
&\lesssim \delta \sum_{
|k|,|\ell|>\kappa} \epsilon(\ell)^{-(1-\eps+t)} \epsilon(k)^{t-\eps} \bigg\|
\frac{(\cN_++1)^{-(1-t+\eps)/2}}{(\epsilon(k) +\epsilon(\ell) + \dd\Gamma(\epsilon))^{t/2-\eps/2}}  a_k  \Phi \bigg\|^2 \nn \\
	&\quad +\frac{1}{\delta}\sum_{
	|k|,|\ell|>\kappa}  \epsilon(k)^{-(1-\eps+t)} \epsilon(\ell)^{t-\eps}  \bigg\|\frac{(\cN_++1)^{(1-t+\eps)/2}}{(\epsilon(k) +\epsilon(\ell) + \dd\Gamma(\epsilon))^{1/2-t/2}}   a_\ell   \Psi \bigg\|^2 \nn\\
	&\lesssim C_\kappa \left(\delta \|\Phi\|^2 + \delta^{-1} \|\dd\Gamma(\epsilon)^{s} \cN_+^{1/2-s+\eps/2}\Psi\|^2\right).
\end{align}
As before, optimizing over $\delta$ gives the desired estimate.
\end{proof}

We conclude this section with the proof of Proposition~\ref{prop:T_N_relative}.
\begin{proof}[Proof of Proposition~\ref{prop:T_N_relative}]
 Recall that $T_{N }=\Theta_{N,0}+\Theta_{N,1}+\Theta_{N,2}$, so in view of Lemmas \ref{lem:T_N bound} and \ref{lem:theta_cv}, we only need to bound the operators $\cN_+^{1/2-s} \dd\Gamma(\epsilon)^s$, $\cN_+^{1-2s} \dd\Gamma(\epsilon)^{2s}$ and $(-\Delta_x + \dd\Gamma(\epsilon))^{s}$, for $0<s<1/4$, relative to $K_{N'}$.
The worst case is $\cN_+^{1-2s} \dd\Gamma(\epsilon)^{2s}$, for which we have
\begin{align}
 \|\cN_+^{1-2s} \dd\Gamma(\epsilon)^{2s} \Psi \| &\leq \|\cN_+^{1-2s} \dd\Gamma(\epsilon)^{2s} (1-G_N)\Psi \| + \|\cN_+^{1-2s} \dd\Gamma(\epsilon)^{2s} G_N\Psi \| \nn\\
 &\lesssim \|K_{N'} \Psi\| + \|(\cN_++2)\Psi \| \lesssim \|(K_{N'}+2) \Psi\|,
\end{align}
by applying first~\eqref{eq:K_0 control} and~\eqref{eq:lem:G_0_bound_1}, and then~\eqref{eq:bound_cN-K_0}.

This implies self-adjointness of $K_\infty+T_\infty$ by the Kato-Rellich theorem.
Resolvent convergence follows from the convergence of $T_{N}$ and the uniform relative bound. A detailed argument is provided in the more general case of Proposition~\ref{prop:HsingW}.
\end{proof}

\subsection{Analysis of the full Hamiltonian} 
\label{sec:H_W}

We now turn to the full Hamiltonian $\cH_N^U$, starting with the singular part $\cH_\mathrm{Sing}$ (which includes $\ud \Gamma(W_{N,x}))$.
In analogy with the previous section and Section~\ref{sect:renorm}, we introduce
\begin{align}
 G_{N,W} = - (-\Delta_x+ \ud \Gamma(\epsilon) +\ud \Gamma(W_{N,x})+1 )^{-1} \1_{\cN_+\leq N^{1-\alpha}} a^*(  \ui \nabla v_{N,x} )^2,
\end{align}
where now the particle number is cut off at $M=N^{1-\alpha}$ as in~\eqref{eq:Hsing_W}. Then, we can write
\begin{align}
 \mathcal H_{\mathrm{Sing}} = K_{N,W} +T_{N,W}+ E_{N,W}-1,
\end{align}
with $E_{N,W}$ defined in~\eqref{eq:E_NW} and
\begin{align}
 K_{N,W} &:= (1-G^*_{N,W})(-\Delta_x+ \ud \Gamma(\epsilon) + \ud\Gamma(W_{N,x}) +1)(1- G_{N,W}), \label{eq:K_NW}\\
T_{N,W} &:=  -\1_{\cN_+\leq N^{1-\alpha}} a(\ui \nabla v_{N,x} )^2(-\Delta_x+ \ud \Gamma(\epsilon)+\ud\Gamma(W_{N,x})+1)^{-1} a^*(\ui \nabla v_{N,x} )^2 \nn \\
&\qquad - E_{N,W}.\label{eq:T_NW}
\end{align}

In this section we prove that $\mathcal H_{\mathrm{Sing}} + R_N-E_{N,W}+1$ converges to the self-adjoint operator $K_{\infty}+T_{\infty} + R_{\infty}$ in norm resolvent sense, where $R_\infty$ is defined by setting $N=\infty$ in $R_N$, as in  Section~\ref{sect:renorm}.

\begin{prop}\label{prop:HsingW}
For  $\kappa$ sufficiently large there exists $C$ so that for all  $0<\alpha\leq 1/2$, $\eps>0$, and $N\in \N$ we have
\begin{align*}
&\| (\cH_N^U- E_{N,W}+1 \pm \ui)^{-1} - (K_{\infty } + T_{\infty } + R_{\infty}\pm \ui)^{-1} \| \leq C  N^{-(\frac14-\eps)\alpha}. 
\end{align*}
Moreover, for $0\leq s<5/8$ there exists $C$ so that for all $N\in \N$
\begin{equation*}
\| K_{N,W}^{s}(\cH_N^U- E_{N,W}\pm \ui)^{-1}\| \leq C.
\end{equation*}
\end{prop}

In order to prove this, we proceed similarly as in  the previous section and first establish bounds on $G_{N,W}$, $K_{N,W}$, as well as their differences to $G_N$, $K_N$.

\begin{lem}[Properties of $G_{N,W}$]\label{lem:G_W_bound}
 For $0\leq s\leq 5/8$ there exists $(C_\kappa)_{\kappa>0}$, $C$ with $\lim_{\kappa\to \infty} C_\kappa=0$ so that for all $\kappa, \alpha>0$, and sufficiently large $N\in \N$
\begin{subequations}
\begin{align}
 \|G_{N,W}^*(-\Delta_x+\ud \Gamma(\epsilon))^s (\cN_++1)^{-s} \| &\leq C_\kappa, \quad &s<1/2,\label{eq:lem:G_W_bound_1} \\
 \|G_{N,W}^*(-\Delta_x+\ud \Gamma(\epsilon))^{s}(\cN_++1)^{-s}\|  &\leq C N^{s-\frac12+\eps}, \quad &s\geq 1/2.\label{eq:lem:G_W_bound_2}
\end{align}
Moreover, for $0\leq s\leq 1/2$ there is $C$ so that for all $\kappa, \alpha>0$, $N\in \N$,
\begin{equation}\label{eq:lem:G_W_bound_diffW}
  \big\|(G_{N,W}^*-G^*_{N}) (-\Delta_x+\ud \Gamma(\epsilon))^{s}(\cN_++1)^{-\frac12} \big\| \leq C N^{(-\frac14+(s-\frac14)_+)(1-\alpha)}
\end{equation}
\end{subequations}
\end{lem}

\begin{proof}
Denote $\mathbb{H}_0=-\Delta_x + \ud \Gamma(\epsilon)+1$, $\W=\ud \Gamma(W_{N,x})$.
 From the proof of Lemma~\ref{lem:G_0_bound}, specifically~\eqref{eq:lem:G_N_bound_5}, using that $\W\geq 0$ we obtain the a priori bound on $G_{N,W}$,
 \begin{align}
 &(\cN_++1)^{-1/2}\mathbb{H}_0^{1/2}  G_{N,W}  G_{N,W}^*\mathbb{H}_0^{1/2} (\cN_++1)^{-1/2}  \notag \\
 &\quad \lesssim C_N^2 \mathbb{H}_0^{1/2} (\mathbb{H}_0 +\W)^{-1} \mathbb{H}_0(\mathbb{H}_0 +\W)^{-1} \mathbb{H}_0^{1/2}\lesssim C_N^2,
 \label{eq:G_W-apriori}
\end{align}
with the constant
\begin{equation}
 C_N = \| \epsilon^{-1/4}|k| \hat v_N\|_2^2 \lesssim \sum_{ |k| \leq N^{\alpha}} \frac{1}{|k|^{3}} + N \| |\nabla|^{1/2} \Wphih\|_2^2 \lesssim \log N,
\end{equation}
by Lemma~\ref{lem:phi_W}. This proves that $\|G_{N,W}^* \mathbb{H}_0^{1/2} (\cN_++1)^{-1/2} \| \lesssim \log N$. 
Now, by the resolvent formula,
\begin{align}
G_{N,W}^*-G^*_{N}\1_{\cN_+\leq N^{1-\alpha}}
&= - a(  \ui \nabla v_{N,x} )^2\1_{\cN_+\leq N^{1-\alpha}} \big((\mathbb{H}_0+\W )^{-1} -  \mathbb{H}_0^{-1} \big) \nn \\
	&= - G_{N,W}^* \W  \mathbb{H}_0^{-1}.\label{eq:G-diff}
\end{align}
As we already have bounds on $\| G^*_{N}\mathbb{H}_0^s(\cN_+ +1)^{-s}\1_{\cN_+\leq N^{1-\alpha}}\|\leq \| G^*_{N}\mathbb{H}_0^s(\cN_+ +1)^{-s}\|$ from Lemma~\ref{lem:G_0_bound}, we will prove the bounds on $G_{N,W}$ by estimating the right hand side.
With~\eqref{eq:G_W-apriori} we have
\begin{equation}
 \|G_{N,W}^* \W  \mathbb{H}_0^{-1+s}  (\cN_+ +1)^{-s} \| \lesssim \log N \| \mathbb{H}_0^{-1/2}\W  \mathbb{H}_0^{s-1}  (\cN_+ +1)^{1/2-s}\1_{\cN_+\leq N^{1-\alpha}} \|.
\end{equation}
To bound this, we use that the Sobolev inequality on the torus gives
\begin{align}\label{eq:bound_WN_onebody}
W_N \lesssim \|W_N\|_{3/(2t)} (-\Delta+1)^t,
\end{align}
for $t \in (0,3/2)$,  from which we obtain that
\begin{align}\label{eq:bound_WN-Laplace}
\dd\Gamma(W_{N,x}) \lesssim N^{1-t} \dd\Gamma(\epsilon^t) \lesssim N^{1-t} \cN_+^{(1-t)_+}\dd\Gamma(\epsilon)^t
\end{align}
uniformly in $x$. Note we assume that $W\in L^1\cap L^2$, which corresponds to $t \in [3/4,3/2]$.
Consider first the case $s\leq 1/4$. Applying the inequality once for $t=1$ and once for $t=3/2-2\eps$, for $0< \varepsilon < 3/8$ yields
\begin{equation}
 \| \mathbb{H}_0^{-1/2}\W  \mathbb{H}_0^{s-1}  (\cN_+ +1)^{1/2-s}\1_{\cN_+\leq N^{1-\alpha}} \|
 \lesssim N^{-\alpha(1/4-\eps)}.
\end{equation}
Similarly, for $1/4\leq s<1/2$, we use $t=2(1-s)>1$ and find
\begin{equation}
 \| \mathbb{H}_0^{-1/2}\W  \mathbb{H}_0^{s-1}  (\cN_+ +1)^{1/2-s}\1_{\cN_+\leq N^{1-\alpha}} \|\lesssim N^{-\alpha(1/2-s)}.
\end{equation}
Choosing $N$ sufficiently large then proves the claim~\eqref{eq:lem:G_W_bound_1}. The inequality~\eqref{eq:lem:G_W_bound_2}, where $s\geq 1/2$, follows from the same the reasoning, but here we do not need to use the cutoff $\1_{\cN_+\leq N^{1-\alpha}}$ or choose $N$ large (the restriction $s\leq5/8$ ensures that only the norms $\|W_N\|_p$ with $1\leq p\leq 2$ intervene).
To prove the final bound~\eqref{eq:lem:G_W_bound_diffW},  we use again the resolvent formula to obtain
\begin{align}
 &\|(G_{N,W}^*-G^*_{N}) \mathbb{H}_0^{s}(\cN_++1)^{-\frac12} \| \nn\\
 &\leq \| G^*_{N} \mathbb{H}_0^{s}(\cN_++1)^{-\frac12}\1_{\cN_+>N^{1-\alpha}}\| + \|G_{N,W}^* \W  \mathbb{H}_0^{-1+s}  (\cN_+ +1)^{-\frac12}\|.
\end{align}
The first term is bounded by $N^{(1-\alpha)(s-1/2)}$ by Lemma~\ref{lem:G_0_bound}, and the second
by $N^{-\frac14+(s-\frac14)_++ \eps}$, as before (without having to use the cutoff on $\cN_+$ due to the larger power of $(\cN_++1)^{-1}$). Both of these quantities are smaller than the claimed one for an appropriate choice of $\eps$ (depending on $\alpha, s$). This proves the claim.
\end{proof}

\begin{lem}[Properties of $K_{N,W}$]\label{lem:K_W-props}
For $\kappa$ sufficiently large and $0<\alpha\leq 1/2$ the following holds. The operators $K_{N,W}$ are self-adjoint and converge to $K_\infty$ in norm-resolvent sense as $N\to \infty$.

Moreover, for $0\leq s\leq 1$ there is $C>0$, such that for large enough $N\in \mathbb{N}$
\begin{subequations}
 \begin{align}\label{eq:K_W control}
 \| (-\Delta_x+\ud \Gamma(\epsilon) +\ud \Gamma(W_{N,x}))^{s}(1- G_{N,W}) K_{N,W}^{-s}\Big\|\leq C, \\
 \|\cN_+^s  K_{N,W}^{-s}\|\leq C, \label{eq:bound_cN-K_W}
\end{align}
and
  for $0\leq s<1/2$, $\eps>0$ there exists $C>0$ so that for $N\in \N$ sufficiently large
 \begin{align}
\label{eq:K_conv2}
\big\|  K_{N,W}^s \big( K_{N,W}^{-1} - K_{N}^{-1}\big)K_{N}^s\big\| &\leq C  N^{(-\frac14+(s-\frac14)_+)(1-\alpha)}
\end{align}
\end{subequations}
\end{lem}
\begin{proof}
 In view of the bounds on $G_{N,W}$ from Lemma~\ref{lem:G_W_bound}, self-adjointness and the estimates~\eqref{eq:K_W control},~\eqref{eq:bound_cN-K_W} follow from exactly the same argument as their analogues for $K_N$ in Lemma~\ref{lem:K_0-props}. Convergence to $K_\infty$ will follow once we have established~\eqref{eq:K_conv2}, since we already know that $K_N$ converges to $K_\infty$.

 It remains to prove (\ref{eq:K_conv2}). Similarly as above, we have with $\mathbb{W}=\ud \Gamma(W_{N,x})$, $\mathbb{H}_0=-\Delta_x+\ud \Gamma(\epsilon)+1$,
\begin{subequations}
 \begin{align}
&K_{N,W}^s \big(K_{N,W}^{-1} - K_N^{-1}\big)K_N^s 
= K_{N,W}^{s-1}\big(K_N-K_{N,W}\big)K_N^{s-1} \notag \\
 &= K_{N,W}^{s-1}\big((1-G_{N,W}^*)(\mathbb{H}_0+\W)(G_N-G_{N,W})\big)K_N^{s-1}\label{eq:K_W-diff-G_W}\\
 & \quad + K_{N,W}^{s-1}\big( (1-G_{N,W}^*)\W(1-G_N)\big) K_N^{s-1} \label{eq:K_W-diff-W}\\
 & \quad + K_{N,W}^{s-1}\big((G_N^*-G_{N,W}^*)\mathbb{H}_0(1-G_N)\big)K_N^{s-1}.
 \label{eq:K_W-diff-W_2}
 \end{align}
\end{subequations}
To bound (\ref{eq:K_W-diff-G_W}), we use that $\|(\mathbb{H}_0+\W)^{s}\Ho^{-s}\|\lesssim 1$ (which holds by~\eqref{eq:bound_WN-Laplace} and for $s=1/2$ by operator monotonicity or interpolation) to obtain
\begin{align}
 \|\eqref{eq:K_W-diff-G_W}\|
 	&\stackrel{\eqref{eq:K_W control}}{\lesssim} \|(\mathbb{H}_0+\W)^{s} (G_N-G_{N,W})K_N^{s-1} \| \lesssim \|\mathbb{H}_0^{s} (G_N-G_{N,W})K_N^{s-1} \|  \notag\\
 	&\stackrel{\eqref{eq:lem:G_W_bound_diffW}}{\lesssim} N^{(-\frac14+(s-\frac14)_+)(1-\alpha)} \| (\cN_++3)^{\frac12} K_N^{s-1}\|.
\end{align}
 The final quantity is bounded by~\eqref{eq:bound_cN-K_0}, which proves the required estimate. The bound on~\eqref{eq:K_W-diff-W_2} is obtained in the same way.

The middle term~\eqref{eq:K_W-diff-W} is bounded using~\eqref{eq:K_0 control},~\eqref{eq:K_W control} and (\ref{eq:bound_WN-Laplace}), to obtain for $\eps >0$
\begin{align}
\| \eqref{eq:K_W-diff-W} \| & \lesssim \|K_{N,W}^{s-1}(1-G_{N,W}^*)  (\mathbb{H}_0+\W)^{1/2}\| \|  (\mathbb{H}_0+\W)^{-1/2} \mathbb{H}_0^{1/2} \| \nn \\
&\quad \times\|\ud \Gamma(\epsilon)^{-1/2}\W\ud \Gamma(\epsilon)^{s-1}\| \|\ud \Gamma(\epsilon)^{1-s} (1-G_N) K_{N,0}^{s-1}\| \nn\\
& \lesssim N^{-\frac14+(s-\frac14)_+ +\eps }.
\end{align}
This concludes the proof of (\ref{eq:K_conv2}), and the proposition.
\end{proof}

\begin{cor}[Comparison estimates for $K_{N,W}$]\label{cor:K_control_kin}
For all $0\leq s<1/2$, and $s\leq t\leq 1$ we have for $\kappa, N$ large enough
\begin{equation*}
 \|\cN_+^{t-s}(-\Delta_x + \dd\Gamma(\epsilon))^{s} K_{N,W}^{-t}\|\lesssim 1.
\end{equation*}
Moreover, for $0<s,\eps \leq 1/8$
\begin{equation} \label{eq:cor_K_control3}
\|(-\Delta_x + \ud\Gamma(\epsilon))^{1/2+s} K_{N,W}^{-1/2-s}\| \lesssim N^{s+\eps}.
\end{equation}
\end{cor}
\begin{proof}
Let us denote $\Ho = -\Delta_x+\ud\Gamma(\epsilon)+1$ and $\W = \dd\Gamma(W_{N,x})$. Using Lemma~\ref{lem:G_W_bound}, 
 we obtain with  $\cN_+ G_{N,0} = G_{N,0}(\cN_++2)$
\begin{align}
  &\|\cN_+^{t-s}\Ho^{s}K_{N,W}^{-t}\| \nn \\
  &\quad \lesssim \|\cN_+^{t-s}(\Ho+\W+1)^{s}(1-G_{N,W})K_{N,W}^{-t}\| + \|\cN_+^{t-s}(\Ho+\W)^{s}G_{N,W}K_{N,W}^{-t}\| \nn \\
  &\quad \lesssim \|(\Ho+\W+1)^{t}(1-G_{N,W})K_{N,W}^{-t}\| +  C_\kappa \|(\cN_++2)^{t}K_{N,W}^{-t}\|. 
\end{align}
We also used that $\W \geq 0$ and that $s < 1/2$. The first bound thus follows from~\eqref{eq:K_W control} and~\eqref{eq:bound_cN-K_W}.
To prove the second bound, the same calculation using~\eqref{eq:lem:G_W_bound_2} yields
\begin{multline}
 \|\Ho^{1/2+s} K_{N,W}^{-1/2-s} \| \\ \lesssim \|\Ho^{1/2+s}(1-G_{N,W})K_{N,W}^{-1/2-s}\|
    +  C N^{s+\eps} \|(\cN_++1)^{1/2+s}K_{N,W}^{-1/2-s}\|.
\end{multline}
The second term on the right hand side is bounded using~\eqref{eq:bound_cN-K_W}. For the first term, we use~\eqref{eq:K_W control} to obtain
\begin{align}
 &\|\Ho^{1/2+s}(1-G_{N,W})K_{N,W}^{-1/2-s}\|\lesssim \| \Ho^{1/2+s} (\Ho+\W)^{-1/2-s}\| \nn \\
 &\lesssim \| \Ho^{-1/2+s} (\Ho+\W)^{1/2-s}\| + \| \Ho^{s-1/2}\W (\Ho+\W)^{-1/2-s}\| \nn \\
 & \lesssim 1 + \|W_{N}\|_{3/2}^{1/2}\|W_N\|_{3/(2-4s)}^{1/2}\lesssim N^{s},
\end{align}
where we used again that $\Ho^{t} (\Ho+\W)^{-t}$ is uniformly bounded for $-1/2\leq t\leq 1/2$, and~\eqref{eq:bound_WN-Laplace}. The restriction $s\leq 1/8$ means that $3/(2-4s)\leq 2$, so the corresponding norm of $W_N$ is finite by hypothesis.
\end{proof}

\begin{lem}\label{lem:T_W-diff}
For all $0<s<1/4$ there are $C$, $(C_{\kappa})_{\kappa>0}$ with  $C_\kappa \to 0$ as $\kappa \to \infty$ such that for all $\kappa>0$, and sufficiently large $N\in \N$, it holds that for all $\Psi\in D(\cN_+^{1/2-s/2}(-\Delta_x + \ud \Gamma(\epsilon))^{s})$
\begin{align} 
\nn
\big|\langle \Psi,T_{N,W} \Psi \rangle\big| \leq C_\kappa  \big\| (-\Delta_x+\dd\Gamma(\epsilon))^s \cN_+^{1/2-s}\Psi\big\|^2
\end{align}
and
\begin{align} 
\nn
\big|\langle \Psi, (T_{N,W}-T_{N}) \Psi \rangle\big| &\leq C  N^{-(1-\alpha)s}\big\|\cN_+^{(1-s)/2}(-\Delta_x +\ud \Gamma(\epsilon))^s\Psi\big\|^2.
\end{align}
\end{lem}
It is important to note that this lemma together with Lemma~\ref{lem:theta_cv} implies that $\lim_{N\to\infty} \langle\Psi,T_{N,W}\Psi\rangle=\langle\Psi,T_\infty\Psi\rangle$, which is crucial for the fact that $\ud \Gamma(W_{N,x})$ does not change the excitation spectrum. It only contributes a scalar term of order one which is $\lim_{N\to\infty} (E_{N,W}-E_{N})$. The fact that $\ud \Gamma(W_{N,x})$ does contribute at order one in this way makes the proof of Lemma~\ref{lem:T_W-diff} rather subtle. This proof is given in Section~\ref{sect:T_NW} below.

\begin{lem}\label{lem:R_N_bound}
 For  $3/8<s\leq 1/2$, $|t| \leq 1/2 - s$ and $\delta>0$, there exists $\kappa_0$, so that for all $\kappa\geq \kappa_0$ there is $C$ with
 \begin{subequations}
  \begin{align}\label{eq:lem:R_N_bound_1}
    |\langle \Phi, R_N\Psi\rangle|
    & \leq\delta \left\|\left\{(-\Delta_x+\ud\Gamma(\epsilon))^{s}(\cN_++1)^{1/2- s + t} + C\right\}\Phi\right\| \nn \\
    &\quad \times  \left\|\left\{(-\Delta_x+\ud\Gamma(\epsilon))^{s}(\cN_++1)^{1/2-s -t}+ C\right\}\Psi\right\|.
    \end{align}
 for all  $N\in \N\cup\{\infty\}$, $\Phi, \Psi\in D((-\Delta_x+\ud\Gamma(\epsilon))^{s}\cN_+^{1/2+t})$.
 Moreover, for $3/8<s< 1/2$ there is $C$ so that for all $N\in \N$, $\kappa>0$
\begin{equation}
  \pm (R_N - R_\infty)
	\leq C N^{-(4s-3/2)\alpha} (-\Delta_x+\ud\Gamma(\epsilon))^{2 s} (\cN_++1)^{1-2s}. \label{eq:lem:R_N_bound_2}
\end{equation}
\end{subequations}
\end{lem}

The proof of this Lemma is given in Section~\ref{sect:R_N} below.
These estimates allow us to define a self-adjoint realization of $K_{N,W} + T_{N',W} + R_{N'}$.

\begin{lem}\label{lem:R_N_pert}
For $\kappa>0$, $N'\in \N\cup\{\infty\}$ sufficiently large and all $N\in \N\cup\{\infty\}$ the quadratic form
 \begin{equation*} \nn
  \langle \Psi, (K_{N,W} +T_{N',W}) \Psi\rangle + \langle \Psi, R_{N'}\Psi\rangle
 \end{equation*}
 defines a unique self-adjoint operator $K_{N,W} +T_{N',W}+R_{N'}$ with domain contained in $Q(K_{N,W})$.

 Moreover, for $3/8 <  t < 5/8$ there exists a constant $C$ so that for all $N\in \N\cup\{\infty\}$ and $N'\in \N\cup\{\infty\}$ sufficiently large
 \begin{equation*} 
\| K_{N,W}^t(K_{N,W}+T_{N',W}+R_{N'}+ \ui)^{-1}K_{N,W}^{1-t} \| \leq C.
 \end{equation*}
\end{lem}

\begin{proof}
 Recall that by Corollary \ref{cor:K_control_kin} we can bound uniformly the quadratic forms $(-\Delta+\dd\Gamma(\epsilon))^s\cN_+^{1-s}$ by $K_{N,W}$ for all $0\leq s < 1$. Therefore, it follows from Lemmas \ref{lem:T_W-diff} and \ref{lem:R_N_bound} that $T_{N',W}$ and  $R_{N'}$ are form bounded by $K_{N,W}$, uniformly in $N,N'$. The same lemmas show that for $\kappa$ and $N'$ large enough the relative constant can be chosen smaller than one. This allows us to apply the KLMN theorem \cite[Theorem X.17]{ReeSim2}: the quadratic form $K_{N,W}+T_{N',W}+R_{N'}$ then defines a unique self-adjoint operator with domain contained in $Q(K_{N,W})$.
This also implies the resolvent bound for $t=1/2$.
Moreover, Lemmas~\ref{lem:R_N_bound} and \ref{lem:T_W-diff} imply that
\begin{equation}
 \begin{aligned}
  \|(K_{N,W}+C)^{t-1} R_{N'} (K_{N,W}+C)^{-t}\| <1\\
  \|(K_{N,W}+C)^{t-1} T_{N',W}(K_{N,W}+C)^{-t}\|<1
 \end{aligned}
\end{equation}
for $3/8<t<5/8$ and $\kappa$, $C$, $N'$ sufficiently large. This implies the bound in the general case by application of the Kato-Rellich Theorem to $K_{N,W}$ and $(K_{N,W}+C)^{t-1} R_{N'} (K_{N,W}+C)^{1-t}$, $(K_{N,W}+C)^{t-1} T_{N',W}(K_{N,W}+C)^{1-t}$.
\end{proof}

We can now prove Proposition \ref{prop:HsingW}.
\begin{proof}[Proof of Proposition \ref{prop:HsingW}.]
Recall that $\cH_N^U - E_{N,W}+1=  K_{N,W} +T_{N,W} + R_N$, hence the resolvent identity gives
\begin{subequations}
 \begin{align}
& (\cH_N^U - E_{N,W}+1\pm \ui  )^{-1} - (K_\infty + T_\infty + R_{\infty} \pm \ui )^{-1} \nn \\
  &=  ( K_{N,W} +T_{N,W} + R_N \pm \ui)^{-1} (T_\infty + R_\infty -T_{N,W} - R_{N}) \nn \\
  &\qquad \qquad \qquad\qquad\qquad\qquad\qquad\qquad  \times ( K_{N,W} + T_\infty + R_{\infty} \pm \ui )^{-1} \label{eq:resolv_1} \\
  &\quad + ( K_{N,W} + T_\infty + R_{\infty} \pm \ui )^{-1} (K_{N} - K_{N,W})  \nn \\
  &\qquad \qquad \qquad\qquad\qquad\qquad\qquad\qquad \times (K_{N} + T_\infty + R_{\infty}\pm  \ui )^{-1} \label{eq:resolv_2} \\
  &\quad +  (K_{N} + T_\infty + R_{\infty}\pm  \ui )^{-1} (K_{\infty} - K_{N})  \nn \\
  &\qquad \qquad \qquad\qquad\qquad\qquad\qquad\qquad  \times (K_{\infty} + T_\infty + R_{\infty}\pm  \ui )^{-1}. \label{eq:resolv_3}
\end{align}
\end{subequations}
Using the bound from Lemma \ref{lem:R_N_pert}, the resolvents can absorb factors of $K_{N,W}^{t}$ for $t<5/8$. Therefore we will obtain the convergence from the estimates on $T_{N,W}-T_{N}$ in Lemma~\ref{lem:T_W-diff}, on $T_{N}-T_\infty$ in Lemma~\ref{lem:theta_cv} and on $R_N - R_\infty$ in Lemma~\ref{lem:R_N_bound}, together with the comparison estimates between $(-\Delta_x+\dd\Gamma(\epsilon))^{s}\cN_+^{t-s}$ and $K_{N,W}^{-t}$ in Corollary~\ref{cor:K_control_kin}.
Explicitly, we obtain from Lemma~\ref{lem:T_W-diff}, Lemma~\ref{lem:theta_cv} with $s=1/8-\eps$ and \eqref{eq:lem:R_N_bound_2} with $s=1/2-\eps$, for $0<\varepsilon<1/8$,
\begin{align}
\|(\ref{eq:resolv_1})\|
	&\lesssim N^{(-1/8+\eps)(1-\alpha)} + N^{-(1/2-4\eps)\alpha} .
\end{align}
From (\ref{eq:K_conv2}) with $s=3/8+\eps$ for any $0<\eps<1/8$,
\begin{align}
\|(\ref{eq:resolv_2})\| \lesssim \|K_{N,W}^{s}(K_{N,0}^{-1} - K_{N,W}^{-1})K_{N,0}^{s}\| \lesssim  N^{(-1/8+\eps)(1-\alpha)}.
\end{align}
and from (\ref{eq:K_conv}) with $s=3/8+\eps$, for any $0<\eps<1/8$
\begin{align}
\|(\ref{eq:resolv_3})\| \lesssim  \|K_{N,0}^{s}(K_{N,0}^{-1} - K_{\infty}^{-1})K_{\infty}^{s}\| \lesssim N^{-(1/4-2\varepsilon)\alpha}.
\end{align}
Gathering the above estimates proves the claimed convergence, and inserting that $M^{-1/8+\eps}\lesssim N^{-(1/4-2\varepsilon)\alpha}$ for $M=N^{1-\alpha}$ and $\alpha\leq1/2$ yields the stated rate. The relative bound for $K_{N,W}^s$ is a direct corollary to Lemma~\ref{lem:R_N_pert}.
\end{proof}

\subsubsection{Proof of Lemma \ref{lem:T_W-diff}}\label{sect:T_NW}

We again use the notation $\Ho = -\Delta_x + \ud \Gamma(\epsilon)+1$, $\W = \ud \Gamma(W_{N,x})$.
We start by proving the bound on the difference $T_{N,W}-T_{N}$, which will also imply the bound on $T_{N,W}$ by using Lemma~\ref{lem:T_N bound}.
We first introduce a cutoff at $\1_{\cN_+\leq M}$ into $T_N$, to match that of $T_{N,W}$. Since $T_N$ commutes with $\cN_+$, Lemma~\ref{lem:T_N bound} gives us for $0<s<1/4$,
\begin{align}
\big|\langle \Psi, T_{N}\1_{\cN_+>M} \Psi \rangle\big| &\leq M^{-s} \big|\langle \Psi, T_{N}\cN_+^{s} \Psi \rangle\big| \nn\\
&\lesssim N^{-(1-\alpha)s} \| \cN_+^{(1+s)/2} \Psi\| \|\cN_+^{(1-s)/2} \Ho^s\Psi  \|. \label{eq:T_N-diff-cutoff}
\end{align}
It thus remains to consider the difference $T_N \1_{\cN_+\leq M}- T_{N,W}$.
Recall that $T_{N,W}$ is given by (\ref{eq:T_NW}).
To evaluate the difference with $T_{N,0}$, we decompose $T_{N,W}$ simarly to \eqref{eq:decompo_TN0}, in the form
\begin{align}
T_{N,W}=\Theta_{N,W,0}+\Theta_{N,W,1}+\Theta_{N,W,2}.
\end{align}
Since $T_{N,W}$ commutes with $\cN_+$ it is characterized by its restriction to the $n$-boson sectors $\mathscr{H}_+^{(n)}$ for $n\leq M$.
Denote by $(\tau_x f)(y) = f(x- y )$ the translation by $x$ on $L^2(\T^3_x)\otimes L^2(\T^3_y)= L^2(\T^3_x, L^2(\T^3_y))$ and set $\mathcal T_x = \tau_x \otimes \tau_x \otimes \1$, which satisfies
\begin{align}
\mathcal T_x^*  (-\Delta_x) \mathcal T_x = |\nabla_x - \nabla_{y_1} -\nabla_{y_2}|^2.
\end{align}
Moreover, denote
\begin{equation}
\begin{aligned}
\mathbb{H}^{(2)} &= -(\nabla_{y_{1}}+\nabla_{y_{2}})^2 +\epsilon(\ui \nabla_{y_1})+\epsilon(\ui \nabla_{y_2})+W_N(y_1)+W_N(y_2)+1. 
\end{aligned}
\end{equation}
For $\Psi^{(n)} \in \mathscr{H}_+^{(n)}\subset L^2\big(\T^3_x,L^2( \T^{3})^{\otimes_{\mathrm{s}} n}\big)  $ we then set
\begin{subequations}
\begin{align}
  &\langle  \Psi^{(n)}, \Theta_{N,W,0} \Psi^{(n)}\rangle  \\
  & = -2 \Big \langle \nabla  v_{N} \otimes  \nabla v_{N} \otimes \Psi^{(n)},\mathcal T_x^*\big( \Ho + \W\big)^{-1} \mathcal T_x-\big(\mathbb{H}^{(2)}\big)^{-1} \nabla  v_{N} \otimes  \nabla v_{N} \otimes \Psi^{(n)} \Big\rangle, \notag \\
   &\langle \Psi^{(n)}, \Theta_{N,W,1} \Psi^{(n)}\rangle \\
   & = - 4 n \Big \langle   \nabla  v_{N,x} \otimes   \nabla v_{N,x} \otimes \Psi^{(n)}, \big( \Ho + \W\big)^{-1}  \nabla v_{N,x}  \otimes \Psi^{(n)}\otimes \nabla v_{N,x} \Big\rangle,\notag \\
  &\langle  \Psi^{(n)}, \Theta_{N,W,2} \Psi^{(n)}\rangle \\
  & = - n(n-1) \Big \langle  \nabla  v_{N,x} \otimes  \nabla v_{N,x} \otimes \Psi^{(n)}, \big( \Ho + \W\big)^{-1}  \Psi^{(n)}\otimes \nabla v_{N,x}  \otimes \nabla v_{N,x}\Big\rangle.\notag
 \end{align}
\end{subequations}
Note that the tensor products here do not involve symmetrization. The coordinate $x$ of the particle is always integrated over at the end, while the tensor products concern the variables $y_1, \dots, y_{n+2}$.
A simple computation shows that these operators sum to $T_{N,W}$.
The operators $\Theta_{N,0}$, $\Theta_{N,1}$, $\Theta_{N,2}$ without $\ud\Gamma(W_{N,x})$ can of course be expressed in the same way, with the difference that $\mathbb{W}$ is set to zero.

When calculating the difference of these operators, the resolvent formula yields expressions involving $W_{N,x}(y_i)$ for $i=1, \dots, n+2$.
The key is to use the regularity of $\Psi^{(n)}$ in either $y_i$ or $x$. We need to take care, however, since not all of the resolvents commute with derivatives in these directions.

 We start by simplifying $\Theta_{N,W,0}$ by replacing  $\Ho + \W$ with  $\Ho+W_{N,x}(y_1)+W_{N,x}(y_2)$, which no longer contains the potentials in $y_3,\dots, y_{n+2}$ and thus commutes with derivatives in these directions.
 To estimate the difference, we will use that
 \begin{multline}
\|(\epsilon^{1/4}\otimes \epsilon^{1/4} \otimes \1)\big(\Ho + \W\big)^{-1/2}\| \\
\leq \| (\epsilon^{1/4} \otimes \1)\Ho^{-1/4}\|^2 \|\Ho^{1/2} \big(\Ho + \W\big)^{-1/2}\| \lesssim 1.
\end{multline}
 We thus take $0<s<1/4$, and bound
\begin{align}
 &\Big|\Big \langle \nabla  v_{N,x} \otimes  \nabla v_{N,x} \otimes \Psi^{(n)}, \nn \\
 &\qquad \big( (\Ho + \W)^{-1}-(\Ho+W_{N,x}(y_1)+W_{N,x}(y_2) )^{-1}\nabla  v_{N,x} \otimes  \nabla v_{N,x} \otimes \Psi^{(n)} \Big\rangle\Big| \notag \\
 &=\Big|\Big \langle \nabla  v_{N,x} \otimes  \nabla v_{N,x} \otimes \Psi^{(n)}(\Ho + \W)^{-1}, \nn \\
 &\qquad \Big(\sum_{j=3}^{n+2}W_{N,x}(y_j)\Big)(\Ho+W_{N,x}(y_1)+W_{N,x}(y_2))^{-1} \nabla v_{N,x} \otimes  \nabla v_{N,x} \otimes \Psi^{(n)}\Big\rangle\Big|\notag \\
 &\leq \|\epsilon^{-1/4}\nabla v_N\|^4 \|\Ho^{-1/2} \W \Ho^{-1/2} \dd\Gamma(\epsilon)^{-s}\| \|\Psi^{(n)}\| \| \ud \Gamma(\epsilon)^{s}\Psi^{(n)} \|
 \notag \\
 &\lesssim  (\log N)^2 N^{-s} \|\Psi^{(n)}\| \| \ud \Gamma(\epsilon)^{s}\Psi^{(n)}\|,
\end{align}
since $\|\epsilon^{-1/4}\nabla v_N\|^2 \lesssim \log N $ (see the definition of $v_N$ in (\ref{eq:w,v-def}) and Lemma~\ref{lem:phi_W}), and $\W$ is bounded using  (\ref{eq:bound_WN-Laplace}).
Another application of the resolvent identity leads to
\begin{align}
 &\langle\Psi^{(n)},\Theta_{N,W,0}\Psi^{(n)}  \rangle =  2 \Big \langle \nabla  v_{N} \otimes  \nabla v_{N} \otimes \Psi^{(n)} (\mathcal T_x^* \Ho \mathcal T_x+W_{N}(y_1)+W_{N}(y_2))^{-1},  \notag \\
 &\qquad \Big(-\Delta_x +2 \nabla_x (\nabla_{y_1}+\nabla_{y_2})+\sum_{j=3}^{n+2} \epsilon(\ui \nabla_{y_j})\Big) (\mathbb{H}^{(2)})^{-1} \nabla  v_{N} \otimes  \nabla v_{N} \otimes \Psi^{(n)}\Big\rangle\notag \\
 &\qquad +\mathcal{O}\big(N^{-s}(\log N)^2 \|\Psi\| \| \ud \Gamma(\epsilon)^{s}\Psi\|\big).
\end{align}
To arrive at the quadratic form of $\Theta_{N,0}$, we have to set $W_{N}$ to zero in the first resolvent and $(H^{(2)})^{-1}$.
 Consider for example the term arising from expanding $(\mathcal T_x^* \Ho \mathcal T_x+W_{N}(y_1)+W_{N}(y_2))^{-1}$. Using that $\mathbb{H}^{(2)}$ commutes with derivatives in $x$ and $y_3,\dots, y_{n+2}$, it is bounded by
\begin{align}
  &\Big|\Big \langle\nabla  v_{N} \otimes  \nabla v_{N} \otimes \Psi^{(n)} (\mathcal T_x^* \Ho \mathcal T_x+W_N(y_1)+W_N(y_2))^{-1}W_N(y_1)(\mathcal T_x^* \Ho \mathcal T_x)^{-1}, \notag \\
  &\quad   \Big(-\Delta_x +2 \nabla_x (\nabla_{y_1}+\nabla_{y_2}) + \sum_{j=3}^{n+2} \epsilon(\ui \nabla_{y_j}) \Big)(\mathbb{H}^{(2)})^{-1}\nabla  v_{N} \otimes  \nabla v_{N} \otimes \Psi^{(n)} \Big\rangle\Big| \nn \\
  &\leq  \| \Ho^{-1/2} (\nabla v_N)^{\otimes 2}\|  \|(\mathbb{H}^{(2)})^{-1/2} (\nabla v_N)^{\otimes 2}\| \|\Ho^{-1/2} W_{N,x}(y_1) \Ho^{-s}\epsilon(\ui \nabla_{y_1})^{-1/2} \| \notag \\
  &\hspace{4pt} \times\bigg\{ \Big( \| \Ho^{-1+s} \mathcal T_x^*\Delta_x^{1-s}\mathcal T_x \|  + \| \Ho^{-1+s} \mathcal T_x^*|\nabla_x|^{1-2s} |\nabla_{y_1}|\mathcal T_x \|\Big) \|(-\Delta_x)^{s}\Psi^{(n)}\|   \|\Psi^{(n)}\| \nn\\
  &\qquad \qquad +  \| \Ho^{-1+s} \dd\Gamma(\epsilon)^{1-s} \|  \|(\dd\Gamma(\epsilon))^{s}\Psi^{(n)}\|  \|\Psi^{(n)}\|\bigg\} \nn \\
    &\leq  \|\epsilon^{-1/4}\nabla v_N\|^4 \|\epsilon^{-1/2} W_N \epsilon^{-1/2-s}\| \|\mathbb{H}_0^s\Psi^{(n)}\| \|\Psi^{(n)}\|  \nn \\
  &\lesssim (\log N)^2 N^{-s}\|\Ho^{s}\Psi^{(n)}\| \|\Psi^{(n)}\| ,
\end{align}
where we applied (\ref{eq:bound_WN_onebody}).
The error terms coming from replacing $\mathbb{H}^{(2)}$ by $-(\nabla_{y_{1}}+\nabla_{y_{2}})^2 +\epsilon(\ui \nabla_{y_1})+\epsilon(\ui \nabla_{y_2})$ are bounded in a similar way, which yields
\begin{align}
 \|\langle \Psi^{(n)}, \big(\Theta_{N,W,0} -\Theta_{N,0}\big) \Psi^{(n)} \rangle\|
 \lesssim (\log N)^2N^{-s} \| \Ho^{s}\Psi^{(n)}\| \|\Psi\| \label{eq:diff_theta_0_fin}
\end{align}
for $0<s<1/4$ and proves the claim for $\Theta_{N,W,0}$.

To bound the difference of $\Theta_{N,W,1}$ and $\Theta_{N,1}$ we use that on $\mathscr{H}_+^{(n+2)}$
\begin{align}
 \Ho^{-1} - &\Big( \Ho + \W \Big)^{-1}
 =  \Big(\Ho +\sum_{j=1}^{n+1} W_{N,x}(y_j)\Big)^{-1} W_{N,x}(y_{n+2}) \Big( \Ho + \W \Big)^{-1}\notag\\
 &\qquad + \Big(\Ho+\sum_{j=1}^{n+1} W_{N,x}(y_j)\Big)^{-1} \sum_{j=2}^{n+1} W_{N,x}(y_j)  \Big(\Ho+W_{N,x}(y_1)\Big)^{-1}   \notag\\
 &\qquad +\Big(\Ho+W_{N,x}(y_1)\Big)^{-1}  W_{N,x}(y_1)  \Ho^{-1}  \label{eq:decompo_theta_1}.
\end{align}
With this decomposition, we can use the regularity of $\Psi$ and the fact that the derivative in the relevant variable commutes with the resolvent on the right to estimate the middle term of each summand. To bound the factor $n$ in the definition of $\Theta_{N,W,1}$, we also need to use the symmetry of $\Psi^{(n)}$. We proceed in the following way, denoting by $S_{a,b}$, for $1\leq a<b \leq n+2$, the projection onto functions symmetric under permutations of $y_a, \dots, y_b$, we have for all $\Phi \in L^2(\T^{3(n+2)})$
\begin{align}
 &\| \epsilon^{1/4}(\ui \nabla_{y_1}) \epsilon^{1/4}(\ui \nabla_{y_{n+2}})\ud \Gamma(\epsilon)^{-1/2} S_{3,n+2}\Phi\|^2 \notag \\
&= n^{-1}\langle S_{3,n+2} \Phi, \ud \Gamma(\epsilon)^{-1} \epsilon(\ui \nabla_{y_1})^{1/2}\sum_{j=3}^{n+2}\epsilon(\ui \nabla_{y_j})^{1/2} S_{3,n+2}\Phi \rangle \notag \\
&\leq n^{-1/2}\langle S_{3,n+2} \Phi, \ud \Gamma(\epsilon)^{-1/2} \epsilon(\ui \nabla_{y_1})^{1/2} S_{3,n+2}\Phi \rangle \notag \\
&\lesssim n^{-1/2} \|\Phi\|^2.
\end{align}
Then, from the decomposition (\ref{eq:decompo_theta_1}) and the estimate above we obtain,  using the symmetry of $\Psi^{(n)}$ and again (\ref{eq:bound_WN-Laplace}),  that for $0\leq s<1/4$
\begin{align}
 & |\langle \Psi^{(n)}, \big(\Theta_{N,W,1} -\Theta_{N,1}\big)\Psi^{(n)} \rangle| \notag\\
 &\lesssim  N^{-s} n \| S_{2,n+1} \ud \Gamma(\epsilon)^{-1/2}\nabla  v_N \otimes  \nabla v_N \otimes \epsilon(\ui \nabla_{y_{n+2}})^s\Psi^{(n)}\| \notag \\
 &\qquad \times \| S_{3,n+2}\ud \Gamma(\epsilon)^{-1/2}\nabla  v_N  \otimes \Psi^{(n)}\otimes  \nabla v_N\| \notag\\
 &\quad + N^{-s} n  \|S_{2,n+1}  \ud \Gamma(\epsilon)^{-1/2}\nabla  v_N \otimes  \nabla v_N \otimes\Psi^{(n)}\| \notag \\
& \qquad \times \| S_{3,n+2}\ud \Gamma(\epsilon)^{-1/2}\nabla  v_N  \otimes \ud \Gamma(\epsilon)^s\Psi^{(n)}\otimes  \nabla v_N\|\notag\\
 &\quad + N^{-s} n  \|S_{2,n+1}  \ud \Gamma(\epsilon)^{-1/2}\nabla  v_N \otimes  \nabla v_N \otimes\Psi^{(n)}\| \notag \\
& \qquad \times \| S_{3,n+2}\ud \Gamma(\epsilon)^{-1/2}\nabla  v_N  \otimes \epsilon(\ui \nabla_{y_{1}})^s \Psi^{(n)}\otimes  \nabla v_N\|\notag\\
 %
 &\lesssim N^{-s} n^{1/2} \|\epsilon^{-1/4} \nabla v_N\|_{2}^4 \|\Psi^{(n)}\| \left(\|\epsilon(\ui\nabla_1)^s\Psi^{(n)}\| + \|\dd\Gamma(\epsilon^s) \Psi^{(n)}\|\right) \nn \\
 &\lesssim (\log N )^2 N^{-s} \|\cN_+^{1/4}\Psi^{(n)}\| \| \cN_+^{1/4}\dd\Gamma(\epsilon)^s \Psi^{(n)}\|. \label{eq:diff_theta_1_fin}
 \end{align}

For the difference of $\Theta_{N,W,2}$ and $\Theta_{N,2}$ we proceed in a similar way, writing on $\mathscr{H}_+^{(n+2)}$
\begin{multline}
 \Ho^{-1} - \Big(\Ho + \W\Big)^{-1} =\Ho^{-1}\sum_{j=3}^{n+2} W_{N,x}(y_j)\Big(\Ho+\sum_{j=3}^{n+2} W_{N,x}(y_j)\Big)^{-1} \\
  + \Big(\Ho + \W \Big)^{-1}\Big(W_{N,x}(y_{1})+W_{N,x}(y_{2})\Big)\Big(\Ho+\sum_{j=3}^{n+2} W_{N,x}(y_j)\Big)^{-1}.
\end{multline}
From this we obtain
\begin{subequations}
 \begin{align}
 & \|\langle \Psi^{(n)}, (\Theta_{N,W,2} -\Theta_{N,2})\Psi^{(n)} \rangle\| \notag\\
 &=n(n-1) \Big| \Big \langle  S_{1,n} \,\big( \tau_x\nabla  v_N \otimes   \tau_x \nabla v_N \otimes \Psi^{(n)} \big) ,\Big(\Ho^{-1}- (\Ho +\W)^{-1}\Big) \notag\\
 &\qquad\qquad\qquad S_{3,n+2}\,\big(\Psi^{(n)}\otimes \tau_x\nabla v_N  \otimes \tau_x\nabla v_N \big)\Big\rangle\Big| \notag\\
 &\lesssim N^{-s}n(n-1)  \| \ud \Gamma(\epsilon)^{-1/2} S_{3,n+2} (\epsilon(\ui \nabla_{y_{1}})^s+\epsilon(\ui \nabla_{y_{2}})^s)\Psi^{(n)} \otimes \nabla v_N \otimes \nabla v_N  \| \notag\\
 &\qquad \times \| \ud \Gamma(\epsilon)^{-1/2}S_{1,n} \nabla v_N \otimes \nabla v_N\otimes \Psi^{(n)}\|  \label{eq:diff_theta_2_est_1}
 \\
 &\quad  +N^{-s}n(n-1)  \| \ud \Gamma(\epsilon)^{-1/2}S_{1,n} \nabla  v_N \otimes  \nabla v_N \otimes(\ud \Gamma(\epsilon)^{s}\Psi^{(n)})\|\notag \\
 &\qquad\times \| \ud \Gamma(\epsilon)^{-1/2}S_{3,n+2}\Psi^{(n)}\otimes \nabla v_N \otimes \nabla v_N\|,\label{eq:diff_theta_2_est_2}
\end{align}
\end{subequations}
where we used the bound (\ref{eq:bound_WN_onebody}).
We now use the symmetry to absorb the factors in $n$. We have for $\Phi \in L^2(\T^{3(n+2)})$
\begin{align}
 &\| \epsilon^{1/4}(\ui \nabla_{y_1}) \epsilon^{1/4}(\ui \nabla_{y_2})\ud \Gamma(\epsilon)^{-1/2} S_{1,n}\Phi\|^2 \notag\\
 &= \langle S_{1,n} \Phi, \ud \Gamma(\epsilon)^{-1} \epsilon(\ui \nabla_{y_1})^{1/2}\epsilon(\ui \nabla_{y_2})^{1/2} S_{1,n}\Phi \rangle \notag \\
 &= \frac{1}{n(n-1)}  \bigg\langle S_{1,n} \Phi, \Big(\sum_{\ell=1}^{n+2} \epsilon(\ui \nabla_{y_\ell})\Big)^{-1}\sum_{i\neq j=1}^n\epsilon(\ui \nabla_{y_i})^{1/2}\epsilon(\ui \nabla_{y_j})^{1/2} S_{1,n}\Phi \bigg\rangle \notag \\
 &\leq \frac{1}{\sqrt{n(n-1)}} \|\Phi\|^2,
\end{align}
where in the final step we used the Cauchy-Schwarz inequality on the double sum. We can now bound (\ref{eq:diff_theta_2_est_1}) and (\ref{eq:diff_theta_2_est_2}). Consider the first factor in (\ref{eq:diff_theta_2_est_1}). We have
\begin{align}
 &\| \ud \Gamma(\epsilon)^{-1/2}S_{1,n} \nabla v_N \otimes \nabla v_N \otimes \Psi^{(n)}\| \notag \\
 &\leq \|\epsilon^{-1/4} \nabla  v_N\|^2 \|\Psi^{(n)}\| \| \epsilon^{1/4}(\ui \nabla_{y_1}) \epsilon^{1/4}(\ui \nabla_{y_2})\ud \Gamma(\epsilon)^{-1/2} S_{1,n} \1_{L^2(\T^{3(n+2)})}\| \nn \\
 &\leq \frac{ \log N}{(n(n-1))^{1/4}}  \|\Psi^{(n)}\|.
\end{align}
Similarly, the second factor in (\ref{eq:diff_theta_2_est_1}) is bounded by
\begin{align}
&\| \ud \Gamma(\epsilon)^{-1/2} S_{3,n+2} (\epsilon(\ui \nabla_{y_{1}})^s+\epsilon(\ui \nabla_{y_{2}})^s)\Psi^{(n)} \otimes \nabla v_N \otimes \nabla v_N  \|\nn \\
	&\lesssim \| \epsilon^{1/4}(\ui \nabla_{y_{n+1}}) \epsilon^{1/4}(\ui \nabla_{y_{n+2}})\ud \Gamma(\epsilon)^{-1/2} S_{3,n+2} \1_{L^2(\R^{3(n+2)})}\| \|\epsilon(\ui \nabla_{y_{1}})^s \Psi^{(n)}\| \nn \\
	&\lesssim \frac{  \log N}{(n(n-1))^{1/4}} \|\dd\Gamma(\epsilon)^s \cN_+^{-s} \Psi^{(n)}\|.
\end{align}
From this we obtain that
\begin{align}
(\ref{eq:diff_theta_2_est_1})
	&\lesssim N^{-s} (\log N)^2 \sqrt{n(n-1)}  \| \Psi^{(n)}\| \|\dd\Gamma(\epsilon)^s\cN_+^{-s} \Psi^{(n)}\| \nn \\
	&\lesssim N^{-s} (\log N)^2  \| \cN_+^{1/2}\Psi^{(n)}\|  \|\cN_+^{1/2-s} \dd\Gamma(\epsilon)^s  \Psi^{(n)}\|.
\end{align}
The same computation leads to
\begin{align}
(\ref{eq:diff_theta_2_est_2})
	&\lesssim N^{-s} (\log N)^2  \| \cN_+^{(1+s)/2}\Psi^{(n)}\| \| \cN_+^{(1-s)/2} \dd\Gamma(\epsilon)^s \Psi^{(n)}\|.
\end{align}
We thus obtain
\begin{align} \label{eq:diff_theta_2_fin}
&\|\langle \Psi^{(n)}, \big(\Theta_{N,W,2} -\Theta_{N,2}\big)\Psi^{(n)} \rangle\| \nn\\
&\quad \lesssim  N^{-s} (\log N)^2 \|\cN_+^{(1+s)/2}\Psi^{(n)}\| \|\cN_+^{(1-s)/2} \dd\Gamma(\epsilon)^s  \Psi^{(n)}\|.
\end{align}
Collecting the estimates (\ref{eq:diff_theta_0_fin}), (\ref{eq:diff_theta_1_fin}) and (\ref{eq:diff_theta_2_fin}) and the bound on the cutoff~\eqref{eq:T_N-diff-cutoff} shows the bound on the difference $T_N-T_{N,W}$.

To prove the estimate of $T_{N,W}$, simply combine the bounds (\ref{eq:diff_theta_0_fin}), (\ref{eq:diff_theta_1_fin}) and (\ref{eq:diff_theta_2_fin}) on $\Theta_{N,W,j}$ with those on $\Theta_{N,j}$ from Lemma~\ref{lem:T_N bound} for $n\leq M$. Moreover, in~\eqref{eq:diff_theta_2_fin} bound $n^s$ by $M^s$, which yields
\begin{equation}
 \|\langle \Psi^{(n)}, \big(\Theta_{N,W,2} -\Theta_{N,2}\big)\Psi^{(n)} \rangle\| \lesssim  (M/N)^{s} (\log N)^2 \|\cN_+^{1/2}\Psi^{(n)}\| \|\cN_+^{1/2-s} \dd\Gamma(\epsilon)^s  \Psi^{(n)}\|.
\end{equation}
This implies the claim when $N$ is chosen large enough.
\qed

\subsubsection{Proof of Lemma \ref{lem:R_N_bound}}\label{sect:R_N}
Recall that
\begin{align*}
 R_N =2 a^*(\ui \nabla v_{N,x}) a(\ui \nabla v_{N,x}) - 2 a^*(\ui \nabla v_{N,x}) \ui \nabla_x  + \hc  +  a(w^\kappa_x) +a^*(w^\kappa_x).
\end{align*}

We will only give the proof of the estimate of $R_N$ in (\ref{eq:lem:R_N_bound_1}), as the one for the difference $R_N - R_{\infty}$ in (\ref{eq:lem:R_N_bound_2}) follows from the same reasoning, replacing $v_N$ by ${v}_N(k) - {v}_\infty$, which using~\eqref{eq:v_N-bound} satisfies
\begin{align}
\| \epsilon^{1/2-2s} |k| |\widehat {v}_N(k) - \widehat {v}_\infty (k)| \|_2 \lesssim \| |\1_{|k| > N^{\alpha}} |k|^{-4s}\|_2 \lesssim N^{-\alpha(4s-3/2)},
\end{align}
for $3/8<s\leq 1/2$.  We now prove (\ref{eq:lem:R_N_bound_1}) by estimating each term.

\emph{The term $a^*(w^\kappa_x)$.}
This term is bounded using that $w^\kappa\in L^2(\T^3)$ (see the definition~\eqref{eq:w,v-def}). We obtain that for all $\delta>0$ using Young's inequality
\begin{align}
|\braket{\Phi,a(w^\kappa_x)\Psi}|  
&\leq \|w^\kappa\|^2 \| (\cN_++1)^{1/4} \Phi\| \| (\cN_++1)^{1/4} \Psi\| \nn \\
&\leq \| (\delta \mathcal N_+^{3/8} + C_\delta ) \Phi \| \| (\delta \mathcal N_+^{3/8} +C_\delta )\Psi \|.
\end{align}

For the other terms in $R_N$, we make the following observation. Each term $X$ in $R_N$ satisfies that $\mathcal N X = X \left(\mathcal N + k\right)$ for some $k \in \mathbb{Z}$, therefore to conclude the proof of the desired bound (\ref{eq:lem:R_N_bound_1}) it is enough to prove estimates of the form
\begin{align}
|\braket{\Phi,X\Psi}| 
	&\lesssim \delta \|((-\Delta_x+\ud\Gamma(\epsilon))^{s} \mathcal N_+^{a}\Phi \| \|((-\Delta_x+\ud\Gamma(\epsilon))^{s} \mathcal N_+^{b} \Psi\|.
\end{align}
for some $0 \leq a,b$ such that $a+b \leq 1-2s$.

\emph{The term $ a^*(\ui \nabla v_{N,x}) \ui \nabla_x$.}
For this term, we decompose
\begin{align}
 | \langle \Phi, \ui \nabla_x  a(\nabla v_{N,x}) \Psi\rangle|  \leq \||\nabla_x|^{2s} \Phi\| \|\| \nabla_x|^{1-2s}a(\nabla v_{N,x}) \Psi\|.
\end{align}
The first factor above is bounded by $ \|((-\Delta_x+\ud\Gamma(\epsilon))^{s}\Phi \|$. To bound the second factor, we commute the gradient to the left,
\begin{align}\label{eq:lemm_RN_comm}
| \nabla_x|^{1-2s}a(\nabla v_{N,x})=  a(\nabla v_{N,x}) |\nabla_x|^{1-2s} + [| \nabla_x|^{1-2s}, a(\nabla v_{N,x})].
\end{align}
For $a > 1/4$, we have from (\ref{eq:w,v-def}) that $\epsilon^{-a}\nabla v_N \in L^2$ and that
\begin{align} \label{eq:C_k_alpha}
C_{\kappa} = \sup_{N\geq 1}\|\epsilon^{-a}\nabla v_N \|_{2} \to 0, \textrm{ as } \kappa \to \infty.
\end{align} Using the Cauchy-Schwarz inequality, we can therefore bound the contribution coming from the first term in (\ref{eq:lemm_RN_comm}) by
\begin{align}
 \| a(\nabla v_{N,x}) |\nabla_x|^{1-2s}{\Psi}\| &\leq \|\epsilon^{-a}\nabla v_N\|_{2} \| \ud \Gamma(\epsilon^{2a})^{1/2} |\nabla_x|^{1-2s}{\Psi}\| \nn \\
 &\leq C_\kappa \| (-\Delta_x+ \ud \Gamma(\epsilon))^{a +1/2 -s}\cN_+^{1/2 - a} {\Psi}\| \nn \\
  &\leq C_\kappa \| (-\Delta_x+ \ud \Gamma(\epsilon))^{s}\cN_+^{1-2s}\Psi\|,
\end{align}
where we took $a = 2s-1/2 > 1/4$ for $s > 3/8$.
To bound the term coming from the commutator in (\ref{eq:lemm_RN_comm}), we use that $\ui \nabla_x \ue^{-\ui kx}=\ue^{- \ui kx} (\ui \nabla_x + k)$ and that ${||p-k|^{1-2s}-|p|^{1-2s} |\lesssim |k|^{1-2s}}$, to obtain
\begin{align}
 \| [| \nabla_x|^{1-2s}, a(\nabla v_{N,x})]\Psi\| &\lesssim \| a(|\nabla|^{2-2s} v_{N,x}) \Psi\| \nn \\
 &\leq \|\epsilon^{-s}|\nabla|^{2-2s} v_N\|_{2} \| \ud \Gamma(\epsilon^{2s})^{1/2}\Psi\|, \nn \\
 &\leq \|\epsilon^{-s}|\nabla|^{2-2s} v_N\|_{2} \| \ud \Gamma(\epsilon)^{s}\cN_+^{1/2-s}\Psi\|,
\end{align}
where the norm tends to zero for $\kappa\to \infty$ since $|k|^{2-4s} \widehat  v_N(k) \lesssim |k|^{-4s}\1_{|k|>\kappa}$.
This proves that for all $\delta>0$,
\begin{align}
 | \langle \Phi, \ui \nabla_x  a(\nabla v_{N,x}) \Psi\rangle|
 \leq \delta \|(-\Delta_x+ \ud \Gamma(\epsilon))^{s} \Phi\|
 \| (-\Delta_x+ \ud \Gamma(\epsilon))^{s}\cN_+^{1-2s}\Psi\|,
\end{align}
for $\kappa $ sufficiently large.

\emph{The term $a^*(\ui \nabla v_{N,x}) a(\ui \nabla v_{N,x})$.}
Finally, we finish the bound (\ref{eq:lem:R_N_bound_1}) on $R_N$ by estimating the term with a creation and an annihilation operator. Similarly as above, we have
\begin{equation}
\braket{ \Phi, a^*(\ui \nabla v_{N,x}) a(\ui \nabla v_{N,x})  \Psi} \leq \|\epsilon^{-s}\nabla v_N \|_{2}^2 \|\ud \Gamma(\epsilon)^s \cN_+^{1/2-s}\Phi\| \|\ud \Gamma(\epsilon)^s \cN_+^{1/2-s}\Psi\|.
\end{equation}
Using (\ref{eq:C_k_alpha}) we conclude the proof of the desired bound (\ref{eq:lem:R_N_bound_1}).

\section{Proof of the main results}\label{sect:proof}

In this section we prove Theorem~\ref{thm:asymptotic}, Theorem~\ref{thm:spect}, and
Corollary~\ref{cor:operators}. We start by recalling the results that form the basis of the proofs, and proving precise upper and lower bounds on the spectrum of $H_N$ as preliminary lemmas.

In Section~\ref{sect:trafo} we computed
\begin{equation}\label{eq:H_N-outro}
 U_X^* H_N U_X = \1_{\cN_+ \leq N} \mathcal H_N  \1_{\cN_+ \leq N},
\end{equation}
with $U_X$ the excitation map (\ref{eq:excitation-map}). We then showed in Proposition~\ref{prop:U_G} and~\eqref{eq:H_NU-error}
\begin{align}\label{eq:U-remind}
 U^*\cH_N U=  4\pi  \mathfrak{a}_{V} (N-1)  + 8\pi  \ao_W \sqrt{N}   + e^{(U)}_N  +\cH_N^U  +\cL_4 +\cE
\end{align}
where $U$ is a unitary operator on $\mathscr{H}_+$. Here, the notation $\cH_N^U$ is defined in~\eqref{eq:H_N-U}, $\cL_4$ is the quartic term defined in Proposition~\ref{prop:H-ex}, $e^{(U)}_N$ is the scalar of order one defined in Proposition~\ref{prop:U_G}, and the error $\cE=\cE^{(U)}+ \mathcal{O}(N^{-3/2+2\alpha} (\cN_++1)^3) $ satisfies
\begin{align}\label{eq:EU-proof}
 \pm \cE & \lesssim N^{-\alpha/2}  \ud \Gamma(W_{N,x})  +N^{-\alpha/4}\mathcal{L}_4 + N^{-1/2-\alpha/4}\cN_+^2 + N^{-3/2+2\alpha}\cN_+^3\nn\\
	& \qquad +   N^{-\alpha/4}\big( \log N \dd\Gamma(|\ui \nabla|)  +(\log N)^2 +(\mathcal N_++1)^{1/2} |\nabla_x|\big).
\end{align}
for any $\alpha\leq1/10$.

Moreover, we proved the following conservation estimates under the transformation $U$.
\begin{lem}\label{lem:U-conserve}
Let $U= U_q U_W U_c U_\mathrm{B} U_\mathrm{G}$ be the product of the unitaries from Section~\ref{sect:trafo}. Under the assumptions of Proposition \ref{prop:U_G}, for any  $t\geq 0$ we have
\begin{equation*}
 U^*\cN_+^{t} U = \mathcal{O}(\cN_+^{t}+1), \qquad U\cN_+^{t} U^* = \mathcal{O}(\cN_+^{t}+1),
\end{equation*}
and it holds
\begin{align*}
 U^*\cL_4 U &= \mathcal{O}\Big( \cN_+ + N^{-1/2-\alpha/2} \cN_+^2 +\cL_4 + N^{-1/2}\log N \ud \Gamma(|\ui \nabla|)+ N\Big),\\
 U^*\ud\Gamma(W_{N,x})U& = \mathcal{O}\Big(\ud\Gamma(W_{N,x})+ \cN_+ + \sqrt{N}\Big).
\end{align*}

\end{lem}
\begin{proof}
 This follows from the various bounds in Section~\ref{sect:trafo}. Specifically: for $\cN_+$, Lemmas~\ref{lem:N_T1}, \ref{lem:Weyl-BB}, \ref{lem:Uc_N},~\ref{lem:U_B-conserve} and Equation~\eqref{eq:U_G-N}; for $\cL_4$ Lemmas~\ref{lem:L4_T1}, \ref{lem:Weyl-BB}, \ref{lem:Uc_p}, \ref{lem:U_B-conserve} and Equations~\eqref{eq:L_4-U_G},~\eqref{eq:U_G-p};
 For $\ud\Gamma(W_{N,x})$ Equations~\eqref{eq:Q_2-T_1-aprio},~\eqref{eq:apriori_Q2_T2}, Lemmas \ref{lem:Uc_p}, \ref{lem:U_B-Q_2} and Equation \eqref{eq:Q_2-U_G}.
\end{proof}

In Proposition~\ref{prop:CV_transf_Ham} we proved that $\cH_N^U-E_{N,W}$ converges to an operator that fiarily equivalent to $H_\mathrm{BF}$ in norm resolvent sense,
where $E_{N,W}$ is given by~\eqref{eq:E_NW} and contains the terms of order $\log N$. Moreover, Proposition~\ref{prop:HsingW} and Corollary~\ref{cor:K_control_kin} imply that for $0\leq s<1$
\begin{equation}\label{eq:H_eff,N-form_bound}
 \cN_+^{1-s}(-\Delta_x + \ud \Gamma(\epsilon))^s \lesssim \cH_N^U-E_{N,W} +C,
\end{equation}
for some $C>0$, independent of $N$.

\begin{lem}[Upper bound]\label{lem:upper}
Assume the hypothesis of Theorem~\ref{thm:asymptotic} and let
\begin{equation*}
 \lambda_N:=4\pi  \ao_V (N-1)  + 8\pi  \ao_W \sqrt{N}   + e^{(U)} + E_{N,W}.
\end{equation*}
Then for any $k\in \N$
\begin{equation*}
 e_k(H_N) - \lambda_N \leq e_k(H_\mathrm{BF}) + \mathcal{O}(N^{-1/80}).
\end{equation*}
\end{lem}
\begin{proof}
 By Proposition~\ref{prop:CV_transf_Ham}, $ \cH_N^U-E_{N,W}$ converges to $(U_\kappa^\infty)  H_\mathrm{BF} (U_\kappa^\infty)^*$ in norm resolvent sense, so the min-max values also converge. Moreover, the spectral projections associated to an interval around $e_k(H_{\mathrm{BF}})$ that contains no other eigenvalues of $H_\mathrm{BF}$ also converge, cf.~\cite[Thm.VII.23]{ReeSim1}. In particular, their ranks are equal for $N$ sufficiently large.
 To deduce the rate of convergence, consider such an interval and write for all eigenvalues in this interval $e_j(\cH_N^U-E_{N,W})=r_j^{-1}-\ui$, where $r_j$ is the corresponding eigenvalue of $(\cH_N^U-E_{N,W}+\ui)^{-1}$.  The eigenvalues $r_j$ all lie in the complement of a disc in $\C$, where the function $z\mapsto z^{-1}-\ui$ is Lipschitz. It is thus sufficient to bound $|r_j- (e_k(H_{\mathrm{BF}})+\ui)^{-1}|$, which is bounded by the norm-difference of the resolvents (cf.~\cite[Theorem 6.3.3]{Bhatia}, which applies to compact operators by the same proof). Proposition~\ref{prop:HsingW} thus gives us the rate of convergence
\begin{align} \label{eq:conv_vp_Htild_Hbf}
 | e_k(H_{\mathrm{BF}}) - e_k(\cH_N^U)-E_{N,W}| 
 & \lesssim N^{(-\frac14+\eps)\alpha},
\end{align}
for $0\leq \alpha \leq 1/2$ and $\varepsilon>0$, where we will take $\alpha=1/10$, the largest value permitted by Proposition~\ref{prop:U_G}. It is thus sufficient to prove the upper bound using $e_k(\cH_N^U)-E_{N,W}$.
Let $Y_k$ denote the span of the first $k+1$ eigenfunction of $\cH_N^U$.
Our test function for $H_N$ is of the form $$\Psi=U_X U\chi_M \Phi$$ with $\Phi\in Y_k$ normalized and $\chi_M$ a cutoff on $\cN_+\leq M=N^{\alpha/9}$. 
More precisely, let $\chi, \eta:\R\to [0,1]$ be two smooth functions with $\chi^2+\eta^2=1$ and $\chi(r)= 1$ for $r\leq1/2$, $\chi(r)=0$ for $r\geq1$, and define $\eta_M=\eta(\cN_+/M)$, $\chi_M=\chi(\cN_+/M)$.

\emph{Normalizing $\Psi$.} We first show that $\Psi=U_X U\chi_M \Phi$ is almost normalized. Recall that $U_X$ is a partial isometry, hence we
first bound, using that $U^*\cN_+U\lesssim \cN_+$ by Lemma~\ref{lem:U-conserve}, and~\eqref{eq:H_eff,N-form_bound},
\begin{align}
\|U\chi_M \Phi \|^2- \|\Psi \|^2 &= \langle U\chi_M\Phi, (1- U_X^*U_X)U\chi_M\Phi\rangle =  \langle U\chi_M\Phi, \1_{\cN_+>N} U\chi_M\Phi\rangle \nn\\
&\leq  N^{-1} \langle \chi_M\Phi, U^* \cN_+ U\chi_M\Phi\rangle 
\nn\\
%
%
&\lesssim N^{-1} \langle \Phi,(\cH_N^U -E_{N,W}+C)\Phi\rangle. \label{eq:Psi_BF-Nbound}
\end{align}
Then, we find that for $\delta>0$ small enough
\begin{align}
 1-\|U\chi_M \Phi \|^2 &= \langle \Phi, \eta_M^2 \Phi\rangle \leq M^{-5/4+\delta} \langle \Phi, \cN_+^{5/4-\delta} \Phi\rangle \nn \\  %
& \lesssim N^{-\alpha/8} \langle \Phi,(\cH_N^U -E_{N,W}+C))^2\Phi\rangle \lesssim N^{-\alpha/8}, \label{eq:eta_M-Phi_bound}
\end{align}
where we used  the second bound of Proposition~\ref{prop:HsingW} together with~\eqref{eq:bound_cN-K_W}, and chose $M=N^{\alpha/9}$.
We also used that the $k+1$ first eigenvalues of $\cH_N^U -E_{N,W}$ are uniformly bounded in $N$ because of the resolvent convergence we discussed before.
In particular, the map $U_X U\chi_M :Y_k\to \mathscr{H}_N$ is injective for large enough $N$. Thus we have for the min-max value
\begin{align}
 e_k(H_N)&=\min_{X \subset \mathscr{H}_N\atop \mathrm{dim}X=k+1} \max_{\Psi \in X} \frac{\langle \Psi, H_N \Psi\rangle }{\|\Psi\|^2} \\
 &\leq \max_{\Phi \in Y_k \atop \|\Phi\|=1} \langle \chi_M\Phi,  U^* U_X ^* H_N U_X U  \chi_M\Phi\rangle(1+ \mathcal O( N^{-\alpha/8})) \nn\\
 &\leq \max_{\Phi \in Y_k \atop \|\Phi\|=1} \langle \Phi,  \chi_M U^* \1_{\cN_+\leq N}\cH_N \1_{\cN_+\leq N} U \chi_M  \Phi\rangle(1+ \mathcal O(N^{-\alpha/8})). \nn
\end{align}

\emph{Removing $\1_{\cN_+\leq N}$.} Since $ \1_{ N\leq \cN_+} \cH_N  \1_{ N> \cN_+} = 0$ and $\1_{ N> \cN_+} \cH_N  \1_{ N> \cN_+} \geq 0$, we can remove the projections for an upper bound.
We thus have with~\eqref{eq:U-remind}
\begin{align}
e_k(H_N) - \lambda_N \leq \max_{\Phi \in Y_k \atop \|\Phi\|=1} \langle \chi_M\Phi,  (\cH_N^U-E_{N,W} + \cL_4 +\cE )  \chi_M\Phi \rangle(1+\mathcal O(N^{-\alpha/8})). \label{eq:minmax_ub}
\end{align}

\emph{Estimating $\cE$.} We can now bound the error term $\cE$ which satisfies (\ref{eq:EU-proof}). We proceed term by term. The powers of the number operator are bounded by $\chi_M \cN_+^k\leq M^{k-1}\chi_M\cN_+$, which is then bounded as in~\eqref{eq:Psi_BF-Nbound}.
Moreover, by~\eqref{eq:H_eff,N-form_bound} we have that
\begin{equation}
 \langle \chi_M\Phi, (\ud \Gamma(|\ui\nabla|) + \cN_+^{1/2} |\nabla_x|)\chi_M\Phi\rangle \lesssim \langle \chi_M\Phi, \cN_+^{1/2} (-\Delta_x + \ud \Gamma(\epsilon))^{1/2}\chi_M\Phi\rangle
\end{equation}
is of order one since $\chi_M$ commutes with the operator in the middle.
For $t=1-\alpha/4 \leq 1$, we have (compare with~\eqref{eq:bound_WN-Laplace})
\begin{equation}
 \ud \Gamma(W_{N,x})\lesssim \ud \Gamma((-\Delta)^t) \|W_N\|_{3/2t}\lesssim \cN_+^{1-t} \ud\Gamma(\epsilon)^{t}\|W_N\|_{3/2t},
\end{equation}
and using the bound~\eqref{eq:H_eff,N-form_bound}, we obtain for some $C>0$,
\begin{equation}\label{eq:W-H_N-bound}
 N^{-\alpha/2}\ud \Gamma(W_{N,x})  \leq C N^{-\alpha/4}(\cH_N^U-E_{N,W}+C).
\end{equation}
Hence from~\eqref{eq:EU-proof}, the above estimates give
\begin{align} \label{eq:chi_M-E-est}
| \langle \chi_M\Phi, \mathcal E \chi_M\Phi\rangle| \lesssim  N^{-\alpha/4} \log N\langle \chi_M \Phi, (\cH_N^U-E_{N,W}+\log N) \chi_M \Phi\rangle.
\end{align}

\emph{Estimating $\mathcal L_4$.} For the remaining term $\cL_4$, we use the analogue of the  bound~\eqref{eq:bound_WN_onebody} on $W_N$ to obtain for $1\leq t<3/2$
\begin{equation}
 \cL_4 \leq \|V_N\|_{3/2t} \ud \Gamma((-\Delta+1)^{t}) \cN_+ \lesssim  N^{2-2t} \ud \Gamma(-\Delta+1)^{t}\cN_+.
\end{equation}
Here, we cannot use $t<1$ since $\cL_4$ does not come with a small constant. However, $\cL_4$ comes from $V_N$ which scales with $N$, while the singularities of $\cH_N^U$ come from $W_N$ which only scales with $\sqrt{N}$, so by using $t>1$ we gain more decay from $V_N$ than we loose form $\cH_N^U$ and this term will still be small.
To see this, we set $t=1+\alpha/4$ and first use the cutoff at $M$, which gives
\begin{equation}
 \langle \chi_M \Phi, \cL_4 \chi_M\Phi\rangle \lesssim  N^{-\alpha/2} M \langle \chi_M \Phi, \ud \Gamma(-\Delta)^{1+\alpha/4} \chi_M\Phi\rangle.
\end{equation}
As $\ud\Gamma(-\Delta)$ commutes with $\cN_+$ and thus $\chi_M$, we can now apply the bound~\eqref{eq:cor_K_control3} from Corollary~\ref{cor:K_control_kin} with the second bound from Proposition~\ref{prop:HsingW}, with $M = N^{\alpha/9}$, to obtain
\begin{align}
 \langle \chi_M \Phi, \cL_4 \chi_M\Phi\rangle &\lesssim  N^{-\alpha/2} M\|\ud \Gamma(-\Delta)^{1/2+\alpha/8} K_{N,W}^{-1/2-\alpha/8}\|^2  \nn \\
	 &\qquad \times\|K_{N,W}^{1/2+\alpha/8} (\cH_N^U-E_{N,W}+\ui)^{-1})\|^2 \nn\\
 &\lesssim N^{-\alpha/4+2\eps}M\lesssim N^{-\alpha/8}, \label{eq:est_L4_ub}
\end{align}
for $\varepsilon>0$ small enough.
Putting these bounds together, the inequality (\ref{eq:minmax_ub}) yields
\begin{align}
 e_k(H_N) - \lambda_N
 \leq \max_{\Phi \in Y_k \atop \|\Phi\|=1} \langle \chi_M \Phi,(\cH_N^U-E_{N,W})\chi_M\Phi\rangle\left(1 + \mathcal{O}(N^{-\alpha/8})\right). \label{eq:upper_first}
\end{align}

\emph{Removing $\chi_M$.} It remains to prove that the quadratic form on the right gives approximately $\langle \Phi, (\cH_N^U-E_{N,W})\Phi\rangle$.
We use the IMS localization formula \cite[Proposition 6.1]{LewNamSerSol-15}
\begin{equation}
 \cH_N^U = \chi_M \cH_N^U \chi_M + \eta_M \cH_N^U \eta_M + \tfrac12 ([\chi_M,[\chi_M,\cH_N^U]]+ [\eta_M,[\eta_M,\cH_N^U]]).\label{eq:IMS}
\end{equation}
Using that the smallest eigenvalue of $\cH_N^U-E_{N,W}$ is approximated by $e_0(H_\mathrm{BF})$ from (\ref{eq:conv_vp_Htild_Hbf}), we find for the term localized on large particle numbers
\begin{equation}
 \langle \Phi,\eta_M( \cH_N^U -E_{N,W})\eta_M \Phi\rangle\geq  e_0(H_\mathrm{BF}) \| \eta_M\Phi\|^2  + \mathcal{O}(N^{-\alpha/8}) = \mathcal{O}(N^{-\alpha/8}),
\end{equation}
where we used the approximate normalization~\eqref{eq:eta_M-Phi_bound}.
Collecting from the explicit expression of $\cH_N^U$ in~\eqref{eq:H_N-U} the terms that do not commute with $\cN_+$, we see that
\begin{align}
[\chi_M,[\chi_M, \cH_N^U]] &=[\chi_M,[\chi_M, a(\ui \nabla v_{N,x})^2  + 2\ui \nabla_x a(\ui \nabla v_{N,x}) +a(w^\kappa)+\hc ] \nn\\
&= \big(2\ui \nabla_x a(\ui \nabla v_{N,x}) +a(w^\kappa)\big)(\chi_M(\cN_++1)-\chi_M(\cN_+))^2+\hc\nn \\
&\qquad +a(\ui \nabla v_{N,x})^2 (\chi_M(\cN_++2)-\chi_M(\cN_+))^2 + \hc
\end{align}
With the bounds of Lemma~\ref{lem:R_N_bound} on the first line and~\eqref{eq:lem:G_N_bound_5} (with $t=1/2+\eps$) together with~\eqref{eq:H_eff,N-form_bound} for the last line, this gives
\begin{align}
 \big|\big\langle\Phi,[\chi_M,[\chi_M, \cH_N^U]] \Phi\big \rangle\big| &\leq C N^{\eps} \langle \Phi, (\cH_N^U-E_{N,W}+C)\Phi\rangle \sup_{x\geq 0}\big(\chi(\tfrac{x+1}{M})-\chi(\tfrac{x}{M})\big)^2 \nn \\
 &\lesssim M^{-2}N^{\eps} \lesssim N^{-\alpha/8}.\label{eq:localise_error}
\end{align}
The same reasoning applies to the commutator with $\eta_M$, and we thus obtain
\begin{align}\label{eq:upper-localise}
 \langle \chi_M \Phi,(\cH_N^U-E_{N,W})\chi_M\Phi\rangle \leq \langle \Phi,(\cH_N^U-E_{N,W})\Phi\rangle + \mathcal{O}(N^{-\alpha/8}).
\end{align}
In view of~\eqref{eq:upper_first} and with the choice $\alpha=1/10$ this proves the claim.
\end{proof}

\begin{lem}[Lower bound]\label{lem:lower}
 Assume the hypothesis of Theorem~\ref{thm:asymptotic} and let
\begin{equation*}
 \lambda_N:=4\pi  \ao_V (N-1)  + 8\pi  \ao_W \sqrt{N}   + e^{(U)} + E_{N,W}.
\end{equation*}
Then for any $k\in \N$
\begin{equation*}
 e_k(H_N) - \lambda_N \geq e_k(H_\mathrm{BF}) + \mathcal{O}(N^{-1/80}).
\end{equation*}
\end{lem}
\begin{proof}
\emph{Localizing in $\mathcal N_+$.}
In order to control the error terms with $\cN_+^2$, $\cN_+^3$ in $\cE$ from~\eqref{eq:EU-proof}, we need to localize on a subspace of $\mathscr{H}_+$ with $\mathcal N_+\leq M$. Let $\chi_M, \eta_M$ be as in Lemma~\ref{lem:upper}, with the choice  of $M=N^{1/2+\alpha/8}$ as the cutoff scale.
We now apply the localization formula~\eqref{eq:IMS} to $\cH_N^{\leq N}:=U_X^* H_N U_X$.
The term localized by $\chi_M$ will be analyzed further, and the other terms need to be estimated as errors.
This is where the condensation hypothesis, Condition~\ref{cond:BEC} enters into play. It translates to (compare with (\ref{eq:equiv_cond}))
\begin{align}\label{eq:BEC_apriori}
 &\cH_N^{\leq N}  \geq (4\pi \mathfrak{a}_V N+  c \cN_+ -d\sqrt{N}) \1_{\cN_+\leq N}.
\end{align}
For the choice $M=N^{1/2+\alpha/8}$, using that $\eta_M^2 \leq M^{-1} \mathcal N_+ \eta_M^2$, we can reconstruct the energy $ \lambda_N$ on the sector of large particle numbers using the gap, that is
\begin{align}
 \eta_M  \cH_N^{\leq N} \eta_M
 &\geq \lambda_N \eta_M^2\1_{\cN_+\leq N} \label{eq:H_N eta}
\end{align}
for $N$ sufficiently large.

To control the commutators, we collect the terms that do not commute with $\cN_+$ from Proposition~\ref{prop:H-ex}, which gives
\begin{align}
 [\eta_M,[\eta_M,\cH_N^{\leq N}]]& = [\eta_M,[\eta_M, (\cL_2 + \cL_3+Q_1)]].
\end{align}
Arguing as in~\eqref{eq:localise_error} and using the bounds $\cL_2+\cL_3 \lesssim \cL_4 + N$, $Q_1\lesssim \ud\Gamma(W_{N,x})+\sqrt{N}$, which follow from the Cauchy-Schwarz inequality (cf.~\eqref{eq:CS_L4}), we can bound the localization errors by
\begin{align}
 [\eta_M,[\eta_M,\cH_N^{\leq N}]] + [\chi_M,[\chi_M,\cH_N^{\leq N}]] &\lesssim M^{-2}(\cL_4 + \ud\Gamma(W_{N,x}) + N)\1_{ M/2\leq \cN_+\leq M}\nn \\
 & \lesssim N^{-1-\alpha/4}(\cH_N^{\leq N}+N).
\end{align}
Here we used that $ \mathds{1}_{\cN_+ \leq N}(\cL_4 + \ud\Gamma(W_{N,x})) \lesssim \cH_N^{\leq N} + N$ which follows again from the Cauchy-Schwarz inequality applied to the terms composing $\cH_N$.

For the term localized at $\cN_+\leq M=N^{1/2+\alpha/8}$, first note that $\chi_M\1_{\cN_+\leq N}=\chi_M$ for $N$ large enough, so we can drop the projections.
Making use of the unitary to rewrite this term as in~\eqref{eq:U-remind} yields together with~\eqref{eq:H_N eta}
\begin{align}
 \cH_N^{\leq N}& \geq \lambda_N \1_{\cN_+\leq N}+ \chi_M U\big(\cH_N^U -E_{N,W}+\cL_4 +\cE\big)U^* \chi_M.
\end{align}
\emph{Estimating $\cE$.} We can now control the error term $\cL_4 + \cE$. Recall the estimate (\ref{eq:EU-proof}). As $\cL_4\geq 0$, we can simply drop $\cL_4+\mathcal{O}(N^{-\alpha/4}\cL_4)\geq 0$ for large $N$. To bound $\cN_+^2$, $\cN_+^3$, we use that by Lemma~\ref{lem:U-conserve}
\begin{equation}
\begin{aligned}
 \chi_MU\cN_+^k U^*\chi_M \lesssim  \chi_M(\cN_++1)^k \chi_M  \lesssim M^{k-1} \chi_M (\cN_++1) \chi_M.
\end{aligned}
\end{equation}
 We have thus obtained
\begin{align}
 &\chi_M U\big(\cL_4 +\cE\big)U^* \chi_M \geq - \chi_MU\Big(N^{-\alpha/2}  \ud \Gamma(W_{N,x})+ (N^{-\alpha/8}+N^{-1/2+9\alpha/4})\cN_+  \nn\\
 &\qquad
	   +   N^{-\alpha/4}\big(\log N \dd\Gamma(|\ui \nabla|)  +(\log N)^2 +(\mathcal N_++1)^{1/2} |\nabla_x|\Big)U^*\chi_M.
 \end{align}
These remaining error terms have to be estimated using $\cH_N^U-E_{N,W}$.
The term $\ud \Gamma(W_{N,x})$ is bounded as in~\eqref{eq:W-H_N-bound}.
With $\ud\Gamma(|\ui \nabla|)\leq \cN_+^{1/2}\ud\Gamma(\epsilon)^{1/2}$ and~\eqref{eq:H_eff,N-form_bound} we also get
\begin{multline}
 N^{-\alpha/8}\cN_+ +   N^{-\alpha/4}\big(\log N \dd\Gamma(|\ui \nabla|)  + (\log N)^2 +(\mathcal N_++1)^{1/2} |\nabla_x|) \\ \lesssim N^{-\alpha/8}(\cH_N^U-E_{N,W}+C).
\end{multline}
We have thus shown that
\begin{align}
 U_X^* H_N U_X \geq  \lambda_N\1_{\cN_+\leq N}  + (1-CN^{-\alpha/8})\chi_M U\big( \cH_N^U-E_{N,W}\big)U^* \chi_M - CN^{-\alpha/8}.\label{eq:lower-temp}
\end{align}

\emph{Removing the localization.}
We are now ready to bring the proof of the lower bound to conclusion. Let $X_k\subset L^2(\T^{3(N+1)})$ denote the span of the first $k+1$ eigenvectors of $H_N$ and let $\Psi \in X_k$ with norm equal to one. We first need to show that $\chi_M U_X^* \Psi\approx U_X^*\Psi$ for large $N$.
We have with $M=N^{1/2+\alpha/8}$ as before, using the a priori bound (\ref{eq:BEC_apriori}) and the upper bound from Lemma~\ref{lem:upper},
\begin{align}
 \| (1-\chi_M) U_X^* \Psi\|^2 &\leq 2M^{-1} \langle(1-\chi_M) U_X^* \Psi, \cN_+(1-\chi_M) U_X^* \Psi \rangle \nn \\
 &\lesssim M^{-1} \langle \Psi, (H_N-4\pi  \mathfrak{a}_{V} (N-1)) \Psi \rangle + \mathcal{O}(N^{-\alpha/8}) \nn \\
 &= \mathcal{O}(N^{-\alpha/8}).
\end{align}
In particular, $\chi_M U_X^*:X_k\to \mathscr{H}_+$ is injective for sufficiently large $N$.
For the min-max values this implies, coming from~\eqref{eq:lower-temp},
\begin{align}
 e_k(H_N)-\lambda_N &\geq \max_{\Psi\in X_k \atop \|\Psi\|=1} (1-CN^{-\alpha/8}) \langle \chi_MU_X^*\Psi ,U( \cH_N^U-E_{N,W}\big)U^*\chi_MU_X^*\Psi\rangle \nn \\
 &\qquad + \mathcal{O}(N^{-\alpha/8}) \nn \\
  %
  %
  &\geq  e_k( \cH_N^U-E_{N,W}) + \mathcal{O}(N^{-\alpha/8}).
\end{align}
Taking $\alpha=1/10$, the largest value permitted by Proposition~\ref{prop:U_G}, proves the claim.

\end{proof}

\subsection{Proof of Theorem~\ref{thm:asymptotic}}\label{sect:Thm1}

\begin{proof}[Proof of Theorem~\ref{thm:asymptotic}] \leavevmode\\
The upper and lower bounds proved in the lemmas above imply that
\begin{equation}
\inf \sigma(H_N) = \lambda_N + \mathcal{O}(1).
\end{equation}
The proof now simply consists in analyzing
\begin{equation}
  \lambda_N:=4\pi  \ao_V (N-1)  + 8\pi  \ao_W \sqrt{N}   + e^{(U)} + E_{N,W},
\end{equation}
where $e^{(U)}$ is given in Proposition \ref{prop:U_G} and $E_{N,W}$ by~\eqref{eq:E_NW}. Using Lemma \ref{lem:aM-difference} and straightforward estimates on the sums, we obtain that $e^{(U)}$ is of order one. Hence, it remains to extract the logarithmic contribution from $E_{N,W}$. Let us write
\begin{align}
\begin{aligned}
  E_{N,W}&= - 2\Big\langle \nabla  v_N \otimes \nabla v_N ,  (\mathbb{H}^{(2)})^{-1} \nabla v_N \otimes \nabla v_N\Big \rangle,  \\
 \mathbb{H}^{(2)} &= -(\nabla_{y_{1}}+\nabla_{y_{2}})^2 +\epsilon(\ui \nabla_{y_1})+\epsilon(\ui \nabla_{y_2})+W_N(y_1)+W_N(y_2)+1.
\end{aligned}
\end{align}
Recall that for large momenta, which are the relevant ones for the divergence, $v_N$ is essentially given by the scattering solution, $v_N=\sqrt{N}\Wphi$  for $N>N^\alpha$ (cf.~\eqref{eq:w,v-def}).
The term of order $\log N$ is already contained in $E_N:=E_{N,0}$, which is obtained by replacing $\mathbb{H}^{(2)} $ with
\begin{equation}
 \mathbb{H}_0^{(2)} = -(\nabla_{y_{1}}+\nabla_{y_{2}})^2 +\epsilon(\ui \nabla_{y_1})+\epsilon(\ui \nabla_{y_2})+1.
\end{equation}
We have
\begin{subequations}
 \begin{align}
 &E_{N,W}-E_{N}\nn \\
 &=4\Big\langle \nabla  v_N \otimes \nabla v_N ,  (\mathbb{H}^{(2)})^{-1} W_N(y_1)(\mathbb{H}^{(2)}_0)^{-1} \nabla v_N \otimes \nabla v_N\Big \rangle \nn \\
 &= 4\Big\langle \nabla  v_N \otimes \nabla v_N ,  (\mathbb{H}^{(2)}_0)^{-1} W_N(y_1)(\mathbb{H}^{(2)}_0)^{-1} \nabla v_N \otimes \nabla v_N\Big \rangle \label{eq:E_N-diff1}\\
 &\quad  - 4\Big\langle \nabla  v_N \otimes \nabla v_N ,  (\mathbb{H}^{(2)}_0)^{-1} W_N(y_1) (\mathbb{H}^{(2)})^{-1}  \nn \\
 &\qquad\qquad \times (W_{N}(y_1)+W_N(y_2))(\mathbb{H}^{(2)}_0)^{-1} \nabla v_N \otimes \nabla v_N\Big \rangle \label{eq:E_N-diff2}.
\end{align}
\end{subequations}
Using that $\|W_N^{1/2}(\mathbb{H}^{(2)})^{-1/2}\|$ is uniformly bounded, we can bound the expression of the last two lines by
\begin{align}
 |\eqref{eq:E_N-diff2}|\lesssim \| W_N(y_1)^{1/2} (\mathbb{H}^{(2)}_0)^{-1} \nabla  v_N \otimes \nabla v_N  \|^2= \tfrac12 \eqref{eq:E_N-diff1}.
\end{align}
This is now bounded using~\eqref{eq:bound_WN_onebody} and Lemma~\ref{lem:phi_W} with $0<s<1/4$
\begin{align}
 \eqref{eq:E_N-diff1} \leq \underbrace{\|\epsilon^{-1/4+s/2} \nabla v_N\|^4}_{\stackrel{\text{Lemma }\ref{lem:phi_W}}{\lesssim} N^{2s}} \underbrace{\|\epsilon^{-1/2-s}W_N \epsilon^{-1/2-s}\|}_{\stackrel{\eqref{eq:bound_WN_onebody}}{\lesssim} N^{-2s}}\lesssim 1.
\end{align}
We have thus reduced the problem to the analysis of
\begin{equation}\label{eq:E_N0}
 E_N=-2\sum_{p,q\in 2\pi \Z^3}\frac{(p\cdot q)^2  |\widehat  v_N(p)|^2 |\widehat  v_N(q)|^2}{((p+q)^2+\epsilon(p)+\epsilon(q)+1}.
\end{equation}
We may cut-off any of the sums at scale $\sqrt{N}$ at the expense of an error of order one, since, with $p^2 \widehat  v(p)=\mathcal{O}(1)$ from~\eqref{eq:v_N-bound} and $\|\nabla v_N\|_{2}=\mathcal{O}(N^{1/4})$ from Lemma~\ref{lem:phi_W},
\begin{align}
 \sum_{|p|> \sqrt{N} }\frac{(p\cdot q)^2  |\widehat  v_N(p)|^2 |\widehat  v_N(q)|^2}{((p+q)^2+\epsilon(p)+\epsilon(q)+1} \leq \|\nabla v_N\|^2_{2} \sum_{|p|> \sqrt{N} } \frac{1}{p^2 \epsilon(p)}\lesssim 1.
\end{align}
To compute the sum, we can replace $\epsilon(p)$ by $p^2$, at the price of another error of order one. We also need to deal with $\widehat  v_{N}$ which we recall is defined in (\ref{eq:w,v-def}). We can replace $\widehat  v_{N} \1_{\kappa < |p| \leq N^{\alpha}}$ by $4\pi \ao_{W}$ in (\ref{eq:E_N0}) up to an error of order one. For the large momentum part $\widehat  v_{N} \1_{|p| > N^{\alpha}} =\sqrt N \Wphih(p)  \ue^{\ui p x}$, we use the scattering equation (\ref{eq:scattW}) and the definition of $\aW$ in (\ref{eq:aW_N}), to obtain
\begin{align}
\Big| \sqrt{N} \Wphih(p) - \frac{4\pi \aW}{p^2} \Big| 
	&\lesssim \frac{1}{2} \left|\int \sqrt{N} W_N(1+ \varphi_{\mathrm{I}}) \frac{(e^{i p \cdot x} - 1)}{p^2}\right| \nn \\
	&\lesssim \| \varphi_{\mathrm{I}} \|_{\infty} \int  N^{3/2} W(\sqrt N x) \frac{|p\cdot x|}{p^2} \lesssim  \frac{1}{|p|\sqrt{N}}.
\end{align}
Therefore, using the cutoff in the sum and also that $\aW = \ao_{W} + \mathcal O(N^{-1/2})$, we can replace entirely $\widehat  v_{N} $ by $4\pi \ao_W/p^2$ in (\ref{eq:E_N0}) up to an error of order one.
Next, we may also replace the sum by the corresponding integral, as the difference of the summand at two different points decays more quickly. Thus,
\begin{align}
 E_N=- \frac{2}{(2\pi)^6} \int\limits_{|p|,|q|\leq \sqrt{N}}  \frac{(4\pi \ao_W)^4 (p\cdot q)^2\ud q\ud p}{p^4q^4((p+q)^2+p^2+q^2+1)}+\mathcal{O}(1).
\end{align}
This integral can now be evaluated up to order $\log N$. Using spherical coordinates and Gaussian integrals, we obtain
\begin{align}
 \int  \frac{ (p\cdot q)^2\ud q}{(p^2q^4((p+q)^2+p^2+q^2)}
 &= 2\pi \int_{-1}^1  \int_0^\infty  \frac{s^2 \ud r\ud s}{(2r^2 + 2p^2 + 2|p| rs)}  \\
 &= \pi \int_{-1}^1  \int_0^\infty \int_0^\infty   s^2 \ue^{-t(r^2+p^2+|p|rs)} \ud t \ud r \ud s \nn \\
 &= \frac{\pi^2}{2|p|} \int_{-1}^1 \frac{s^2 \ud s}{\sqrt{1-s^2/4}} = \frac{\pi^2}{|p|} \Big(4\arcsin\big(\tfrac12\big) -2\sqrt{\tfrac34}\Big).\nn
\end{align}
Integrating over $p$, this yields
\begin{align}
 E_N= -32\pi\ao_W^4\Big(\frac{2\pi}{3} -\sqrt{3}\Big) \log N + \mathcal{O}(1).
\end{align}
This completes the proof of Theorem~\ref{thm:asymptotic}.
\end{proof}

\subsection{Proof of Theorem~\ref{thm:spect} and Corollary~\ref{cor:operators}}

\begin{proof}[Proof of Theorem~\ref{thm:spect}] \leavevmode\\
This follows immediately by combining the upper bound of Lemma~\ref{lem:upper} with the lower bound of Lemma~\ref{lem:lower} and cancelling all the scalar contributions to $e_k(H_N)$, $e_0(H_N)$.
\end{proof}

\begin{proof}[Proof of Corollary~\ref{cor:operators}] \leavevmode\\

\noindent
\textit{a)} Convergence of eigenfunctions up to subsequences follows immediately from the convergence of spectral projections of part \textit{b)}. Indeed,  denote $e_0:=e_0(H_\mathrm{BF})$ and let $B$ be an interval containing $e_k(H_{\mathrm{BF}})-e_0$, but no different eigenvalues of $H_\mathrm{BF}-e_0$. Then, for $N$ sufficiently large, $e_k(H_N)-e_0(H_N)\in B$ and $U_\mathrm{B}^* U_X^* \Psi_N = U_\mathrm{B}^* U_X^* \1_{B}(H_N-e_0(H_N))\Psi_N$.
With \textit{b)} the difference to $  \1_{B}(H_\mathrm{BF}-e_0)U_\mathrm{B}^* U_X^*\Psi_N$ tends to zero. The latter sequence has a convergent subsequence in the finite-dimensional range of $\1_{B}(H_\mathrm{BF}-e_0)=\1(H_\mathrm{BF}=e_k(H_\mathrm{BF}))$.

\medskip\noindent
\textit{b)} Recall from (\ref{eq:U-remind}) that $ U_X^* H_N U_X = \1_{\cN_+ \leq N} \mathcal H_N  \1_{\cN_+ \leq N}$, with
\begin{equation}
 U^* \mathcal H_N U =  \cH_N^U-E_{N,W} + \lambda_N + \mathcal L_4  + \cE.\nn
\end{equation}
We already know from the norm-resolvent convergence of Proposition~\ref{prop:CV_transf_Ham} that
\begin{equation}\label{eq:proj-H_BF-converge}
 \lim_{N\to \infty} \1_{B}(\cH_N^U-E_{N,W}-e_0)=(U_\kappa^\infty)^* \1_{B}(H_\mathrm{BF}-e_0)U_\kappa^\infty. 
\end{equation}
We thus need to take care of the interaction term $\mathcal L_4$, the error term $\cE$ and the convergence of the unitary $U= U_q U_W U_c U_\mathrm{B}U_\mathrm{G}$, which depends on $N$ via the generators of its factors defined in Section~\ref{sect:trafo}.

\emph{Convergence of $U$.} Concerning the unitary, we have
\begin{equation}
 \lim_{N\to \infty} U\1_{B}(\cH_N^U-E_{N,W}-e_0)U^* = \lim_{N\to \infty} U_\mathrm{B}U_\kappa^\infty\1_{B}(\cH_N^U-E_{N,W}-e_0)(U_\mathrm{B}U_\kappa^\infty)^*. \label{eq:U-converge}
\end{equation}
The reason is that for the transformations $U_q$, $U_W$ and $U_c$ acting on large momenta,
\begin{align}
 0=\lim_{N\to \infty} (U_q-1) (\cN_++1)^{-1/2}&=\lim_{N\to \infty} (U_c -1)(\cN_++1)^{-1/2} \nn  \\ &= \lim_{N\to \infty} (U_W -1)(\cN_++1)^{-1/2} ,
\end{align}
which follows easily using Duhamel's formula as in Sections~\ref{sect:quadratic},~\ref{sect:cubic} in the first two cases and a simple bound on the generator (compare~\eqref{eq:U_kappa-conv}) in the last case.
By the same reasoning, the transformations $U_\mathrm{G}$ and $U_\mathrm{B}$ that act on small momenta converge to their limits, defined by setting $N=\infty$ in the generators, when regularised by $(\cN_++1)^{-1/2}$ (for $N=\infty$, $U_\mathrm{G}$ equals $U_\kappa^\infty$ as defined in~\eqref{eq:def_U_kappa_Lambda}). Since all the unitaries are uniformly bounded on $Q(\cN_+)$ (cf.~Lemma~\ref{lem:U-conserve}) and $(\cN_++1)^{1/2}$ is uniformly bounded on the range of the spectral projection by~\eqref{eq:H_eff,N-form_bound} this proves~\eqref{eq:U-converge}.

\emph{Convergence of the spectral projections.} It remains to prove that
\begin{equation}
 \lim_{N\to \infty}\Big( \1_{B}(\cH_N^U-E_{N,W}-e_0))-U^*U_X^* \1_{B}(H_N-e_0(H_N))U_XU\Big)=0.
\end{equation}
We will argue by induction on the number $k$ of the min-max value. For any $k$ let $B_{k}$ be an interval containing $e_k(H_\mathrm{BF})-e_0$, but no other part of the spectrum of $H_\mathrm{BF}-e_0$, as in the hypothesis. Define $\delta_k$ by
$\sup{B_{k}}= e_k(H_\mathrm{BF}) -e_0 +\delta_k$,
and then for $j\in \N_0$
\begin{equation}
\begin{aligned}
\mathbb{P}_k &= U^* U^*_X \1\big( H_N\leq e_k(H_N)+\delta_k\big)U_X U  = 1-\mathbb{Q}_k, \\
\widetilde{\mathbb{P}}_j &= \1\big(\cH_N^U\leq e_k(\cH_N^U)+\delta_k\big) = 1 - \widetilde{\mathbb{Q}}_k.
\end{aligned}
\end{equation}
Then we have for $N$ sufficiently large
\begin{align}
\begin{aligned}
 U^* U^*_X \1_{B_{k}}(H_N-e_0(H_N) ) U_X U &= \mathbb{P}_k - \mathbb{P}_{k-1}, \\
\1_{B_{k}}(\cH_N^U-E_{N,W}-e_0) &= \widetilde{\mathbb{P}}_k - \widetilde{\mathbb{P}}_{k-1},
\end{aligned}
\end{align}
with the convention that $\mathbb{P}_{-1} = \widetilde{\mathbb{P}}_{-1} = 0$. Therefore, it is enough to prove that $\mathbb{P}_k - \widetilde{\mathbb{P}}_k = \widetilde{\mathbb{Q}}_k \mathbb{P}_k - \widetilde{\mathbb{P}}_k \mathbb{Q}_k$ converges to zero for every $k$. In fact, we will prove that $\widetilde{\mathbb{Q}}_k \mathbb{P}_k$ and $\widetilde{\mathbb{P}}_k \mathbb{Q}_k$ converge to zero.
Let us now assume the statement that has been proved up to $k-1$ for some $k\geq 0$ (it is obvious for $k<0$, as $\mathbb{P}_{-1}=0$).
In Lemma~\ref{lem:upper}, we showed that
\begin{align}
  &\mathrm{Tr} \Big((\cH_N^U - E_{N,W}-e_0)\widetilde{\mathbb{P}}_{k}\Big) \nn\\
 &\!\!\!\! \stackrel{\eqref{eq:upper-localise}}{\geq} \mathrm{Tr} \Big(\widetilde{\mathbb{P}}_k\chi_M(\cH_N^U-E_{N,W}-e_0)\chi_M \widetilde{\mathbb{P}}_k\Big) + \mathcal{O}(N^{-1/80}) \nn\\
 &\stackrel{}{\geq}\mathrm{Tr} \Big(\widetilde{\mathbb{P}}_k\chi_M (U^*U_X H_N U_X U -\lambda_N - e_0) \chi_M \widetilde{\mathbb{P}}_k\Big) + \mathcal{O}(N^{-1/80}) \nn \\
 & \geq \mathrm{Tr}\Big(\widetilde{\mathbb{P}}_k\chi_M  (U^*U_X H_N U_X U -e_0(H_N))\mathbb{P}_k\chi_M \widetilde{\mathbb{P}}_k\Big) \nn \\
 &\qquad +(e_{k+1}(H_N)-e_0(H_N))\mathrm{Tr}\Big(\widetilde{\mathbb{P}}_k\chi_M  \mathbb{Q}_k\chi_M \widetilde{\mathbb{P}}_k\Big)+ \mathcal{O}(N^{-1/80}), \label{eq:spectral-PQ}
\end{align}
where we specifically used (\ref{eq:minmax_ub}) together with (\ref{eq:est_L4_ub}) and (\ref{eq:chi_M-E-est}) for the second inequality to get rid of $\mathcal L_4$ and $\cE$, as well as the lower bound on $e_0(H_N)$ from Lemma~\ref{lem:lower}.
We may remove the localization $\chi_M$ in the two terms in (\ref{eq:spectral-PQ}) thanks to the presence of $\mathbb{P}_k$ and because
\begin{equation}
 \| \widetilde{\mathbb{P}}_j(1-\chi_M)\|^2 \stackrel{\eqref{eq:eta_M-Phi_bound}}{=} \mathcal{O}(N^{-1/80}).
\end{equation}
Let $d_k$ be the dimension of $\mathrm{ran}\,\1_{B_{k}}(\cH_N^U-E_{N,W}-e_0)$, which equals the multiplicity of $e_k(H_\mathrm{BF})$ by~\eqref{eq:proj-H_BF-converge} for $N$ large enough.
By the induction hypothesis, which implies in particular that $\mathbb{P}_{k-1}\widetilde{\mathbb{P}}_k-\mathbb{P}_{k-1}\to0$, we then have by Theorem \ref{thm:spect},
\begin{equation}
 \lim_{N\to \infty} \mathrm{Tr}\Big(\widetilde{\mathbb{P}}_{k}\chi_M  (U^*U_X H_N U_X U -e_0(H_N))\mathbb{P}_{k-1}\chi_M \widetilde{\mathbb{P}}_{k}\Big)= \sum_{j=0}^{k-1}d_j (e_j(H_\mathrm{BF})-e_0).
\end{equation}
Moreover, by (\ref{eq:conv_vp_Htild_Hbf}) we also have
\begin{equation}
 \lim_{N\to \infty} \mathrm{Tr} \Big((\cH_N^U - E_{N,W}-e_0)\widetilde{\mathbb{P}}_{k}\Big) = \sum_{j=0}^k d_j (e_j(H_\mathrm{BF})-e_0).
\end{equation}
We can thus take the limit in~\eqref{eq:spectral-PQ}, whence the terms with $j<k$ cancel, giving
\begin{align}
d_k(e_k(H_\mathrm{BF})-e_0)  \geq  \limsup_{N\to \infty} \Big(& (e_k(H_\mathrm{BF})-e_0)\mathrm{Tr}\big((\widetilde{\mathbb{P}}_{k}-\widetilde{\mathbb{P}}_{k-1}) \mathbb{P}_k\big) \\
 &\quad  + (e_{k+1}(H_\mathrm{BF})-e_0)\mathrm{Tr}\big((\widetilde{\mathbb{P}}_{k}-\widetilde{\mathbb{P}}_{k-1}) \mathbb{Q}_k \big)\Big).\nn
\end{align}
As $ \mathbb{P}_k +\mathbb{Q}_k=1$ and $d_k=\mathrm{Tr}(\widetilde{\mathbb{P}}_{k}-\widetilde{\mathbb{P}}_{k-1})$ for large $N$, dividing by $d_k$ and eliminating $e_0$ expresses $e_k(H_{BF})$ as a convex combination of itself and $e_{k+1}(H_\mathrm{BF})$. But since $e_k(H_\mathrm{BF})<e_{k+1}(H_\mathrm{BF})$ this implies,
\begin{equation}
0= \lim_{N\to \infty} \mathrm{Tr}\big((\widetilde{\mathbb{P}}_{k}-\widetilde{\mathbb{P}}_{k-1}) \mathbb{Q}_k \big)=\lim_{N\to \infty} \mathrm{Tr}\big(\widetilde{\mathbb{P}}_{k} \mathbb{Q}_k \big)\geq\lim_{N\to \infty} \| \widetilde{\mathbb{P}}_{k} \mathbb{Q}_k\|,
\end{equation}
where we used the induction hypothesis once more.
The argument for the convergence of $\mathbb{P}_k  \widetilde{\mathbb{Q}}_k$ is essentially the same, using the inequalities proved in Lemma~\ref{lem:lower}.

\medskip\noindent\textit{c)}
This follows from \textit{b)} by an approximation argument.
To be precise, let $\eps >0$ and let $N$ be large enough for $\|U_X^* \Psi_N-\Psi\|<\eps$ to hold. Choose  $k$  so that
\begin{equation}
\| \1(H_\mathrm{BF}>e_k(H_\mathrm{BF})+\delta)U^*_\mathrm{B} \Psi\|<\eps,
\end{equation}
where $\delta$ is such that part \textit{b)} applies to $B_\delta(e_k(H_\mathrm{BF})H_\mathrm{BF}))$.
By the convergence of spectral projections proved there, we have
\begin{equation}
 \| \1(H_N>e_k(H_N)+\delta) U_X \Psi\|<2\eps, 
\end{equation}
for sufficiently large $N$.
Consequently, using unitarity of the group,
\begin{align}
&\Big\| U_X^*\ue^{-\ui t(H_N-e_0(H_N))}\Psi_N - U_\mathrm{B} \ue^{-\ui t (H_\mathrm{BF}-e_0(H_\mathrm{BF}))} U_\mathrm{B}^*\Psi\| \nn \\
&\leq \Big\|  \sum_{j=0}^k \Big(\ue^{-\ui t(e_j(H_N)-e_0(H_N))} U_X^*\1(H_N=e_j(H_N))U_X \Psi \nn\\
&\qquad - \ue^{-\ui t(e_j(H_\mathrm{BF})-e_0(H_\mathrm{BF}))}U_\mathrm{B}\1(H_N=e_j(H_N))U_\mathrm{B}^*\Psi\Big) \Big\| +4\eps,
\end{align}
which is less than $5\eps$ for sufficiently large $N$ by the convergence of the eigenvalue differences from Theorem~\ref{thm:spect} and the spectral projections from part \textit{b)}.
\end{proof}

\appendix

\section{The scattering problem on the torus}\label{app:scatt} 

\subsection{Properties of the scattering solution}

Let $ 0\leq v \in L^2(\T^3)$ be even and let us define
\begin{align}
\varphi = - \frac{1}{2} \frac{1}{P^{\perp}(-\Delta + \frac{1}{2}v)P^{\perp}} P^{\perp} v \in P^{\perp} L^2(\T^3), \label{eq:phi_scatt}
\end{align}
where $P^{\perp}$ is the orthogonal projection to the complement of the constant $\1$, i.e, functions that integrate to zero. 
The Fourier coefficients of $\phi$ then solve the scattering equation (\ref{eq:scattV}) with potential $v$,
\begin{equation}\label{eq:scat:gen}
 p^2 \widehat\phi(p) + \tfrac12 \sum_{q\in 2\pi \in \Z^3 \setminus\{0\}} \widehat\phi(q) \widehat v(p-q) = - \tfrac12 \widehat v(p)
\end{equation}

\begin{lem}[Regular estimates]
	\label{lem:reg_est}
There is $C>0$ so that for $0\leq v\in L^2(\T^3)$ 
\begin{subequations}
 \begin{align}
\|\phi\|_{2} &\leq C  \|v\|_{2}^{2/3} \|v\|_{1}^{4/3}, \qquad \| \nabla \phi\|_{2} \leq C  \|v\|_{2}^{1/3} \|v\|_{1}^{2/3},    \label{eq:est_phi_2_p} \\
| \widehat\phi (p)|
	&\leq \frac{C}{p^2}   \|v\|_{2}^{2/3} \|v\|_{1}^{4/3}, \quad p \in  2\pi\Z^3\setminus\{0\}, \label{eq:est_phi_pw} \\
 \| \widehat\phi\|_1 
	&\leq C  \|v\|_{1}^{1/3} \|v\|_2^{2/3} \left( 1+  (1+\|v\|_{2}^2) \|v\|_{1} \right). \label{eq:est_phi_1_infty}
\end{align}
\end{subequations}
\end{lem}

\begin{rem} \label{rem:Lp-scaling}
Taking $v_n(x) = n^2 v(n x)$ for $n \geq 1$, that is $\widehat v_n (p) = n^{-1} \widehat  v (p/n)$, we obtain
\begin{align*}
\|v_n\|_{2}^{2/3} \|v_n\|_{1}^{4/3}& = n^{-1} \|v\|_{2}^{2/3} \|v\|_{1}^{4/3},\\
\|v_n\|_{1}^{1/3} \|v_n\|_2^{2/3} &=  \|v\|_{1}^{1/3} \|v\|_2^{2/3}, \\
 (1+\|v_n\|_{2}^2) \|v_n\|_{1} &=  (n^{-1}+\|v\|_{2}^2) \|v\|_{1}.
\end{align*}
\end{rem}

\begin{proof}
Since $v\geq 0$ is in $L^2$, we have $\phi\in L^2(\T^3)$. 
Writing the scalar product of $-\Delta \phi$ with $\phi$ in Fourier space, we obtain with $v\geq 0$
\begin{align}
\| p \widehat\phi \|_2^2  \leq \| p \widehat\phi \|_2^2 + \tfrac{1}{2}\|\sqrt v \phi \|_2^2 
	&= - \tfrac{1}{2} \langle \widehat v, \widehat\phi\rangle \leq \left(\sum_{p \in 2\pi\Z^3\setminus\{0\}} \frac{| \widehat v(p)|^2}{p^2}\right)^{1/2} \|p \widehat\phi\|_2.
\end{align}
We deduce that $\phi \in H^1(\T^3)$ and $\|p \widehat\phi\|_2  \leq \| |p|^{-1} \widehat v\|_{2}$. Let us now bound, for any $A\geq 1$,
\begin{align}
\| |p|^{-1} \widehat v\|_{2}^2 
	&= \sum_{p \in 2\pi\Z^3\setminus\{0\}} \frac{| \widehat v(p)|^2}{p^2} = \sum_{|p| > A} \frac{| \widehat v(p)|^2}{p^2} +  \sum_{|p| \leq A}\frac{| \widehat v(p)|^2}{p^2}\nn \\
	&\leq  A^{-2} \|\widehat v\|_{2}^2 + C A \|\widehat v\|_{\infty}^2 
	\leq C \|v\|_{2}^{2/3} \|v\|_{1}^{4/3},
\end{align}
where we took $A = \|v\|_{2}^{2/3} \|v\|_1^{-2/3} \geq 1$, used that $\|\widehat v\|_{\infty} \leq \|v\|_1 \leq \|v\|_2$. This proves the second estimate in (\ref{eq:est_phi_2_p}). The same computation gives also $\| |p|^{-2} \ast |\widehat v|^2 \|_\infty \leq C \|v\|_{2}^{2/3} \|v\|_{1}^{4/3}$,
\begin{align} \label{eq:bound_v_ast_phi}
\|\widehat v \ast \widehat\phi\|_{\infty} 
	&\leq \sup_{q \in 2\pi\Z^3\setminus\{0\}}\left(\sum_{p \in 2\pi\Z^3\setminus\{0\}} \frac{| \widehat v(p-q)|^2}{p^2}\right)^{1/2}\|p \widehat\phi\|_2 \leq C \|v\|_{2}^{2/3} \|v\|_{1}^{4/3},
\end{align}
From this and the scattering equation (\ref{eq:scat:gen}), we obtain the point-wise bound on $\widehat\varphi$ (\ref{eq:est_phi_pw}) as well as the first estimate in (\ref{eq:est_phi_2_p}).
Multiplying (\ref{eq:scat:gen}) by $|p|^{-2}$ and summing over $p$, we also obtain 
\begin{align} \label{eq:phi_1}
2 \|\widehat \phi\|_1 
	&\leq \||p|^{-2} \widehat v  \|_1 + \||p|^{-2} \widehat v \ast \widehat\phi\|_1.
\end{align}
The first term is bounded as follows, picking $A = (\|v\|_2/\|v\|_1)^{2/3} \geq 1$, we obtain
\begin{align}
\||p|^{-2} \widehat v  \|_1 \leq \| \1_{|p|\leq A} |p|^{-2} \|_1 \|\widehat v\|_{\infty} +  \| \1_{|p|> A} |p|^{-2} \|_2 \|v\|_{2} \leq C \|v\|_{1}^{1/3} \|v\|_2^{2/3}.
\end{align}
By iterating the equation (\ref{eq:phi_scatt}) and using that $\||p|^{-2} \widehat{v} \ast h \|_{6p/(p+6)+\varepsilon} \leq C_\varepsilon \|h\|_{p}$ for all $p\geq 6/5$, $\varepsilon>0$, we can estimate the second term in (\ref{eq:phi_1}) and already obtain that $\| \widehat\phi\|_1  < \infty$. To obtain a quantitative estimate, we decompose the second term for $A \geq 1$ as
\begin{align}
\||p|^{-2} \widehat v \ast \widehat\phi\|_1 
	&\leq \||p|^{-2}\1_{|p|\leq A} \widehat v \ast \widehat\phi\|_1 + \||p|^{-2}\1_{|p|> A} \widehat v \ast \widehat\phi\|_1 \nn \\
	&\leq \||p|^{-2} \1_{|p|\leq A}\|_1 \|\widehat v \ast \widehat\phi\|_\infty + \||p|^{-2}\1_{|p|> A}\|_2 \|\widehat v \ast \widehat\phi\|_2 \nn \\
	&\leq C \left( A  \|v\|_{2}^{2/3} \|v\|_{1}^{4/3}  + A^{-1/2}  \|\widehat v\|_2 \| \widehat\phi\|_1\right).
\end{align}
For $A =  C^2\|v\|_2^2 + 1> C^2 \|v\|_2^2$ this  yields
\begin{align}
\| \widehat\phi\|_1 
	&\lesssim (1-A^{-1/2} \|v\|_{2})^{-1} \left( \|v\|_{1}^{1/3} \|v\|_2^{2/3} + A \|v\|_{2}^{2/3} \|v\|_{1}^{4/3}\right) \nn\\
	&\lesssim \|v\|_{1}^{1/3} \|v\|_2^{2/3} \left( 1+  (1+\|v\|_{2}^2) \|v\|_{1} \right),
\end{align}
which concludes the proof.
\end{proof}

We now consider, for $ m \geq 1$, the truncated scattering solution 
\begin{equation}
\widehat\phi_{m}(p) = \1_{|p|> m} \widehat\phi(p). 
\end{equation}

\begin{lem}[Estimates on $\phi_m$]
	\label{lem:est_phi_alpha}
There is $C>0$ so that for $0\leq v\in L^2(\T^3)$ and all $m \geq 1$, we have
\begin{subequations}
 
\begin{align}
&\|\phi_m\|_{2} \leq C m^{-1/2} \|v\|_{2}^{2/3} \|v\|_{1}^{4/3}, \qquad \| \nabla \phi_m\|_{2} \leq C  \|v\|_{2}^{1/3} \|v\|_{1}^{2/3}, \label{eq:phi_alpha} \\
&\| \phi - \phi_m \|_{\infty} \leq C m  \|v\|_{2}^{2/3} \|v\|_{1}^{4/3} \label{eq:diff_phi_infty}
\end{align}
Moreover, for $0\leq s < 1$ there is $C>0$ so that if $s< 1/2$
\begin{align}
\| |p|^s \widehat\phi_m\|_2
	&\leq C m^{-(1/2-s)} \|v\|_{2}^{2/3} \|v\|_{1}^{4/3} \label{eq:s_leq}\\
\| |p|^{1/2} \widehat\phi\|_2 
	&\leq C  \|v\|_{2}^{2/3} \|v\|_{1}^{4/3} \sqrt{|\log \|v\|_2| + 1},
\end{align}
and for $1/2 < s \leq 1$
\begin{align*}
\| |p|^{s} \widehat\phi_m\|_2 
	&\leq C m^{-(1-s)/2} \|v\|_{2}^{2-s/3} \|v\|_{1}^{2(2-s)/3}.
\end{align*}

\end{subequations}

\end{lem}

\begin{proof}

The estimates (\ref{eq:phi_alpha}), (\ref{eq:diff_phi_infty}) and  (\ref{eq:s_leq}) follow easily from (\ref{eq:est_phi_pw}) and (\ref{eq:est_phi_2_p}). The estimate for $1/2<s<1$ follows by Hölder's inequality $\| |p|^{s} \widehat  \phi_m\|_2  \leq \| |p| \widehat  \phi_m\|_2^{s} \|  \widehat  \phi_m\|_2^{1-s}$ together with the two bounds in (\ref{eq:phi_alpha}).

It remains to prove the case $s=1/2$. We use the equation (\ref{eq:scat:gen}) to obtain
\begin{align}
2 \||p|^{1/2} \widehat {\phi}\|_2 \leq \| |p|^{-3/2} \widehat  v \|_2 + \| |p|^{-3/2} \widehat  v \ast \widehat \phi \|_2.
\end{align}
We bound the first term as follows, for any $A\geq 1$, we have
\begin{align}
 \| |p|^{-3/2} \widehat  v \|_2^2
 	&\leq \| \1_{|p|\leq A}|p|^{-3/2} \widehat  v \|_2^2 + \|\1_{|p|>A} |p|^{-3/2} \widehat  v \|_2^2 \nn \\
 	&\leq C \log A \|\widehat  v\|_\infty^2 + A^{-3} \|\widehat  v \|_2^2\nn \\
	&\leq C   \|\widehat  v\|_\infty^2(\log ( \|\widehat  v\|_2/ \|\widehat  v\|_\infty) + 1)
\end{align}
where we chose $A =  \|\widehat  v\|_2^{2/3} /  \|\widehat  v\|_\infty^{2/3} \geq 1$. Similarly, we have
\begin{align}
\| |p|^{-3/2} \widehat  v \ast \widehat \phi \|_2
	&\leq C \|\widehat  v \ast \widehat \phi\|_\infty \sqrt{\log ( \|\widehat  v \ast \widehat \phi\|_2/ \|\widehat  v \ast \widehat \phi\|_\infty) + 1}\nn \\
	&\leq C  \|v\|_{2}^{2/3} \|v\|_{1}^{4/3} \sqrt{|\log \|v\|_2| + 1}.
\end{align}
where we used the bounds on $\|\widehat  v \ast \widehat  \phi \|_\infty \leq \|v\|_{2}^{2/3} \|v\|_{1}^{4/3} \leq \|v\|_{2}^2$ from (\ref{eq:bound_v_ast_phi}) and 
\begin{align}
\|\widehat  v \ast \widehat  \phi \|_2  \leq C \|v\|_2 \left( \||p|^{-2} \widehat  v\|_{1} +  (1+\|v\|_{2}^2) \|v\|_{2}^{1/3} \|v\|_{1}^{2/3}\right) \leq C (\|v\|_2+1)^4.
\end{align}
Finally combing the above estimates, we obtain
\begin{align}
\||p|^{1/2} \widehat {\phi}\|_2  \leq C  \|v\|_{2}^{2/3} \|v\|_{1}^{4/3} \sqrt{|\log \|v\|_2| + 1}.
\end{align}
\end{proof}

\subsection{Comparison between torus and free space scattering lengths}

For a generic potential $0\leq v\in L^2(\T^3)$, we define the torus scattering length by minimization problem
\begin{align}
4\pi \ao_{\T^3}(v) = \inf_{\varphi \in H^1(\T^3) \atop \int \phi =0 } \left\{ \int_{\T^3} |\nabla \varphi|^2 + \frac{1}{2} \int_{\T^3} v |1+\varphi|^2 \right\}
\end{align}
By the Lax-Milgram Theorem this problem has a unique minimizer $\varphi_{\T^3}$, which satisfies the equation~\eqref{eq:scat:gen}.
 We may extend $v$ beyond the unit cube by zero, to obtain a potential $0\leq  v\in L^1(\R^3)$. Denote by  $a_{\mathbb{R}^{3}}(v)$ the scattering length on the full space defined in (\ref{eq:scatt_length-intro}) of this potential.

\begin{lem}\label{lem:diff_scatt_gen}
Let $0 \leq v \in L^2(\T^3)$. There exists $C>0$ so that, denoting $v_n(x) = n^{2} v(nx)$, we have for all  $n\geq 1$
\begin{align*}
| n \ao_{\T^3}(v_n) - \ao_{\mathbb{R}^{3}}(v)| \leq C n^{-1}.
\end{align*}
\end{lem}
\begin{proof}
Let $\varphi_{\mathbb{R}^{3}}$ be the solution to (\ref{eq:phi_full_space}), then from \cite[Theorem 6]{NamRicTri-23} we know that
\begin{align}
 -\frac{1}{|x|+1} \lesssim \varphi_{\mathbb{R}^{3}}(x) \leq 0, \quad |\nabla \varphi_{\mathbb{R}^{3}}(x)| \lesssim \frac{1}{|x|^2 + 1}
\end{align}
for all $x \in \mathbb{R}^{3}$.
Let $0\leq \chi \in C^\infty(\R^3)$ such that $\chi(x) = 1$ for $|x| \leq 1/3$ and $\chi(x) = 0$ for $|x| \geq 1/2$. We set $\varphi_{n}(x) = \chi(x)\varphi_{\mathbb{R}^{3}}(n x)+c_n$, with $c_n$ chosen so that the integral of $\phi_n$ over the unit cube equals zero. Then $\phi_n$ satisfies for $n\geq 2$
\begin{align}
-\Delta \varphi_n + \tfrac{1}{2} v_n(1+ \varphi_n) = (-\Delta \chi) \varphi(n\cdot) - 2 n \nabla \chi \cdot \nabla \varphi(n\cdot) +\tfrac12 v_n c_n
\end{align}
since $\supp \, v_n \subset \{|x|\leq 1/3\} \subset \{\chi \equiv 1\}$ for $n\geq 2$. 
Note that 
\begin{equation}
 |c_n| \leq \int |\chi \phi_{\R^3}(n\cdot)| \lesssim  \int\limits_{[-1/2,1/2]^3}\frac{1}{n |x|+1} \lesssim n^{-1},
\end{equation}
so, since $\phi_{\R^3}$ is bounded and $\int v_n = n^{-1} \int v$, 
\begin{equation}
 \int v_n (|1+\phi_n|^2 -|1+\chi \phi_{\R^3}(n\cdot)|^2)  =\mathcal{O}(n^{-2}).
\end{equation}
Then, using the variational characterization of $4\pi \ao_{\T^3}(v_n)$ and (\ref{eq:rel_scat_R3}), we obtain
\begin{align}
4\pi \ao_{\T^3}(v_n) 
	&\leq \int_{\T^3} |\nabla \varphi_n|^2 + \frac{1}{2} \int_{\T^3} v_n |1+\varphi_n|^2 \nn\\
	&=  4\pi \ao_{\mathbb{R}^{3}}(v_n) + \int_{\mathbb{R}^{3}} \left((-\Delta \chi) \varphi_{\R^3}(n\cdot) - 2 n \nabla \chi \cdot \nabla \varphi_{\R^3}(n\cdot)\right) \varphi_n +\mathcal{O}(n^{-2})  \nn \\
	&\leq 4\pi \ao_{\mathbb{R}^{3}}(v_n) + C \int\limits_{[-1/2,1/2]^3} \left(\frac{1}{(n |x|+1)^2} +\frac{n}{(n |x|+1)^3}\right)  +\mathcal{O}(n^{-2})   \nn\\
	&\leq 4\pi \ao_{\mathbb{R}^{3}}(v) n^{-1}+\mathcal{O}(n^{-2}) ,
\end{align}
where we used that $\supp \chi \subset [-1/2,1/2]^3$ and the exact scaling relation $n \ao_{\mathbb{R}^{3}}(v_n) = \ao_{\mathbb{R}^{3}}(v)$.

Reciprocally, take this time  $\varphi_{n}(x) = \chi(x)\varphi_{\T^3}^{(n)}(x)$ where 
\begin{equation}
\varphi_{\T^3}^{(n)} = - \frac12 \frac{1}{P^{\perp}(-\Delta + \frac{1}{2}v_n)P^{\perp}} P^{\perp} v_n \in P^{\perp}\in L^2(\T^3) \cong L^2([-1/2,1/2]^3). 
\end{equation}
The same computation as above now leads to, since $\phi_{\T^3}$ is real,
\begin{align}
4\pi \ao_{\mathbb{R}^{3}}(v_n) -  4\pi  \ao_{\T^3}(v_n) 
	&\leq \int_{\mathbb{R}^{3}} \left((-\Delta \chi) \varphi_{\T^3}^{(n)} - 2 \nabla \chi \cdot \nabla \varphi_{\T^3}^{(n)}\right) \chi \varphi_{\T^3}^{(n)}\nn \\
	&= \int_{\mathbb{R}^{3}} |\nabla \chi|^2 |\varphi_{\T^3}^{(n)}|^2\nn \\ 
	&\leq  C \|\varphi_{\T^3}^{(n)}\|_2^2 =\mathcal{O}( n^{-2}),
\end{align}
where for the last estimate we used Lemma \ref{lem:reg_est} together with Remark \ref{rem:Lp-scaling}.
\end{proof}

\subsection*{Funding}
This work was partially funded by the Deutsche Forschungsgemeinschaft (DFG, German Research Foundation) – Project-ID 470903074 – TRR 352 and supported by the Agence National de la Recherche (ANR, project MaBoP ANR-23-CE40-0025). J.L. acknowledges financial support by the EIPHI Graduate School (ANR-17-EURE-0002) and Bourgogne-Franche-Comté Region through the project SQC.




\begin{thebibliography}{10}



\bibitem{albeverio-1982}
S. Albeverio,  F. Gesztesy,  and R. H{\o}egh-Krohn.
\newblock The low energy expansion in nonrelativistic scattering theory
\newblock {\em Ann. inst. Henri Poincar{\'e} A}, 37(1):1--28, 1982.

\bibitem{basti-21}
Giulia Basti, Serena Cenatiempo, and Benjamin Schlein.
\newblock A new second-order upper bound for the ground state energy of dilute
  {B}ose gases.
\newblock {\em Forum Math. Sigma}, 9:e74, 2021.

\bibitem{BenOliSch-15}
Benedikter Niels, Gustavo de Oliveira, and Benjamin Schlein.
\newblock Quantitative derivation of the {G}ross--{P}itaevskii equation.
\newblock{\em Comm. Pure Appl. Math.} 68(8):1399–1482, 2015.

\bibitem{Bhatia}
Rajendra Bhatia.
\newblock {\em Matrix analysis}, volume 169.
\newblock Springer Science \& Business Media, 2013.

\bibitem{BocBreCenSch-18}
Chiara Boccato, Christian Brennecke, Serena Cenatiempo, and Benjamin Schlein.
\newblock {B}ogoliubov theory in the {G}ross-{P}itaevskii limit.
\newblock {\em Acta Math.}, 222(2):219--335, 2019.

\bibitem{BocBreCenSch-20}
Chiara Boccato, Christian Brennecke, Serena Cenatiempo, and Benjamin Schlein.
\newblock Optimal rate for {B}ose--{E}instein condensation in the
  {G}ross--{P}itaevskii regime.
\newblock {\em Comm. Math. Phys.}, 376:1311--1395, 2020.

\bibitem{boccato-23}
Chiara Boccato, Andreas Deuchert, and David Stocker.
\newblock Upper bound for the grand canonical free energy of the {B}ose gas in
  the {G}ross-{P}itaevskii limit.
\newblock {\em SIAM Journal on Mathematical Analysis}, 56(2), 2024.

\bibitem{BocCenSch-17}
Chiara Boccato, Serena Cenatiempo, and Benjamin Schlein.
\newblock Quantum many-body fluctuations around nonlinear Schrödinger dynamics.
\newblock {\em  Annales Henri Poincaré}, (18):113–191, 2017.


\bibitem{Bogoliubov-47}
N.~N. {B}ogoliubov.
\newblock On the theory of superfluidity.
\newblock {\em J. Phys. (USSR)}, 11:23, 1947.

\bibitem{bossmann-23}
Lea Bo{\ss}mann, Nikolai Leopold, S{\"o}ren Petrat, and Simone Rademacher.
\newblock Ground state of {B}ose gases interacting through singular potentials.
\newblock {\em arXiv preprint arXiv:2309.12233}, 2023.

\bibitem{bossmann-22}
Lea Bo{\ss}mann, S{\"o}ren Petrat, Peter Pickl, and Avy Soffer.
\newblock Beyond {B}ogoliubov dynamics.
\newblock {\em Pure and Applied Analysis}, 3(4):677--726, 2022.

\bibitem{BosPetSei-21}
Lea Bo{\ss}mann, S{\"o}ren Petrat and Robert Seiringer,
\newblock Asymptotic expansion of low-energy excitations for weakly interacting bosons
\newblock {\em Forum Math. Sigma}, 9:e28, 2021

\bibitem{BreLeeNam-24}
Christian Brennecke, Jinyeop Lee, and Phan~Th{\`a}nh Nam. 
\newblock Second Order Expansion of Gibbs State Reduced Density Matrices in the Gross–Pitaevskii Regime.
\newblock {\em SIAM Journal on Mathematical Analysis} 56.4 : 5262--5284, 2024

\bibitem{BreSch-19}
Christian Brennecke and Benjamin Schlein.
\newblock Gross-{P}itaevskii dynamics for {B}ose-{E}instein condensates.
\newblock {\em Anal. PDE}, 12(6):1513--1596, 2019.

\bibitem{BreSchSch-22}
Christian Brennecke, Benjamin Schlein, and Severin Schraven.
\newblock {B}ogoliubov theory for trapped bosons in the {G}ross-{P}itaevskii
  regime.
\newblock {\em Ann. H. Poincar\'{e}}, 23(5):1583--1658, 2022.

\bibitem{brooks-23}
Morris Brooks.
\newblock Diagonalizing {B}ose gases in the {G}ross-{P}itaevskii regime and beyond.
\newblock {\em arXiv preprint arXiv:2310.11347}, 2023.

\bibitem{CarCenSch-23}
Cristina Caraci, Serena Cenatiempo, and Benjamin Schlein. 
\newblock{The excitation spectrum of two-dimensional Bose gases in the Gross–Pitaevskii regime.} 
\newblock {\em Ann. H. Poincar\'{e}}, 24(8):2877--2928, 2023.

\bibitem{CapDeu-23}
Marco Caporaletti and Andreas Deuchert
\newblock Upper bound for the grand canonical free energy of the Bose gas in the Gross-Pitaevskii limit for general interaction potentials
\newblock {\em arXiv preprint arXiv:2310.12314}, 2023.

\bibitem{CarOlgAubSch-23}
Cristina Caraci, Alessandro Olgiati, Diane~Saint Aubin, and Benjamin Schlein.
\newblock Third order corrections to the ground state energy of a {B}ose gas in
  the {G}ross-{P}itaevskii regime.
\newblock {\em arXiv preprint arXiv:2311.07433}, 2023.

\bibitem{CarOldSch-24}
Cristina Caraci, Jakob Oldenburg, and Benjamin Schlein.
\newblock Quantum fluctuations of many-body dynamics around the {G}ross—{P}itaevskii equation.
\newblock {\em Ann. Inst. Henri Poincaré C}, 2024.

\bibitem{christensen2015}
R.S. Christensen, J.~Levinsen, and G.M. Bruun.
\newblock Quasiparticle properties of a mobile impurity in a {B}ose-{E}instein
  condensate.
\newblock {\em Phys. Rev. Lett.}, 115(16):160401, 2015.

\bibitem{DeFrPiPi-16}
Dirk-Andr{\'e} Deckert, J{\"u}rg Fr{\"o}hlich, Peter Pickl, and Alessandro
  Pizzo.
\newblock Dynamics of sound waves in an interacting {B}ose gas.
\newblock {\em Adv. Math.}, 293:275--323, 2016.

\bibitem{DerNap-14}
Jan Derezi\'{n}ski and Marcin Napi\'{o}rkowski.
\newblock Excitation spectrum of interacting bosons in the mean-field
  infinite-volume limit.
\newblock {\em Ann. H. Poincar\'{e}}, 15(12):2409--2439, 2014.

\bibitem{Eglietal.2013}
Daniel Egli, J{\"u}rg Fr{\"o}hlich, Zhou Gang, Arick Shao, and Israel~Michael
  Sigal.
\newblock {H}amiltonian dynamics of a particle interacting with a wave field.
\newblock {\em Commun. Part. Diff. Eq.}, 38(12):2155--2198, 2013.

\bibitem{ErdSchYau-08}
Erdős, László, Benjamin Schlein, and Horng-Tzer Yau.
\newblock{Ground-state energy of a low-density {B}ose gas: {A} second-order upper bound.}
\newblock{ \em Phys. Rev. A}, 78(5):053627, 2008

\bibitem{fournais-22}
S{\o}ren Fournais, Théotime Girardot, Lukas Junge, Leo Morin, and Marco Olivieri.
\newblock The ground state energy of a two-dimensional {B}ose gas.
\newblock {\em arXiv preprint arXiv:2206.11100}, 2022.

\bibitem{fournais-24}
S{\o}ren Fournais, Théotime Girardot, Lukas Junge, Leo Morin, Marco Olivieri and Arnaud Triay.
\newblock The free energy of dilute Bose gases at low temperatures interacting via strong potentials
\newblock{\em arXiv preprint arXiv:2408.14222}, 2024.

\bibitem{FouSol-20}
S{\o}ren Fournais and Jan~Philip Solovej.
\newblock The energy of dilute {B}ose gases.
\newblock {\em Ann. Math. (2)}, 192(3):893--976, 2020.

\bibitem{FouSol-22}
S{\o}ren Fournais and Jan~Philip Solovej.
\newblock The energy of dilute {B}ose gases {II}: The general case.
\newblock {\em Invent. Math.}, pages 1--132, 2022.

\bibitem{FrohlichGang.2014b}
J{\"u}rg Fr{\"o}hlich and Zhou Gang.
\newblock Ballistic motion of a tracer particle coupled to a {B}ose gas.
\newblock {\em Adv. Math.}, 259:252--268, 2014.

\bibitem{GreSei-13}
Philip Grech, Robert Seiringer. 
\newblock The excitation spectrum for weakly interacting bosons in a trap.
\newblock {\em Comm. Math. Phys.} vol. 322, p. 559--591, 2013.

\bibitem{GrWu-18}
M.~Griesemer and A.~W{\"u}nsch.
\newblock On the domain of the {N}elson {H}amiltonian.
\newblock {\em J. Math. Phys.}, 59(4):042111, 2018.

\bibitem{GrilMacMar-10}
Manoussos Grillakis, Matei Machedon, and Dionisios Margetis.
\newblock Second-order corrections to mean field evolution of weakly
  interacting bosons. {I}.
\newblock {\em Comm. Math. Phys.} 294.1 (2010): 273–301.

\bibitem{GriMacMar-11}
Manoussos Grillakis, Matei Machedon, and Dionisios Margetis.
\newblock Second-order corrections to mean field evolution of weakly
  interacting bosons. {II}.
\newblock {\em Adv. Math.}, 228(3):1788--1815, 2011.

\bibitem{grusdt2016}
F.~Grusdt and E.~Demler.
\newblock New theoretical approaches to {B}ose polarons.
\newblock In M.~Inguscio, W.~Ketterle, S.~Stringari, and G.~Roati, editors,
  {\em Proceedings of the international school of physics "Enrico Fermi"},
  pages 325--411. Societ{\`a} Italiana di Fisica, 2016.

\bibitem{grusdt2015}
F.~Grusdt, Y.E. Shchadilova, A.N. Rubtsov, and E.~Demler.
\newblock Renormalization group approach to the {F}r{\"o}hlich polaron model:
  application to impurity-{BEC} problem.
\newblock {\em Sci. Rep.}, 5:12124, 2015.

\bibitem{HabHaiNamSeiTri-23}
Florian Haberberger, Christian Hainzl, Phan~Th{\`a}nh Nam, Robert Seiringer,
  and Arnaud Triay.
\newblock The free energy of dilute {B}ose gases at low temperatures.
\newblock {\em arXiv preprint arXiv:2304.02405}, 2023.

\bibitem{HabHaiSchTri-24}
Florian Haberberger, Christian Hainzl, Benjamin Schlein,
  and Arnaud Triay.
\newblock Upper Bound for the Free Energy of Dilute {B}ose Gases at Low Temperature
\newblock {\em arXiv preprint arXiv:2405.03378}, 2024.


\bibitem{Hai-21}
Christian Hainzl.
\newblock Another proof of {BEC} in the {GP}-limit.
\newblock {\em J. Math. Phys.}, 62(5):051901, 2021.

\bibitem{HaiSchTri-22}
Christian Hainzl, Benjamin Schlein, and Arnaud Triay.
\newblock {B}ogoliubov theory in the {G}ross-{P}itaevskii limit: a simplified
  approach.
\newblock {\em Forum Math. Sigma}, 10:Paper No. e90, 39, 2022.

\bibitem{HinLam2023}
Benjamin Hinrichs and Jonas Lampart.
\newblock A lower bound on the critical momentum of an impurity in a
  {B}ose-{E}instein condensate.
\newblock {\em C.R. Math.}, 362: 1399--1411, 2024.

\bibitem{HuPi-59}
NM~Hugenholtz and David Pines.
\newblock Ground-state energy and excitation spectrum of a system of
  interacting bosons.
\newblock {\em Phys. Rev.}, 116(3):489, 1959.

\bibitem{LaSch-19}
J.~Lampart and J.~Schmidt.
\newblock On {N}elson-type {H}amiltonians and abstract boundary conditions.
\newblock {\em Commun. Math. Phys.}, 367(2):629--663, 2019.

\bibitem{La-19}
Jonas Lampart.
\newblock A nonrelativistic quantum field theory with point interactions in
  three dimensions.
\newblock {\em Ann. H. Poincar{\'e}}, 20:3509--3541, 2019.

\bibitem{La-20}
Jonas Lampart.
\newblock The renormalized {B}ogoliubov--{F}r{\"o}hlich {H}amiltonian.
\newblock {\em J. Math. Phys.}, 61(10):101902, 2020.

\bibitem{LaPi-22}
Jonas Lampart and Peter Pickl.
\newblock Dynamics of a tracer particle interacting with excitations of a
  {B}ose--{E}instein condensate.
\newblock {\em Ann. H. Poincar{\'e}}, 23(8):2855--2876, 2022.

\bibitem{LaTr-24}
Jonas Lampart and  Arnaud Triay. Validity of the Fr\" ohlich model for a mobile impurity in a Bose-Einstein condensate. {\em arXiv preprint arXiv:2411.11655}, 2024.

\bibitem{Landau-41}
L.~Landau.
\newblock {Theory of the Superfluidity of Helium II}.
\newblock {\em Phys. Rev.}, 60:356--358, Aug 1941.

\bibitem{LeeHuangYang-57}
Tsin~D Lee, Kerson Huang, and Chen~N Yang.
\newblock {Eigenvalues and eigenfunctions of a {B}ose system of hard spheres
  and its low-temperature properties}.
\newblock {\em Physical Review}, 106(6):1135, 1957.

\bibitem{Leger.2020}
Tristan L{\'e}ger.
\newblock Scattering for a particle interacting with a {B}ose gas.
\newblock {\em Commun. Part. Diff. Eq.}, 45(10):1381--1413, 2020.

\bibitem{LewNamSch-15}
Mathieu Lewin, Phan~Thanh Nam, and Benjamin Schlein
\newblock Fluctuations around {H}artree states in the mean-field regime.
\newblock {\em Amer. J. Math.}, 137(6), 1613–1650, 2015.

\bibitem{LewNamSerSol-15}
Mathieu Lewin, Phan~Thanh Nam, Sylvia Serfaty, and Jan~Philip Solovej.
\newblock {B}ogoliubov spectrum of interacting {B}ose gases.
\newblock {\em Comm. Pure Appl. Math.}, 68(3):413--471, 2015.

\bibitem{LieSei-02}
Elliott~H. Lieb and Robert Seiringer.
\newblock {Proof of {B}ose-{E}instein Condensation for Dilute Trapped Gases}.
\newblock {\em Phys. Rev. Lett.}, 88(17):170409, April 2002.

\bibitem{LieSei-06}
Elliott~H. Lieb and Robert Seiringer.
\newblock {Derivation of the Gross-Pitaevskii equation for rotating Bose gases. }.
\newblock {\em Commun. Math. Phys.}, vol. 264, p. 505--537, 2006

\bibitem{LieSeiYng-00}
Elliott~H. Lieb, Robert Seiringer and Jakob Yngvason.
\newblock {Bosons in a trap: A rigorous derivation of the Gross-Pitaevskii energy functional}.
\newblock {\em Phys. Rev. Lett. A}, 61, 043602 (2000).

\bibitem{LieYng-98}
Elliott~H. Lieb and Jakob Yngvason.
\newblock Ground state energy of the low density {B}ose gas.
\newblock {\em Phys. Rev. Lett.}, 80(12):2504--2507, Mar 1998.

\bibitem{MySe-20}
Krzysztof My{\'s}liwy and Robert Seiringer.
\newblock Microscopic derivation of the {F}r{\"o}hlich {H}amiltonian for the
  {B}ose polaron in the mean-field limit.
\newblock {\em Ann. H. Poincar{\'e}}, 21:4003--4025, 2020.

\bibitem{NamNap-19}
Phan~Th\`anh Nam and Marcin Napi\'{o}rkowski.
\newblock Norm approximation for many-body quantum dynamics: focusing case in low dimensions.
\newblock In {\em Adv. Math.}, (350), 547-587, 2019.

\bibitem{NamRicTri-22b}
Phan~Th\`anh Nam, Julien Ricaud, and Arnaud Triay.
\newblock Dilute {B}ose gas with three-body interaction: recent results and
  open questions.
\newblock {\em J. Math. Phys.}, 63(6):Paper No. 061103, 13, 2022.

\bibitem{NamRicTri-23}
Phan~Th\`anh Nam, Julien Ricaud, and Arnaud Triay.
\newblock The condensation of a trapped dilute {B}ose gas with three-body
  interactions.
\newblock {\em Probab. Math. Phys.}, 4(1):91--149, 2023.

\bibitem{NamRouSei-16}
Phan~Th\`anh Nam, Nicolas Rougerie and Robert Seiringer.
\newblock Ground states of large bosonic systems: the Gross–Pitaevskii limit revisited. 
\newblock{ \em Analysis \& PDE}, 9(2), 459-485, 2016

\bibitem{NamTri-23}
Phan~Th\`anh Nam and Arnaud Triay.
\newblock {B}ogoliubov excitation spectrum of trapped {B}ose gases in the
  {G}ross-{P}itaevskii regime.
\newblock {\em J. Math. Pures Appl. (9)}, 176:18--101, 2023.

\bibitem{Nelson-64}
Edward Nelson.
\newblock Interaction of nonrelativistic particles with a quantized scalar
  field.
\newblock {\em J. Math. Phys.}, 5:1190--1197, 1964.

\bibitem{ReeSim1}
Michael Reed and Barry Simon.
\newblock {\em Methods of {M}odern {M}athematical {P}hysics. {I}. Functional
  analysis}.
\newblock Academic Press, 1972.

\bibitem{ReeSim2}
Michael Reed and Barry Simon.
\newblock {\em Methods of {M}odern {M}athematical {P}hysics. {II}. {F}ourier
  analysis, self-adjointness}.
\newblock Academic Press, New York, 1975.

\bibitem{RodSch-09}
Rodnianski, Igor, and Benjamin Schlein. 
\newblock{\em Quantum fluctuations and rate of convergence towards mean field dynamics}.
\newblock Commun. Math. Phys. 291, 31 (2009).

\bibitem{sawada-59}
Katuro Sawada.
\newblock Ground-state energy of {B}ose-{E}instein gas with repulsive
  interaction.
\newblock {\em Phys. Rev.}, 116(6):1344, 1959.

\bibitem{Sei-11}
Robert Seiringer.
\newblock The excitation spectrum for weakly interacting bosons.
\newblock {\em Commun. Math. Phys.}, 306(2):565--578, 2011.

\bibitem{tempere2009}
Tempere, J., Casteels, W., Oberthaler, M. K., Knoop, S., Timmermans, E., \& Devreese, J. T. (2009). 
\newblock Feynman path-integral treatment of the BEC-impurity polaron. 
\newblock {\em Phys. Rev. Lett. B}, 80(18):184504, 2009

\bibitem{wu1959}
Tai~Tsun Wu.
\newblock Ground state of a {B}ose system of hard spheres.
\newblock {\em Phys. Rev.}, 115(6):1390, 1959.

\bibitem{YauYin-09}
Horng-Tzer Yau and Jun Yin.
\newblock The second order upper bound for the ground energy of a {B}ose gas.
\newblock {\em J. Stat. Phys.}, 136(3):453--503, 2009.

\end{thebibliography}

\end{document}